\documentclass[12pt]{book}

\usepackage{geometry}
\usepackage{amsfonts, bm, amsmath, graphicx, amssymb, amsthm, mathdots}
\usepackage[T1]{fontenc}
\usepackage[utf8]{inputenc}
\usepackage[english,spanish,french]{babel}
\usepackage[osf,noBBpl]{mathpazo}
\usepackage{euler}
\usepackage{enumerate}
\usepackage{float}
\usepackage[table]{xcolor}
\usepackage{hyperref}
\usepackage{comment}
\usepackage{subcaption}
\usepackage[section]{placeins}
\usepackage{graphicx}
\usepackage{indentfirst}
\usepackage{thmtools, thm-restate}

\makeatletter\@addtoreset{chapter}{part}\makeatother%
\usepackage{fancyhdr}
\theoremstyle{plain}
\newtheorem{teo}{Theorem}[chapter]
\newtheorem*{teo*}{Theorem}
\newtheorem{conj}{\textit{Conjecture}}[chapter]
\newtheorem{prop}[teo]{Proposition}
\newtheorem{claim}[teo]{Claim}
\newtheorem{cor}[teo]{Corollary}
\newtheorem{lema}[teo]{Lemma}
\newtheorem{defn}[teo]{Definition}

\newcounter{cases}
\newcounter{subcases}[cases]
\newcounter{subsubcases}[subcases]
\newenvironment{mycases}
  {%
    \setcounter{cases}{0}%
    \setcounter{subcases}{0}%
    \setcounter{subsubcases}{0}%
    \def\case
      {%
        \par\noindent
        \refstepcounter{cases}%
        \textit{\textbf{Case (\thecases})}
      }%
    \def\subcase
      {%
        \par\noindent
        \refstepcounter{subcases}%
        \textbf{\textit{Case (\thecases.\thesubcases) }}
      }%
     \def\subsubcase
      {%
        \par\noindent
        \refstepcounter{subsubcases}%
        \textbf{\textit{Case (\thecases.\thesubcases.\thesubsubcases) }}
      }%
  }
  {%
    \par
  }
\renewcommand*\thecases{\arabic{cases}}
\renewcommand*\thesubcases{\arabic{subcases}}
\renewcommand*\thesubsubcases{\arabic{subsubcases}}

\theoremstyle{remark}

\newtheorem{remark}[teo]{\textbf{\textit{Remark}}}

\newcommand*{\QED}{\hfill\ensuremath{\square}}%
\newcommand*{\tagg}{\mathrm{tag}}%

\newenvironment{abstract}%
{\clearpage\null \vfill\begin{center}%
\bfseries \abstractname \end{center}}%
{\vfill\null}

\author{Lic. Nina Pardal}

\definecolor{dark-blue}{RGB}{60,80,170}
\definecolor{dark-red}{RGB}{240,50,60}
\definecolor{dark-orange}{RGB}{255,147,23}
\definecolor{dark-green}{RGB}{65, 173, 57}

\begin{document}

\thispagestyle{empty}

\begin {center}

\medskip
\textit{\textbf{UNIVERSITÉ SORBONNE PARIS NORD} } \\

\textit{ Laboratoire d'Informatique de Paris-Nord}

\vspace{2cm}

\underline{\textit{T H E S E}} \\

\vspace{.7cm}

pour obtenir le grade de \\

\vspace{.5cm}

\textbf{DOCTEUR DE L'UNIVERSITE SORBONNE PARIS NORD }\\

\vspace{1cm}

Discipline: Informatique \\

\vspace{.5cm}
présentée et soutenue publiquement par \\

\vspace{.5cm}
\textbf{Nina PARDAL} \\

\vspace{1cm}

\textsc{\textbf{\Large Caractérisation structurelle de quelques problèmes dans les graphes de cordes et d'intervalles}} 

\vspace{1.5cm}

\underline{\textbf{MEMBRES DU JURY:}}
\vspace{.5cm}

\end{center}

\small{
\hskip-1cm \begin{tabular}{lll}
Min Chih LIN & Universidad de Buenos Aires & Examinateur \\
Valeria LEONI & Universidad de Rosario & Rapporteur \\
Christophe PICOULEAU & Conservatoire National des Arts et Métiers, Paris & Rapporteur \\ 
Frédérique BASSINO & Université Sorbonne Paris Nord & Examinatrice \\
Daniel PERRUCCI & Universidad de Buenos Aires & Examinateur \\
\\
\\
\\
Mario VALENCIA-PABON & Université Sorbonne Paris Nord & Directeur \\ 
Guillermo DURÁN & Universidad de Buenos Aires & Co-directeur
\end{tabular}
}

\newpage

\thispagestyle{empty}

\begin {center}

\includegraphics[scale=.3]{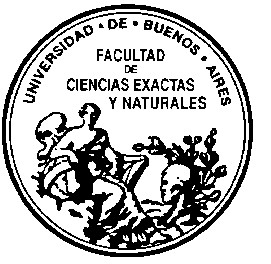}

\medskip
UNIVERSIDAD DE BUENOS AIRES

Facultad de Ciencias Exactas y Naturales

Departamento de Matem\'atica

\vspace{3cm}

\textbf{\large Caracterización estructural de algunos problemas en grafos circle y de intervalos}


\vspace{2cm}

Tesis presentada para optar al t\'\i tulo de Doctor de la Universidad de Buenos Aires en el \'area Ciencias Matem\'aticas 

\vspace{2cm}

\textbf{Nina Pardal}

\end {center}

\vspace{1.5cm}

\noindent Director de tesis: Guillermo A. Dur\'an 

\noindent Director Asistente: Mario Valencia-Pabon 

\noindent Consejero de estudios: Mariano Su\'arez \'Alvarez

\vspace{1cm}

\noindent Lugar de trabajo: Instituto de C\'alculo, FCEyN 


\vspace{1cm}

\noindent Buenos Aires, 2020 

\newpage

	


\selectlanguage{spanish}%
\begin{abstract}

\begin{center}
\textbf{Caracterización estructural de algunos problemas en grafos círculo y de intervalos}
\end{center}

Dada una familia de conjuntos no vacíos $\mathcal{S} = \{S_i\}$, se define el \emph{grafo de intersección} de la familia $\mathcal{S}$ como el grafo obtenido al representar con un vértice a cada conjunto $S_i$ de forma tal que dos vértices son adyacentes sí y sólo si los conjuntos correspondientes tienen intersección no vacía. Un grafo se dice \emph{círculo} si existe una familia de cuerdas $L= \{ C_v \}_{v\in G}$ en un círculo de modo que dos vértices son adyacentes si las cuerdas correspondientes se intersecan. Es decir, es el grafo de intersección de la familia de cuerdas $L$.
Existen diversas caracterizaciones de los mismos mediante operaciones como complementación local o descomposición split. Sin embargo, no se conocen aún caracterizaciones estructurales de los grafos círculo por subgrafos inducidos minimales prohibidos. En esta tesis, damos una caracterización de los grafos círculo por subgrafos inducidos prohibidos, restringido a que el grafo original sea split. 
Una matriz de $0$'s y $1$'s tiene la \emph{propiedad de los unos consecutivos (C$1$P) para sus filas} si existe una permutación de sus columnas de forma tal que para cada fila, todos sus $1$'s se ubiquen de manera consecutiva.
En esta tesis desarrollamos caracterizaciones por submatrices prohibidas de matrices de $0$'s y $1$'s con la C$1$P para sus filas que además son $2$-coloreables bajo una cierta relación de adyacencia, y caracterizamos estructuralmente algunas subclases de grafos círculo auxiliares que se desprenden de estas matrices.

Dada una clase de grafos $\Pi$, una \emph{$\Pi$-completación }de un grafo $G = (V,E)$ es un grafo $H = (V, E \cup F)$ tal que $H$ pertenezca a $\Pi$. 
Una $\Pi$-completación $H$ de $G$ es minimal si $H'= (V, E \cup F')$ no pertenece a $\Pi$ para todo $F'$ subconjunto propio de $F$. 
Una $\Pi$-completación $H$ de $G$ es mínima si para toda $\Pi$-completación $H' = (V, E \cup F')$  de $G$, se tiene que el tamaño de $F$ es inferior o igual al tamaño de $F'$. 
En esta tesis estudiamos el problema de completar de forma minimal a un grafo de intervalos propios, cuando el grafo de input es de intervalos. Hallamos condiciones necesarias que caracterizan una completación minimal en este caso, y dejamos algunas conjeturas para considerar en el futuro. 

\vspace{2mm}
\textbf{Palabras clave:} grafos, círculo, subgrafos prohibidos, completación, minimal.

\end{abstract}


\selectlanguage{english}%
\begin{abstract}

\begin{center}
\textbf{Structural characterization of some problems on circle and interval graphs}
\end{center}

Given a family of nonempty sets $\mathcal{S} = \{S_i\}$, the \emph{intersection graph of the family $\mathcal{S}$} is the graph with one vertex for each set $S_i$, such that two vertices are adjacent if and only if the corresponding sets have nonempty intersection. A graph is \emph{circle} if there is a family of chords in a circle $L= \{ C_v \}_{v\in G}$ such that two vertices are adjacent if the corresponding chords cross each other. In other words, it is the intersection graph of the chord family $L$.
There are diverse characterizations of circle graphs, many of them using the notions of local complementation or split decomposition. However, there are no known structural characterization by minimal forbidden induced subgraphs for circle graphs. In this thesis, we give a characterization by forbidden induced subgraphs of those split graphs that are also circle graphs.
A $(0,1)$-matrix has the \emph{consecutive-ones property (C1P) for the rows }if there is a permutation of its columns such that the $1$'s in each row appear consecutively. In this thesis, we develop characterizations by forbidden subconfigurations of $(0,1)$-matrices with the C1P for which the rows are $2$-colorable under a certain adjacency relationship, and we characterize structurally some auxiliary circle graph subclasses that arise from these special matrices. 

Given a graph class $\Pi$, a \emph{$\Pi$-completion} of a graph $G = (V,E)$ is a graph $H = (V, E \cup F)$ such that $H$ belongs to $\Pi$. 
A $\Pi$-completion $H$ of $G$ is minimal if $H'= (V, E \cup F')$ does not belong to $\Pi$ for every proper subset $F'$ of $F$. 
A $\Pi$-completion $H$ of $G$ is minimum if for every $\Pi$-completion $H' = (V, E \cup F')$ of $G$, the cardinal of $F$ is less than or equal to the cardinal of $F'$.
In this thesis, we study the problem of completing minimally to obtain a proper interval graph when the input is an interval graph. We find necessary conditions that characterize a minimal completion in this particular case, and we leave some conjectures for the future.

\vspace{2mm}
\end{abstract}


\selectlanguage{french}%
\begin{abstract}

\begin{center}
\textbf{Caractérisation structurelle de quelques problèmes dans les graphes de cordes et d’intervalles}
\end{center}

Étant donnée une famille d'ensembles non vides $\mathcal{S} = \{S_i\}$, le \emph{graphe d’intersection de la famille $\mathcal{S}$} est celui pour lequel chaque sommet représent un ensemble $S_i$ de tel façon que deux sommets sont adjacents si et seulement si leurs ensembles correspondants ont une intersection non vide. Un graphe est dit \emph{graphe de cordes} s'il existe une famille de cordes $L= \{ C_v \}_{v\in G}$ dans un cercle tel que deux sommets sont adjacents si les cordes correspondantes se croisent. Autrement dit c'est le graphe d'intersection de la famille de cordes $L$. Ils existent différentes caractérisations des graphes de cordes qui utilisent certaines opérations dont notamment la complémentation locale ou encore la dé\-com\-po\-si\-tion split. Cependant on ne connaît pas encore aucune caractérisation structurelle des graphes de cordes par sous-graphes induits interdits minimales. Dans cette thèse nous donnons une caractérisation des graphes de cordes par sous-graphes induits interdits dans le cas où le graphe original est un graphe split. Une matrice binaire possède la \emph{propriété des unités consécutives en lignes (C1P)} s'il existe une permutation de ses colonnes de sorte que les $1$’s dans chaque ligne apparaissent consécutivement. Dans cette thèse nous développons des caractérisations par sous-matrices interdites de matrices binaires avec C1P pour lesquelles les lignes sont $2$-coloriables sous une certaine condition d’ad\-ja\-cence et nous caractérisons structurelle\-ment quelques sous-classes auxiliaires de graphes de cordes qui découlent de ces matrices. 

Étant donnée une classe de graphes $\Pi$, une \emph{$\Pi$-complétion} d’un graphe $G = (V,E)$ est un graphe $H = (V,E \cup F)$ tel que $H$ appartient à $\Pi$. Une $\Pi$-complétion $H$ de $G$ est minimale si $H' = (V,E \cup F')$ n’appartient pas à $\Pi$ pour tout $F'$ sous-ensemble propre de $F$. Une $\Pi$-complétion $H$ de $G$ est minimum si pour toute $\Pi$-complétion $H' = (V, E \cup F')$ de $G$ la cardinalité de $F$ est plus petite ou égale à la cardinalité de $F'$. Dans cette thèse nous étudions le problème de complétion minimale d'un graphe d'intervalles propre quand le graphe d'entrée est un graphe d’intervalles quelconque. Nous trouvons des conditions nécessaires qui caractérisent une complétion minimale dans ce cas particulier, puis nous laissons quelques conjectures à considérer dans un futur. 

\vspace{2mm}
\textbf{Mots clés:} graphes, cordes, intervalles, sous-graphes interdits, complétion minimale. 
\end{abstract}



\selectlanguage{spanish}%

\chapter*{Agradecimientos}

Nunca fui buena con estas cosas, probablemente olvide a alguien o descuide las formas. Sepan que en mi mente el orden es casi arbitrario y carece de importancia real.

A mi amor y compañero Guido. A mi madre y a mi padre. A Pablo, el hermano que mis padres no me dieron. 
A mis directores Willy y Mario, por la confianza, la libertad y el aguante. A Martín Safe, un enorme maestro en tiempos de gran confusión. A Flavia, por su generosidad y calidez.
A la música, mi válvula de escape y generadora infinita de felicidad. A mis hermanos de la música, Nacho, Emma y Sami, y a todas las hermosas personas que conocí gracias a la música en estos años. A los maestros, Flor, Pablo, Edgardo.
A los afectos nuevos, a quienes han retornado a mi vida y a aquellos valientes que nunca se han ido: Emi, Pri, Germán, Alessandra, Melina, Leti, Lucas, Juan. A mis compañeros de oficina y doctorado por la interminable catarsis. A mis gatitas, por la tierna compañía al escribir y pensar. A la tía Cora, si bien postiza, anque la más tía. A mis jurados, por la paciencia para leer esta tesis kilométrica. A Deby y Pablo, por la buena predisposición y ayuda para domar a esta cotutela.

A todos ustedes, simplemente ''gracias'' por lo que cada uno ha hecho por mí en estos largos cinco años. Por la infinita tolerancia y paciencia, por los cafés y cervezas, por las charlas, consejos y enseñanzas. Por evitar las crisis consumadas en mi mente, por los abrazos y la contención. Por ayudarme a crecer cada día un poquito --y una vez más, por qué no-- por creer en mí cuando yo no tenía la capacidad de hacerlo.


\selectlanguage{english}%
\tableofcontents


\selectlanguage{spanish}%
\chapter*{Introducción general}

La teoría estructural de grafos estudia caracterizaciones y descomposiciones de clases de grafos particulares, y utiliza estos resultados para demostrar propiedades teóricas para esas clases de grafos, así como para obtener distintos resultados sobre la complejidad algorítmica de ciertos problemas.
Son tópicos comunes en este área el estudio de los menores y threewidth, descomposición modular y clique-width, caracterización de familias de grafos por configuraciones prohibidas, entre otros.

Esta tesis consiste de dos partes, en las cuales nos concentramos en el estudio de dos tópicos distintos de la teoría estructural de grafos: caracterización por subgrafos inducidos prohibidos y caracterización de completaciones mínimas y minimales.

\vspace{3mm}
\textbf{Parte I: Caracterización por subgrafos inducidos prohibidos}
\vspace{3mm}

Dada una familia no vacía de conjuntos $\mathcal{S} = \{S_i\}$, el \emph{grafo de intersección de la familia $\mathcal{S}$} es el grafo que tiene un vértice por cada conjunto $S_i$, de moto tal que dos vértices son adyacentes si y sólo si los conjuntos correspondientes tienen intersección no vacía. Un grafo es \emph{círculo (o circle)} si existe una familia de cuerdas en un círculo $L= \{ C_v \}_{v\in G}$ tal que dos vértices resultan adyacentes si las cuerdas correspondientes se cruzan dentro del círculo. En otras palabras, es el grafo de intersección de la familia de cuerdas $L$.
Existen diversas caracterizaciones de los grafos círculo, muchas de ellas utilizan las nociones de complementación local o descomposición split. A pesar de contar con muchas caracterizaciones diversas, no existe una caracterización completa de los grafos círculo por subgrafos inducidos minimales prohibidos.
La investigación actual en esta dirección se concentra en hallar caracterizaciones parciales de esta clase de grafos. Es decir, se conocen algunas caracterizaciones para los grafos círculo por subgrafos in\-du\-ci\-dos minimales prohibidos cuando el grafo considerado en principio además pertenece a otra subclase particular, como por ejemplo, grafos $P_4$-tidy, grafos linear-domino, grafos sin diamantes, para dar algunos ejemplos. 
En esta tesis, damos una caracterización por subgrafos inducidos prohibidos para aquellos grafos split que además son circle. 
La motivación para estudiar esta clase en particular viene de los grafos cordales, los cuales son aquellos grafos que no contienen ciclos inducidos de longitud mayor a $3$. Los grafos cordales han sido ampliamente estudiados y son una clase de grafos muy interesante, que además son un subconjunto de los grafos perfectos. Pueden ser reconocidos en tiempo polinomial, y muchos problemas que son difíciles de resolver en otras clases de grafos, como por ejemplo problemas de coloreo, pueden ser resueltos en en tiempo polinomial cuando el grafo de input es cordal. Es por ello que surge naturalmente la pregunta de hallar una lista de subgrafos inducidos prohibidos para la clase de grafos círculo cuando el grafo además es cordal. 
A su vez, los grafos split son aquellos grafos cuyo conjunto de vértices se puede particionar en un subconjunto completo y un subconjunto independiente, y además son una subclase de los grafos cordales. Más aún, los grafos split son aquellos grafos cordales cuyo complemento es también es cordal. Es por ello que estudiar cómo caracterizar los grafos círculo por subgrafos inducidos prohibidos cuando el grafo es split es un buen comienzo para obtener resultados en la dirección de la caracterización de los cordales que además son grafos círculo.

Una matriz binaria tiene la \emph{propiedad de los unos consecutivos (C$1$P) para las filas }si existe una permutación de sus columnas de modo tal que los $1$'s en cada fila aparezcan de forma consecutiva. Para caracterizar los grafos split que son círculo, desarrollamos caracterizaciones por subconfiguraciones prohibidas para las matrices binarias que tienen la C$1$P y para las cuales las filas además admiten una asignación de color (de entre dos colores distintos) que depende de una cierta relación de adyacencia. Esto deviene en una caracterización estructural para algunas subclases auxiliares de los grafos círculo que surgen a partir de estas matrices especiales. 

\vspace{6mm}
\textbf{Parte II: El problema de la $\Pi$-completación}
\vspace{3mm}

Dada una propiedad $\Pi$, el $\Pi$-problema de modificación de grafos se define de la siguiente manera. 
Dado un grafo $G$ y una propiedad de grafos $\Pi$, necesitamos borrar (o agregar o editar) un subconjunto de vértices (o aristas) de forma tal que el grafo resultante tenga la propiedad $\Pi$.

Dado que los grafos pueden ser utilizados para representar tanto problemas del mundo real como de estructuras teóricas, no es difícil ver que un problema de mo\-di\-fi\-ca\-ción puede ser utilizado para modelar un gran número de aplicaciones prácticas en distintos campos. En particular, muchos problemas fundamentales de la teoría de grafos pueden expresarse como problemas de mo\-di\-fi\-ca\-ción de grafos. Por ejemplo, el problema de Conectividad es el problema de hallar un número mínimo de vértices o aristas que al ser removidos desconecten al grafo, o el problema de Matching Máximo Inducido puede ser formulado como el problema de remover la menor cantidad de vértices del grafo para obtener una colección disjunta de aristas. 

Un problema de modificación de grafos particular es la $\Pi$-completación. Dada una clase de grafos $\Pi$, una $\Pi$-completación de un grafo $G = (V,E)$ es un grafo $H = (V, E \cup F)$ tal que $H$ pertenece a la clase $\Pi$. 
Una $\Pi$-completación $H$ se dice $G$ minimal si $H'= (V, E \cup F')$ no pertenece a la clase $\Pi$ para todo subconjunto propio $F'$ de $F$. 
Una $\Pi$-completación $H$ de $G$ es mínima si para toda $\Pi$-completación $H' = (V, E \cup F')$ de $G$, el cardinal de $F$ es menor o igual que el cardinal de $F'$.

El problema de calcular una completación mínima en un grafo arbitrario a una clase específica de grafos ha sido ampliamente estudiado. Desafortunadamente, se ha de\-mos\-tra\-do que las completaciones mínimas de grafos arbitrarios a clases específicas de grafos como los cografos, grafos bipartitos, grafos cordales, etc. son NP-hard de computar \cite{NSS01,BBD06,Y81}. 
Por esta razón, la investigación actual en este tópico se concentra en hallar com\-ple\-ta\-cio\-nes minimales de grafos arbitrarios a clases de grafos específicas de la forma más eficiente posible, desde el punto de vista computacional. Más aún, aunque el pro\-ble\-ma de completar minimalmente es y ha sido muy estudiado, se desconocen caracterizaciones estructurales para la mayoría de los problemas para los cuales se ha dado un algoritmo polinomial para hallar tal completación. Estudiar la estructura de las completaciones minimales puede permitir hallar algoritmos de reconocimiento eficientes. 
En particular, las completaciones minimales de un grafo arbitratio a grafos de intervalos y grafos de intervalos propios han sido estudiadas en \cite{CT13,RST06}.
En esta tesis estudiamos el problema de completar de forma minimal un grafo de intervalos para obtener un grafo de intervalos propios. Hallamos condiciones necesarias que caracterizan una completación minimal en este caso particular, y dejamos algunas conjeturas abiertas para trabajo futuro.


\selectlanguage{english}%
\chapter*{A general introduction}
\addcontentsline{toc}{chapter}{A general introduction}

Structural graph theory studies characterizations and decompositions of particular graph classes, and uses these results to prove theoretical properties from such graph classes as well as to derive various algorithmic consequences.
Typical topics in this area are graph minors and treewidth, modular decomposition and clique-width, characterization of graph families by forbidden configurations, among others. 

This thesis consists on two parts, in each of which we focus on the study of two distinct topics in structural graph theory: characterization by forbidden induced subgraphs and characterization of minimal and minimum completions.

\vspace{3mm}
\textbf{Part I: Characterization by forbidden induced subgraphs}
\vspace{3mm}

Given a family of nonempty sets $\mathcal{S} = \{S_i\}$, the \emph{intersection graph of the family $\mathcal{S}$} is the graph with one vertex for each set $S_i$, such that two vertices are adjacent if and only if the corresponding sets have nonempty intersection. A graph is \emph{circle} if there is a family of chords in a circle $L= \{ C_v \}_{v\in G}$ such that two vertices are adjacent if the corresponding chords cross each other. In other words, it is the intersection graph of the chord family $L$.
There are diverse characterizations of circle graphs, many of them using the notions of local complementation or split decomposition. In spite of having many diverse characterizations, there is no known complete characterization of circle graphs by minimal forbidden induced subgraphs.
Current research on this direction focuses on finding partial characterizations of this graph class. In other words, some characterizations by minimal forbidden induced subgraphs for circle graphs are known when the graph we consider in the first place also belongs to another certain subclass, such as $P_4$-tidy graphs, linear-domino graphs, diamond-free graphs, to give some examples.
In this thesis, we give a characterization by forbidden induced subgraphs of those split graphs that are also circle graphs.
The motivation to study this particular graph class comes from chordal graphs, which are those graphs that contain no induced cycle of length greater than $3$. Chordal graphs are a widely studied and interesting graph class, which is also a subset of perfect graphs. They may be recognized in polynomial time, and several problems that are hard on other classes of graphs such as graph coloring may be solved in polynomial time when the input is chordal. This is why the question of finding a list of forbidden induced subgraphs for the class of circle graphs when the graph is also chordal arises naturally. In turn, split graphs are those graphs whose vertex set can be split into a complete set and an independent set, and they are a subclass of chordal graphs. Moreover, split graphs are those chordal graphs whose complement is also a chordal graph. Thus, studying how to characterize circle graphs by forbidden induced subgraphs when the graph is split seemed a good place to start in order to find such a characterization for chordal circle graphs.

A $(0,1)$-matrix has the \emph{consecutive-ones property (C$1$P) for the rows }if there is a permutation of its columns such that the $1$'s in each row appear consecutively. In order to characterize those split graphs that are circle, we develop characterizations by forbidden subconfigurations of $(0,1)$-matrices with the C$1$P for which the rows admit a color assignment of two distinct colors under a certain adjacency relationship. This leads to structurally characterize some auxiliar circle graph subclasses that arise from these special matrices. 

\vspace{6mm}
\textbf{Part II: The $\Pi$-completion problem}
\vspace{3mm}

For a graph property $\Pi$, the $\Pi$-graph modification problem is defined as follows. 
Given a graph $G$ and a graph property $\Pi$, we need to delete (or add or edit) a subset of vertices (or edges) so that the resulting graph has the property $\Pi$.
As graphs can be used to represent diverse real world and theoretical structures, it is not difficult to see that a modification problem can be used to model a large number of practical applications in several different fields. In particular, many fundamental problems in graph theory can be expressed as graph modification problems. For instance, the Connectivity problem is the problem of finding the minimum number of vertices or edges that disconnect the graph when removed from it, or the Maximum Induced Matching problem can be seen as the problem of removing the smallest set of vertices from the graph to obtain a collection of disjoint edges.

A particular graph modification problem is the $\Pi$-completion. Given a graph class $\Pi$, a $\Pi$-completion of a graph $G = (V,E)$ is a graph $H = (V, E \cup F)$ such that $H$ belongs to $\Pi$. 
A $\Pi$-completion $H$ of $G$ is minimal if $H'= (V, E \cup F')$ does not belong to $\Pi$ for every proper subset $F'$ of $F$. 
A $\Pi$-completion $H$ of $G$ is minimum if for every $\Pi$-completion $H' = (V, E \cup F')$ of $G$, the cardinal of $F$ is less than or equal to the cardinal of $F'$.

The problem of calculating a minimum completion in an arbitrary graph to a specific graph class has been rather studied. Unfortunately, minimum completions of arbitrary graphs to specific graph classes, such as cographs, bipartite graphs, chordal graphs, etc., have been showed to be NP-hard to compute \cite{NSS01,BBD06,Y81}. 
For this reason, current research on this topic is focused on finding minimal completions of arbitrary graphs to specific graph classes in the most efficient way possible from the computational point of view. And even though the minimal completion problem is and has been rather studied, structural characterizations are still unknown for most of the problems for which a polynomial algorithm to find such a completion has been given. Studying the structure of minimal completions may allow to find efficent recognition algorithms.
In particular, minimal completions from an arbitrary graph to interval graphs and proper interval graphs have been studied in \cite{CT13,RST06}.
In this thesis, we study the problem of completing minimally to obtain a proper interval graph when the input is an interval graph. We find necessary conditions that characterize a minimal completion in this particular case, and we leave some conjectures for the future.

\part{Split circle graphs}


\selectlanguage{spanish}%
\chapter*{Introducción}

Los grafos círculo \cite{EI71} son los grafos de intersección de cuerdas en un círculo. En otras palabras, un grafo es círculo si existe una familia de cuerdas $L= \{ C_v \}_{v\in G}$ en un círculo de modo tal que dos vértices en $G$ son adyacentes si y sólo si las cuerdas correspondientes se intersecan en el interior del círculo. 
Estos grafos fueron definidos por Even e Itai \cite{EI71} para resolver un problema presentado por Knuth, el cual consiste en resolver un problema de ordenamiento con el mínimo número de pilas paralelas sin la restricción de cargar antes de que se complete la descarga. Se demostró que este problema puede ser traducido al problema de hallar el número cromático de un grafo círculo. 
Por su parte, en 1985, Naji \cite{N85} caracterizó los grafos círculo en términos de resolubilidad de un sistema lineal de ecuaciones, dando un algoritmo de reconocimiento de orden $\mathcal{O}(n^7)$.

El \emph{complemento local }de un grafo $G$ con respecto a un vértice $u \in V(G)$ es el grafo $G*u$ que surge de $G$ al reemplazar el subgrafo inducido $G\left[N(u)\right]$ por su complemento. Dos grafos $G$ y$H$ son \emph{localmente equivalentes} si y sólo si $G$ se obtiene de $H$ a través de la aplicación de una secuencia finita de complementaciones locales. Los grafos círculo fueron caracterizados por Bouchet \cite{B94} en 1994 por subgrafos inducidos prohibidos bajo complementación local. Inspirados por este resultado, Geelen y Oum \cite{GO09} dieron una nueva caracterización de los grafos círculo en términos de \emph{pivoteo}. El resultado de pivotear un grafo $G$ con respecto a una arista $uv$ es el grafo $G \times uv = G * u * v * u$. 

Los grafos círculo son una superclase de los grafos de \emph{permutación}. De hecho, los grafos de permutación pueden ser definidos como aquellos grafos círculo que admiten un modelo para el cual se puede adicionar una cuerda de modo tal que esta cuerda cruce a todas las cuerdas del modelo. Por otro lado, los grafos de permutación son aquellos grafos de comparabilidad cuyo complemento también es un grafo de comparabilidad \cite{EPL72}. Dado que los grafos de comparabilidad han sido caracterizados por subgrafos inducidos prohibidos \cite{G67}, aquella caracterización implica una caracterización por subgrafos inducidos prohibidos para los grafos de permutación.

A pesar de todos estos resultados, no se conoce aún una caracterización completa por subgrafos inducidos prohibidos para la clase de grafos círculo en su totalidad. Sí se conocen algunas caracterizaciones parciales, es decir, existen algunas caractericiones de los grafos círculo por subgrafos inducidos prohibidos cuando el grafo además pertenece a una cierta subclase. Ejemplos de estas subclases son los grafos $P_4$-tidy, grafos Helly circle, grafos linear-domino, entre otros. 

Para extender estos resultados, consideramos el problema de caracterizar por sub\-grafos inducidos prohibidos aquellos grafos circle que además son split.
La motivación para estudiar esta clase en particular viene de los grafos cordales, los cuales son aquellos grafos que no contienen ciclos inducidos de longitud mayor a $3$. Los grafos cordales --que a su vez son una importante subclase de los grafos perfectos --han sido ampliamente estudiados y son una clase de grafos muy interesante. Pueden ser reconocidos en tiempo polinomial, y muchos problemas que son difíciles de resolver en otras clases de grafos, como por ejemplo problemas de coloreo, pueden ser resueltos en en tiempo polinomial cuando el grafo de input es cordal. Otra propiedad interesante de los grafos cordales, es que el treewidth de un grafo arbitrario puede ser caracterizado por el tamaño de las cliques de los grafos cordales que lo contienen. Los grafos bloque son una subclase particular de los grafos cordales, y además son grafos círculo. Sin embargo, no todo grafo cordal es círculo. Todas estas razones nos han llevado a considerar a los grafos cordales como una restricción natural a la hora de estudiar caracterizaciones parciales de los grafos círculo por subgrafos inducidos prohibidos. Algo similar sucede con los grafos split, la cual es una subclase muy interesante de los grafos cordales. 
Los grafos split son aquellos grafos cuyo conjunto de vértices se puede particionar en un subconjunto completo y un subconjunto independiente. Más aún, los grafos split son aquellos grafos cordales cuyo complemento es también es cordal. El Capítulo $2$ damos un ejemplo de un grafo cordal que no es split ni círculo. Es por ello que estudiar cómo caracterizar los grafos círculo por subgrafos inducidos prohibidos cuando el grafo es split es un buen comienzo para obtener resultados en la dirección de la caracterización de los cordales que además son grafos círculo.

Comenzamos considerando un grafo split $H$ minimalmente no círculo. Dado que los grafos de permutación son una subclase de los grafos círculo, en particular $H$ no es un grafo de permutación. Observemos que los grafos de permutación son aquellos grafos de comparabilidad cuyo com\-ple\-men\-to también es un grafo de comparabilidad. Es fácil ver que los grafos de permutación son precisamente aquellos grafos círculo que admiten un modelo con un ecuador. 
Usando la lista de subgrafos inducidos minimales prohibidos para los grafos de comparabilidad (ver Figuras \ref{fig:forb_comparability1} y \ref{fig:forb_comparability2}) y el hecho de que $H$ es un grafo split, concluímos que $H$ debe contener un subgrafo inducido isomorfo al tent, al $4$-tent, al co-$4$-tent o al net (ver Figura \ref{fig:forb_permsplit_base}).
En el Capítulo $2$, dado un grafo split $G=(K,S)$ y un subgrafo inducido $H$ que puede ser un tent, un $4$-tent o un co-$4$-tent, definimos particiones de los conjuntos $K$ y $S$ de acuerdo a la relación de adyacencia con $H$ y demostramos que estas particiones están bien definidas. 

Una matriz binaria tiene la \emph{propiedad de los unos consecutivos (C$1$P) para las filas} si existe una permutación de sus columnas de forma tal que los $1$'s en cada fila aparezcan consecutivamente. 
Para caracterizar aquellos grafos círculo que contienen un tent, un $4$-tent, un co-$4$-tent o un net como subgrafo inducido, primero estudiamos el problema de caracterizar aquellas matrices que admiten un ordenamiento C$1$P para sus filas y para las cuales existe una asignación de colores particular, teniendo exactamente $2$ colores para ello. Esta asignación de colores se define en el Capítulo \ref{chapter:2nested_matrices}, considerando que deben cumplirse ciertas propiedades especial que están basadas estrictamente en la relación de orden parcial dada por los vecindarios de los vértices independientes de un grafo split. Estas propiedades se encuentran plasmadas en la definición de admisibilidad. 

En el Capítulo \ref{chapter:2nested_matrices}, definimos y caracterizamos las matrices $2$-nested por sus sub\-matrices minimales prohibidas. Esta caracterización lleva a una caracterización por sub\-grafos inducidos prohibidos para la clase de grafos asociada, la cual es una subclase tanto de los grafos split como de los grafos círculo. Para llevar a cabo esta caracterización, definimos el concepto de matriz enriquecida, que son aquellas matrices binarias para las cuales algunas filas pueden tener una etiqueta con la letra L (que representa \textit{izquierda}) o R (que representa \textit{derecha}) o LR (que representa \textit{izquierda-derecha}), y algunas de estas filas etiquetadas además pueden estar coloreadas con rojo o azul, cada una. 
En las primeras secciones del Capítulo \ref{chapter:2nested_matrices}, definimos y caracterizamos las nociones de admisibilidad, LR-ordenable y parcialmente $2$-nested. Estas nociones nos permiten definir qué es un '`pre-coloreo válido'' y caracterizar aquellas matrices enriquecidas con pre-coloreos válidos que admiten un LR-ordenamiento, que es la propiedad de admitir un ordenamiento lineal de las columnas $\Pi$ tal que, cuando se ordenan acorde a $\Pi$, las filas no-LR y los complementos de las LR cumplen la C$1$P, las filas etiquetadas con L empiezan en la primera columna y las filas etiquetadas con R terminan en la última columna. 
Esto deviene en una caracterización de las matrices $2$-nested  por submatrices prohibidas, ya que las matrices $2$-nested resultan ser exactamente aquellas matrices parcialmente $2$-nested para las cuales el $2$-coloreo previo de las filas puede ser extendida a un $2$-coloreo total de todas las filas de la matriz, manteniendo ciertas propiedades.
El Capítulo \ref{chapter:2nested_matrices} es crucial para determinar cuáles son los subgrafos inducidos prohibidos para aquellos grafos círculo que también son split. 

En el Capítulo \ref{chapter:split_circle_graphs}, nos referimos al problema de caracterizar por subgrafos inducidos prohibidos de un grafo círculo que contiene un tent, un $4$-tent, un co-$4$-tent o un net como subgrafo inducido. En cada sección vemos un caso del teorema, llegando así a obtener un teorema de caracterización y terminando en cada caso dando las claves para dibujar un modelo circular.


\selectlanguage{english}%
\chapter{Introduction}

Circle graphs \cite{EI71} are intersection graphs of chords in
a circle. In other words, a graph is circle if there is a family of chords $L= \{ C_v \}_{v\in G}$ in a circle such that two vertices in $G$ are adjacent if and only if the corresponding chords cross each other. 
These graphs were defined by Even and Itai \cite{EI71} to solve a problem stated by Knuth, which consists in solving an ordering problem with the minimum number of parallel stacks without the restriction of loading before unloading is completed. It was proven that this problem can be translated into the problem of finding the chromatic number of a circle graph. 
For its part, in 1985, Naji \cite{N85} characterized circle graphs in terms of the solvability of a system of linear equations, yielding a $\mathcal{O}(n^7)$-time recognition algorithm for this class.

The \emph{local complement }of a graph $G$ with respect to a vertex $u \in V(G)$ is the graph $G*u$ that arises from $G$ by replacing the induced subgraph $G\left[N(u)\right]$ by its complement. Two graphs $G$ and $H$ are \emph{locally equivalent} if and only if $G$ arises from $H$ by a finite sequence of local complementations. Circle graphs were characterized by Bouchet \cite{B94} in 1994 by forbidden induced subgraphs under local complementation. Inspired by this result, Geelen and Oum \cite{GO09} gave a new characterization of circle graphs in terms of \emph{pivoting}. The result of pivoting a graph $G$ with respect to an edge $uv$ is the graph $G \times uv = G * u * v * u$. 

Circle graphs are a superclass of \emph{permutation} graphs. Indeed, permutation graphs can be defined as those circle graphs having a circle model such that a chord can be added in such a way that this chord meets every chord belonging to the circle model. On the other hand, permutation graphs are those comparability graphs whose complement graph is also a comparability graph \cite{EPL72}. Since comparability graphs have been characterized by forbidden induced subgraphs \cite{G67}, such a characterization implies a forbidden induced subgraphs characterization for the class of permutation graphs.

In spite of all these results, there are no known characterizations for the entire class of circle graphs by forbidden induced subgraphs. Some partial characterizations of circle graphs have been given. In other words, there are some characterizations of circle graphs by forbidden minimal induced subgraphs when these graphs also belong to a certain subclass, such as $P_4$-tidy graphs, Helly circle graphs, linear-domino graphs, among others. In Chapter 4 we give a brief introduction to these known results. 

In order to extend these results, we considered the problem of characterizing by minimal forbidden induced subgraphs those circle graphs that are also split graphs. 
The motivation to study circle graphs restricted to this particular graph class came from chordal graphs, which are defined as those graphs that contain no induced cycles of length greater than $3$. 
Chordal graphs --which is a subset of perfect graphs-- is a very widely studied graph class, for which there are several interesting characterizations. They may be recognized in polynomial time, and several problems that are hard on other classes of graphs --such as graph coloring-- may be solved in polynomial time when the input is chordal. Another interesting property of chordal graphs, is that the treewidth of an arbitrary graph may be characterized by the size of the cliques in the chordal graphs that contain it. 
Block graphs are a particular subclass of chordal graphs, and are also circle. However, not every chordal graph is a circle graph. All these reasons lead to consider chordal graphs as a natural restriction to study a partial characterization of circle graphs by forbidden induced subgraphs. Something similar happens with split graphs, which is an interesting subclass of chordal graphs. More precisely, split graphs are those chordal graphs for which its complement is also a chordal graph. In Chapter $2$ we give an example of a chordal graph that is neither circle nor split. Hence, studying those split graphs that are also circle is a good first step towards a characterization of those chordal graphs that are also circle. 

We started by considering a split graph $H$ such that $H$ is minimally non-circle. Since permutation graphs are a subclass of circle graphs, in particular $H$ is not a permutation graph. Notice that permutation graphs are those comparability graphs for which their complement is also a comparability graph. It is easy to prove that permutation graphs are precisely those circle graphs having a circle model with an equator.
Using the list of minimal forbidden subgraphs of comparability graphs (see Figures \ref{fig:forb_comparability1} and \ref{fig:forb_comparability2}) and the fact that $H$ is also a split graph, we conclude that $H$ contains either a tent, a $4$-tent, a co-$4$-tent or a net as an induced subgraph (See Figure \ref{fig:forb_permsplit_base}).
In Chapter \ref{chapter:partitions}, given a split graph $G=(K,S)$ and an induced subgraph $H$ that can be either a tent, a $4$-tent or a co-$4$-tent, we define partitions of $K$ and $S$ according to the adjacencies and prove that these partitions are well defined.

A $(0,1)$-matrix has the \emph{consecutive-ones roperty (C$1$P) for the rows} if there is a permutation of its columns such that the ones in each row appear consecutively. 
In order to characterize those circle graphs that contain a tent, a $4$-tent, a co-$4$-tent or a net as an induced subgraph, we first address the problem of characterizing those matrices that can be ordered with the C$1$P for the rows and for which there is a particular color assignment for every row, having exactly $2$ colors to do so. Such a color assignment is defined in Chapter \ref{chapter:2nested_matrices}, considering the fullfillment of some special properties which are purely based on the partial ordering relationship that must hold between the neighbourhoods of the vertices in the independent partition of a split graph. These properties are contemplated in the definition of admissibillity.

In Chapter \ref{chapter:2nested_matrices}, we define and characterize $2$-nested matrices by minimal forbidden submatrices. This characterization leads to a minimal forbidden induced subgraph characterization for the associated graph class, which is a subclass of split and circle graphs. 
In order to do this, we define the concept of enriched matrix, which are those $(0,1)$-matrices for which some rows are labeled with a letter L (standing for \textit{left}) or R (standing for \textit{right}) or LR (standing for \textit{left-right}), and some of these labeled rows may also be colored with either red or blue each. 
In the first sections of Chapter \ref{chapter:2nested_matrices}, we define and characterize the notions of admissibility, LR-orderable and partially $2$-nested. This notions allowed to define what is a '`valid pre-coloring'' and characterize those enriched matrices with valid pre-colorings that admit an LR-ordering, which is the property of having a lineal ordering $\Pi$ of the columns such that, when ordered according to $\Pi$, the non-LR-rows and the complements of the LR-rows have the C$1$P, those rows labeled with L start in the first column and those rows labeled with R end in the last column. 
This leads to a characterization of $2$-nested matrices by forbidden induced submatrices. $2$-nested matrices are those partially $2$-nested matrices for which the given $2$-coloring of the rows can be extended to a total proper $2$-coloring of all the matrix, while maintaining certain properties.
Chapter \ref{chapter:2nested_matrices} is crucial in order to determine which are the forbidden induced subgraphs for those circle graphs that are also split.

In chapter \ref{chapter:split_circle_graphs}, we address the problem of characterizing the forbidden induced subgraphs of a circle graph that contains either a tent, $4$-tent, co-$4$-tent or a net as an induced subgraph. In each section we see a case of the theorem, proving a characterization theorem and finishing with the guidelines to draw a circle model for each case.

\section{Known characterizations of circle graphs} \label{sec:circle1}

Recall that a graph is circle if it is the intersection graph of a family of chords in a circle. 
The characterization of the entire class of circle graphs by forbidden minimal induced subgraphs is still an open problem. However, some partial characterizations are known. In this section, we state some of the known characterizations for circle graphs, including those that are partial, and give the necessary definitions to understand these results.

A \emph{double occurrence word }is a finite string of symbols in which each symbol appears precisely twice. Let $\Gamma = (\pi_1 , \pi_2 , \ldots , \pi_{2n})$ be a double occurrence word. The \emph{shift operation on $\Gamma$} transforms $\Gamma$ into $(\pi_{2n},\pi_1,\pi_2,\ldots,\pi_{2n-1})$. The \emph{reverse operation }transforms $\Gamma$ into $\overline{\Gamma} = (\pi_{2n}, \pi_{2n-1}, \ldots , \pi_2, \pi_1)$. With each double occurrence word $\Gamma$ we associate a graph $G\left[\Gamma\right]$ whose vertices are the symbols in $\Gamma$ and in which two vertices are adjacent if and only if the corresponding symbols appear precisely once between the two occurrences of the other. Clearly, a graph is circle if and only if it is isomorphic to $G\left[\Gamma\right]$ for some double occurrence word. Those graphs that are isomorphic to $G\left[\Gamma\right]$ for some double occurrence $\Gamma$ are also called alternance graphs. A graph $G$ is overlap interval if there exists a bijective function $f : V \rightarrow I (f(v) = I_v)$ where $I = \{I_v\}_\{I \in V(G)\}$ is a family of intervals on the real line, such that $uv \in E$ if and only if $I_u$ and $I_v$ overlap; i.e., $I_u \cap I_v \neq \emptyset$, $I_u \nsubseteq I_v$ and $I_v \nsubseteq I_u$. It is well known that circle graphs and overlap interval graphs are the same class (see \cite{G04}).
Moreover, circle graphs are also equivalent to alternance graphs. 

Given a double alternance word $\Gamma$, we denote by $\overline{\Gamma}$ the word that arises by traversing $\Gamma$ from right to left, for instance, if $\Gamma = abcadcd$, then $\overline{\Gamma} = dcdacba$.
Given a graph $G$ and a vertex $v$ of $G$. The local complement of $G$ at $v$, denoted by $G*v$, is the graph that arises from $G$ by replacing $N(v)$ by its complementary graph. Two graphs $G$ and $H$ are \emph{locally equivalent} if and only if $G$ arises from $H$ by a finite sequence of local complementations.
This operation is strongly linked with circle graphs; namely, \emph{if $G$ is a circle graph, then $G*v$ is a circle graph}. This is because, if $a$ represents the vertex $v$ in $\Gamma$ and $\Gamma = AaBaC$ where $A$, $B$ and $C$ are subwords of $\Gamma$, then $G\left[AaBaC \right]$ is a double alternance model for $G*v$. Bouchet proved the following theorem.

\begin{teo} \cite{B94} \label{teo:bouchet}
 Let $G$ be a graph. $G$ is a circle graph if and only if any graph locally equivalent to $G$ has no induced subgraph isomorphic to $W_5$, $W_7$, or $BW_3$ (see Figure \ref{fig:bouchet_lc}).
\end{teo} 

\begin{figure}[h]
\centering
\includegraphics[scale=.65]{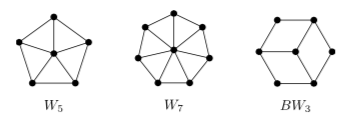}
\caption{The graphs $W_5$, $W_7$ and $BW_3$.} \label{fig:bouchet_lc}
\end{figure} 

Bouchet also proved the following property of circle graphs. Let $G = (V, E)$ and let $A = \{A_{vw} : v, w \in V \}$ be an antisymmetric integer matrix \cite{B87}. For $W \subseteq V$, we denote $A\left[W\right] = \{A_{vw} : v, w \in W \}$. The matrix $A$ satisfies the property $\alpha$ if the following property (related to unimodularity) holds: $\det(A\left[W\right]) \in \{-1, 0, 1\}$ for all $W\subseteq V$. Graph $G$ is unimodular if there is an orientation of $G$ such that the resulting digraph satisfies property $\alpha$. Bouchet proved that every circle graph admits such an orientation \cite{B87}.
Moreover, it was also Bouchet who proved that, if $G$ is a bipartite graph such that its complement is circle, then $G$ is a circle graph \cite{B99}. In \cite{ES19}, the authors give a new and shorter prove for this result.
       
Inspired by Theorem \ref{teo:bouchet}, Geelen and Oum gave a new characterization of circle graphs in terms of \emph{pivoting} \cite{GO09}. The result of pivoting a graph $G$ with respect to an edge $uv$ is the graph $G \times uv = G * u * v * u$, where $*$ stands for local complementation. 
A graph $G'$ is \emph{pivot equivalent }to $G$ if $G'$ arises from $G$ by a sequence of pivoting operations. They proved, with the aid of a computer, that $G$ is a circle graph if and only if each graph that is pivot equivalent to $G$ contains none of 15 prescribed induced subgraphs.

Let $G_1$ and $G_2$ be two graphs such that $|V(G_i)| \geq 3$, for each $i = 1,2$, and assume that $V(G_1) \cap V(G_2) = \emptyset$. Let $v_i$ be a distinguished vertex of $G_i$, for each $i = 1, 2$. The \emph{split composition }of $G_1$ and $G_2$ with respect to $v_1$ and $v_2$ is the graph $G_1 \circ G_2$ whose vertex set is $V(G_1 \circ G_2) = (V(G_1)\cup V(G_2)) \setminus \{v_1,v_2\}$ and whose edge set is $E(G_1 \circ G_2) = E(G_1 \setminus \{v_1\}) \cup E(G_2 \setminus \{v_2\}) \cup \{uv : u \in N_{G_1}(v_1) \mbox{ and } v \in N_{G_2}(v_2)\}$. The vertices $v_1$ and $v_2$ are called the \emph{marker vertices}. We say that $G$ has a \emph{split decomposition} if there exist two graphs $G_1$ and $G_2$ with $|V(G_i)| \geq 3$, $i = 1,2$, such that $G = G_1 \circ G_2$ with respect to some pair of marker vertices. If so, $G_1$ and $G_2$ are called the factors of the split decomposition. Those graphs that do not have a split decomposition are called \emph{prime graphs}. The concept of split decomposition is due to Cunningham \cite{C82}. 

Circle graphs turned out to be closed under this decomposition \cite{B87} and in 1994 Spinrad presented a quadratic-time recognition algorithm for circle graphs that exploits this peculiarity \cite{S94}. Also based on split decomposition, Paul \cite{P06} presented an $\mathcal{O}((n + m)\alpha(n + m))$-time algorithm for recognizing circle graphs, where $\alpha$ is the inverse of the Ackermann function.

In \cite{F84} De Fraysseix presented a characterization of circle graphs, which leads to a novel in\-ter\-pre\-ta\-tion of circle graphs as the intersection graphs of induced paths of a given graph. A \emph{cocycle }of a graph $G$ with vertex set $V$ is the set of edges joining a vertex of $V_1$ to a vertex of $V_2$ for some bipartition $(V_1,V_2)$ of $V$. A \emph{cocyclic-path }is an induced path whose set of edges constitutes a cocycle. A \emph{cocyclic-path intersection graph } is a simple graph with vertex set being a family of cocyclic-paths of a given graph, two vertices being adjacent if and only if the corresponding cocyclic-paths have an edge in common. Notice that the definition is restricted to those graphs covered by cocyclic-paths any two of which have at most a common edge. Fraysseix proved the following characterization of circle graphs as cocyclic-path intersection graphs.

\begin{teo}\cite{F84} 
Let $G$ be a graph. $G$ is a circle graph if and only if $G$ is a cocyclic-path intersection graph.
\end{teo}

A \emph{diamond} is the complete graph with 4 vertices minus one edge. 
A claw is the graph with 4 vertices that has 1 vertex with degree 3 and a 
maximum independent set of size 3.
\emph{Prisms} are the graphs that arise from the cycle $C_6$ by subdividing the edges that link the triangles.

A graph is \emph{Helly circle }if it has a circle model whose chords
are all different and every subset of pairwise intersecting chords has a point in common. 
A characterization by minimal forbidden induced subgraphs for Helly circle graphs, inside circle graphs,
was conjectured in \cite{D03} and was proved some years later in \cite{DGR10}. Notice that this characterization does not solve the general characterization of Helly circle graphs by forbidden subgraphs.


\begin{teo} \cite{DGR10} 
Let $G$ be a circle graph. $G$ is Helly circle if and only if $G$ is diamond-free.
\end{teo}

A graph $G$ is \emph{domino }if each of its vertices belongs to at most
two cliques. In addition, if each of its edges belongs to at most one
clique, $G$ is \emph{linear-domino}. Linear-domino graphs coincide with
$\{$claw,diamond$\}$-free graphs.

There are no known characterizations for the class of circle graphs by minimal forbidden induced subgraphs. In order to obtain some results in this direction, this problem was addressed by attempting to characterize circle graphs by minimal forbidden induced subgraphs when given a graph that belongs to a certain graph class. This is known as a partial structural characterization. 
Some results in this direction are the following.

\begin{teo}\cite{BDGS10}
Let $G$ be a linear domino graph. Then, $G$ is a circle graph if and only if $G$ contains no induced prisms.
\end{teo}

The proof given in \cite{BDGS10} is based on the fact that circle graphs are closed under split decomposition \cite{B87}. As a corollary of
the above theorem, the following partial characterization of Helly circle graphs is obtained.

\begin{cor} \cite{BDGS10}
Let $G$ be a claw-free graph. Then, $G$ is a Helly circle graph if and only if $G$ contains no induced prism and no induced diamond.
\end{cor}

A graph is \emph{cograph} if it is $P_4$-free. 
A graph is \emph{tree-cograph} if it can be constructed from trees by disjoint union and complement operations.
Let $A$ be a $P_4$ in some graph $G$. A \emph{partner of $A$ in $G$ }is a vertex $v$ in $G\setminus A$ such that $A+v$ induces at least two $P_4$'s. A graph $G$ is $P_4$-tidy if any $P_4$ has at most one partner.

\begin{teo} \cite{BDGS10} 
Let $G$ be a $P_4$-tidy graph. Then, $G$ is a circle graph if and only if $G$ contains no $W_5$, net$+K_1$, tent$+K_1$, or tent-with-center as induced subgraph.
\end{teo}

\begin{teo} \cite{BDGS10} 
Let $G$ be a tree-cograph. Then, $G$ is a circle graph if and only if $G$ contains no induced (bipartite-claw)$+K_1$ and no induced co-(bipartite-claw).
\end{teo}

\section{Basic definitions and notation} \label{section:basic_defs}

Let $A=(a_{ij})$ be a $n\times m$ $(0,1)$-matrix.
We denote by $a_{i.}$ and $a_{.j}$ the $i$th row and the $j$th column of matrix $A$. From now on, we associate each row $a_{i.}$ with the set of columns in which $a_{i.}$ has a $1$. For example, the \emph{intersection} of two rows $a_{i.}$ and $a_{j.}$ is the subset of columns in which both rows have a $1$.
Two rows $a_{i.}$ and $a_{k.}$ are \emph{disjoint} if there is no $j$ such that $a_{ij} = a_{kj} = 1$.
We say that $a_{i.}$ is \emph{contained} in $a_{k.}$ if for each $j$ such that $a_{ij} = 1$ also $a_{kj} = 1$. We say that $a_{i.}$ and $a_{k.}$ are \emph{nested} if $a_{i.}$ is contained in $a_{k.}$ or $a_{k.}$ is contained in $a_{i.}$.
We say that a row $a_{i.}$ is \emph{empty} if every entry of $a_{i.}$ is $0$, and we say that $a_{i.}$ is \emph{nonempty} if there is at least one entry of $a_{i.}$ equal to $1$.
We say that two nonempty rows \emph{overlap} if they are non-disjoint and non-nested.
For every nonempty row $a_{i.}$, let $l_i = \min\{ j \colon\,a_{ij} = 1 \}$ and $r_i = \max\{ j \colon\,a_{ij} = 1 \}$ for each $i\in\{1,\ldots,n\}$.
Finally, we say that $a_{i.}$ and $a_{k.}$ \emph{start} (resp.\ \emph{end}) \emph{in the same column} if $l_i = l_k$ (resp.\ $r_i = r_k$), and we say $a_{i.}$ and $a_{k.}$ \emph{start (end) in different columns}, otherwise.

We say a $(0,1)$-matrix $A$ has the \emph{consecutive-ones property for the rows }(for short, C$1$P) if there is permutation of the columns of $A$ such that the 1's in each row appear con\-sec\-u\-tive\-ly. Any such permutation of the columns of $A$ is called a \emph{consecutive-ones ordering }for $A$.
In \cite{T72}, Tucker characterized all the minimal forbidden submatrices for the C$1$P, later known as \emph{Tucker matrices}. For the complete list of Tucker matrices, see Figure \ref{fig:tucker_matrices}.
 
\begin{figure}[h!] 
	\centering
	\begin{align*}
			M_I(k)&= \begin{pmatrix}
				110...00\\
				011...00\\
				.   .   .   .   . \\
				.   .   .   .   . \\
				.   .   .   .   . \\
				000...11\\
				100...01\\
			\end{pmatrix}
			&
			M_{II}(k)&= \begin{pmatrix}
				011...111\\
				110...000\\
				011...000\\
				.   .   .   .   . \\
				.   .   .   .   . \\
				000...110\\
				111...101\\
			\end{pmatrix}
			&
			M_{III}(k)&= \begin{pmatrix}
				110...000\\
				011...000\\
				.   .   .   .   . \\
				.   .   .   .   . \\
				000...110\\
				011...101\\
			\end{pmatrix}
			\\
			\end{align*}
			\begin{align*}
			M_{IV}&= \begin{pmatrix}
				110000\\
				001100\\
				000011\\
				010101\\
			\end{pmatrix}
			&
			M_{V}&= \begin{pmatrix}
				11000\\
				00110\\
				11110\\
				10011\\
			\end{pmatrix}
	\end{align*}
	\caption{Tucker matrices $M_{I}(k) \in \{0,1\}^{k \times k}$, $M_{III}(k) \in \{0,1\}^{k \times (k+1)}$ with $k \geq 3$, and $M_{II}(k) \in \{0,1\}^{k \times k}$ with $k \geq 4$} \label{fig:tucker_matrices}
\end{figure}


Let $A$ and $B$ be $(0,1)$-matrices. We say that $B$ is a \emph{subconfiguration} of $A$ if there is a permutation of the rows and the columns of $B$ such that $B$ with this permutation results equal to a submatrix of $A$. 
Given a subset of rows $R$ of $A$, we say that $R$ \emph{induces a matrix} $B$ if $B$ is a subconfiguration of the submatrix of $A$ given by selecting only those rows in $R$.

All graphs in this work are simple, undirected, with no loops or multiple edges. The pair $(K,S)$ is a \emph{split partition} of a graph $G$ if $\{K,S\}$ is a partition of the vertex set of $G$ and the vertices of $K$ (resp.\ $S$) are pairwise adjacent (resp.\ nonadjacent), and we denote it $G=(K,S)$. A graph $G$ is a \emph{split graph} if it admits some split partition. Let $G$ be a split graph with split partition $(K,S)$, $n=\vert S\vert$, and $m=\vert K\vert$.
Let $s_1, \ldots, s_n$ and $v_1, \ldots, v_m$ be linear orderings of $S$ and $K$, respectively. Let $A= A(S,K)$ be the $n\times m$ matrix defined by $A(i,j)=1$ if $s_i$ is adjacent to $v_j$ and $A(i,j)=0$, otherwise.
From now on, we associate the row (resp.\ column) of the matrix $A(S,K)$ with the corresponding vertex in the independent set (resp.\ vertex in the complete set) of the partition.

Given a graph $H$, the graph $G$ is $H$-free if it does not contain $H$ as induced subgraph. For a family of graphs $\mathcal{H}$, a graph
$G$ is $\mathcal{H}$-free if $G$ is $H$-free for every $H \in \mathcal{H}$.


\selectlanguage{spanish}%
\chapter*{Preliminares}

Consideremos un grafo split $G=(K,S)$ y supongamos que $G$ es minimalmente no circle. Equivalentemente, todo subgrafo inducido $H$ es circle. Si $G$ no es circle, entonces en particular $G$ no es un grafo de permutación. 
Los grafos de permutación son exactamente aquellos grafos de comparabilidad cuyo complemento también es un grafo de comparabilidad \cite{EPL72}. Los grafos de comparabilidad han sido caracterizados por subgrafos inducidos prohibidos en \cite{G67}.

\begin{teo}[\cite{G67}]
Un grafo es de comparabilidad si y sólo si no contiene como subgrafo inducido a ninguno de los subgrafos de la Figura \ref{fig:forb_comparability1_} y su complemento no contiene como subgrafo inducido a ninguno de los subgrafos en la Figura \ref{fig:forb_comparability2_}.
\end{teo}

\begin{figure}[h]
\centering
\includegraphics[scale=.75]{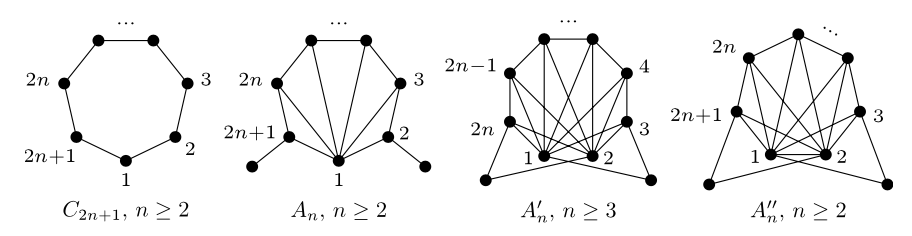}
\caption{Subgrafos inducidos prohibidos para los grafos de comparabilidad.} \label{fig:forb_comparability1_}
\end{figure} 

\begin{figure}[h]
\centering
\includegraphics[scale=.7]{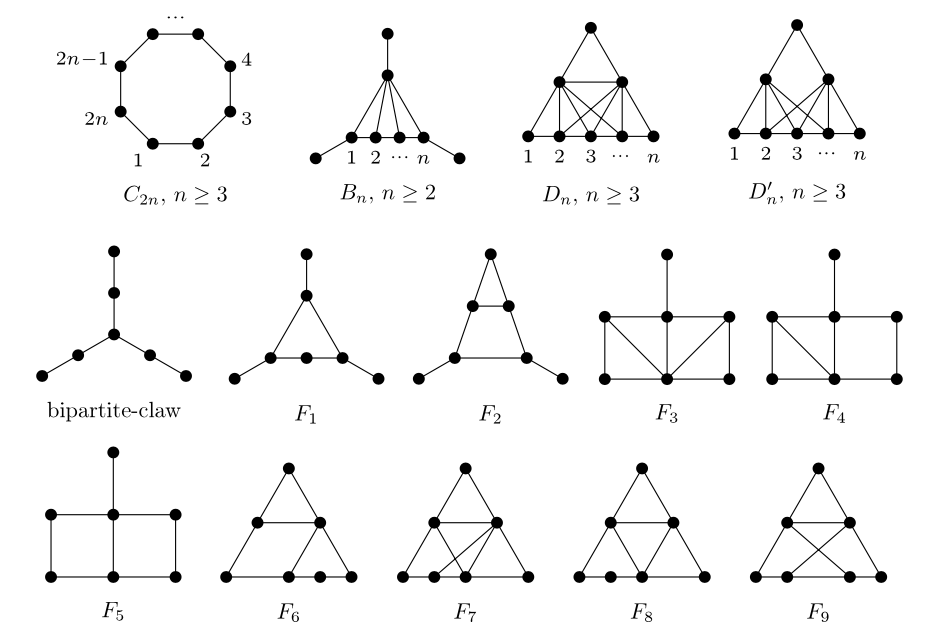}
\caption{Subgrafos inducidos prohibidos para los grafos de comparabilidad.} \label{fig:forb_comparability2_}
\end{figure} 

Esta caracterización de los grafos de comparabilidad induce una caracterización por subgrafos inducidos prohibidos para la clase de los grafos de permutación.
Por lo tanto, como los grafos de permutación son una subclase de los grafos circle, en particular $G$ no es un grafo de permutación. Usando la lista de subgrafos inducidos minimales prohibidos para los grafos de comparabilidad dada en las Figuras \ref{fig:forb_comparability1_} y \ref{fig:forb_comparability2_}  y el hecho de que $G$ también es un grafo split, concluímos que $G$ debe contener como subgrafo inducido un tent, un $4$-tent, un co-$4$-tent o un net (Ver Figura \ref{fig:forb_permsplit_base_}).

\begin{figure}[h]
\centering
\includegraphics[scale=.5]{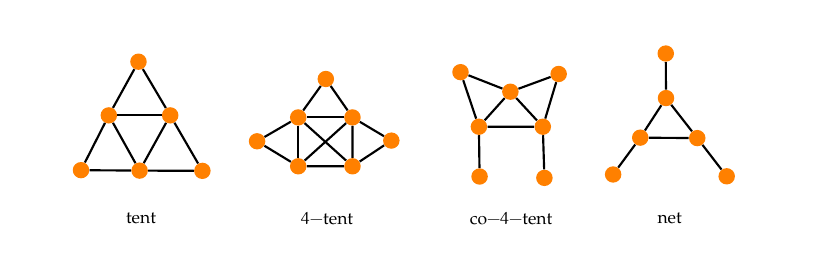}
\caption{Subgrafos inducidos prohibidos para permutación $\cap$ grafos split.
} \label{fig:forb_permsplit_base_}
\end{figure} 

En este capítulo, se define una partición de los conjuntos $K$ y $S$ basada en el hecho de que $G$ contiene un tent, $4$-tent o co-$4$-tent $H$ como subgrafo inducido. En cada caso, la partición de $K$ depende de las adyacencias a los vértices de $V(T) \cap S$, y la partición de $S$ a su vez depende de las adyacencias de cada vértice independiente a las distintas particiones de $K$.
Se prueba que estos subconjuntos efectivamente inducen una partición de $K$ y $S$, respectivamente. 
Estos resultados serán de gran utilidad en el Capítulo 3 para dar motivación a la teoría de matrices desarrollada a lo largo del mismo, y luego en el Capítulo 4, donde se enuncia y se demuestra la caracterización pos subgrafos inducidos prohibidos para grafos split circle.


\selectlanguage{english}%
\chapter{Preliminaries} \label{chapter:partitions}

Let us consider a split graph $G=(K,S)$ and suppose that $G$ is minimally non-circle. Equivalently, any proper induced subgraph of $H$ is circle. If $G$ is not circle, then in particular $G$ is not a permutation graph. 
Permutation graphs are exactly those comparability graphs whose complement graph is also a comparability graph \cite{EPL72}. Comparability graphs have been characterized by forbidden induced subgraphs in \cite{G67}.

\begin{teo}[\cite{G67}]
A graph is a comparability graph if and only if it does not contain as an induced subgraph any graph in Figure \ref{fig:forb_comparability1} and its complement does not contain as an induced subgraph any graph in Figure \ref{fig:forb_comparability2}.
\end{teo}

\begin{figure}[h]
\centering
\includegraphics[scale=.75]{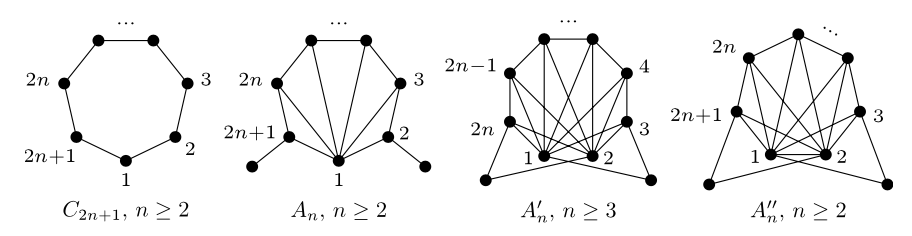}
\caption{Forbidden induced subgraphs for comparability graphs.} \label{fig:forb_comparability1}
\end{figure} 

\begin{figure}[h]
\centering
\includegraphics[scale=.7]{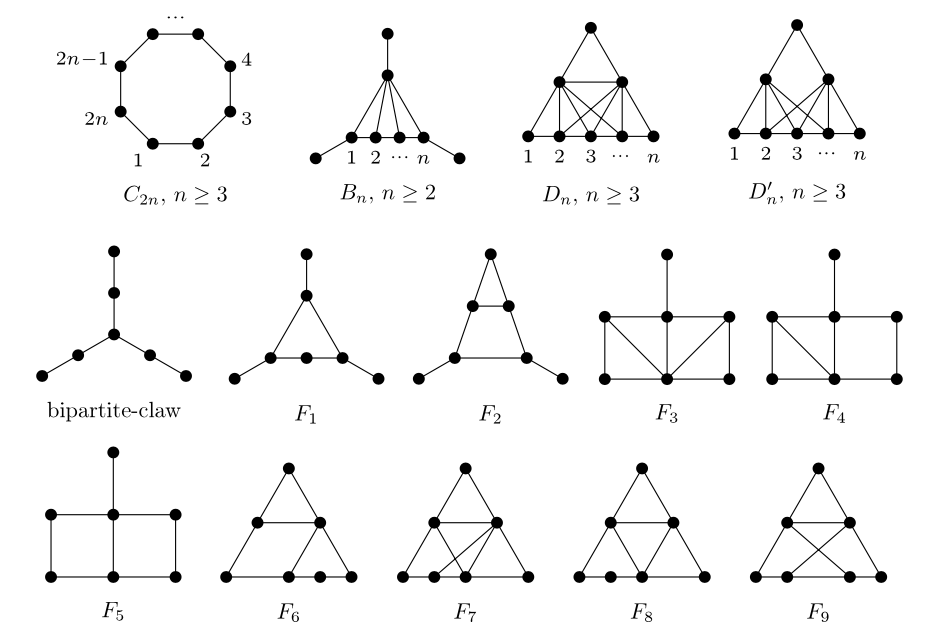}
\caption{Forbidden induced subgraphs for comparability graphs.} \label{fig:forb_comparability2}
\end{figure}

This characterization of comparability graphs leads to a forbidden induced subgraph characterization for the class of permutation graphs.
Hence, since permutation graphs is a subclass of circle graphs, in particular $G$ is not a permutation graph. Using the list of minimal forbidden subgraphs for comparability graphs given in Figures \ref{fig:forb_comparability1} and \ref{fig:forb_comparability2} and the fact that $G$ is also a split graph, we conclude that $G$ contains either a tent, a $4$-tent, a co-$4$-tent or a net as an induced subgraph (See Figure \ref{fig:forb_permsplit_base}). 

\begin{figure}[h]
\centering
\includegraphics[scale=.5]{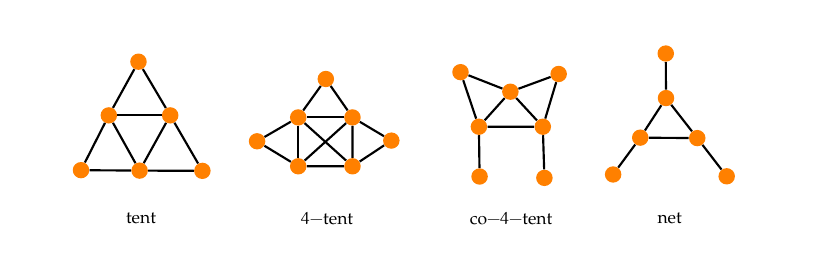}
\caption{Forbidden subgraphs for permutation $\cap$ split graphs.
} \label{fig:forb_permsplit_base}
\end{figure} 

\begin{figure}[h]
\centering
\begin{subfigure}{.328\textwidth}
\centering
\includegraphics[scale=.25]{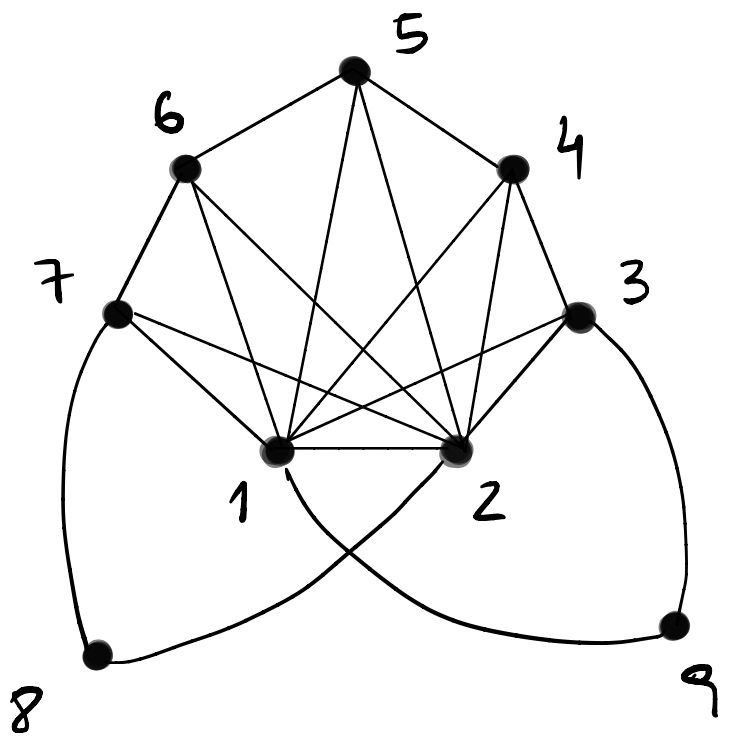}
\caption{The graph $A''_3$}
\end{subfigure}
\begin{subfigure}{.328\textwidth}
\centering
\includegraphics[scale=.25]{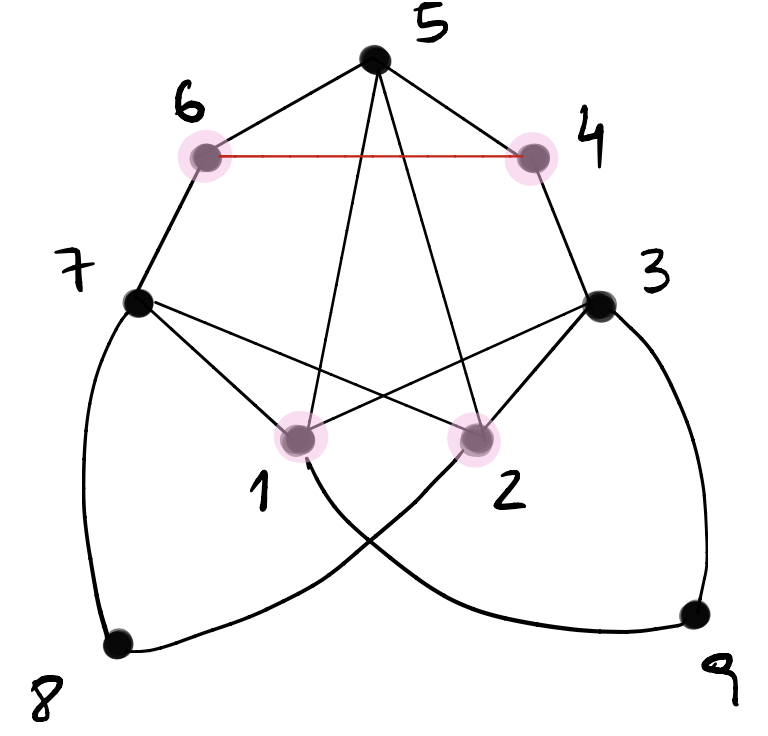}
\caption{Local complementation by $5$}
\end{subfigure}
\begin{subfigure}{.328\textwidth}
\centering
\includegraphics[scale=.25]{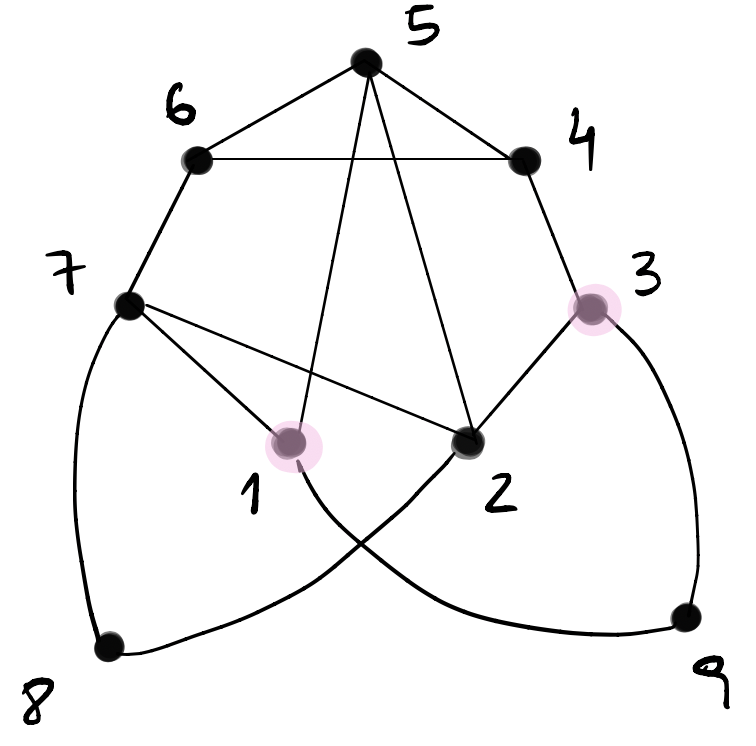}
\caption{Local complementation by $9$}
\end{subfigure} \\
\begin{subfigure}{.327\textwidth}
\centering
\includegraphics[scale=.25]{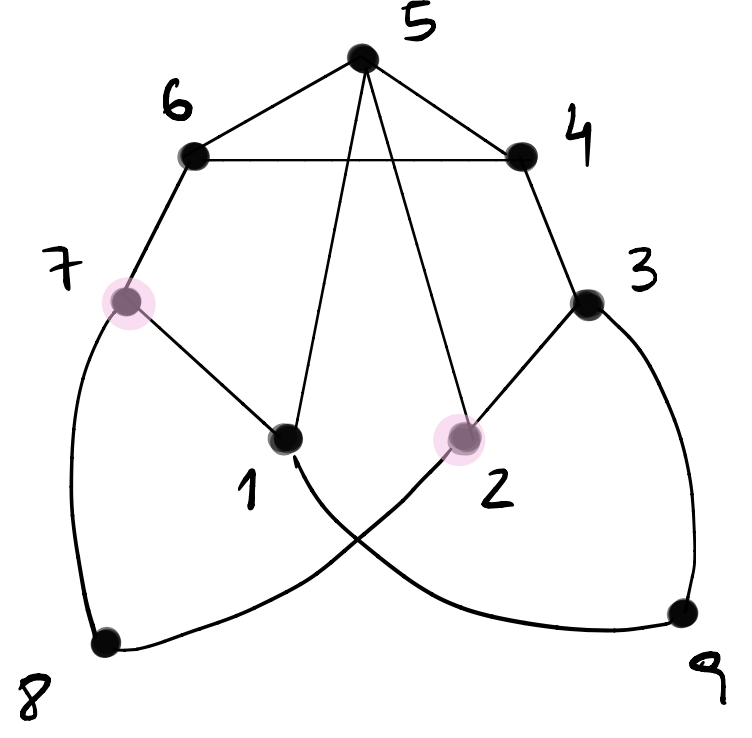}
\caption{Local complementation by $8$}
\end{subfigure} 
\begin{subfigure}{.327\textwidth}
\centering
\includegraphics[scale=.25]{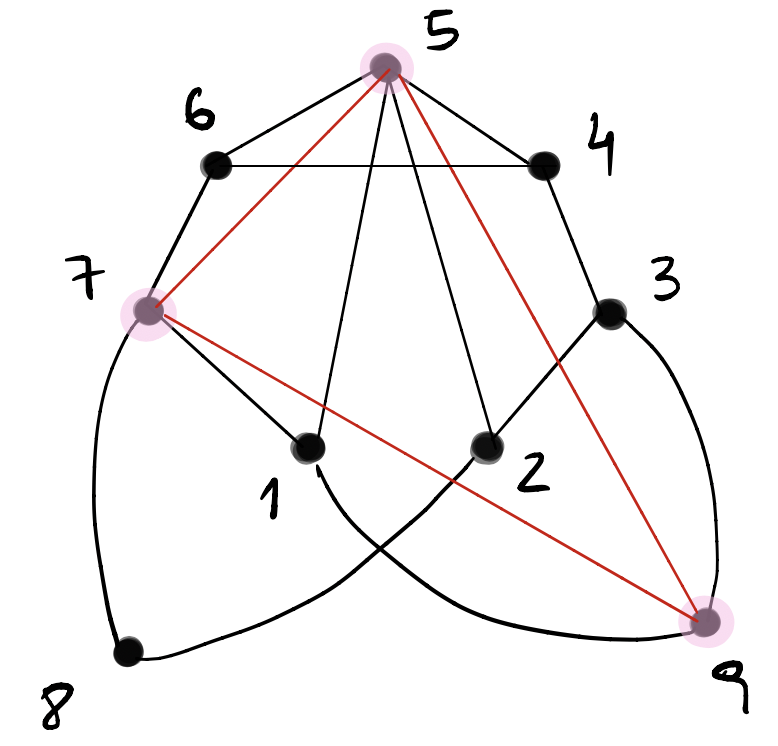}
\caption{Local complementation by $1$}
\end{subfigure} 
\begin{subfigure}{.327\textwidth}
\centering
\includegraphics[scale=.25]{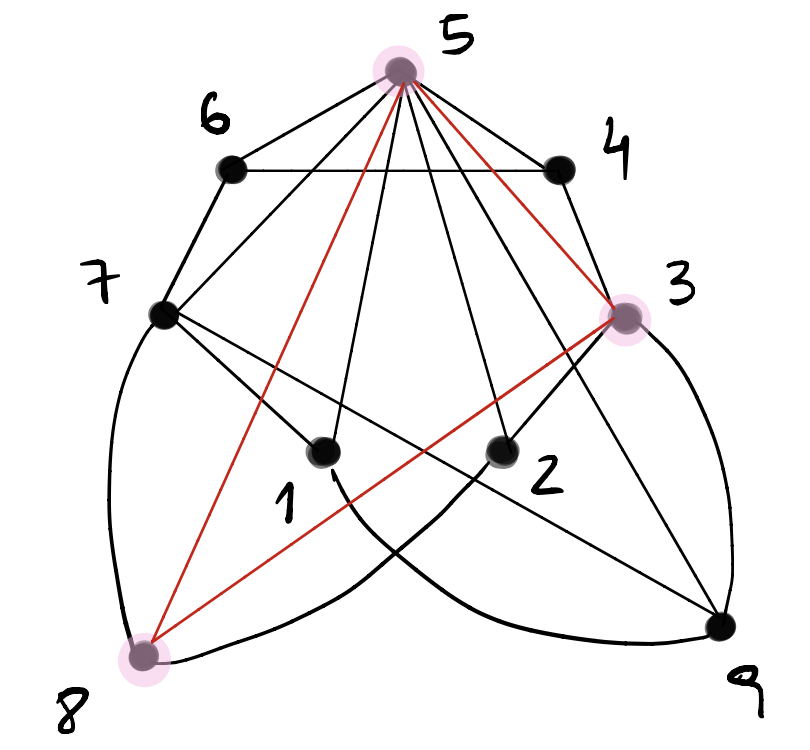}
\caption{Local complementation by $2$}
\end{subfigure}
\caption{Sequence of local complementations applied to $A''_3$.} \label{fig:local_complement_A''3}
\end{figure} 

As previously mentioned, the motivation to study circle graphs restricted to split graphs came from chordal graphs. Remember that split graphs are those chordal graphs for which its complement is also a chordal graph.
Let us consider the graph $A''_n$ for $n = 3$ depicted in Figure \ref{fig:forb_comparability1}. 

This is a chordal graph since $A''_3$ contains no cycles of length greater than $3$. Moreover, it is easy to see that $A''_3$ is not a split graph. This follows from the fact that the maximum clique has size $4$, and the removal of any such clique leaves out a non-independent set of vertices. The same holds for any clique of size smaller than $4$. 
Furthermore, if we apply local complement of the graph sequentially on the vertices $5$, $9$, $8$, $1$ and $2$, then we find $W_5$ induced by the subset $\{ 5$, $3$, $4$, $6$, $7$, $8 \}$. For more detail on this, see Figure \ref{fig:local_complement_A''3}. It follows from the characterization given by Bouchet in \ref{teo:bouchet} that $A''_3$ is not a circle graph.

This shows an example of a graph that is neither circle nor split, but is chordal. In particular, it follows from this example (which is minimally non-circle) that whatever list of forbidden subgraphs found for split circle graphs is not enough to characterize those chordal graphs that are also circle. Therefore, studying split circle graphs is a good first step towards characterizing those chordal graphs that are also circle.

Throughout the following sections, we will define some subsets in both $K$ and $S$ depending on whether $G$ contains an induced tent, $4$-tent or co-$4$-tent $H$ as an induced subgraph. We will prove that these subsets induce a partition of both $K$ and $S$. In each case, the vertices in the complete partition $K$ are split into subsets according to the adjacencies with the independent vertices of $H$, and the vertices in the independent partition $S$ are split into subsets according to the adjacencies with each partition of $K$.
These partitions will be useful in Chapter \ref{chapter:2nested_matrices}, in order to give motivation for the matrix theory developed in that chapter, and in Chapter \ref{chapter:split_circle_graphs}, when we give the proof of the characterization by forbidden induced subgraphs for split circle graphs.
Notice that we do not consider the case in which $G$ contains an induced net in order to define the partitions of $K$ and $S$, for it will be explained in detail in Section \ref{sec:circle5} that this case can be reduced using the cases in which $G$ contains a tent, a $4$-tent and a co-$4$-tent.

In Figures \ref{fig:forb_T_graphs} and \ref{fig:forb_F_graphs}, we define two graph families that will be central throughout the sequel. These graphs are necessary to state the main result of this part, which is the following characterization by forbidden induced subgraphs for those split graphs that are also circle.

\begin{teo}
Let $G=(K,S)$ be a split graph. Then, $G$ is a circle graph if and only if $G$ is $\{ \mathcal{T}, \mathcal{F}\}$-free. 
\end{teo}

\begin{figure}[h]
\centering
\includegraphics[scale=.33]{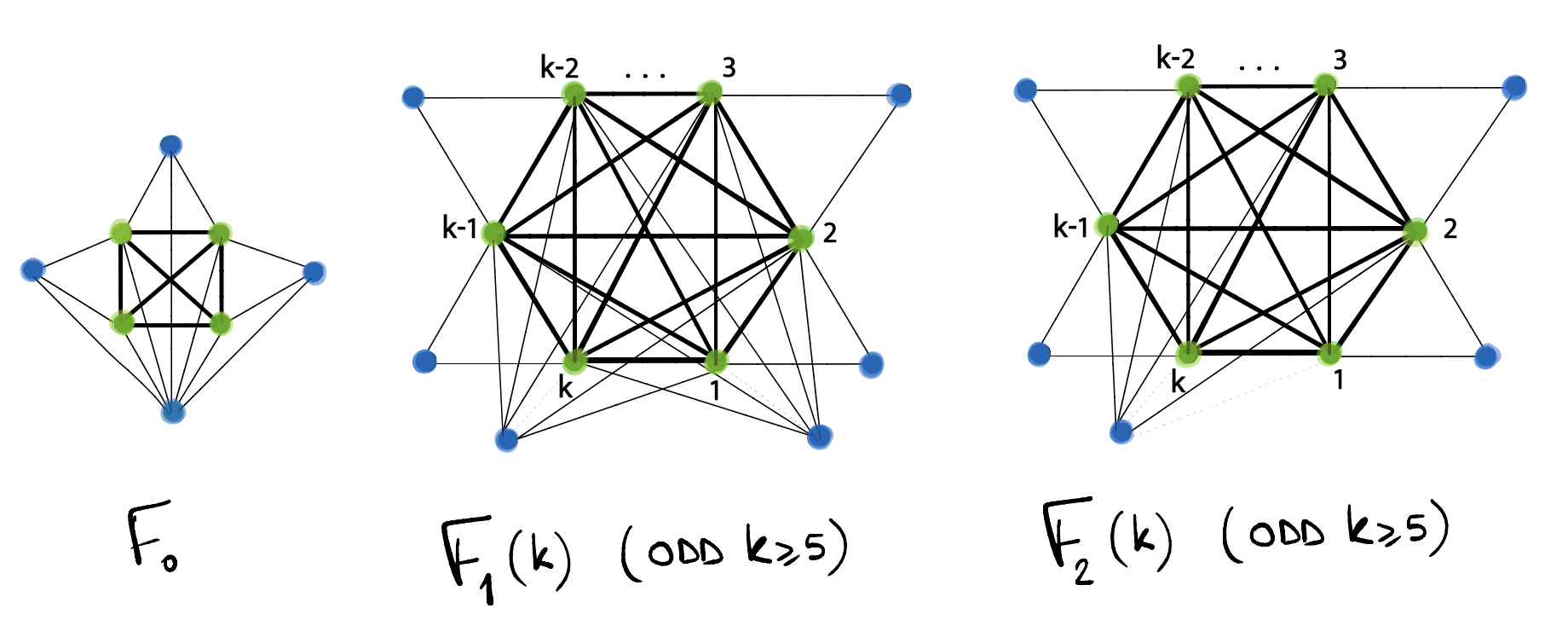} 
\caption{The graphs in the family $\mathcal{F}$.} \label{fig:forb_F_graphs}
\end{figure}

\begin{figure}[h]
\centering
\includegraphics[scale=.35]{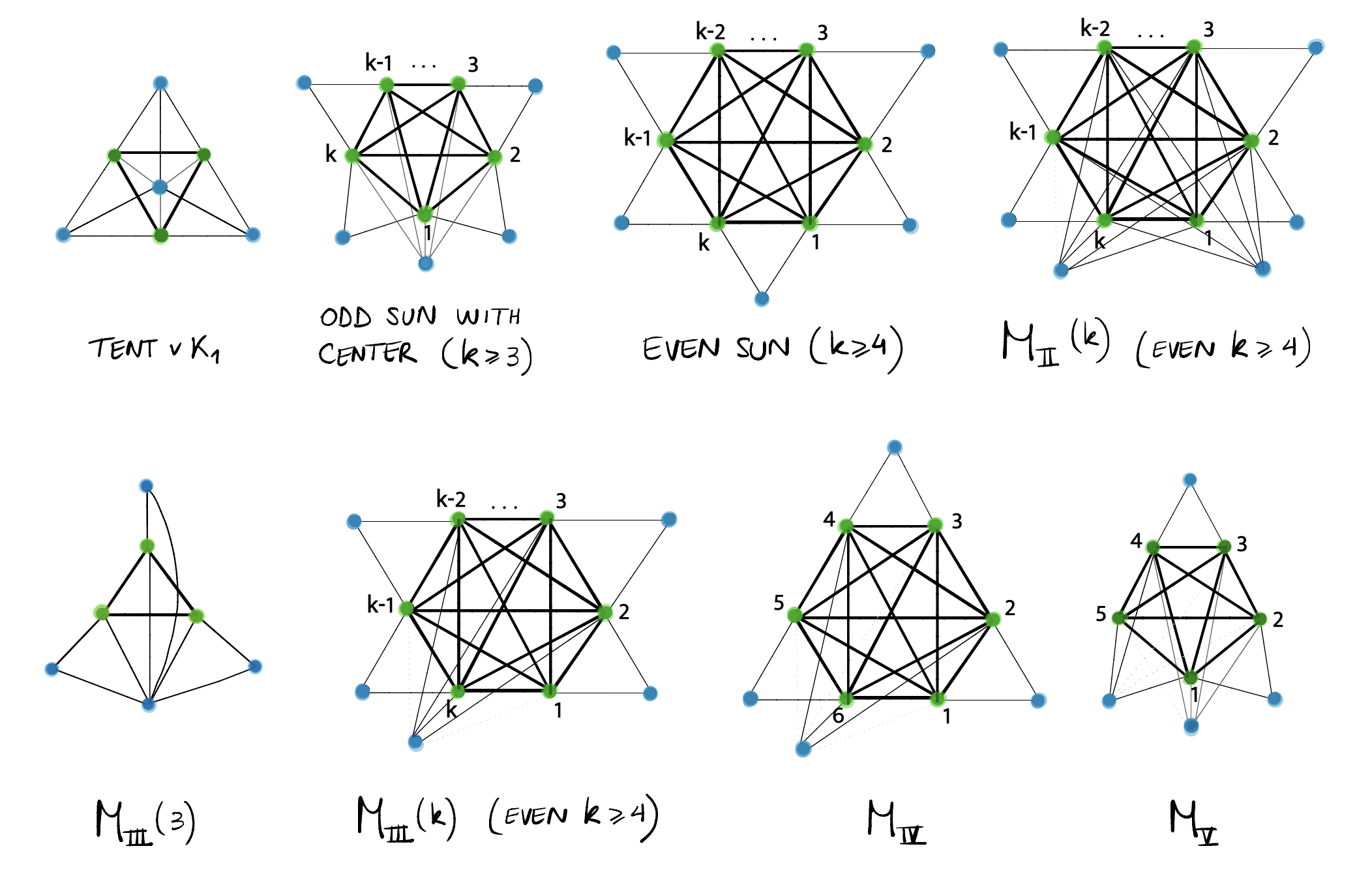} 
\caption{The graphs in the family $\mathcal{T}$.} \label{fig:forb_T_graphs}
\end{figure}

\section{Partitions of S and K for a graph containing an induced tent} \label{sec:tent_partition}

Let $G=(K,S)$ be a split graph where $K$ is a clique and $S$ is an independent set. Let $H$ be an induced subgraph of $G$ isomorphic to a tent. Let $V(T)=\{k_1,$ $k_3,$ $k_5,$ $s_{13},$ $s_{35},$ $s_{51}\}$ where $k_1,$ $k_3,$ $k_5\in K$, $s_{13},$ $s_{35},$ $s_{51}\in S$, and the neighbors of $s_{ij}$ in $H$ are precisely $k_i$ and $k_j$.

We introduce sets $K_1,K_2,\ldots,K_6$ as follows.
\begin{itemize}
 \item For each $i\in\{1,3,5\}$, let $K_i$ be the set of vertices of $K$ whose neighbors in $V(T)\cap S$ are precisely $s_{(i-2)i}$ and $s_{i(i+2)}$ (where subindexes are modulo~$6$).
 \item For each $i\in\{2,4,6\}$, let $K_i$ be the set of vertices of $K$ whose only neighbor in $V(T)\cap S$ is $s_{(i-1)(i+1)}$ (where subindexes are modulo~$6$).
\end{itemize}
See Figure~\ref{fig:tent_ext} for a graphic idea of this. Notice that $K_1$, $K_3$ and $K_5$ are always nonempty sets.
We say a vertex $v$ is \emph{complete to} the set of vertices $X$ if $v$ is adjacent to every vertex in $X$, and we say $v$ is \emph{anticomplete to} $X$ if $v$ has no neighbor in $X$.
We say that $v$ is \emph{adjacent to} $X$ if $v$ has at least one neighbor in $X$. Notice that complete to $X$ implies adjacent to $X$ if and only if $X$ is nonempty.
For $v$ in $S$, let $N_i(v) = N(v) \cap K_i$.
Given two vertices $v_1$ and $v_2$ in $S$, we say that $v_1$ and $v_2$ are \emph{nested} if either $N(v_1) \subseteq N(v_2)$ or $N(v_2) \subseteq N(v_1)$.
In particular, given $i \in \{1, \ldots, 6\}$, if either $N_{i}(v_1) \subseteq N_i(v_2)$ or $N_{i}(v_2) \subseteq N_i(v_1)$, then we say that
$v_1$ and $v_2$ are \emph{nested in} $K_i$.
Additionally, if $N(v_1) \subseteq N(v_2)$, then we say that \emph{$v_1$ is contained in $v_2$}.

\begin{figure}[h!]
	\begin{center}
		\includegraphics[scale=1]{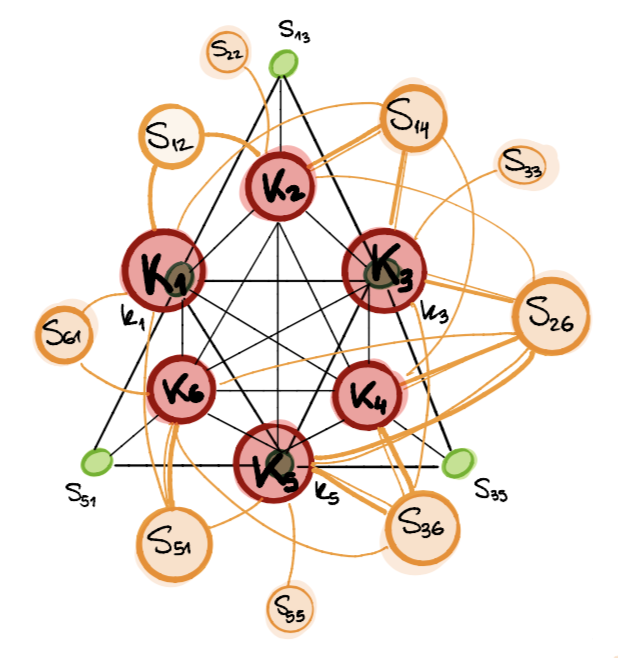}
	\end{center}
	\caption{Tent $H$ and the split graph $G$ according to the given extensions}
	\label{fig:tent_ext}
\end{figure}

\begin{lema} \label{lema:tent_1}
	If $G$ is $\{ \mathcal{T}, \mathcal{F} \}$-free, then $\{K_1,K_2,\ldots,K_6\}$ is a partition of $K$.
\end{lema}

\begin{proof}
    Every vertex of $K$ is adjacent to precisely one or two vertices of $V(T) \cap S$, for if not we find either a tent${}\vee{}K_1$ or a $3$-sun with center as induced subgraph of $G$, a contradiction.
\end{proof}

Let $i,j\in\{1,\ldots,6\}$ and let $S_{ij}$ be the set of vertices of $S$ that are adjacent to some vertex in $K_i$ and some vertex in $K_j$, are complete to $K_{i+1}$,$K_{i+2},\ldots,K_{j-1}$, and are anticomplete to $K_{j+1}$,$K_{j+2},\ldots,K_{i-1}$ (where subindexes are modulo~$6$).
The following claims are necessary to prove Lemma~\ref{lema:tent_2}, that states, on the one hand, which sets of $\{S_{ij}\}_{i,j \in \{1, \ldots, 6\}}$ may be nonempty, and, on the other hand, that the sets $\{S_{ij}\}_{i,j \in \{1, \ldots, 6\}}$ indeed induce a partition of $S$.
This shows that the adjacencies of a vertex of $S$ have a circular structure with respect to the defined partition of $K$.

	\begin{claim} \label{claim:tent_1}
		If $G$ is $\{ \mathcal{T}, \mathcal{F} \}$-free, then there is no vertex $v$ in $S$ such that $v$ is simultaneously adjacent to $K_1$, $K_3$ and $K_5$.
        Moreover, there is no vertex $v$ in $S$ adjacent to $K_2$, $K_4$ and $K_6$ such that $v$ is anticomplete to any two of $K_j$, for $j \in \{1, 3, 5 \}$.
	\end{claim}
	
	 Let $v$ in $S$ and let $w_i$ in $K_i$ for each $i\in\{1,3,5\}$, such that $v$ is adjacent to each $w_i$. Hence, $\{w_1,$ $w_3,$ $w_5,$ $s_{13},$ $s_{35},$ $s_{51},$ $v\}$ induce in $G$ a $3$-sun with center, a contradiction.

    To prove the second statement, let $w_i$ in $K_i$ such that $v$ is adjacent to $w_i$ for every $i\in\{2,4,6\}$. Suppose that $v$ is anticomplete to $K_3$ and $K_5$. Thus, we find a $4$-sun induced by the set $\{ w_2$, $k_3$, $k_5$, $w_6$, $s_{13}$, $s_{35}$, $s_{51}$, $v \}$.
    If instead $v$ is anticomplete to $K_1$ and $K_3$, then we find a $4$-sun induced by $\{ k_1$, $k_3$, $w_4$, $w_6$, $s_{13}$, $s_{35}$, $s_{51}$, $v \}$, and if $v$ is anticomplete to $K_1$ and $K_5$, then a $4$-sun is induced by $\{ k_1$, $w_2$, $w_4$, $k_5$, $s_{13}$, $s_{35}$, $s_{51}$, $v \}$. \QED
    
	\begin{claim} \label{claim:tent_2}
		If $G$ is $\{ \mathcal{T}, \mathcal{F} \}$-free and $v$ in $S$ is adjacent to $K_i$ and $K_{i+3}$, then $v$ is complete to $K_j$, either for $j \in \{ i+1, i+2 \}$ or for $j \in \{ i-1, i-2\}$.
    \end{claim}

    Let $w_i$ in $K_i$, $w_{i+3}$ in $K_{i+3}$ such that $v$ is adjacent to $w_i$ and $w_{i+3}$. Notice that the statement is exactly the same for $i = j$ and for $i = j+3$, so let us assume that $i$ is even.

    If $w_j$ in $K_j$ is a non-neighbor of $v$ for each $j \in \{ i-1, i+1 \}$, then we find an induced $M_{III}(3)$. Hence, $v$ is complete to $K_j$ for at least one of $j \in \{i-1, i+1 \}$.
    Suppose that $v$ is complete to $K_{i+1}$. If $K_{i+2} = \emptyset$, then the claim holds. Hence, suppose that $K_{i+2} \neq \emptyset$ and suppose $w_{i+2}$ in $K_{i+2}$ is a non-neighbor of $v$. In particular, since $v$ is adjacent to $w_{i+3}$ and $k_{i+1}$, then $v$ is anticomplete to $K_{i-1}$ by Claim~\ref{claim:tent_1}. However, in this case we find $M_{III}(3)$ induced by $\{s_{(i-1)(i+3)}$, $s_{(i+3)(i-1)}$, $v$, $k_{i-1}$, $k_{i+1}$, $w_{i+2}$, $w_{i+3} \}$.
    It follows analogously if instead $v$ is complete to $K_{i-1}$ and is not complete to $K_{i-2}$, for we find the same induced subgraphs. Notice that the proof is independent on whether $K_j = \emptyset$ or not, for every even $j$. \QED
	
	\begin{claim} \label{claim:tent_3}
		If $G$ is $\{ \mathcal{T}, \mathcal{F} \}$-free and $v$ in $S$ is adjacent to $K_i$ and $K_{i+2}$, then either $v$ is complete to $K_{i+1}$, or $v$ is complete to $K_j$ for $j \in \{ i-1, i-2, i-3 \}$.
	\end{claim}
	
	Once more, we assume without loss of generality that $K_j$ is nonempty, for all $j \in \{1, \ldots, 6\}$. 
	Given the simmetry of the odd-indexed and even-indexed sets $K_j$, we may also separate in two cases without losing generality: if $v$ is adjacent to $K_1$ and $K_3$ and if $v$ is adjacent to $K_2$ and $K_4$.
	
	Suppose first that $v$ is adjacent to $K_1$ and $K_3$. By Claim \ref{claim:tent_1}, $v$ is anticomplete to $K_5$. 
	If $v$ is nonadjacent to some vertex $w_2$ in $K_2$, then the set $\{ s_{35}$, $v$,  $s_{51}$, $w_1$, $w_{3}$, $w_{5}$, $w_{2} \}$ induces a tent with center. Hence, $v$ is complete to $K_2$.
	
	Suppose now that $v$ is adjacent to $K_2$ and $K_4$.
	First, notice that $v$ is complete to either $K_1$ or $K_5$, for if not we find a $4$-sun induced by $\{ s_{13}$, $s_{51}$, $s_{35}$, $v$, $w_2$, $w_1$, $w_5$, $w_4 \}$.
	Suppose that $v$ is complete to $K_1$. If $v$ is not complete to $K_3$, then $v$ is complete to $K_5$ and $K_6$, for if not there is $M_{III}(3)$ induced by $\{ s_{13}$, $s_{51}$, $v$, $w_1$, $w_3$, $w_4$, $w_j \}$ for $j=5,6$. \QED
	
	\begin{remark}
		As a consequence of the previous claims we also proved that, If $G$ is $\{ \mathcal{T}, \mathcal{F} \}$-free, then:
		\begin{itemize}
			\item For each $i \in \{1, 2, \ldots, 6 \}$, the sets $S_{i,i-2}$ are empty, for if not, there is a vertex $v$ in $S$ such that $v$ is adjacent to $K_1$, $K_3$ and $K_5$ (Claim \ref{claim:tent_1}). Moreover, the same holds for $S_{i(i-2)}$, for each $i \in \{1, 3, 5\}$.
			\item For each $i \in \{2, 4, 6 \}$, the sets $S_{i(i+2)}$ are empty since every vertex $v$ in $S$ such that $v$ is adjacent to $K_i$ and $K_{i+2}$ is necessarily complete to either $K_{i-1}$ or $K_{i+3}$ (Claim \ref{claim:tent_3}).
		\end{itemize}
	\end{remark}			
		
	\begin{claim} \label{claim:tent_4}
		If $G$ is $\{ \mathcal{T}, \mathcal{F} \}$-free, then for each $i \in \{1, 3, 5\}$, every vertex in $S_{i(i+3)} \cup S_{(i+3)i}$ is complete to $K_i$.
	\end{claim}
	
	   We will prove this claim without loss of generality for $i = 1$.

    Let $v$ in $S_{14}$. By definition, $v$ is adjacent to $k_3$ and nonadjacent to $k_5$. Towards a contradiction, let $w_{11}$ and $w_{12}$ in $K_1$ such that $v$ is nonadjacent to $w_{11}$ and $v$ is adjacent to $w_{12}$, and let $w_4$ in $K_4$ such that $v$ is adjacent to $w_4$. In this case, we find $F_0$ induced by the set $\{ s_{13}$, $s_{35}$, $v$, $w_{11}$, $w_{12}$, $k_3$, $w_4$, $k_5 \}$.

    Analogously, if $v$ is in $S_{41}$, then $F_0$ is induced by $\{ s_{35}$, $s_{51}$, $v$, $w_{11}$, $w_{12}$, $k_3$, $w_4$, $k_5 \}$. \QED

The following Lemma is a straightforward consequence of Claims~\ref{claim:tent_1} to~\ref{claim:tent_4}.

\begin{lema} \label{lema:tent_2} 
Let $G=(K,S)$ be a split graph that contains an induced tent. If $G$ is $\{ \mathcal{T}, \mathcal{F} \}$-free, then all the following assertions hold:
 \begin{itemize}
  \item $\{S_{ij}\}_{i,j\in\{1,2,\ldots,6\}}$ is a partition of $S$.
  \item For each $i\in\{1,3,5\}$, $S_{i(i-1)}$ and $S_{i(i-2)}$ are empty.
  \item For each $i\in\{2,4,6\}$, $S_{i(i-1)}$ and $S_{i(i+2)}$ are empty.
  \item For each $i\in\{1,3,5\}$, $S_{i(i+3)}$ and $S_{(i+3)i}$ are complete to $K_i$.
 \end{itemize}
 
\begin{figure}[h]
\begin{center}
	\begin{tabular}{ c | c c c c c c } 
		 \hline
		 $i\setminus j$ & 1 & 2 & 3 & 4 & 5 & 6 \\ 
		  \hline
		 1 & \checkmark & \checkmark & \checkmark & \checkmark & $\emptyset$ & $\emptyset$ \\ 
		 2 & $\emptyset$ & \checkmark & \checkmark & $\emptyset$ & \checkmark & \checkmark \\
 		 3 & $\emptyset$ & $\emptyset$ & \checkmark & \checkmark & \checkmark & \checkmark \\
		 4 & \checkmark & \checkmark & $\emptyset$ & \checkmark & \checkmark & $\emptyset$ \\
		 5 & \checkmark & \checkmark & $\emptyset$ & $\emptyset$ & \checkmark & \checkmark \\
		 6 & \checkmark & $\emptyset$ & \checkmark & \checkmark & $\emptyset$ & \checkmark \\
	\end{tabular}
\caption{The (possibly) nonempty parts of $S$ in the tent case. The orange checkmarks denote those $S_{ij}$ for which every vertex is complete to $K_i$ or $K_j$.}
\end{center}
\end{figure}

\end{lema}

\section{Partitions of S and K for a graph containing an induced $4$-tent} \label{sec:4tent_partition}

Let $G=(K,S)$ be a split graph where $K$ is a clique and $S$ is an independent set. Let $H$ be a $4$-tent induced subgraph of $G$. Let $V(T)=\{k_1,k_2,k_4,k_5,s_{12},s_{24},s_{45}\}$ where $k_1,k_2,k_4,k_5\in K$, $s_{12},s_{24},s_{45}\in S$, and the neighbors of $s_{ij}$ in $H$ are precisely $k_i$ and $k_j$.

We introduce sets $K_1,K_2,\ldots,K_6$ as follows.
\begin{itemize}
 \item Let $K_1$ be the set of vertices of $K$ whose only neighbor in $V(T)\cap S$ is $s_{12}$.
 Analogously, let $K_3$ be the set of vertices of $K$ whose only neighbor in $V(T)\cap S$ is $s_{24}$, and let $K_5$ be the set of vertices of $K$ whose only neighbor in $V(T)\cap S$ is $s_{45}$.
 \item For each $i\in \{2, 4\}$, let $K_i$ be the set of vertices of $K$ whose neighbors in $V(T)\cap S$ are precisely $s_{ji}$ and $s_{ik}$, for $i=2$, $j=1$ and $k=2$ or $i=4$, $j=2$ and $k=5$.
 
 \item Let $K_6$ be the set of vertices of $K$ that are anticomplete to $V(T)\cap S$.
\end{itemize}

\begin{figure}[h!]
\centering
    \includegraphics[scale=1]{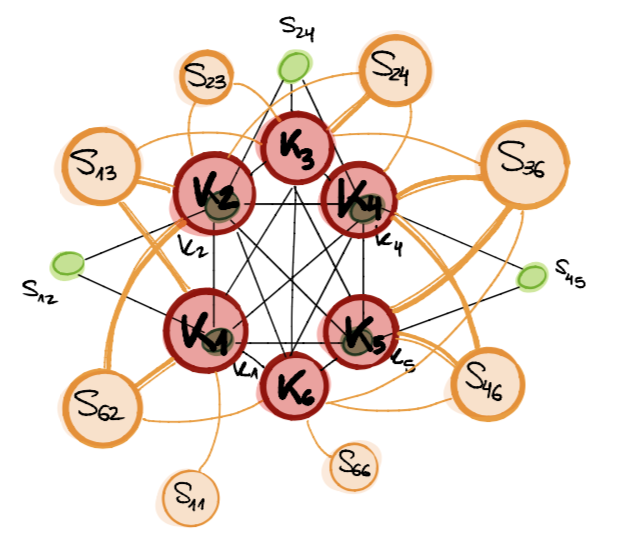}
  \caption{Some of the possible extensions of the $4$-tent graph.}		
  \label{fig:4tent_ext}    
\end{figure}

The following Lemma is straightforward.

\begin{lema}\label{lema:4tent_0}
	If $G$ is $\{ \mathcal{T}, \mathcal{F} \}$-free, then $\{K_1, K_2, \ldots, K_6\}$ is a partition of $K$.
\end{lema}

\begin{proof}
Every vertex in $K$ is adjacent to precisely one, two or no vertices of $V(T)\cap S$, for if not we find a $4$-tent${}\vee{}K_1$. 
\end{proof}

Let $i,j \in \{1, \ldots, 6\}$ and let $S_{ij}$ defined as in the previous section. We denote $S_{[ij}$ (resp.\ $S_{ij]}$) to the set of vertices in $S$ that are adjacent to $K_j$ and complete to $K_i$, $K_{i+1}, \ldots, K_{j-1}$ (resp.\ adjacent to $K_i$ and complete to $K_{i+1}, \ldots, K_{j-1}, K_j$). 
We denote $S_{[ij]}$ to the set of vertices in $S$ that are complete to $K_i, \ldots, K_j$.

 In particular, we consider separately those vertices adjacent to $K_6$ and complete to $K_1, K_2, \ldots, K_5$: we denote $S_{[16}$ to the set that contains these vertices, and $S_{16}$ to the subset of vertices of $S$ that are adjacent but not complete to $K_1$. Furthermore, we consider the set $S_{65}$ as those vertices in $S$ that are adjacent but not complete to $K_5$.  

\begin{claim} \label{claim:4tent_0}
	If $v$ in $S$ fullfils one of the following conditions:
	\begin{itemize}
		\item $v$ is adjacent to $K_i$ and $K_{i+2}$ and is anticomplete to $K_{i+1}$, for $i=1,3$
		\item $v$ is adjacent to $K_1$ and $K_4$ and is anticomplete to $K_2$
		\item $v$ is adjacent to $K_2$ and $K_5$ and is anticomplete to $K_4$  
	\end{itemize}
	Then, there is an induced tent in $G$.
\end{claim} 

If $v$ is adjacent to $K_1$ and $K_3$ and is anticomplete to $K_2$, then we find a tent induced by $\{s_{12}$, $s_{24}$, $v$, $k_1$, $k_2$, $k_3\}$. If instead $v$ is adjacent to $K_3$ and $K_5$ and is anticomplete to $K_4$, then the tent is induced by $\{s_{45}$, $s_{24}$, $v$, $k_3$, $k_4$, $k_5\}$.

If $v$ is adjacent to $K_1$ and $K_4$ and is anticomplete to $K_2$, then we find a tent induced by $\{s_{12}$, $s_{24}$, $v$, $k_1$, $k_2$, $k_4\}$.

Finally, if $v$ is adjacent to $K_2$ and $K_5$ and is anticomplete to $K_4$, then the tent is induced by the set $\{s_{45}$, $s_{24}$, $v$, $k_2$, $k_4$, $k_5\}$. \QED

\vspace{1mm}
As a direct consequence of the previous claim, we will assume without loss of gen\-er\-al\-i\-ty that the subsets $S_{31}$, $S_{41}$, $S_{52}$ and $S_{53}$ of $S$ are empty.

\begin{claim} \label{claim:4tent_1}
	If $G$ is $\{ \mathcal{T}, \mathcal{F} \}$-free, then $S_{51}$ is empty. Moreover, if $K_3 \neq \emptyset$, then $S_{42}$ is empty.
\end{claim} 

Suppose there is a vertex $v$ in $S_{51}$, let $k_1$ in $K_1$ and $k_5$ in $K_5$ be vertices adjacent to $v$. Thus, we find a $4$-sun induced by the set $\{ s_{12}$, $s_{24}$, $s_{45}$, $k_1$, $k_2$, $k_4$, $k_5$, $v \}$.

If $K_3 \neq \emptyset$, suppose $v$ in $S_{42}$, and let $k_2$ in $K_2$, $k_4$ in $K_4$ be vertices adjacent to $v$. Notice that, by definition, $v$ is complete to $K_5$ and $K_1$, and anticomplete to $K_3$. Then, we find $M_V$ induced by the set $\{ s_{12}$, $s_{24}$, $s_{45}$, $k_1$, $k_2$, $k_4$, $k_5$, $k_3$, $v \}$. \QED

We want to prove that $\{S_{ij}\}$ is indeed a partition of $S$, analogously as in the tent case. Towards this purpose, we state and prove the following claims.

\begin{claim} \label{claim:4tent_2}
	If $G$ is $\{ \mathcal{T}, \mathcal{F} \}$-free and $v$ in $S$ is adjacent to $K_i$ and $K_{i+2}$ and anticomplete to $K_j$ for $j < i$ and $j>i+2$, then:
	\begin{itemize}
		\item If $i \equiv 0 \pmod{3}$, then $v$ is complete to $K_{i+1}$ and $K_{i+2}$.
		\item If $i \equiv 1 \pmod 3$, then $v$ is complete to $K_i$ and $K_{i+1}$.
		\item If $i \equiv 2 \pmod{3}$, then $v$ lies in $S_{24}$.
	\end{itemize}
\end{claim}

Let $v$ in $S$ adjacent to some vertices $k_1$ in $K_1$ and $k_3$ in $K_3$, such that $v$ is anticomplete to $K_4$, $K_5$ and $K_6$. By the previous Claim, we know that $v$ is complete to $K_2$ for if not there is an induced tent.
Moreover, suppose that $v$ is not complete to $K_1$. Let $k_2$ in $K_2$, $k_4$ in $K_4$ and let $k'_1$ in $K_1$ be a vertex nonadjacent to $v$. Then, we find $F_0$ induced by $\{ s_{12}$, $s_{24}$, $v$, $k_1$, $k'_1$, $k_2$, $k_3$, $k_4 \}$.
The proof is analogous for $v$ adjacent to $K_3$ and $K_5$, and anticomplete to $K_1$, $K_2$ and $K_6$.

\vspace{1mm}
Let $v$ in $S$ be a vertex adjacent to $k_4$ in $K_4$ and $k_6$ in $K_6$, such that $v$ is anticomplete to $K_1$, $K_2$ and $K_3$ (it is indistinct if $K_3 = \emptyset$). Suppose there is a vertex $k_5$ in $K_5$ nonadjacent to $v$. In this case, we find a net${}\vee{}K_1$ induced by $\{s_{24}$, $s_{45}$, $v$, $k_2$, $k_4$, $k_5 \}$.
Moreover, suppose that $v$ is not complete to $K_4$. Let $k'_4$ in $K_4$ nonadjacent to $v$. Thus, we find $F_0$ induced by $\{ s_{24}$, $s_{45}$, $v$, $k_2$, $k'_4$, $k_4$, $k_5$, $k_6 \}$.
The proof is analogous for $v$ adjacent to $K_6$ and $K_2$, and anticomplete to $K_3$, $K_4$ and $K_5$.

\vspace{1mm}
Finally, we know that in the third statement either $i=2$ or $i=5$.
If $i=5$, then $v$ is a vertex adjacent to $K_5$ and $K_1$ such that $v$ is anticomplete to $K_2$, $K_3$ (if nonempty) and $K_4$. Hence, as a direct consequence of the proof of Claim \ref{claim:4tent_1}, we find a $4$-sun. It follows that there is no such vertex $v$ adjacent to $K_5$ and $K_1$ and thus necessarily $i=2$. 
Let $v$ in $S$ adjacent to $k_2$ in $K_2$ and $k_4$ in $K_4$ such that $v$ is anticomplete to $K_5$, $K_6$ and $K_1$ (it is indistinct if $K_6 = \emptyset$).
If $K_3 \neq \emptyset$, let $k_3$ in $K_3$ and suppose that $v$ is nonadjacent to $K_3$. Then, we find $M_{III}(4)$ induced by $\{s_{12}$, $s_{24}$, $s_{45}$, $v$, $k_1$, $k_2$, $k_4$, $k_5$, $k_3\}$. \QED

\begin{claim} \label{claim:4tent_3}
	If $G$ is $\{ \mathcal{T}, \mathcal{F} \}$-free and $v$ in $S$ is adjacent to $K_i$ and $K_{i+3}$ and $v$ is anticomplete to $K_j$ for $j < i$ and $j>i+3$, then:
	\begin{itemize}
		\item If $i \equiv 0 \pmod{3}$, then $v$ is complete to $K_{i+1}$ and $K_{i+2}$.
		\item If $i \equiv 1 \pmod 3$, then $v$ lies in $S_{14]}$.
		\item If $i \equiv 2 \pmod{3}$, then $v$ lies in $S_{25] }$.
	\end{itemize}
\end{claim}

\begin{proof}
Suppose first  that $i \equiv 0 \pmod{3}$. Let $v$ in $S$ such that $v$ is adjacent to some vertices $k_3$ in $K_3$ and $k_6$ in $K_6$ and $v$ has at least one non-neighbor in $K_1$ and $K_2$ each. Let $k_1$ in $K_1$ and $k_2$ in $K_2$ be two such vertices.
	If there are vertices $k_4$ in $K_4$ and $k_5$ in $K_5$ such that $k_4$ and $k_5$ are both nonadajcent to $v$, then we find $M_{IV}$ induced by the set $\{ s_{12}$, $s_{24}$, $v$, $s_{45}$, $k_1$, $k_2$, $k_6$, $k_3$, $k_5$, $k_4 \}$.
	If instead $v$ is adjacent to a vertex $k_5$ in $K_5$ and $v$ is nonadjacent to a vertex $k_4$ in $K_4$, then we find a tent with center induced by $\{s_{24}$, $v$, $s_{45}$, $k_1$, $k_3$, $k_4$, $k_5\}$. Conversely, if $v$ is adjacent to $k_4$ in $K_4$ and is nonadjacent to some $k_5$ in $K_5$, then we find a net${}\vee{}K_1$ induced by the set $\{ s_{12}$, $s_{24}$, $v$, $s_{45}$, $k_2$, $k_6$, $k_5$, $k_4 \}$. The proof is analogous by symmetry for $v$ in $S_{63}$.
	
	Let us see now the case $i \equiv 1 \pmod 3$, thus either $i = 1$ or $i=4$. If $i=4$, then $v$ is adjacent to $K_4$ and $K_1$ and $v$ is anticomplete to $K_2$ and $K_3$ (if nonempty). Thus, by Claim \ref{claim:4tent_0} we may discard this case.
	Let $v$ in $S$ such that $v$ is adjacent to some vertices $k_1$ in $K_1$, $k_4$ in $K_4$ and $v$ is nonadjacent to a vertex $k_5$ in $K_5$. Suppose that $v$ is not complete to $K_2$ and $K_3$. 
	Whether $K_3 = \emptyset$ or not, if there is a vertex $k_2$ in $K_2$ that is nonadjacent to $v$, then we find a net${}\vee{}K_1$ induced by $\{ s_{24}$, $s_{45}$, $v$, $k_1$, $k_5$, $k_2$, $k_4 \}$.
	If $K_3 \neq \emptyset$, we find a net${}\vee{}K_1$ by replacing the vertex $k_2$ in the previous set for any vertex in $K_3$ that is nonadjacent to $v$. 
	Let us see that $v$ is also complete to $K_4$. If this is not true, then there is a vertex $k'_4$ in $K_4$ nonadjacent to $v$. However, we find $F_0$ induced by $\{ s_{24}$, $s_{45}$, $v$, $k_1$, $k_2$, $k_4$, $k'_4$, $k_5 \}$.
	
	The proof for the third statement is analogous by symmetry.
\end{proof}

\begin{claim} \label{claim:4tent_4}
	If $G$ is $\{ \mathcal{T}, \mathcal{F} \}$-free, $v$ in $S$ is adjacent to $K_i$ and $K_{i+4}$ and $v$ is anticomplete to $K_{i-1}$, then:
	\begin{itemize}
		\item If $i \equiv 0 \pmod{3}$, then $v$ lies in $S_{64]}$.
		\item If $i \equiv 1 \pmod 3$ and $K_3 \neq \emptyset$, then $v$ lies in $S_{15}$. 
		\item If $i \equiv 2 \pmod{3}$, then $v$ lies in $S_{[26}$.
	\end{itemize}
\end{claim}

\begin{proof}
	Suppose first that $i \equiv 0 \pmod{3}$. In this case, either $i=3$ or $i=6$. If $i=3$, then $v$ is adjacent to $K_3$ and $K_1$ and is anticomplete to $K_2$. By Claim \ref{claim:4tent_0}, this case is discarded for there is an induced tent.
	Hence, necessarily $i=6$. Equivalently, $v$ is adjacent to $K_6$ and $K_4$ and $v$ is anticomplete to $K_5$.
	Suppose there is a vertex $k_2$ in $K_2$ nonadjacent to $v$. Then, we find a net${}\vee{}K_1$ induced by $\{ k_2$, $k_5$, $k_6$, $v$, $s_{45}$, $s_{24}$, $k_4 \}$, and thus $v$ must be complete to $K_2$.
	Furthermore, suppose $v$ is not complete to $K_1$, thus there is a vertex $k_1$ in $K_1$ nonadjacent to $v$. Since $v$ is complete to $K_2$, we find $M_{III}(4)$ induced by the set $\{ k_1$, $k_2$, $k_4$, $k_5$, $k_6$, $s_{12}$, $s_{24}$, $s_{45}$, $v \}$.
	If $K_3 \neq \emptyset$, then $v$ is complete to $K_3$, for if not we find a net ${}\vee{}K_1$ induced by $\{ k_3$, $k_5$, $k_6$, $v$, $s_{24}$, $s_{45}$, $k_4 \}$.
	Finally, if $v$ is not complete to $K_4$, then there is a vertex $k'_4$ in $K_4$ nonadjacent to $v$. In this case, we find $F_0$ induced by $\{ k_1$, $k_2$, $k_4$, $k'_4$, $k_5$, $s_{24}$, $s_{45}$, $v \}$ and this finishes the proof of the first statement.
	
	Suppose now that $i \equiv 1 \pmod 3$. Thus, either $i=1$ or $i=4$.
	By hypothesis, $K_3 \neq \emptyset$. Suppose that $i=4$, thus $v$ is adjacent to $K_4$ and $K_2$ and $v$ is anticomplete to $K_3$. By Claim \ref{claim:4tent_1}, if $K_3 \neq \emptyset$, then $v$ is anticomplete to $K_1$ and $K_5$, for if not we find an induced $M_V$. Hence, if $K_3 \neq \emptyset$, then we find $M_{III}(4)$ induced by $\{ k_1$, $k_2$, $k_4$, $k_5$, $s_{12}$, $s_{24}$, $s_{45}$, $v \}$. Thus, if $i=4$, then necessarily $K_3 = \emptyset$.
	Let $i = 1$. Suppose that $v$ is adjacent to $K_1$ and $K_5$ and that $v$ is anticomplete to $K_6$. If $v$ is nonadjacent to some vertices $k_2$ in $K_2$ and $k_4$ in $K_4$, then we find a $4$-sun induced by $\{ k_1$, $k_2$, $k_4$, $k_5$, $s_{12}$, $s_{24}$, $s_{45}$, $v \}$. If $v$ is not complete to $k_2$ in $K_2$, then we find a tent induced by $\{k_2$, $k_4$, $k_5$, $v$, $s_{45}$, $s_{24}\}$. The same holds for $K_4$ by replacing the vertex $k_5$ for some vertex $k_1$ in $K_1$ adjacent to $v$ and $s_{45}$ by $s_{12}$. Notice that it was not necessary for the argument that $K_6 \neq \emptyset$.
	
	Finally, suppose that $i \equiv 2 \pmod{3}$. By Claim \ref{claim:4tent_0} we can discard the case where $i=5$, thus we may assume that $i=2$. However, the proof for $i=2$ is analogous to the proof of the first statement.

\end{proof}

\begin{claim} \label{claim:4tent_5}
	Let $v$ in $S$ such that $v$ is adjacent to at least one vertex in each nonempty $K_i$, for all $i\in \{1, \ldots, 6\}$.
	
	If $G$ is $\{ \mathcal{T}, \mathcal{F} \}$-free, then the following statements hold:
	\begin{itemize}
		\item The vertex $v$ is complete to $K_2$ and $K_4$, regardless of whether $K_3$ and $K_6$ are empty or not.
		\item If $K_3 \neq \emptyset$, then $v$ is complete to $K_3$. 
		\item If $K_6 \neq \emptyset$, then either $v$ is complete to $K_1$ or $v$ is complete to $K_5$. 	
	\end{itemize}
\end{claim} 

\begin{proof}
		Let $k_i$ in $K_i$ be a vertex adjacent to $v$, for each $i=1, 2, 4, 5$, which are always nonempty sets. If $v$ is not complete to $K_2$, then there is a vertex $k'_2$ in $K_2$ nonadjacent to $v$. Thus, we find a tent induced by $\{ s_{12}$, $v$, $s_{24}$, $k'_2$, $k_4$, $ k_1\}$. The same holds if $v$ is not complete to $K_4$, and thus the first statement holds. Notice that, in fact, this holds regardless of $K_3$ or $K_6$ being empty or not.
		
		Suppose now that $K_3 \neq \emptyset$, and that there is a vertex $k_3$ in $K_3$ such that $v$ is nonadjacent to $k_3$. Then, we find $M_V$ induced by $\{ s_{12}$, $s_{45}$, $v$, $s_{24}$, $k_2$, $k_1$, $k_5$, $k_1$, $k_3 \}$.
		
		Finally, let us suppose that $K_6 \neq \emptyset$ and toward a contradiction, let $k'_1$ in $K_1$ and $k'_5$ in $K_5$ be two non-neighbors of $v$, and $k_6$ in $K_6$ adjacent to $v$.	 Then, we find $M_{III}(4)$ induced by $\{ s_{12}$, $s_{24}$, $s_{45}$, $v$, $k'_1$, $k_2$, $k_4$, $k'_5$, $k_6 \}$. Notice that this holds even if $K_3 = \emptyset$. 
\end{proof}

As a consequence of Claims \ref{claim:4tent_0} to \ref{claim:4tent_5}, we have the following Lemma.

\begin{lema} \label{lema:4tent_1} 
Let $G=(K,S)$ be a split graph that contains an induced $4$-tent and contains no induced tent. If $G$ is $\{ \mathcal{T}, \mathcal{F} \}$-free, then all the following assertions hold:
 \begin{itemize}
  \item $\{S_{ij}\}_{i,j\in\{1,2,\ldots,6\}}$ is a partition of $S$.
  \item For each $i\in\{2,3,4,5\}$, $S_{i1}$ is empty.
  \item For each $i\in\{3,4,5\}$, $S_{i2}$ is empty.
  \item The subsets $S_{43}$, $S_{53}$ and $S_{54}$ are empty. 
  \item The following subsets coincide: $S_{13}= S_{[13}$, $S_{14}=S_{14]}$, $S_{25}=S_{[25}$, $S_{26}=S_{[26}$, $S_{35}=S_{35]}$, $S_{46} = S_{[46}$, $S_{62} = S_{62]}$ and $S_{64} = S_{64]}$.
 \end{itemize}
 
\begin{figure}[h!]
	\begin{center}
		\begin{tabular}{ c | c c c c c c } 
			 \hline
			 $i\setminus j$ & 1 & 2 & 3 & 4 & 5 & 6 \\ 
			  \hline
			 1 & \checkmark & \checkmark & \textcolor{dark-orange}{\checkmark} & \textcolor{dark-orange}{\checkmark} & \checkmark & \checkmark \\ 
			 2 & $\emptyset$ & \checkmark & \checkmark & \checkmark & \textcolor{dark-orange}{\checkmark} & \textcolor{dark-orange}{\checkmark} \\
	 		 3 & $\emptyset$ & $\emptyset$ & \checkmark & \checkmark & \textcolor{dark-orange}{\checkmark} & \checkmark \\
			 4 & $\emptyset$ & $\emptyset$ & $\emptyset$ & \checkmark & \checkmark & \textcolor{dark-orange}{\checkmark} \\
			 5 & $\emptyset$ & $\emptyset$ & $\emptyset$ & $\emptyset$ & \checkmark & \checkmark \\
			 6 & \checkmark & \textcolor{dark-orange}{\checkmark} & \checkmark & \textcolor{dark-orange}{\checkmark} & \checkmark & \checkmark \\
		\end{tabular}
	\end{center}
	\caption{The (possibly) nonempty parts of $S$ in the $4$-tent case. The orange checkmarks denote those sets $S_{ij}$ complete to either $K_i$ or $K_j$.} 
\end{figure}

\end{lema}

\section{Partitions of S and K for a graph containing an induced co-$4$-tent} \label{sec:co4tent_partition}

Let $G=(K,S)$ be a split graph where $K$ is a clique and $S$ is an independent set, and suppose that $G$ contains no induced tent or $4$-tent. 
Let $H$ be a co-$4$-tent induced subgraph of $G$. Let $V(T)=\{k_1$,$k_3$,$k_5$,$s_{13}$,$s_{35}$,$s_1$,$s_5\}$ where $k_1,k_3,k_5\in K$, $s_{13}$,$s_{35}$,$s_1$, $s_5$ in $S$ such that the neighbors of $s_{ij}$ in $H$ are precisely $k_i$ and $k_j$ and the neighbor of $s_i$ in $H$ is precisely $k_i$.

We introduce sets $K_1,K_2,\ldots,K_{15}$ as follows.
\begin{itemize}
 \item Let $K_1$ be the set of vertices of $K$ whose only neighbors in $V(T)\cap S$ are $s_1$ and $s_{13}$. Analogously, let $K_5$ be the set of vertices of $K$ whose only neighbors in $V(T)\cap S$ are $s_5$ and $s_{35}$, and let $K_3$ be the set of vertices of $K$ whose only neighbors in $V(T)\cap S$ are $s_{13}$ and $s_{35}$. Let $K_{13}$ be the set of vertices of $K$ whose only neighbors in $V(T)\cap S$ are $s_1$ and $s_5$, $K_{14}$ be the set of vertices of $K$ whose only neighbors in $V(T)\cap S$ are $s_{13}$ and $s_5$ and $K_{15}$ be the set of vertices of $K$ whose only neighbors in $V(T)\cap S$ are $s_1$ and $s_{35}$.
 \item Let $K_2$ be the set of vertices of $K$ whose neighbors in $V(T)\cap S$ are precisely $s_1$, $s_{13}$ and $s_{35}$, and let $K_4$ be the set of vertices of $K$ whose neighbors in $V(T)\cap S$ are precisely $s_5$, $s_{13}$ and $s_{35}$. Let $K_9$ be the set of vertices of $K$ whose neighbors in $V(T)\cap S$ are precisely $s_1$, $s_{13}$ and $s_5$, and let $K_{10}$ be the set of vertices of $K$ whose neighbors in $V(T)\cap S$ are precisely $s_1$, $s_{35}$ and $s_5$.
 \item Let $K_6$ be the set of vertices of $K$ whose only neighbor in $V(T)\cap S$ is precisely $s_{35}$, and let $K_8$ be the set of vertices of $K$ whose only neighbor in $V(T)\cap S$ is precisely $s_{13}$. Let $K_{11}$ be the set of vertices of $K$ whose only neighbor in $V(T)\cap S$ is precisely $s_1$, and let $K_{12}$ be the set of vertices of $K$ whose only neighbor in $V(T)\cap S$ is precisely $s_5$.
 \item Let $K_7$ be the set of vertices of $K$ that are anticomplete to $V(T) \cap S$.
 
\end{itemize}

\begin{remark} \label{obs:co4tent_0}
	If $K_{13} \neq \emptyset$, then there is an induced $4$-sun in $G$.
	If $K_{14} \neq \emptyset$, $K_{15} \neq \emptyset$, $K_{9} \neq \emptyset$ or $K_{10} \neq \emptyset$, then we find an induced tent.
	Moreover, if $K_{11} \neq \emptyset$ or $K_{12} \neq \emptyset$, then we find an induced $4$-tent in $G$.
	Hence, in virtue of the previous chapters, we assume that $K_9, \ldots, K_{15}$ are empty sets.
\end{remark} 

The following Lemma is straightforward.

\begin{lema}\label{lema:co4tent_0}
	Let $G=(K,S))$ be a split graph that contains no induced tent or $4$-tent. If $G$ is $\{ \mathcal{T}, \mathcal{F} \}$-free, then $\{K_1, K_2, \ldots, K_8\}$ is a partition of $K$.
\end{lema}

\begin{proof}
Every vertex in $K$ is adjacent to precisely one, two, three or no vertices of $V(T)\cap S$, for if it is adjacent to every vertex in $V(T) \cap S$, then we find a co-$4$-tent${}\vee{}K_1$. Moreover, by the previous remark, the only possibilities are the sets $K_1, \ldots, K_8$.
\end{proof}

Let $i,j \in \{1, \ldots, 8\}$ and let $S_{ij}$ defined as in the previous sections.

\begin{remark} \label{obs:co4tent_1}
If $K_4 = \emptyset$, then there is a split decomposition of $G$. Let us consider the subset $K_5$ on the one hand, and on the other hand a vertex $u \not\in G$ such that $u$ is complete to $K_5$ and is anticomplete to $V(G) \setminus K_5$. Let $G_1$ and $G_2$ be the subgraphs induced by the vertex subsets $V_1 = V(G) \setminus S_{55}$ and $V_2 = \{u\} \cup K_5 \cup S_{55}$, respectively. Hence, $G$ is the result of the split composition of $G_1$ and $G_2$ with respect to $K_5$ and $u$. The same holds if $K_2 = \emptyset$ considering the subgraphs induced by the vertex subsets $V_1 = V(G) \setminus S_{11}$ and $V_2 = \{u\} \cup K_1 \cup S_{11}$, where in this case $u$ is complete to $K_1$ and is anticomplete to $V(G) \setminus K_1$.
	
	If we consider $H$ a minimally non-circle graph, then $H$ is a prime graph, for if not one of the factors should be non-circle and thus $H$ would not be minimally non-circle \cite{B94}. Hence, in order characterize those circle graphs that contain an induced co-$4$-tent, we will assume without loss of generality that $G$ is a prime graph, and therefore $K_2 \neq \emptyset$ and $K_4 \neq \emptyset$. 
\end{remark}

\begin{figure}[h]
	\centering
	\includegraphics[scale=1]{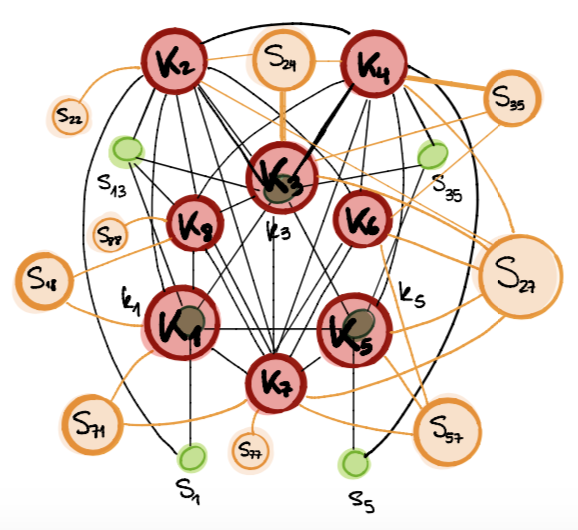}  
	\caption{The partition of $K$ and some of the subsets of $S$ according to the adjancencies with $H$.} 	\label{fig:co4tent_ext} 
\end{figure}

\begin{claim} \label{claim:co4tent_0}
	Let $G=(K,S)$ be a split graph that contains an induced co-$4$-tent. If $v$ in $S$ fullfils one of the following conditions:
	\begin{itemize}
	\item $v$ is adjacent to $K_1$ and $K_5$ and is not complete to $K_3$ (resp.\ $K_2$ or $K_4$)
	\item $v$ is adjacent to $K_1$ and $K_4$ and is not complete to $K_2$
	\item $v$ is adjacent to $K_2$ and $K_5$ and is not complete to $K_4$ 
	\item $v$ is adjacent to $K_i$ and $K_{i+2}$ and is not complete to $K_{i+1}$, for $i=1,3$ 
	\item $v$ is adjacent to $K_2$ and $K_4$ and is not complete to $K_1$ and $K_5$
	\item $v$ is adjacent to $K_1$, $K_3$ and $K_5$ and is not complete to $K_2$ and $K_4$
	\end{itemize}
	Then, $G$ contains either an induced tent or $4$-tent.
\end{claim} 

\begin{proof}
If $v$ is adjacent to $K_1$ and $K_5$ and is not complete to $K_3$, then we find a tent induced by $\{s_{13}$, $s_{35}$, $v$, $k_1$, $k_3$, $k_5\}$. If instead it is not complete to $K_2$, then we find a tent induced by $\{s_{1}$, $s_{35}$, $v$, $k_1$, $k_2$, $k_5\}$. It is analogous by symmetry if $v$ is not complete to $K_4$.
If $v$ is adjacent to $K_1$ and $K_4$ and is not complete to $K_2$, then the tent is induced by $\{s_{1}$, $s_{35}$, $v$, $k_1$, $k_2$, $k_4\}$.
It is analogous by symmetry if $v$ is adjacent to $K_2$ and $K_5$ and is not complete to $K_4$.
If $v$ is adjacent to $K_1$ and $K_3$ and is not complete to $K_2$, then we find a tent induced by $\{s_{1}$, $s_{35}$, $v$, $k_1$, $k_2$, $k_3\}$. It is analogous by symmetry if $i=3$.
If $v$ is adjacent to $K_2$ and $K_4$ and is not complete to $K_1$ and $K_5$, then we find a 4-tent induced by the set $\{s_{1}$, $s_{5}$, $v$, $k_1$, $k_2$, $k_4$, $k_5\}$.
Finally, if $v$ is adjacent to $K_1$, $K_3$ and $K_5$ and is not complete to $K_2$ and $K_4$, then there are tents induced by the sets $\{s_{13}$, $s_{5}$, $v$, $k_3$, $k_4$, $k_5\}$ and $\{s_{35}$, $s_{1}$, $v$, $k_1$, $k_2$, $k_3 \}$. 
\end{proof}

It follows from the previous claim that the following subsets are empty:
$S_{51}$, $S_{52}$, $S_{53}$, $S_{41}$, $S_{31}$, $S_{24}$.

Moreover, the following subsets coincide:
$S_{54} = S_{54]}$, $S_{42} = S_{42]}$, $S_{43} = S_{[43]}$, $S_{32} = S_{32]}$.

\begin{claim} \label{claim:co4tent_1} 
	If there is a vertex $v$ in $S$ such that $v$ belongs to either $S_{61}$, $S_{71}$, $S_{81}$, $S_{56}$, $S_{57}$, $S_{58}$, $S_{67}$, $S_{68}$ or $S_{78}$, then there is an induced tent or a $4$-tent in $G$.
\end{claim} 
\begin{proof}
	If $v$ in $S_{61}$, then we find a tent induced by the set $\{s_{1}$, $s_{35}$, $v$, $k_1$, $k_2$, $k_6\}$. If $v$ in $S_{71}$, then we find a $4$-tent induced by the set $\{s_{1}$, $s_{35}$, $v$, $k_7$, $k_1$, $k_2$, $k_3\}$. If $v$ in $S_{81}$, then we find a $4$-tent induced by the set $\{s_{1}$, $s_{35}$, $v$, $k_8$, $k_1$, $k_2$, $k_4\}$.  
If $v$ in $S_{56}$, then we find a $4$-tent induced by the set $\{s_{13}$, $s_{5}$, $v$, $k_3$, $k_4$, $k_5$, $k_6\}$. It is analogous for $v$ in $S_{57}$, swaping $k_6$ for $k_7$. If $v$ in $S_{58}$, then we find a $4$-tent induced by the set $\{s_{35}$, $s_{1}$, $v$, $k_1$, $k_2$, $k_5$, $k_8\}$.
If $v$ in $S_{67}$, then we find a $4$-tent induced by the set $\{s_{13}$, $s_{35}$, $v$, $k_1$, $k_2$, $k_6$, $k_7\}$. If $v$ in $S_{68}$, then we find a 4-tent induced by the set $\{s_{13}$, $s_{35}$, $v$, $k_1$, $k_5$, $k_6$, $k_7\}$. Finally, If $v$ in $S_{78}$, then we find a $4$-tent induced by the set $\{s_{13}$, $s_{5}$, $v$, $k_4$, $k_5$, $k_7$, $k_8\}$. 
\end{proof}

\vspace{1mm}
As a direct consequence of the previous claims, we will assume without loss of generality that the following subsets are empty:
$S_{5i}$ for $i=1, 2, 3, 6, 7, 8$, $S_{4i}$ for $i=1, 2$, $S_{6i}$ for $i=1, 7, 8$, $S_{7i}$ for $i=1, 8$, $S_{81}$, $S_{31}$ and $S_{24}$.

\begin{claim} \label{claim:co4tent_2} 
	If $G$ is $\{ \mathcal{T}, \mathcal{F} \}$-free and contains no induced tent or $4$-tent, then $S_{64}= \emptyset$, $S_{54} = S_{54]}$ and $S_{65}=S_{65]}$. Moreover, if $K_6 \neq \emptyset$, then $S_{54} = S_{[54]}$.
\end{claim} 
\begin{proof}
	Let $v$ in $S_{64}$, $k_i$ in $K_i$ for $i=1, 4, 5, 6$ such that $v$ is adjacent to $k_1$, $k_4$ and $k_6$ and is nonadjacent to $k_5$. Hence, we find $M_{II}(4)$ induced by the vertex set $\{ k_1$, $k_4$, $k_5$, $k_6$, $v$, $s_5$, $s_{13}$, $s_{35} \}$. Hence, $S_{64} = \emptyset$. 
Notice that this also implies that, if $K_6 \neq \emptyset$, then every vertex in $S_{54}$ or $S_{65}$ is complete to $K_5$. Suppose now that $v$ lies in $S_{54}$ and is not complete to $K_4$. Thus, there is a vertex $k_4$ in $K_4$ such that $v$ is nonadjacent to $k_4$. Let $k_1$ in $K_1$ and $k_5$ in $K_5$ such that $k_1$ and $k_5$ are adjacent to $v$. Hence, we find a tent induced by $\{ k_1$, $k_4$, $k_5$, $v$, $s_{13}$, $s_{35} \}$. 
\end{proof}

\vspace{1mm}
As a consequence of the previous claim, we will assume througout the proof that $S_{54} = \emptyset$. This follows from the fact that the vertices in $S_{54}$ are exactly those vertices in $S_{76}$ that are complete to $K_6$ and $K_7$, since the endpoints of both subsets coincide. The same holds for the vertices in $S_{65}$, which are those vertices in $S_{76}$ that are complete to $K_7$.

\vspace{1mm}
We want to prove that $\{S_{ij}\}$ is indeed a partition of $S$. Towards this purpose, we need the following claims.

\begin{claim} \label{claim:co4tent_3} 
	If $G$ is $\{ \mathcal{T}, \mathcal{F} \}$-free and $v$ in $S$ is adjacent to $K_i$ and $K_{i+2}$ and anticomplete to $K_j$ for $j < i$ and $j>i+2$, then $v$ is complete to $K_{i+1}$.
\end{claim}

\begin{proof}
	Once discarded the subsets of $S$ that induce a tent or $4$-tent and those that are empty, the remaining cases are $i=4, 8$. 

Let $v$ in $S$ adjacent to $k_4$ in $K_4$ and $k_6$ in $K_6$, and suppose there is a vertex $k_5$ in $K_5$ nonadjacent to $v$. Then, we find a net${}\vee K_1$ induced by $\{k_6$, $k_4$, $k_5$, $k_3$, $v$, $s_5$, $s_{13}\}$. 

If instead $v$ in $S$ is adjacent to $k_8$ in $K_8$ and $k_2$ in $K_2$ and is nonadjacent to some $k_1$ in $K_1$, then we find a net${}\vee K_1$ induced by $\{k_8$, $k_1$, $k_5$, $k_2$, $v$, $s_1$, $s_{35}\}$.
\end{proof}

\begin{claim} \label{claim:co4tent_4} 
	If $G$ is $\{ \mathcal{T}, \mathcal{F} \}$-free and $v$ in $S$ is adjacent to $K_i$ and $K_{i+3}$ and anticomplete to $K_j$ for $j < i$ and $j>i+3$, then:
	\begin{itemize}
		\item If $i \equiv 0 \pmod{3}$, then $v$ lies in $S_{36}$.
		\item If $i \equiv 1 \pmod 3$, then $v$ lies in $S_{[14}$.
		\item If $i \equiv 2 \pmod{3}$, then $v$ lies in $S_{25]}$ or $S_{83}$.
	\end{itemize}
\end{claim}

\begin{proof}
	Suppose first that $i \equiv 0 \pmod{3}$. Equivalently, either $i=3$ or $i= 6$. 
	Let $v$ in $S$ such that $v$ is adjacent to some vertices $k_3$ in $K_3$ and $k_6$ in $K_6$. If there is a vertex $k_4$ in $K_4$ nonadjacent to $v$ and a vertex $k_5$ in $K_5$ adjacent to $v$, then we find a tent induced by $\{ k_3$, $k_4$, $k_5$, $v$, $s_{13}$, $s_5 \}$. If instead $k_5$ is nonadjacent to $v$, then we find a $4$-tent induced by $\{k_3$, $k_6$, $k_4$, $k_5$, $v$, $s_5$, $s_{13}\}$.
	If instead $v$ is nonadjacent to $k_4$ and is adjacent to $k_5$, then we find a $4$-tent induced by $\{k_1$, $k_4$, $k_5$, $k_6$, $v$, $s_5$, $s_{13}\}$. Hence, $v$ is complete to $K_4$ and $K_5$.
	 If $i=6$, since $v$ is anticomplete to $K_2$, then we find a tent induced by $\{k_6$, $k_1$, $k_1$, $v$, $s_1$, $s_{35}\}$.

	Let us prove the second statement. If $i \equiv 1 \pmod 3$, then either $i = 1$, $i=4$ or $i=7$.
	First, we need to see that $v$ is complete to $K_{i+1}$ and $K_{i+2}$. If $i=4,7$, then $K_7 \neq \emptyset$. If $i=4$, then there are vertices $k_4 \in K_4$ and $k_7 \in K_7$ adjacent to $v$. Suppose that $v$ is nonadjacent to some vertex in $K_5$. Then, we find a net${}\vee K_1$ induced by $\{k_3$, $k_4$, $k_5$, $k_7$, $v$, $s_5$, $s_{13}\}$. If instead there is a vertex $k_6 \in K_6$ nonadjacent to $v$, then there is a net${}\vee K_1$ induced by $\{k_1$, $k_4$, $k_6$, $k_7$, $v$, $s_{35}$, $s_{13}\}$. It is analogous by symmetry if $i=7$.
	However, by Claim \ref{claim:co4tent_2}, $S_{47}$ and $S_{72}$ are empty sets.
	Suppose now that $i=1$, let $k_1$ in $K_1$ and $k_4$ in $K_4$ be vertices adjacent to $v$ and $k_3$ in $K_3$ nonadjacent to $v$. Then, we find $M_{II}(4)$ induced by $\{k_1$, $k_4$, $k_5$, $k_3$, $v$, $s_{35}$, $s_{13}$, $s_5 \}$. It is analogous if $v$ is nonadjacent to some vertex in $K_2$. Notice that, if $v$ is not complete to $K_1$, we find a $4$-tent induced by $\{k_1$, $k_1'$, $k_4$, $k_5$, $v$, $s_5$, $s_1 \}$.

	 Finally, suppose that $i \equiv 2 \pmod{3}$. Hence, either $i=2, 5, 8$.
	 Suppose $i=2$. Let $k_2$ in $K_2$ and $k_5$ in $K_5$ be vertices adjacent to $v$, and let $k_3$ in $K_3$ and $k_4$ in $K_4$. If $k_4$ is nonadjacent to $v$, then we find a $4$-tent induced by $\{k_1$, $k_2$, $k_4$, $k_5$, $v$, $s_5$, $s_1 \}$. Hence, $v$ is complete to $K_4$. If instead $v$ is nonadjacent to $k_3$, then we find $M_{II}(2)$ induced by $\{k_1$, $k_2$, $k_3$, $k_5$, $v$, $s_1$, $s_{35}$, $s_{13}\}$ and therefore $v$ lies in $S_{25}$.
	 
	 Suppose now that $i=5$. Notice that, in this case, there is no vertex $v$ adjacent to $K_5$ and $K_8$ such that $v$ is anticomplete to $K_1, \ldots, K_4$, since in that case we find a tent induced by $\{k_5$, $k_8$, $k_4$, $v$, $s_5$, $s_{13}\}$. Hence, we discard this case.
	 
	 Suppose that $i=8$. Let $k_3$ in $K_3$ and $k_8$ in $K_8$ adjacent to $v$, and let $k_1$ in $K_1$ and $k_2$ in $K_2$. If both $k_1$ and $k_2$ are nonadjacent to $v$, then we find a $4$-tent induced by $\{k_8$, $k_1$, $k_2$, $k_3$, $v$, $s_1$, $s_{35} \}$. Hence, either $v$ is complete to $K_1$ or $K_2$. If $k_1$ is nonadjacent to $v$, then we find a net${}\vee K_1$ induced by $\{k_8$, $k_1$, $k_2$, $k_5$, $v$, $s_1$, $s_{35} \}$. If instead $k_2$ is nonadjacent to $v$, then we find a tent induced by $\{k_1$, $k_2$, $k_5$, $v$, $s_1$, $s_{35} \}$, and therefore $v$ lies in $S_{83}$.
\end{proof}

\begin{claim} \label{claim:co4tent_5} 
	If $G$ is $\{ \mathcal{T}, \mathcal{F} \}$-free and $v$ in $S$ is adjacent to $K_i$ and $K_{i+4}$ and anticomplete to $K_j$ for $j<i$ and $j>i+4$, then either $v$ lies in $S_{15}$ (or $S_{51}$ if $K_6, K_7, K_8 = \emptyset$), or $S_{26}$ or $S_{84}$.
\end{claim}
\begin{proof}
	Notice that, if $v$ is adjacent to $k_3$ in $K_3$ and $k_7$ in $K_7$ and nonadjacent to $k_2$ in $K_2$, then $v$ is complete to $K_1$ for if not we find $4$-tent induced by $\{k_7$, $k_1$, $k_2$, $k_3$, $v$, $s_1$, $s_{35} \}$. However, if $k_1$ in $K_1$ is adjacent to $v$, then we find a tent induced by $\{ k_1$, $k_2$, $k_3$, $v$, $s_1$, $s_{35} \}$. Hence, we discard this case.
	Suppose $i=4$. If $v$ is adjacent to $k_4$ in $K_4$ and $k_8$ in $K_8$ and is nonadjacent to $k_3$ in $K_3$ and $k_5$ in $K_5$, then we find $M_{II}(4)$ induced by $\{k_8$, $k_3$, $k_4$, $k_5$, $v$, $s_5$, $s_{13}$, $s_{35} \}$. However, if $k_5$ is adjacent to $v$, then we find a tent induced by $\{ k_3$, $k_5$, $k_8$, $v$, $s_{13}$, $s_{35} \}$. 
	Suppose $i=5$. Let $k_5$ in $K_5$ and $k_1$ in $K_1$ are adjacent to $v$ and let $k_4$ in $K_4$ nonadjacent to $v$. Thus, we find a tent induced by $\{ k_1$, $k_4$, $k_5$, $v$, $s_5$, $s_{13} \}$. 
	Suppose $i=6$. If $k_6$ in $K_6$ and $k_2$ in $K_2$ are adjacent to $v$, and $k_4$ in $K_4$ and $k_5$ in $K_5$ are nonadjacent to $v$, then we find a $4$-tent induced by $\{ k_4$, $k_5$, $k_6$, $k_2$, $v$, $s_5$, $s_{13} \}$. 
	Suppose $i=7$. Let $k_7$ in $K_7$ and $k_3$ in $K_3$ adjacent to $v$, and $k_4$ in $K_4$ and $k_5$ in $K_5$ nonadjacent to $v$. Thus, we find a $4$-tent induced by $\{ k_7$, $k_3$, $k_4$, $k_5$, $v$, $s_5$, $s_{13} \}$. 
	Suppose $i=8$. Let $k_8$ in $K_8$ and $k_4$ in $K_4$ adjacent to $v$ and let $k_j$ in $K_j$ for $j=1,2,3$. If $k_j$ is nonadjacent to $v$, then we find $M_{II}(4)$ induced by $\{ k_5$, $k_4$, $k_8$, $k_j$,  $v$, $s_5$, $s_{13}$, $s_{35} \}$, for each $j=1, 2, 3$. Hence, $v$ lies in $S_{84}$.
	Suppose $i=1$. Let $k_1$ in $K_1$ and $k_5$ in $K_5$ be adjacent to $v$, and $k_j$ in $K_j$ for $i=2,3,4$. If $v$ is nonadjacent to $k_j$, then we find a tent induced by $\{ k_1$, $k_3$, $k_5$, $v$, $s_{35}$, $s_{13} \}$. Hence, if $K_6, K_7, K_8 = \emptyset$, then $v$ lies in $S_{15}$ or $S_{51}$, and if $K_j \neq \emptyset$ for any $j=6,7,8$, then $v$ lies in $S_{15}$.
	Finally, suppose $i=2$. Let $k_2$ in $K_2$ and $k_6$ in $K_6$ adjacent to $v$, and let $k_j$ in $K_j$ for $j=3,4,5$. If $v$ is nonadjacent to both $k_4$ and $k_5$, then we find a $4$-tent induced by $\{ k_2$, $k_4$, $k_5$, $k_6$, $v$, $s_5$, $s_{13} \}$. Thus, either $v$ is complete to $K_4$ or $K_5$. If $v$ is complete to $K_5$ and not complete to $K_4$, then we find a tent induced by $\{ k_2$, $k_4$, $k_5$, $v$, $s_5$, $s_{13} \}$. If instead $v$ is complete to $K_4$ and not complete to $K_5$, then we find a net${}\vee K_1$ induced by $\{ k_1$, $k_6$, $k_4$, $k_5$, $v$, $s_5$, $s_{13} \}$.
	If $k_3$ is nonadjacent to $v$, then we find $M_{II}(k)$ induced by $\{ k_1$, $k_3$, $k_2$, $k_6$, $v$, $s_1$, $s_{35}$, $s_{13} \}$. Hence, $v$ lies in $S_{26}$.
\end{proof}

\begin{claim} \label{claim:co4tent_6} 
	If $G$ is $\{ \mathcal{T}, \mathcal{F} \}$-free, then the sets $S_{i1}$ for $i=2, 3, 4$, $S_{ij}$ for $j=2, 3, 4, 5$ and $i=j+1, \ldots, 7$, $S_{i7}$ for $i= 3,4$ and $S_{i8}$ for $i= 2,3,4$ are empty, unless $v$ in $S_{[32]}$ or $S_{21}=S_{[21]}$. 
\end{claim} 

\begin{proof}
	Let $v$ in $S_{i1}$ for $i=2,3,4$. If $i=2$ and $v$ is not complete to every vertex in $K$, then there is either a vertex $k_1$ in $K_1$ or a vertex in $k_2$ in $K_2$ that are nonadjacent to $v$.
	  Suppose there is such a vertex $k_2$, and let $k_1'$ in $K_1$ adjacent to $v$. Thus, we find a tent induced by $\{ v$, $s_1$, $s_{35}$, $k_1'$, $k_2$, $k_3 \}$. Similarly, we find a tent if there is a vertex in $K_1$ nonadjacent to $v$.
	  If $i=3$, then we find a tent induced by $\{ v$, $s_{13}$, $s_{35}$, $k_3'$, $k_5$, $k_3 \}$ where $k_3$, $k'_3 \in K_3$, $k_3$ is adjacent to $v$ and $k'_3$ is nonadjacent to $v$. Similarly, we find a tent if $i=4$ considering two analogous vertices $k_4$ and $k'_4$ in $K_4$.

	Let $v$ in $S_{i2}$ for $i=3, 4, 5, 6, 7$ and let us assume in the case where $i=3$ that $v$ is not complete to every vertex in $K$. Thus, there is a vertex $k_3$(or maybe a vertex $k_2$ in $K_2$ if $i=3$, which is indistinct to this proof) in $K_3$ that is nonadjacent to $v$. If $i=3,4,5,6$, then there are vertices $k_1$ in $K_1$ and $k_5$ in $K_5$ (resp.\ $k_6$ in $K_6$ if $i=6$) such that $k_1$ and $k_5$ (resp.\ $k_6$) are adjacent to $v$. Hence, we find a tent induced by $\{ v$, $s_{13}$, $s_{35}$, $k_1$, $k_3$, $k_5 (k_6) \}$.
	If instead $i=7$, then we find a $4$-tent induced by $\{ v$, $s_{13}$, $s_{35}$, $k_1$, $k_3$, $k_5$, $k_7 \}$, where $k_l$ in $K_l$ is adjacent to $v$ for $l=1,7$ and $k_n$ in $K_n$ is nonadjacent to $v$ for $n=3,5$.
	
	Let $v$ in $S_{i7}$ for $i=3,4$. In either case, there are vertices $k_1$ in $K_1$ and $k_2$ in $K_2$ nonadjacent to $v$, and vertices $k_l$ in $K_l$ for $l=4,5,7$ adjacent to $v$.
	Thus, we find $F_0$ induced by $\{ v$, $s_{13}$, $s_{35}$, $k_1$, $k_2$, $k_4$, $k_5, k_7 \}$.
	
	Finally, let $v$ in $S_{i8}$ for $i=2, 3, 4$. Suppose first that $i=3,4$ or that $v$ is not complete to $K_2$, thus there is a vertex $k_2$ in $K_2$ nonadjacent to $v$. In that case, we find a tent induced by $\{ v$, $s_{13}$, $s_{35}$, $k_2$, $k_5$, $k_8 \}$. If instead $v$ in $S_{28}$ and is complete to $K_2$, then we find $M_{II}(4)$ induced by $\{ v$, $s_1$, $s_{13}$, $s_{35}$, $k_1$, $k_2$, $k_5$, $k_8 \}$.
\end{proof}

\begin{remark} \label{obs:coinciden_los_completos}
It follows from the previous proof that $S_{32} = S_{[32]}$ and $S_{21} = S_{[21]}$. 
\end{remark}


\begin{claim} \label{claim:co4tent_7} 
	If $G$ is $\{ \mathcal{T}, \mathcal{F} \}$-free, $K_6 \neq \emptyset$ and $K_8 \neq \emptyset$, then $S_{15} = \emptyset$.
	Moreover, if $K_8 = \emptyset$, then $S_{15} = S_{[15}$, and if $K_6 = \emptyset$, then $S_{15}=S_{15]}$.
\end{claim} 

\begin{proof}
Let $v$ in $S_{15}$, and $k_6$ in $K_6$ and $k_8$ in $K_8$ be vertices nonadjacent to $v$.
Since there are vertices $k_1$ in $K_1$, $k_2$ in $K_2$ and $k_5$ in $K_5$ adjacent to $v$, then we find $F_0$ induced by $\{ v$, $s_{13}$, $s_{35}$, $k_8$, $k_1$, $k_2$, $k_5$, $k_6 \}$.

The proof is analogous if $K_8 = \emptyset$ (resp.\ if $K_6 = \emptyset$) considering two vertices $k_{11}$ and $k_{12}$ in $K_1$ ($k_{51}$, $k_{52}$ in $K_5$) such that $v$ is adjacent to $k_{11}$ (resp.\ $k_{51}$) and is nonadjacent to $k_{22}$ (resp.\ $k_{52}$).
\end{proof}


\begin{claim} \label{claim:co4tent_8} 
	Let $v$ in $S_{ij}$ such that $v$ is adjacent to at least one vertex in each nonempty $K_l$, for every $l\in \{1, \ldots, 8\}$.
	If $G$ is $\{ \mathcal{T}, \mathcal{F} \}$-free, then the following statements hold:
	\begin{itemize}
		\item The vertex $v$ is complete to $K_2$, $K_3$ and $K_4$.
		\item If $K_j \neq \emptyset$ for some $j=6,8$, then $v$ is complete to $K_5$. Moreover, $v$ is either complete to $K_i$ or $K_j$. 
	\end{itemize}
\end{claim} 

\begin{proof}
 The first statement follows as a direct consequence of Claim \ref{claim:co4tent_0}: if $v$ is adjacent to $K_1$, $K_3$ and $K_5$, then $v$ is complete to $K_2$ and $K_4$. Moreover, $v$ is complete to $K_3$. 
 

To prove the second statement, suppose first  that $K_6 \neq \emptyset$ and $K_7, K_8 = \emptyset$. Let us see that $v$ is complete to $K_5$. Suppose there is a vertex $k_{5}$ in $K_5$ such that $v$ is nonadjacent to $k_{5}$, and let $k_i$ in $K_i$ adjacent to $v$ for each $i=1,4,6$. We find $M_{II}(4)$ induced by $\{ v$, $s_{13}$, $s_{35}$, $s_5$, $k_1$, $k_4$, $k_5$, $k_6 \}$.  
The proof is analogous if $K_8 \neq \emptyset$ and $K_7, K_6 = \emptyset$.

Let us suppose now that $v$ is not complete to $K_1$ and $K_6$. We find $F_0$ induced by $\{ v$, $s_{13}$, $s_{35}$, $k_{11}$, $k_{12}$, $k_3$, $k_{61}$, $k_{62} \}$, where $k_{1j}$ in $K_1$, $k_{6j}$ in $K_6$ for each $j=1,2$ and $v$ is adjacent to $k_{i1}$ and is nonadjacent to $k_{i2}$ for each $i=1,6$. The proof is analogous if $K_8 \neq \emptyset$ and $K_7, K_6 = \emptyset$ and if $K_6, K_8 \neq \emptyset$, independently on whether $K_7=\emptyset$ or not. 


\end{proof}

By simplicity, we will also consider that every vertex in $S_{[32]}$ and $S_{[21]}$ lies in $S_{76]}$, and that in particular, if $K_7 \neq \emptyset$, then such vertices are complete to $K_7$. This follows from Claim \ref{claim:co4tent_8} and Remark \ref{obs:coinciden_los_completos}.
As a consequence of Claims \ref{claim:co4tent_0} to \ref{claim:co4tent_8}, we have the following Lemma.

\begin{lema} \label{lema:co4tent_1} 
Let $G=(K,S)$ be a split graph that contains an induced co-$4$-tent and contains no induced tent or $4$-tent. If $G$ is $\{ \mathcal{T}, \mathcal{F} \}$-free, then all the following assertions hold:
 \begin{itemize}
  \item $\{S_{ij}\}_{i,j\in\{1,2,\ldots,8\}}$ is a partition of $S$.
  \item For each $i\in\{2,3,4,5,6,7,8 \}$, $S_{i1}$ is empty.
  \item For each $i\in\{3,4,5,6,7 \}$, $S_{i2}$ is empty.
  \item For each $i\in\{4,5,6,7 \}$, $S_{i3}$ is empty, and $S_{56}$ is also empty.
  \item For each $i\in\{3,4,5,6 \}$, $S_{i7}$ is empty.
  \item For each $i\in\{2,3,4,5,6,7 \}$, $S_{i8}$ is empty.
  \item The subsets $S_{64}$, $S_{54}$ and $S_{56}$ are empty.
  \item The following subsets coincide: $S_{1i}= S_{[1i}$ for $i=3,4,8$;
  $S_{16}=S_{16]}$, $S_{25}=S_{25]}$, $S_{27}=S_{[27}$, $S_{35}=S_{35]}$, $S_{46} = S_{[46}$, $S_{82} = S_{82]}$ and $S_{85} = S_{[85}$ (as the case may be, according to whether $K_i \neq \emptyset$ or not, for $i=6,7,8$).
 \end{itemize}
 
\end{lema}

Since $S_{18} = S_{[18}$, we will consider these vertices as those in $S_{87}$ that are complete to $K_7$ and $S_{18} = \emptyset$. Moreover, those vertices that are complete to $K_1, \ldots, K_6, K_8$ and are adjacent to $K_7$ will be considered as in $S_{76]}$, thus $S_{87}$ is the set of independent vertices that are complete to $K_1, \ldots, K_7$ and are adjacent but not complete to $K_8$. These results are summarized in Figure~\ref{fig:tabla_co4tent_1}.
 
\begin{figure}[h!]	 
\begin{center}
	\begin{tabular}{ c | c c c c c c c c} 
		 \hline
		 $i\setminus j$ & 1 & 2 & 3 & 4 & 5 & 6 & 7 & 8 \\ 
		  \hline
		 1 & \checkmark & \checkmark & \textcolor{orange}{\checkmark} & \textcolor{orange}{\checkmark} & $\emptyset$ & \textcolor{orange}{\checkmark} & \checkmark & $\emptyset$ \\ 
		 2 & $\emptyset$ & \checkmark & \checkmark & $\emptyset$ & \textcolor{orange}{\checkmark} & \checkmark & \textcolor{orange}{\checkmark} & $\emptyset$ \\
 		 3 & $\emptyset$ & $\emptyset$ & \checkmark & \checkmark & \textcolor{orange}{\checkmark} & \checkmark & $\emptyset$ & $\emptyset$ \\
		 4 & $\emptyset$ & $\emptyset$ & $\emptyset$ & \checkmark & \checkmark & \textcolor{orange}{\checkmark} & $\emptyset$ & $\emptyset$ \\
		 5 & $\emptyset$ & $\emptyset$ & $\emptyset$ & $\emptyset$ & \checkmark & $\emptyset$ & $\emptyset$ & $\emptyset$ \\
		 6 & $\emptyset$  & $\emptyset$  & $\emptyset$  & $\emptyset$  & $\emptyset$ & \checkmark & $\emptyset$  &  $\emptyset$  \\
		7 & $\emptyset$ & $\emptyset$ & $\emptyset$ & \textcolor{orange}{\checkmark} & \checkmark & \checkmark & \checkmark & $\emptyset$ \\
		8 & $\emptyset$ & \textcolor{orange}{\checkmark} & \checkmark & \checkmark & \textcolor{orange}{\checkmark} & \checkmark & \checkmark & \checkmark \\
	\end{tabular}
\end{center} 
\caption{The (possibly) nonempty parts of $S$ in the co-$4$-tent case. The orange checkmarks denote those subsets $S_{ij}$ that are either complete to $K_i$ or $K_j$.} \label{fig:tabla_co4tent_1}
\end{figure}


\selectlanguage{spanish}%
\chapter*{Matrices 2-nested}

En este capítulo se definen y caracterizan las matrices $2$-nested, las cuales son fun\-da\-men\-ta\-les para describir cada porción del modelo circle para aquellos grafos split que también son circle. 
Los resultados de este capítulo son cruciales para la prueba del resultado más importante del próximo capítulo, el cual da una caracterización completa de aquellos grafos split circle por subgrafos inducidos prohibidos.
La definición de matriz $2$-nested es muy técnica, por lo cual omitiré los detalles en este resumen. Sin embargo, daré una versión reducida de la definición y un ejemplo para entender un poco qué motivó el estudio de estas matrices.

Las matrices enriquecidas se definen como aquellas matrices de $0$'s y $1$'s para las cuales algunas de sus filas cuentan con una etiqueta L, R o LR y algunas de ellas a su vez tienen asignado el color rojo o azul.
A su vez, dada una matriz enriquecida $A$, decimos que $A$ es una matriz LR-ordenable si admite un ordenamiento $\Pi$ que cumpla con: (1) la propiedad de los unos consecutivos para sus filas, (2) la tira de $1$'s de aquellas filas etiquetadas con L (resp.\ con R) comienzan en la primera columna (resp.\ terminan en la última columna) y (3) aquellas filas etiquetadas con LR, o bien quedan como ordenadas como una única tira de $1$'s que comienza en la primera columna o termina en la última columna, o bien como dos tiras de $1$'s separadas, una que comienza en la primera columna y otra que termina en la última columna.  
Esto nos lleva a considerar los bloques de una matriz enriquecida, una vez ordenada según $\Pi$. Para cada fila de $A$ etiquetada con L o LR y que tiene un $1$ en la primera columna del orden $\Pi$, llamamos L-bloque al subconjunto maximal de columnas consecutivas empezando desde la primera en la cual la fila tiene un $1$. Los R-bloques se definen de manera análoga. Para toda fila sin etiqueta, su U-bloque es el conjunto maximal de columnas consecutivas en las que hay un $1$ en esa fila.

Las matrices $2$-nested son aquellas matrices LR-ordenables que admiten un bi-coloreo total de sus bloques que además cumple con una lista de 9 propiedades.
Estas propiedades se derivan de las condiciones necesarias que se requieren para determinar la ubicación de las cuerdas en una cierta porción del modelo círculo para un grafo split circle.
Para entender un poco más sobre estas condiciones, consideremos el grafo split $G=(K,S)$ de la Figura \ref{fig:example_graph1_}. 

\begin{figure}[h!] 	\centering
	\includegraphics[scale=.9]{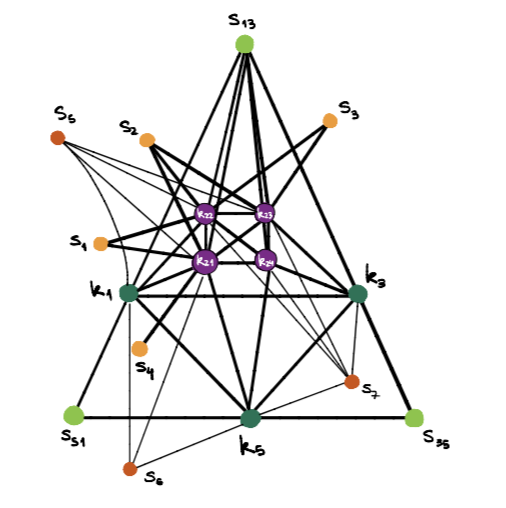}
	\caption{Ejemplo $2$: el grafo split circle $G$.}
	\label{fig:example_graph1_}
\end{figure}

Este grafo es split y también es circle. Para comenzar a dar un modelo circle del mismo, consideramos primero poner las cuerdas del subgrafo isomorfo al tent, el cual está dado por aquellos vértices de color verde claro y oscuro. Consideramos también los demás vértices de $G$ de acuerdo a las particiones $K_1, K_2 \ldots, K_6$ y $\{ S_{ij} \}_{1 \leq i,j \leq 6}$ descriptas en el capítulo anterior. En nuestro ejemplo, tenemos que los vértices violetas $k_{21}$, $k_{22}$, $k_{23}$ y $k_{24}$ pertenecen al subconjunto $K_2$, y los vértices naranjas $s_1, \ldots, s_7$ son los únicos vértices independientes en $G\setminus T$,  y todos son adyacentes a vértices en $K_2$. Los vértices $s_1$, $s_2$, $s_3$ y $s_4$ sólo son adyacentes a $K_2$ y por lo tanto pertenecen a $S_{22}$. Por su parte, los vértices $s_5$, $s_6$ y $s_7$ también son adyacentes a $K_1$, $K_1$ y $K_5$ y $K_3$ y $K_5$, respectivamente. 
 
Para colocar las cuerdas correspondientes a todo vértice de $S_{22}$ es necesario primero dibujar aquellas cuerdas que representan vértices en $K_2$. Esto se deduce del hecho de que las cuerdas correspondientes a $K_2$ tienen un extremo en el arco $k_1 k_3$ y el otro extremo en el arco $s_{51} s_{35}$, y las cuerdas de vértices en $S_{22}$ tienen ambos extremos dentro del arco $k_1 k_3$ o ambos extremos dentro del arco $s_{51} s_{35}$, y siempre intersecando cuerdas que representan vértices de $K_2$.
Luego, para colocar las cuerdas que corresponden a $K_2$ debemos establecer un ordenamiento de los vértices de este conjunto que respete el orden parcial dado por la inclusión de los vecindarios de aquellos vértices en $S_{22}$. 
Por ejemplo, $N(s_1) \subseteq N_(s_2)$, luego sabemos que un ordenamiento para las cuerdas de $K_2$ que nos permita dar un modelo circle debe contener alguna de las siguientes subsecuencias: $(k_{21}$, $k_{22}$, $k_{23})$, o $(k_{22}$, $k_{21}$, $k_{23})$, o $(k_{23}$, $k_{21}$, $k_{22})$, o $(k_{23}$, $k_{22}$, $k_{21})$. Además, como $N(s_2) \cap N_(s_3) \neq \emptyset$ y $N(s_2)$ y $N_(s_3)$ no están anidados, las cuerdas correspondientes a $s_2$ y $s_3$ deben ser dibujadas en distintas porciones del modelo circle, ya que representan vértices independientes y por lo tanto sus cuerdas no pueden intersecarse. 
Por su parte, $s_4$ es adyacente sólo a $k_{21}$, luego $N(s_4)$ está contenido en $N(s_1)$ y $N(s_2)$ y es disjunto a $N(s_3)$. Luego, la cuerda que represente $s_4$ puede ser colocada indistintamente en cualquiera de las dos porciones del círculo que corresponden a la partición $S_{22}$.

Por lo tanto, cuando consideramos cómo colocar las cuerdas, nos encontramos con dos de\-ci\-sio\-nes importantes: (1) en qué orden debemos poner las cuerdas que representan vértices en $K_2$ para que esto nos permita dibujar aquellos vértices independientes adyacentes a $K_2$, y (2) en qué porción del modelo circle debemos poner ambos extremos de las cuerdas que representan vértices en $S_{22}$. 

Al contrario que las cuerdas correspondientes a los vértices $s_1$, $s_2$, $s_3$ y $s_4$, las cuerdas que representan a $s_5$, $s_6$ y $s_7$ tienen sólo uno de sus extremos en algún arco designado para $K_2$, más precisamente, en los arcos $k_1 k_3$ y $s_{51} s_{35}$. Más aún, cada uno de estos vértices tiene una única ubicación posible para cada extremo de su cuerda. 
Si queremos colocar las cuerdas co\-rres\-pon\-dien\-tes a los vértices $s_5$, $s_6$ y $s_7$, debemos notar primero que la condición ``anidados o disjuntos'' también debe valer para los vecindarios de los vértices $s_1, \ldots, s_7$, tanto si nos restringimos a $K_2$ como en las demás particiones de $K$. Más precisamente, como $s_5$ es adyacente a $k_{24}$, $k_{23}$ y $k_1$, y $s_1$ no es adyacente a $k_1$, entonces necesariamente $s_1$ debe estar contenido en $s_5$. Algo similar ocurre con $s_7$ y $s_3$, mientras que $s_6$ y $s_3$ son disjuntos. 

Si consideramos la matriz enriquecida $A(S,K_2)$, las primeras cuatro filas corresponden a $s_1, \ldots, s_4$ y son filas sin etiqueta y sin color. Asimismo, los vértices $s_5$ y $s_6$ corresponderían a filas etiquetadas con L y coloreadas de rojo y azul, respectivamente, $s_7$ correspondería a una fila etiquetada con R y coloreada de azul. 

\vspace{-5mm}
\[
A(S,K_2) = 
\bordermatrix{ &  \cr
	& 1 1 0 0 \cr
	& 1 1 1 0 \cr
	& 0 1 1 0 \cr
	& 1 0 0 0 \cr
	\textbf{L} & 1 1 1 0 \cr
	\textbf{L} & 1 0 0 0 \cr
	\textbf{R} & 0 1 1 1 }\,
	\begin{matrix} 
   \cr  \cr \cr  \cr \textcolor{red}{\bullet} \cr \textcolor{blue}{\bullet} \cr \textcolor{blue}{\bullet} 
\end{matrix}
\]

Observemos que, como $s_6$ es adyacente a $k_{21}$, $k_1$ y $k_5$, la cuerda que corresponde a $k_{21}$ es forzada a estar primera entre todas las cuerdas correspondientes a vértices de $K_2$. Esto se deduce de que $s_6$ tiene una única ubicación posible, en la que uno de sus extremos va dentro del arco $s_{51} s_{35}$, por lo que necesitamos que $k_{21}$ sea la primera cuerda de $K_2$ que viene justo después de $s_{51}$. Más aún, esto se refuerza por el hecho de que $s_5$ es adyacente a $k_1$ y $k_{21}$, por lo que la cuerda que representa a $k_{21}$ debe ser dibujada primero cuando consideramos el orden dado por los vecindarios de aquellos vértices independientes que tienen al menos un extremo en el arco $k_1 k_3$. Luego, que $k_{21}$ sea el primero de los vértices en cualquier ordenamiento válido para $K_2$ es una condición necesaria al buscar un ordenamiento que cumpla con la propiedad de los unos consecutivos para la matriz $A(S,K_2)$. Esta es una de las propiedades que se consideran en la definición de $2$-nested y LR-ordenamiento, que está asociada a las etiquetas de las filas. A continuación damos un modelo circle para el grafo $G$ en la Figura \ref{fig:example_model1_}.

\begin{figure}[h!] 	\centering
	\includegraphics[scale=1]{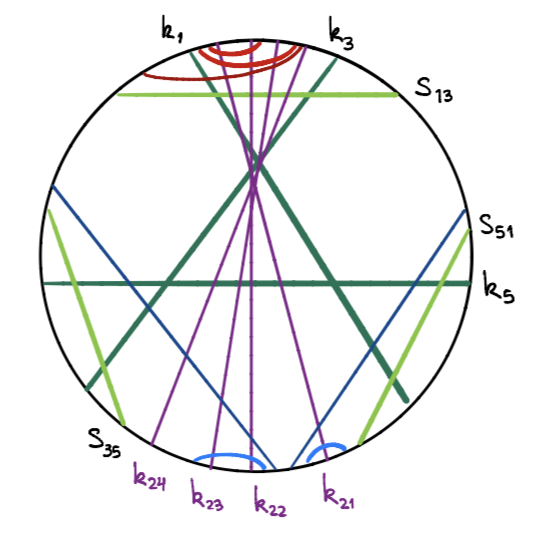}
	\caption{Un modelo circle para el grafo split $G$.}
	\label{fig:example_model1_}
\end{figure}


\selectlanguage{english}%
\chapter{$2$-nested matrices} \label{chapter:2nested_matrices}


In this chapter, we will define and characterize nested and $2$-nested matrices, which are of fundamental importance to describe each portion of a circle model for those split graphs that are also circle. The results in this chapter are crucial for the proof of the main result in the next chapter, which gives a complete characterization of split circle graphs by minimal forbidden induced subgraphs. 

In order to give some motivation for the definitions on this chapter, let us consider the split graph $G=(K,S)$ represented in Figure \ref{fig:example_graph0}. Since $G$ contains an induced tent $H$, we can consider the partitions $K_1, K_2 \ldots, K_6$ and $\{ S_{ij} \}_{1 \leq i,j \leq 6}$ as defined in Section \ref{sec:tent_partition}.

\begin{figure}[h!] 
	\centering
	\includegraphics[scale=.6]{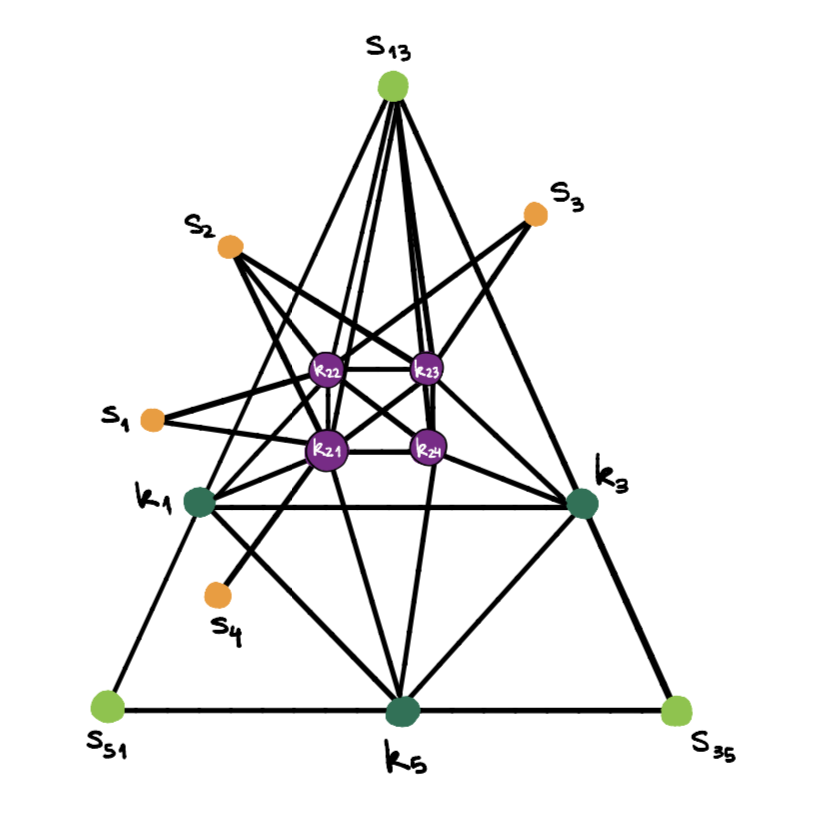}
	\caption{Example $1$: a split circle graph $G$.}
		\label{fig:example_graph0}
\end{figure}

Notice that every vertex in the complete partition of $G \setminus T$ lies in $K_2$, for the only adjacency of these vertices with regard to $S$ is the vertex $s_{13}$. 
Thus, $K_2 = \{ k_{21}$, $k_{22}$, $k_{23}$, $k_{24}\}$. 
Moreover, the orange vertices are the only independent vertices in $G\setminus T$ and these vertices are adjacent only to vertices in $K_2$, thus they all lie in $S_{22}$.
Furthermore, the graph $G$ is also a circle graph. Indeed, we would like to give a circle model for $G$. 
The tent is a prime graph, and as such, $H$ admits a unique circle model. Hence, let us begin by considering a circle model as the one presented in Figure \ref{fig:example_tentmodel}, having only the chords that represent the subgraph $H$. 
We will consider the arcs and chords of a model described clockwise. For example, in Figure \ref{fig:example_tentmodel} the arc $k_1 k_3$ is the portion of the circle that lies between $k_1$ and $k_3$ when traversing the circumference clockwise.

\begin{figure}[h!] 
	\centering
	\includegraphics[scale=.5]{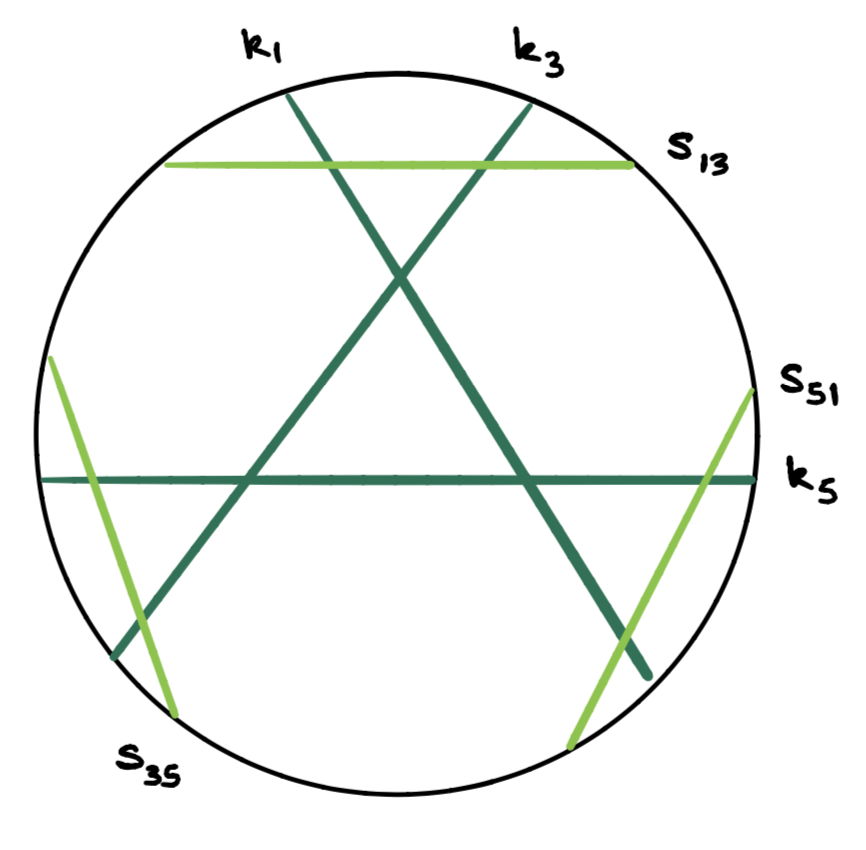}
	\caption{A circle model for the tent graph $H$.}
		\label{fig:example_tentmodel}
\end{figure}


To place the chords corresponding to each vertex in $S_{22}$, first we need to place the chords that represent every vertex in $K_2$. This follows from the fact that a chord re\-pre\-sen\-ting a vertex in $K_2$ has one endpoint between the arc $k_1 k_3$ and the other endpoint between the arc $s_{51} s_{35}$, and a chord representing a vertex in $S_{22}$ has either both endpoints inside the arc $k_1 k_3$ or both endpoints inside the arc $s_{51} s_{35}$, always intersecting chords representing vertices in $K_2$.
Thus, in order to place the chords corresponding to each vertex of $K_2$, we need to establish an ordering of the vertices in $K_2$ that respects the partial ordering relationship given by the neighbourhoods of the vertices in $S_{22}$. 
For example, since $N(s_1) \subseteq N(s_2)$, it follows that an ordering of the chords in $K_2$ that allows us to give a circle model must contain one of the following subsequences: $(k_{21}$, $k_{22}$, $k_{23})$ or $(k_{22}$, $k_{21}$, $k_{23})$ or $(k_{23}$, $k_{21}$, $k_{22})$ or $(k_{23}$, $k_{22}$, $k_{21})$. Moreover, since $N(s_2) \cap N(s_3) \neq \emptyset$ and $N(s_2)$ and $N(s_3)$ are not nested, then the chords corresponding to $s_2$ and $s_3$ must be drawn in distinct portions of the circle model, for they represent independent vertices and thus the chords cannot intersect. 
The vertex $s_4$ is adjacent only to $k_{21}$, thus $N(s_4)$ is contained in both $N(s_1)$ and $N(s_2)$ and is disjoint with $N(s_3)$. Hence, the chord that represents $s_4$ may be placed indistinctly in any of the two portions of the circle corresponding to the partition $S_{22}$.

Therefore, when considering the placement of the chords, we find ourselves in front of two important decisions: (1) in which order should we place the chords corresponding to the vertices in $K_2$ so that we can draw the chords of those independent vertices adjacent to $K_2$, and (2) in which portion of the circle model should we place both endpoints of the chords corresponding to vertices in $S_{22}$. We give a circle model for $G$ in Figure \ref{fig:example_model0} 

\begin{figure}[h!] 
	\centering
	\includegraphics[scale=1]{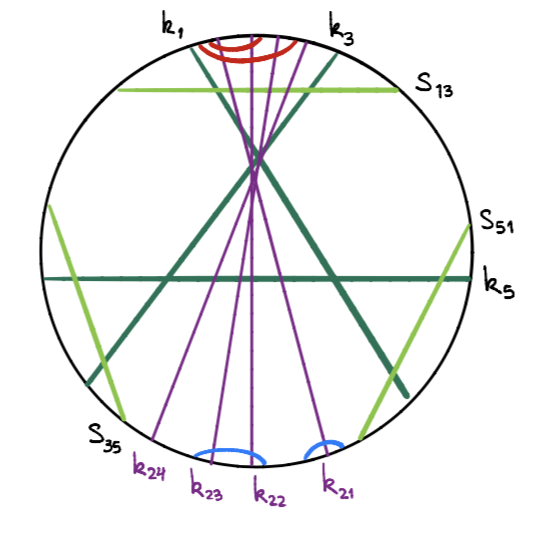}
	\caption{A circle model for the split graph $G$.}
		\label{fig:example_model0}
\end{figure}

Yet in this small example of a split graph that is circle, it becomes evident that there is a property that must hold for every pair of independent vertices that have both of its endpoints placed within the same arc of the circumference. This led to the definition of nested matrices, which was the first step in order to translate some of these problems to having certain properties in the adjacency matrix $A(S,K)$ (See Section \ref{section:basic_defs} for more details on the definition of $A(S,K)$). 

\begin{defn} \label{def:nested_m}
	Let $A$ be a $(0,1)$-matrix. We say $A$ is \emph{nested} if there is a consecutive-ones ordering for the rows and every two rows are disjoint or nested.
\end{defn} 

\begin{defn} \label{def:nested_g}
A split graph $G = (K,S)$ is \emph{nested} if and only if $A(S,K)$ is a nested matrix.
\end{defn}

\begin{teo} \label{teo:nested_caract}
A $(0,1)$-matrix is nested if and only if it contains no $0$-gem as a
submatrix (See Figure \ref{fig:forb_nested}).
\end{teo}

\begin{proof}
Since no Tucker matrix has the C$1$P and the rows of the $0$-gem are neither disjoint nor nested, no nested matrix contains a Tucker matrix or a $0$-gem as submatrices. Con\-verse\-ly, as each Tucker matrix contains a $0$-gem as a submatrix, every matrix containing no $0$-gem as a submatrix is a nested matrix.
\end{proof}

\begin{figure}[h!] 
	\centering  
	\includegraphics[scale=.6]{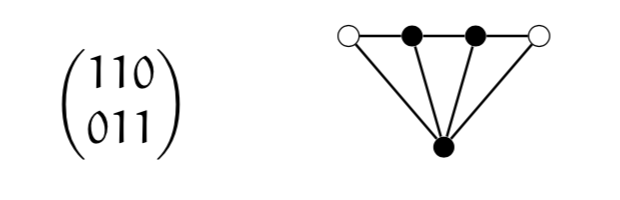}
	\caption{\mbox{The} $0$-gem \mbox{matrix and the associated gem graph.}} \label{fig:forb_nested}
\end{figure}

Let us consider the matrix $A(S,K_2)$ that corresponds to the example given in Figure \ref{fig:example_graph0}, where the rows are given by $s_1$, $s_2$, $s_3$ and $s_4$, and the columns are $k_{21}$, $k_{22}$, $k_{23}$ and $k_{24}$.
\vspace{-.5mm}
\[
A(S,K_2) = \begin{pmatrix}
	1 1 0 0 \\
	1 1 1 0 \\
	0 1 1 0 \\
	1 0 0 0 
\end{pmatrix}
\]

Notice that the existance of a C$1$P for the columns of the matrix $A(S,K_2)$ is a necessary condition to find an ordering of the vertices in $K_2$ that is compatible with the partial ordering given by containment for the vertices in $S_{22}$. 
Moreover, if the matrix $A(S,K_2)$ is nested, then any two independent vertices are either nested or disjoint. In other words, if $A(S,K_2)$ is nested, then we can draw every chord corresponding to an independent vertex in $G \setminus T$ in the same arc of the circumference. However, this is not the case in the previous example, for the vertices $s_1$ and $s_3$ are neither disjoint nor nested, and thus they cannot be drawn in the same portion of the circle model. Hence, $A(S,K)$ is not a nested matrix, and thus the notion of nested matrix is not enough to determine whether there is a circle model for a given split graph or not.

Let us see one more example. Consider $H$ to be the split graph presented in Figure \ref{fig:example_graph1}. 
Notice that this graph is equal to $G$ plus three new independent vertices.

\begin{figure}[h!] 	\centering
	\includegraphics[scale=.9]{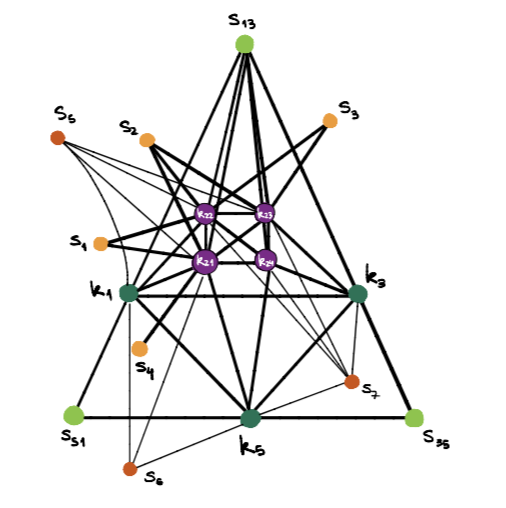}
	\caption{Example $2$: the split circle graph $H$.}
	\label{fig:example_graph1}
\end{figure}


Moreover, unlike $s_1$, $s_2$, $s_3$ and $s_4$, the chords that represent these new independent vertices $s_5$, $s_6$ and $s_7$ have only one of its endpoints in the arcs corresponding to the area of the circle designated for $K_2$, this is, in the arcs $k_1 k_3$ and $s_{51} s_{35}$. Furthermore, each of these new vertices has a unique possible placement for each endpoint of their corresponding chord.
If we consider the rows given by the vertices $s_1, \ldots, s_7$ and the columns given by $k_{21}, \ldots, k_{24}$, then the adjacency matrix $A(S,K_2)$ in this example is as follows:
\vspace{-.5mm}
\[
A(S,K_2) = \begin{pmatrix}
	1 1 0 0 \\
	1 1 1 0 \\
	0 1 1 0 \\
	1 0 0 0 \\
	1 1 1 0 \\
	1 0 0 0 \\
	0 1 1 1 
\end{pmatrix}
\]

As in the previous example, $A(S,K_2)$ is not a nested matrix. Furthermore, notice that in this case not every adjacency of each independent vertex $s_1, \ldots, s_7$ is depicted in this matrix, since $s_5$, $s_6$ and $s_7$ all are adjacent to at least one vertex in $K \setminus K_2$. 

Let us concentrate in the placement of the endpoints of the chords representing $s_5$, $s_6$ and $s_7$ that lie between the arcs $k_1 k_3$ and $s_{51} s_{35}$. Notice that the ``nested or disjoint'' property must still hold, and not only for those vertices in $K_2$. More precisely, since $s_5$ is adjacent to $k_{24}$, $k_{23}$ and $k_1$ and $s_1$ is nonadjacent to $k_1$ and adjacent to $k_{23}$ and $k_{24}$, then necessarily $s_1$ must be contained in $s_5$. Something similar occurs with $s_7$ and $s_3$, whereas $s_6$ and $s_3$ are disjoint. 

There is one situation in this example that did not occur in the previous one. Since $s_6$ is adjacent to $k_{21}$, $k_1$ and $k_5$, then the chord corresponding to the vertex $k_{21}$ is forced to be placed first within every chord corresponding to $K_2$. This follows from the fact that a chord that represents $s_6$ has a unique possible placement inside the arc $s_{51} s_{35}$, for we need $k_{21}$ to be the first chord of $K_2$ that comes right after $s_{51}$. Moreover, this is confirmed by the fact that $s_5$ is adjacent to $k_1$ and $k_{21}$, thus the chord corresponding to the vertex $k_{21}$ must be drawn first when considering the ordering given by the neighbourhoods of those independent vertices that have at least one endpoint lying in $k_1 k_3$. It follows from the previous that $k_{21}$ being the first vertex in the ordering is a necessary condition when searching for a consecutive-ones ordering for the matrix $A(S,K_2)$. See Figure \ref{fig:example_model1}, where we give a circle model for the graph $H$.

\begin{figure}[h!] 	\centering
	\includegraphics[scale=1]{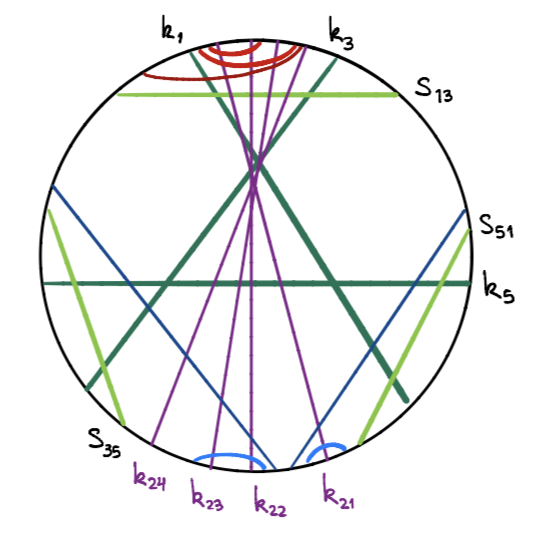}
	\caption{A circle model for the split graph $H$.}
	\label{fig:example_model1}
\end{figure}

The previously described situations must also hold for each partition $K_i$ of $K$. We translated the problem of giving a circle model to the fullfilment of some properties for each of the matrices $A(S,K_i)$, where $K= \bigcup_{i} K_i$ and these partitions depend on whether $G$ contains an induced tent, $4$-tent or co-$4$-tent.
This led to the definition of enriched matrices, which allowed us to model some of the above mentioned properties, and also others that came up when considering split graphs containing a $4$-tent and a co-$4$-tent. 

\begin{defn} \label{def:enriched_matrix}
	Let $A$ be a $(0,1)$-matrix. We say $A$ is an \emph{enriched matrix} if all of the following conditions hold:
	\begin{enumerate}
		\item Each row of $A$ is either unlabeled or labeled with one of the following labels: L or R or LR. We say that a row is an \emph{LR-row  (resp.\ L-row, R-row)} if it is labeled with LR (resp.\ L, R).
		\item Each row of $A$ is either uncolored or colored with either blue or red.
		\item The only colored rows may be those labeled with L or R, and those LR-rows having a $0$ in every column.
		\item The LR-rows having a $0$ in every column are all colored with the same color. 
	\end{enumerate}
	
	The \emph{underlying matrix of $A$} is the $(0,1)$-matrix that coincides with $A$ that has neither labels nor colored rows.
\end{defn}

We will denote the color assignment for a row with a colored bullet at the right side of the matrix.

The color assignment for some of the rows represents in which arc of the circle corresponding to $K_i$ we must draw one or both endpoints when considering the placement of the chords. Some of the independent vertices have a unique possible placement, and some of them can be --a priori-- drawn in either two of the arcs corresponding to $K_i$.
Moreover, the labeling of the rows explains `'from which direction does the chord come from'' if we are standing in a particular portion of the circle. 
For example, the following is the matrix $A(S,K_2)$ for the graph represented in Figure \ref{fig:example_model1} considered as an enriched matrix --taking into account all the information on the placement of the chords: 

\vspace{-5mm}
\[
A(S,K_2) = 
\bordermatrix{ &  \cr
	& 1 1 0 0 \cr
	& 1 1 1 0 \cr
	& 0 1 1 0 \cr
	& 1 0 0 0 \cr
	\textbf{L} & 1 1 1 0 \cr
	\textbf{L} & 1 0 0 0 \cr
	\textbf{R} & 0 1 1 1 }\,
	\begin{matrix} 
   \cr  \cr \cr  \cr \textcolor{red}{\bullet} \cr \textcolor{blue}{\bullet} \cr \textcolor{blue}{\bullet} 
\end{matrix}
\]
\vspace{.5mm}
\begin{defn} \label{def:LR-orderable}
	Let $A$ be an enriched matrix. We say $A$ is \emph{LR-orderable} if there is a linear ordering $\Pi$ for the columns of $A$ such that each of the following assertions holds:
	\begin{itemize}
	    \item $\Pi$ is a consecutive-ones ordering for every non-LR row of $A$.
	    
		\item The ordering $\Pi$ is such that the ones in every nonempty row labeled with L (resp.\ R) start in the first column (resp.\ end in the last column). 
		
		\item $\Pi$ is a consecutive-ones ordering for the complements of every LR-row of $A$.
     \end{itemize} 
Such an ordering is called an \emph{LR-ordering}.
For each row of $A$ labeled with L or LR and having a $1$ in the first column of $\Pi$, we define its \emph{L-block (with respect to $\Pi$)} as the maximal set of consecutive columns of $\Pi$ starting from the first one on which the row has a 1. \emph{R-blocks} are defined on an entirely analogous way.
For each unlabeled row of $A$, we say its \emph{U-block (with respect to $\Pi$)} is the set of columns having a $1$ in the row.
The blocks of $A$ with respect to $\Pi$ are its L-blocks, its R-blocks and its U-blocks. 
\end{defn} 

\begin{defn} \label{def:block-bicoloring}
	Let $A$ be an enriched matrix. We say an \emph{L-block (resp.\ R-block, U-block) is colored} if there is a 1-color assignment for every entry of the block. 
	
	A \emph{block bi-coloring for the blocks of $A$} is a color assignment with either red or blue for some L-blocks, U-blocks and R-blocks of $A$.
		A block bi-coloring is \emph{total }if every L-block, R-block and U-block of $A$ is colored, and is \emph{partial }if otherwise.
\end{defn}


Notice that for every enriched matrix, the only colored rows are those labeled with L or R and those empty LR-rows. Moreover, for every LR-orderable matrix, there is an ordering of the columns such that every row labeled with L (resp.\ R) starts in the first column (resp.\ ends in the last column), and thus all its $1$'s appear consecutively.  
Thus, if an enriched matrix is also LR-orderable, then the given coloring induces a partial block bi-coloring (see Figure \ref{fig:example_LR-ord}), in which every empty LR-row remains the same, whereas for every nonempty colored labeled row, we color all its $1$'s with the color given in the definition of the matrix.

\begin{figure}[h!]
\begin{align*}
	A = \bordermatrix{ &  \cr
	\textbf{LR} & 1  0  0  0  1 \cr
	\textbf{LR} & 1  1  0  0  1 \cr
					& 0  1 1  0  0 \cr
	\textbf L & \textcolor{red}{1}  \textcolor{red}{1}  \textcolor{red}{1}  0  0 \cr
	\textbf{LR} & 0  0  0   0  0 \cr
	\textbf R & 0  0  \textcolor{blue}{1} \textcolor{blue}{1} \textcolor{blue}{1} }\,
	\begin{matrix} 
   \cr  \cr \cr  \cr \textcolor{red}{\bullet} \cr \textcolor{blue}{\bullet} \cr \textcolor{blue}{\bullet} 
\end{matrix}
	&&
	B = \bordermatrix{ & \cr
	\textbf{LR} & 1  0  1  0 1 \cr
	\textbf L & \textcolor{red}{1}  \textcolor{red}{1}  0  0  0 \cr
	\textbf R & 0  0  0  \textcolor{blue}{1}  \textcolor{blue}{1} \cr
				& 0  0  1  1  0 }\,
	\begin{matrix} 
   \cr  \textcolor{red}{\bullet} \cr \textcolor{blue}{\bullet} \cr \cr
\end{matrix}
\end{align*}
\caption{Example: An enriched LR-orderable matrix $A$, where the column ordering given from left to right is a consecutive-ones ordering. 
$B$ is an enriched non-LR-orderable matrix.} \label{fig:example_LR-ord} 
\end{figure}

We now define $2$-nested matrices, which will allow us to address and solve both the problem of ordering the columns in each adjacency matrix $A(S,K_i)$ of a split graph for each partition $K_i \subset K$, and the problem of deciding if there is a feasible distribution of the independent vertices adjacent to $K_i$ between the two portions of the circle corresponding to $K_i$. This allows to give a circle model for the given graph. We give a complete characterization of these matrices by forbidden subconfigurations at the end of this chapter.

\begin{defn} \label{def:2-nested}
	Let $A$ be an enriched matrix. We say $A$ is \emph{$2$-nested} if there exists an LR-ordering $\Pi$ of the columns and an assignment of colors red or blue to the blocks of $A$ such that all of the following conditions hold:
\begin{enumerate}
	\item If an LR-row has an L-block and an R-block, then they are colored with distinct colors. \label{item:2nested1}
	\item For each colored row $r$ in $A$, any of its blocks is colored with the same color as $r$ in $A$. \label{item:2nested2}
	\item If an L-block of an LR-row is properly contained in the L-block of an L-row, then both blocks are colored with different colors. \label{item:2nested3}
	\item Every L-block of an LR-row and any R-block are disjoint. The same holds for an R-block of an LR-row and any L-block. \label{item:2nested4} 
	\item If an L-block and an R-block are not disjoint, then they are colored with distinct colors.	\label{item:2nested5} 
	\item Each two U-blocks colored with the same color are either disjoint or nested. \label{item:2nested6}
	\item If an L-block and a U-block are colored with the same color, then either they are disjoint or the U-block is contained in the L-block. The same holds replacing L-block for R-block. \label{item:2nested7}
	\item If two distinct L-blocks of non-LR-rows are colored with distinct colors, then every LR-row has an L-block. The same holds replacing L-block for R-block.  \label{item:2nested8} 
	\item If two LR-rows overlap, then the L-block of one and the R-block of the other are colored with the same color. \label{item:2nested9}
	
\end{enumerate}	

An assignment of colors red and blue to the blocks of $A$ that satisfies all these properties is called a \emph{(total) block bi-coloring}.
\end{defn}

\begin{figure}
\begin{align*}
	A = \bordermatrix{  &  \cr
	\textbf{LR} &  \textcolor{blue}{1}  0  0  0  \textcolor{red}{1} \cr
	\textbf{LR} &  \textcolor{blue}{1}   \textcolor{blue}{1}  0  0   \textcolor{red}{1} \cr
					 & 0  \textcolor{red}{1}  \textcolor{red}{1}  0  0 \cr
	\textbf L &  \textcolor{red}{1}   \textcolor{red}{1}   \textcolor{red}{1}  0  0 \cr
	\textbf{LR} &  0   0   0   0   0 \cr
	\textbf R & 0  0  \textcolor{blue}{1}  \textcolor{blue}{1}  \textcolor{blue}{1} } \,
	\begin{matrix}
	\\ \\  \\  \\  \textcolor{red}{\bullet} \\ \textcolor{blue}{\bullet} \\ \textcolor{blue}{\bullet} \\
	\end{matrix}
\end{align*}
\caption{Example of a total block bi-coloring of the blocks of the matrix in Figure \ref{fig:example_LR-ord}, considering the columns ordered from left to right. Moreover, $A$ is $2$-nested considering this LR-ordering and total block bi-coloring.}
\end{figure}

\begin{remark}
We will give some insight on which properties we are modeling with Definition \ref{def:2-nested}, which are necessary conditions that each matrix $A(S,K_i)$ must fullfil in order to give a circle model for any split graph containing a tent, $4$-tent or co-$4$-tent.

The LR-rows represent independent vertices that have both endpoints in the arcs corresponding to $K_i$. The difference between these rows and those that are unlabeled, is that one endpoint of the chords must be placed in one of the arcs corresponding to $K_i$ and the other endpoint must be placed in the other arc corresponding to $K_i$. Hence, the first property ensures that, when deciding where to place the chord corresponding to an LR-row, if the ordering indicates that the chord intersects some of its adjacent vertices in one arc and the other in the other arc, then the distinct blocks corresponding to the row must be colored with distinct colors.

With the second property, we ensure that the colors that are pre-assigned are respected, since they correspond to independent vertices with a unique possible placement.

The third property refers to the ordering given by containement for the vertices. We will further on see that every LR-row represents vertices that are adjacent to almost every vertex in the complete partition $K$ of $G$. Hence, when dividing the LR-rows into blocks, we need to ensure that each of its block is not properly contained in the neighbourhoods of vertices that are nonadjacent to at least one partition of $K$. Something similar must hold for L-rows (resp.\ R-rows) and U-rows, and L-rows (resp.\ R-rows) and LR-rows. This is modeled by properties \ref{item:2nested7} and \ref{item:2nested8}.

The properties \ref{item:2nested4}, \ref{item:2nested5}, \ref{item:2nested6} and \ref{item:2nested9} refer to the previously discussed `'nested or disjoint'' property that we need to ensure in order to give a circle model for $G$.
\end{remark}

This chapter is organized as follows. In Section \ref{section:defs_para2nested} we give some more definitions which are necessary to state a characterization of $2$-nested matrices. In Section \ref{section:admissibility} we define and characterize admissible matrices, which give necessary conditions for a matrix to admit a total block bi-coloring. In Section \ref{section:part2nested} we define and characterize LR-orderable and partially $2$-nested matrices, and then we prove some properties of LR-orderings in admissible matrices. Finally, in Section \ref{section:2nestedmatrices} we prove Theorem \ref{teo:2-nested_caract_bymatrices}, which characterizes $2$-nested matrices by forbidden subconfigurations. 

\section{A characterization for $2$-nested matrices} \label{section:defs_para2nested}

In this section, we begin by giving some definitions and examples that are necessary to state Theorem \ref{teo:2-nested_caract_bymatrices}, which is presented at the end of this section and is the main result of this chapter. The proof of Theorem \ref{teo:2-nested_caract_bymatrices} will be given in Section \ref{section:2nestedmatrices}.

\begin{defn}
Let $A$ be an enriched matrix. The dual matrix of $A$ is defined as the enriched matrix $\tilde{A}$ that coincides with the underlying matrix of $A$ and for which every row of $A$ that is labeled with L (resp.\ R) is now labeled with R (resp.\ L) and every other row remains the same. Also, the color assigned to each row remains as in $A$.
\end{defn}

\begin{figure}[h]
\begin{align*}
	A = \bordermatrix{ &        \cr
	\textbf{LR} & 1   0   0   0   1 \cr
	\textbf{LR} & 1   1   0   0   1 \cr
					& 0   1   1   0   0 \cr
	\textbf L & \textcolor{red}{1}   \textcolor{red}{1}   \textcolor{red}{1}   0   0 \cr
	\textbf{LR} & 0   0   0    0   0 \cr
	\textbf R & 0   0   \textcolor{blue}{1}   \textcolor{blue}{1}   \textcolor{blue}{1} }\,
	\begin{matrix} 
   \\  \\ \\  \\ \textcolor{red}{\bullet} \\ \textcolor{blue}{\bullet} \\ \textcolor{blue}{\bullet} 
\end{matrix}
	&&
	\tilde{A} = \bordermatrix{ &         \cr
	\textbf{LR} & 1   0   0   0   1 \cr
	\textbf{LR} & 1   1   0   0   1 \cr
					& 0   1   1   0   0 \cr
	\textbf R & \textcolor{red}{1}   \textcolor{red}{1}   \textcolor{red}{1}   0   0 \cr
	\textbf{LR} & 0   0   0    0   0 \cr
	\textbf L & 0   0   \textcolor{blue}{1}   \textcolor{blue}{1}   \textcolor{blue}{1} }\,
	\begin{matrix} 
   \\  \\ \\  \\ \textcolor{red}{\bullet} \\ \textcolor{blue}{\bullet} \\ \textcolor{blue}{\bullet} 
\end{matrix}
\end{align*}
\caption{Example: $A$ and its dual matrix.} \label{fig:example_dualmatrix}
\end{figure}

\vspace{2mm}
The $0$-gem, $1$-gem and $2$-gem are the following enriched matrices:
\[ 		\begin{pmatrix}
		 1 1 0 \cr
		0 1 1
		\end{pmatrix}, \qquad
		\bordermatrix{ & \cr
		 & 1 0 \cr
		 & 1 1 }\ , \qquad
		\bordermatrix{ & \cr
		\textbf{LR} & 1 1 0 \cr
		\textbf{LR} & 1 0 1  }\  \]
respectively.

\begin{defn}  \label{def:gems}
	Let $A$ be an enriched matrix. We say that $A$ \emph{contains a gem} (resp.\ \emph{doubly-weak gem}) if it contains a $0$-gem (resp.\ a $2$-gem) as a subconfiguration.
	We say that $A$ \emph{contains a weak gem} if it contains a $1$-gem such that, either the first is an L-row (resp.\ R-row) and the second is a U-row, or the first is an LR-row and the second is a non-LR-row.
	We say that a $2$-gem is \emph{badly-colored }if the entries in the column in which both rows have a $1$ are in blocks colored with the same color. 
\end{defn}

Let $r$ be an LR row of $A$. We denote with $\overline{r}$ to \emph{the complement of $r$}, this is, the row that has a $1$ in each coordinate of $r$ that has a 0, and has a $0$ in each coordinate of $r$ that has a 1.

\begin{defn} \label{def:A*}
	Let $A$ be an enriched matrix and let $\Pi$ be a LR-ordering. 
	We define $A^*$ as the enriched matrix that arises from $A$ by:
\begin{itemize}
	\item Replacing each LR-row by its complement.
	\item Adding two distinguished rows: both rows have a $1$ in every column, one is labeled with L and the other is labeled with R. 
\end{itemize}	 
\end{defn}

In Figures \ref{fig:forb_D}, \ref{fig:forb_F}, \ref{fig:forb_Si}, \ref{fig:forb_P} and \ref{fig:forb_LR-orderable} we define some matrices, for they play an important role in the sequel. 
We will use green and orange to represent red and blue or blue and red, respectively. For every enriched matrix represented in the figures of this chapter, if a row labeled with L or R appears in black, then it may be colored with either red or blue indistinctly.
Moreover, whenever a row is labeled with \textbf{L (LR)} (resp.\ \textbf{R (LR)}), then such a row may be either a row labeled with L or LR (resp.\ R or LR) indistinctly.

\begin{figure}[h]
	\centering
	\includegraphics[scale=.5]{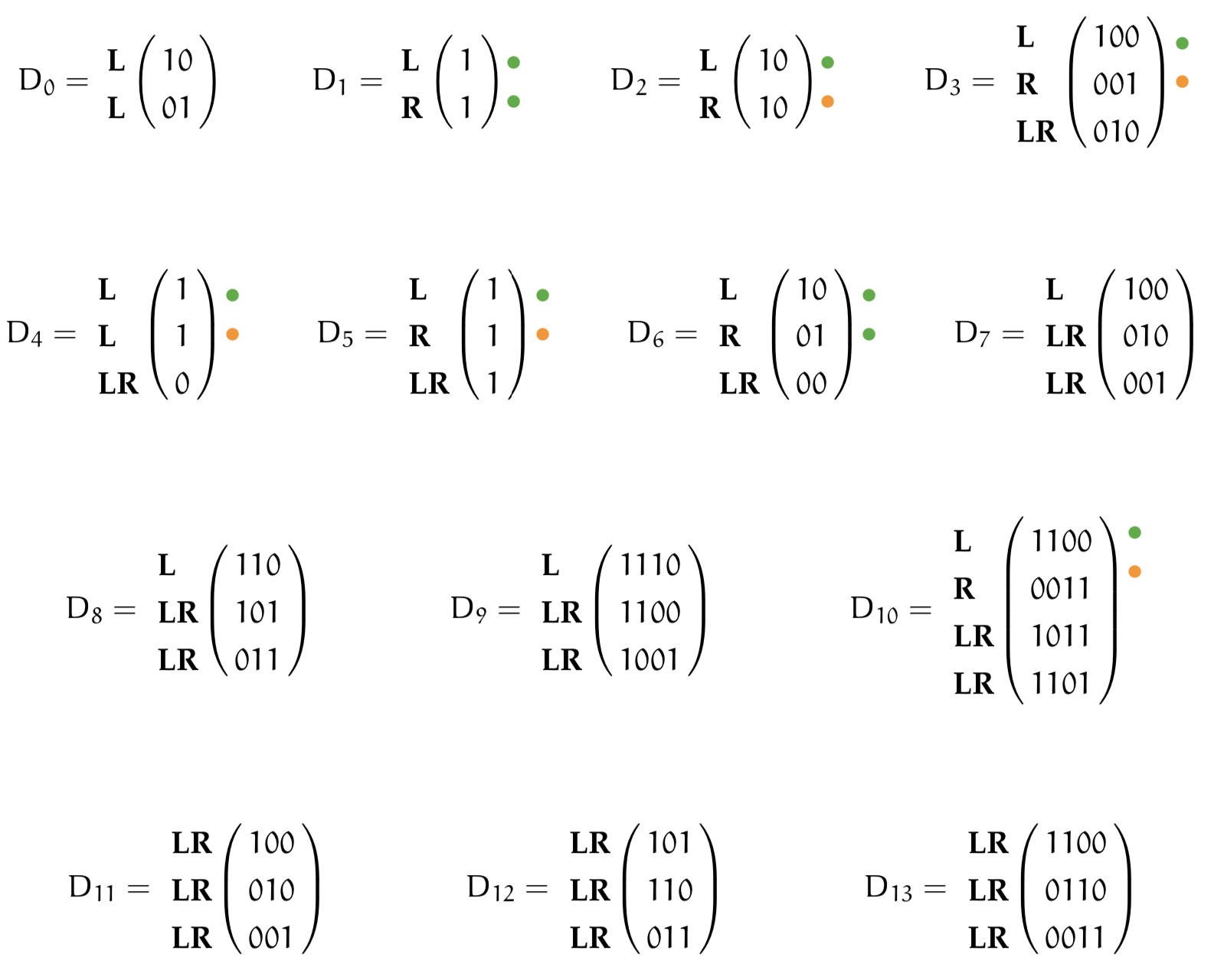}
	\caption{The family of enriched matrices $\mathcal{D}$.} 
	\label{fig:forb_D}
\end{figure}

\begin{figure}[H] 
\small{
	\centering
	\begin{align*}
			F_0= \begin{pmatrix}
				11100\\
				01110\\
				00111\\
			\end{pmatrix}
			&&
			F_1(k)= \begin{pmatrix}
				011...111\\
				111...110\\
				000...011\\
				000...110\\
				.   .   .   .   . \\
				.   .   .   .   . \\
				.   .   .   .   . \\
				110...000\\
			\end{pmatrix}
			&&
			F_2(k)= \begin{pmatrix}
				0111...10\\
				1100...00\\
				0110...00\\
				.   .   .   .   . \\
				.   .   .   .   . \\
				.   .   .   .   . \\
				0000...11\\
			\end{pmatrix} 
			\end{align*}
			\begin{align*}
			F'_0= \bordermatrix{ & \cr
				\textbf{L (LR)} & 1 1 0 0 \cr
				 &1 1 1 0 \cr
				&0 1 1 1 }\			 
			&&
			F''_0= \bordermatrix{ & \cr
				\textbf{L} & 1 1 0 \cr
				 &1 1 1 \cr
				\textbf{R} &0 1 1 }\	
			&&
			F'_1(k)= \bordermatrix{ & \cr
				&11\ldots1111\cr
				\textbf{L (LR)}&11\ldots1110\cr
				&00\ldots0011\cr
				&00\ldots0110\cr
				& \iddots \cr
				\textbf{L (LR)}&10 \ldots 0000 }\
	\end{align*}
	\begin{align*}
			F'_2(k)= \bordermatrix{ & \cr
				&111\ldots10\cr
				\textbf{L (LR)}&100\ldots00\cr
				&110\ldots00\cr
				&  \ddots \cr
				&000\ldots11 }\
	\end{align*}
	}
	\caption{The enriched matrices of the family $\mathcal{F}$.}  \label{fig:forb_F}
\end{figure}

The matrices $\mathcal{F}$ represented in Figure \ref{fig:forb_F} are defined as follows:
$F_1(k) \in \{0,1\}^{k \times (k-1)}$, $F_2(k) \in \{0,1\}^{k \times k}$, $F'_1(k) \in \{0,1\}^{k \times (k-2)}$ and $F'_2(k) \in \{0,1\}^{k \times (k-1)}$, for every odd $k \geq 5$. In the case of $F'_0$, $F'_1(k)$ and $F'_2(k)$, the labeled rows may be either L or LR indistinctly, and in the case of their dual matrices, the labeled rows may be either R or LR indistinctly.

The matrices $\mathcal{S}$ in Figure \ref{fig:forb_Si} are defined as follows. 
If $k$ is odd, then $S_1(k) \in \{0,1\}^{(k+1) \times k}$ for $k \geq 3$, and if $k$ is even, then $S_1(k) \in \{0,1\}^{k \times (k-2)}$ for $k\geq 4$.
The remaining matrices have the same size whether $k$ is even or odd: $S_2(k) \in \{0,1\}^{k \times (k-1)}$ for $k \geq 3$, 
$S_3(k) \in \{0,1\}^{k \times (k-1)}$ for $k \geq 3$, 
$S_5(k) \in \{0,1\}^{k \times (k-2)}$ for $k \geq 4$, 
$S_4(k) \in \{0,1\}^{k \times (k-1)}$, $S_6(k) \in \{0,1\}^{k \times k}$ for $k\geq 4$,
$S_7(k) \in \{0,1\}^{k \times (k+1)}$ for every $k \geq 3$
and $S_8(2j) \in \{0,1\}^{2j \times (2j)}$ for $j \geq 2$. 
With regard to the coloring of the labeled rows, if $k$ is even, then the first and last row of $S_2(k)$ and $S_3(k)$ are colored with the same color, and in $S_4(k)$ and $S_5(k)$ are colored with distinct colors.

\begin{figure}[H]
	\centering
	\small{
			\begin{align*}
			S_1(2j) &= \bordermatrix{ & \cr
				\textbf{L}& 1 0 \ldots 00  \cr
							  & 1 1 \ldots 00 \cr
						       &   \ddots  \cr
						       & 0 0 \ldots 11 \cr
				\textbf{LR} &0 0  \ldots 01 \cr
				\textbf{L} & 1  1  \ldots  11 }\
			&
			S_1(2j+1) &=  \bordermatrix{ & \cr
				\textbf{L} & 1 0 \ldots 0 0 \cr
							& 1 1 \ldots 0 0 \cr
							& \ddots \cr
							& 0 0 \ldots 1 1 \cr
				\textbf{LR} &0 0 \ldots 0 1 }\
			&
		S_2(k) &= \bordermatrix{ & \cr
				\textbf{L}& 1 0 \ldots 0 0 \cr
							& 1 1  \ldots 0 0 \cr
							&  \ddots \cr
					 	    &0 0 \ldots 1 1 \cr
				\textbf{L} & 1 1 \ldots 1 0 }\
				\begin{matrix}
			 \textcolor{dark-orange}{\bullet} \\ \\ \\ \\ \\ \textcolor{dark-green}{\bullet} \\
				\end{matrix}
				\end{align*}
				\begin{align*}
			S_3(k) &=  \bordermatrix{ & \cr
				\textbf{L} & 1 0 \ldots 0 0 \cr
							& 1 1  \ldots 0 0 \cr
							&  \ddots \cr
							&0 0 \ldots 1 1 \cr
				\textbf{R} & 0 0 \ldots 0 1 }\
				\begin{matrix}
			\textcolor{dark-orange}{\bullet} \\ \\ \\ \\ \\ \textcolor{dark-green}{\bullet} \\
				\end{matrix}
			&
		S_4(k) &= \bordermatrix{ & \cr
		\textbf{LR}& 1 1 \ldots 1 1 \cr
		\textbf{L} & 1 0 \ldots 0 0 \cr
						& 1 1  \ldots  0 0 \cr
						&  \ddots \cr
						& 0 0  \ldots 1 1 \cr
		  \textbf{R}& 0 0 \ldots 0 1  }\
				\begin{matrix}
			\\ \textcolor{dark-orange}{\bullet} \\ \\ \\ \\ \\ \textcolor{dark-green}{\bullet} \\
				\end{matrix}
		&		
		S_5(k) &= \bordermatrix{ & \cr
		\textbf{L}& 1 0 \ldots 0 0 \cr
					  & 1 1 \ldots 0 0 \cr
					  &  \ddots \cr
					  & 0 0  \ldots 1 1 \cr
		\textbf{LR}& 1 1 \ldots 1 0 \cr
		\textbf{L}& 1 1 \ldots 1 1  }\
				\begin{matrix}
			 \textcolor{dark-orange}{\bullet} \\ \\ \\ \\ \\ \\  \textcolor{dark-green}{\bullet} \\
				\end{matrix}
		\end{align*}
		\begin{align*}
		S_6(3) &= \bordermatrix{ & \cr
		\textbf{LR} & 1 1 0 \cr
		  \textbf{R} & 0 1 1  \cr
						  & 1 1 0 }\
		&
		S_6'(3) &= \bordermatrix{ & \cr
		\textbf{LR} & 1 1 0 \cr
		  \textbf{R} & 0 1 1  \cr
						  & 1 1 1 }\
		&
		S_6(k) &= \bordermatrix{ & \cr
		\textbf{LR}& 1 1 1 \ldots 1 1 0 \cr
		  \textbf{R}& 0 1 1 \ldots 1 1 1 \cr
						& 1 1 0  \ldots 0 0 0  \cr
						&  \ddots \cr
						& 0 0 0 \ldots 0 1 1 }\
				\begin{matrix}
			\\ \\ \textcolor{dark-orange}{\bullet} \\ \\ \\ \\ \\ \\
				\end{matrix}
	\end{align*}
		\begin{align*}
		S_7(3) &= \bordermatrix{ & \cr
		\textbf{LR} & 1 1 0 0 1 \cr
		\textbf{LR} & 1 0 0 1 1  \cr
						 & 1 1 1 0 0 }\
		&
		S_7(2j) &= \bordermatrix{ & \cr
		\textbf{LR}& 1 1 0 0 \ldots 0 0 0 \cr
		\textbf{LR}& 1 0 0 0  \ldots 0 0 1 \cr
						& 0 1 1 0 \ldots 0 0 0  \cr
						&  \ddots \cr
						& 0 0 0 0 \ldots 0 1 1 }\
		&
		S_8(2j) &= \bordermatrix{ & \cr
		\textbf{LR}& 1 0 0 \ldots 0 0 1 \cr
						& 1 1 0  \ldots 0 0 0 \cr
						&  \ddots \cr
						& 0 0 0 \ldots 0 1 1 }\
	\end{align*}
		}
\caption{ The family of matrices $\mathcal{S}$ for every $j\geq2$ and every odd $k \geq 5$} \label{fig:forb_Si}
\end{figure}

\begin{figure}[H]
	\centering
	\footnotesize{
		\begin{align*}	
		P_0(k,0) =  \bordermatrix{ & \cr
		\textbf{L} & 1 1 0 0 0  \ldots 0 0 0   \cr
		\textbf{LR} & 1 0 0 1 1 \ldots 1 1 1 \cr
					     & 0 0  1  1  0  \ldots 0 0 0 \cr
						 &  \ddots  \cr
						& 0 0 0 0 0 \ldots  0 1 1  \cr
		\textbf{R} & 0 0 0 0 0 \ldots 0 0 1  }\
	 \begin{matrix}
	  \textcolor{dark-green}{\bullet} \\ \\ \\ \\ \\ \\ \textcolor{dark-green}{\bullet}
	  \end{matrix}
		&&
		P_0(k,l) =  \bordermatrix{ & \cr
		\textbf{L}& 1 0 0 \ldots 0 0 0 0 \ldots 0 \cr
					  & 1 1 0 \ldots 0 0 0 0  \ldots 0 \cr
					  &  \ddots \cr
					 & 0 0 0 \ldots 1 1 0 0 \ldots 0 \cr
	\textbf{LR} & 1 1 1 \ldots 1 0 0 1 \ldots 1 \cr
					 & 0 0 0 \ldots 0 0 1 1  \ldots 0 \cr
					 & \ddots \cr
					& 0 0 0 \ldots 0 0 \ldots 0 1 1 \cr
		\textbf{R} & 0 0 0 \ldots 0 0 \ldots 0 0 1 }\
	\begin{matrix}
	  \textcolor{dark-green}{\bullet} \\ \\ \\ \\ \\ \\ \\ \\ \\ \\ \textcolor{dark-green}{\bullet}
	  \end{matrix}
	\end{align*}
	\begin{align*}
		P_1(k,0) = \bordermatrix{ & \cr
		\textbf{L} & 1 1 0 0  \ldots 0 0 0   \cr
		\textbf{LR} & 1 0 1 1 \ldots 1 1 1 \cr
		\textbf{LR} & 1 1 0 1 \ldots 1 1 1 \cr
					     & 0 0  1  1  0  \ldots 0 0 0 \cr
						 &  \ddots  \cr
						& 0 0 0 0 0 \ldots  0 1 1  \cr
		\textbf{R} & 0 0 0 0 \ldots 0 0 1  }\
	 \begin{matrix}
	  \textcolor{dark-green}{\bullet} \\ \\ \\ \\ \\ \\ \\ \textcolor{dark-green}{\bullet}
	  \end{matrix}
		&&
			P_1(k,l) =  \bordermatrix{ & \cr
		\textbf{L}& 1 0 0 \ldots 0 0 0 0 \ldots 0 \cr
					  & 1 1 0 \ldots 0 0 0 0  \ldots 0 \cr
					  &  \ddots \cr
					 & 0 0 0 \ldots 1 1 0 0 \ldots 0 \cr
	\textbf{LR} & 1 1 1 \ldots 1 0 1 1 \ldots 1 \cr
	\textbf{LR} & 1 1 1 \ldots 1 1 0 1 \ldots 1 \cr
					 & 0 0 0 \ldots 0 0 1 1  \ldots 0 \cr
					 & \ddots \cr
					& 0 0 0 \ldots 0 0 \ldots 0 1 1 \cr
		\textbf{R} & 0 0 0 \ldots 0 0 \ldots 0 0 1 }\
	\begin{matrix}
	  \textcolor{dark-green}{\bullet} \\ \\ \\ \\ \\ \\ \\ \\ \\ \\ \\ \textcolor{dark-green}{\bullet}
	  \end{matrix}
	\end{align*}
	\begin{align*}
		P_2(k,0) =  \bordermatrix{ & \cr
		\textbf{L} & 1 1 0 0 0 0 \ldots 0 0 0 \cr
		\textbf{LR} & 1 0 1 1 1 1 \ldots 1 1 1 \cr
		\textbf{LR} & 1 1 1 0 1 1 \ldots 1 1 1 \cr
		\textbf{LR} & 1 1 0 1 1 1 \ldots 1 1 1 \cr
		\textbf{LR} & 1 1 1 0 0 1 \ldots 1 1 1 \cr
					     & 0 0  0 0 1  1 \ldots 0 0 0 \cr
						 &  \ddots  \cr
						& 0 0 0 0 0 \ldots  0 1 1  \cr
		\textbf{R} & 0 0 0 0 0 \ldots 0 0 1  }\
	 \begin{matrix}
	  \textcolor{dark-green}{\bullet} \\ \\ \\ \\ \\ \\ \\ \\ \\ \\ \textcolor{dark-green}{\bullet}
	  \end{matrix}
	&&
			P_2(k,l) = \bordermatrix{ & \cr
		\textbf{L}& 1 0 0 \ldots 0 0 0 0 0 \ldots 0 \cr
					  & 1 1 0 \ldots 0 0 0 0 0  \ldots 0 \cr
					  &  \ddots \cr
					 & 0 0 0 \ldots 1 1 0 0 0  \ldots 0 \cr
	\textbf{LR} & 1 1 1 \ldots 1 0 0 1 1 \ldots 1 \cr
	\textbf{LR} & 1 1 1 \ldots 1 1 1 0 1 \ldots 1 \cr
	\textbf{LR} & 1 1 1 \ldots 1 1 0 1 1 \ldots 1 \cr
	\textbf{LR} & 1 1 1 \ldots 1 1 0 0 1 \ldots 1 \cr
					 & 0 0 0 \ldots 0 0 0 1 1 \ldots 0 \cr
					 & \ddots \cr
					& 0 0 0 \ldots 0 0 0 \ldots 0 1 1 \cr
		\textbf{R} & 0 0 0 \ldots 0 0 0  \ldots 0 0 1 }\
	\begin{matrix}
	  \textcolor{dark-green}{\bullet} \\ \\ \\ \\ \\ \\ \\ \\\ \\ \\ \\ \\ \\ \\ \textcolor{dark-green}{\bullet}
	  \end{matrix}
	\end{align*}

	}
	\caption{The family of enriched matrices $\mathcal{P}$ for every odd $k$.} 
	\label{fig:forb_P}
	\end{figure}

In the matrices $\mathcal{P}$, the integer $l$ represents the number of unlabeled rows between the first row and the first LR-row. The matrices $\mathcal{P}$ described in Figure \ref{fig:forb_P} are defined as follow: 
$P_0(k,0) \in \{0,1\}^{k \times k}$ for every $k \geq 4$, $P_0(k,l) \in \{0,1\}^{k \times (k-1)}$ for every $k \geq 5$ and $l >0$; 
$P_1(k,0) \in \{0,1\}^{k \times (k-1)}$ for every $k \geq 5$, $P_1(k,l) \in \{0,1\}^{k \times (k-2)}$ for every $k \geq 6$, $l > 0$;
$P_2(k,0) \in \{0,1\}^{k \times (k-1)}$ for every $k \geq 7$, $P_2(k,l) \in \{0,1\}^{k \times (k-2)}$ for every $k \geq 8$ and $l > 0$. 
If $k$ is even, then the first and last row of every matrix in $\mathcal{P}$ are colored with distinct colors.

\begin{figure}[h]
	\centering
	\includegraphics[scale=.495]{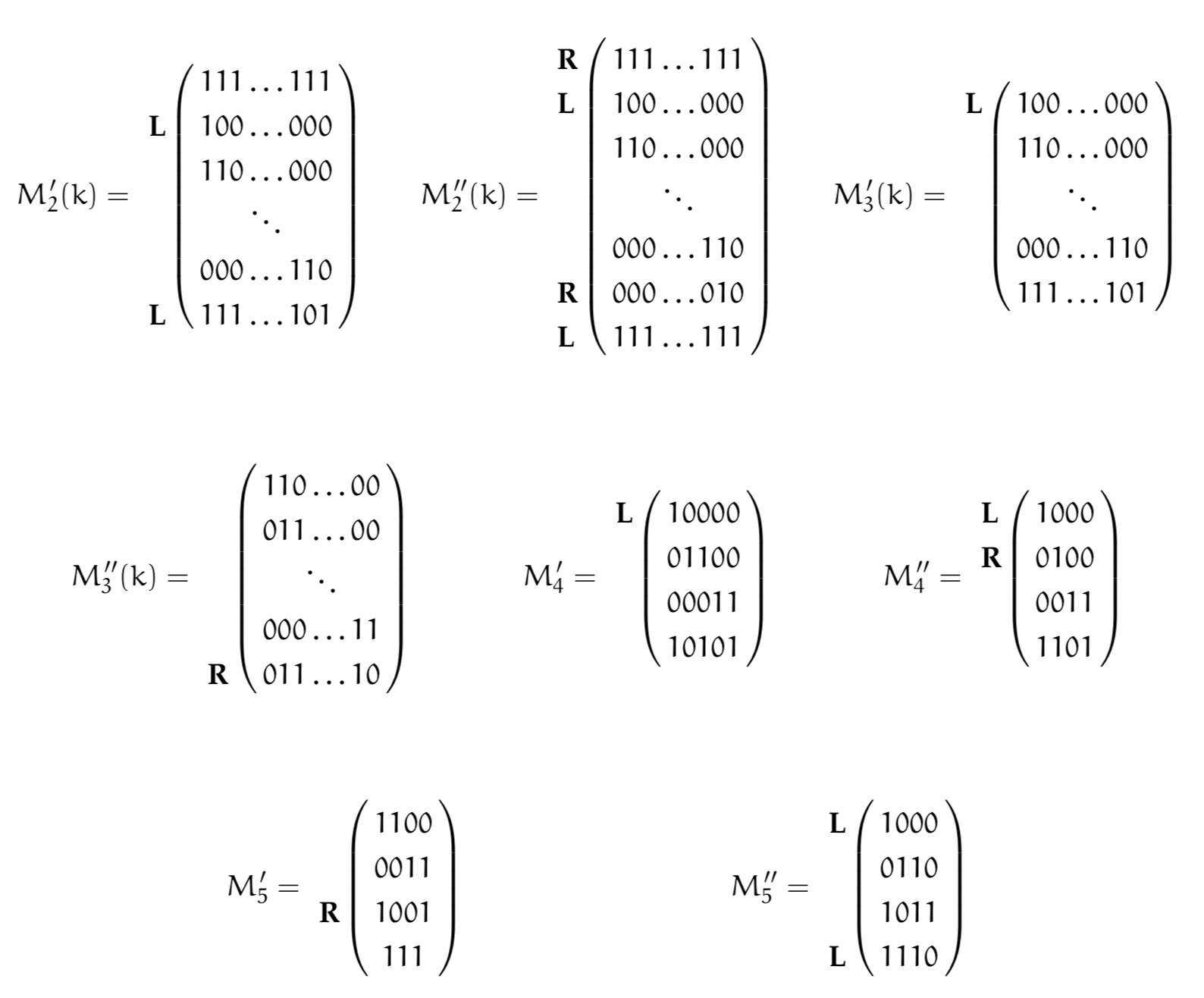}
	\caption{The enriched matrices in family $\mathcal{M}$: $M_{2}'(k)$, $M_{3}'(k)$, $M_{3}''(k)$, $M_{3}'''(k)$ for $k \geq 4$, and $M_{2}''(k)$ for $k \geq 5$. }  \label{fig:forb_LR-orderable}
\end{figure}

\begin{figure}[H]
	\centering
	\footnotesize{
	\begin{align*}
			M_0 = \bordermatrix{ & \cr
						 & 1 0 1 1 \cr
		  				 & 1 1 1 0  \cr
						 & 0 1 1 1 }\
			&&
			M_{II}(4) = \bordermatrix{ & \cr
						 & 0 1 1 1 \cr
		  				 & 1 1 0 0  \cr
						 & 0 1 1 0 \cr
						 & 1 1 0 1 }\
			&&
			M_V = \bordermatrix{ & \cr
						 & 1 1 0 0 0 \cr
		  				 & 0 0 1 1 0  \cr
						 & 1 1 1 1 0 \cr
						 & 1 0 0 1 1 }\
			&&
			S_0(k)&= \begin{pmatrix}
			 	1 11...11\\
				110...00\\
				011...00\\
				.   .   .   .   . \\
				.   .   .   .   . \\
				.   .   .   .   . \\
				000...11\\
				100...01\\
			\end{pmatrix}
	\end{align*}
	}
	\caption{The matrices $M_0$, $M_{II}(4)$, $M_V$ and $S_0(k) \in \{0,1\}^{((k+1)\times k}$ for any even $k \geq 4$.} 
	\label{fig:forb_M_chiquitas}
	\end{figure}

Now we are in conditions to state Theorem \ref{teo:2-nested_caract_bymatrices}, which characterizes $2$-nested matrices by forbidden subconfigurations and is the main result of this chapter. The proof for this theorem will be given at the end of the chapter.
 
\begin{teo}[label={teo:2-nested_caract_bymatrices}]
	Let $A$ be an enriched matrix. Then, $A$ is $2$-nested if and only if $A$ contains none of the following listed matrices or their dual matrices as subconfigurations: 
	\begin{itemize}		
	\item $M_0$, $M_{II}(4)$, $M_V$ or $S_0(k)$ for every even $k$ (See Figure \ref{fig:forb_M_chiquitas})
	\item Every enriched matrix in the family $\mathcal{D}$ (See Figure \ref{fig:forb_D}) 
	\item Every enriched matrix in the family $\mathcal{F}$ (See Figure \ref{fig:forb_F}) 
	\item Every enriched matrix in the family $\mathcal{S}$ (See Figure \ref{fig:forb_Si})
	\item Every enriched matrix in the family $\mathcal{P}$ (See Figure \ref{fig:forb_P})
	\item Monochromatic gems, monochromatic weak gems, badly-colored doubly-weak gems 
	\end{itemize}
and $A^*$ contains no Tucker matrices and none of the enriched matrices in $\mathcal{M}$ or their dual matrices as subconfigurations (See Figure \ref{fig:forb_LR-orderable}).

\end{teo} 

Throughout the following sections we will give some definitions and characterizations that will allow us to prove this theorem. In Section \ref{section:admissibility} we will define and characterize the notion of admissibility, which englobes all the properties we need to consider when coloring the blocks of an enriched matrix. In Section \ref{section:part2nested}, we give a characterization for LR-orderable matrices by forbidden subconfigurations. Afterwards, we define and characterize partially $2$-nested matrices, which are those enriched matrices that admit an LR-ordering and for which the given pre-coloring of those labeled rows induces a partial block bi-coloring.
These definitions and characterizations allow us to prove Lemmas \ref{lema:2-nested_if} and \ref{lema:2-nested_onlyif}, which are of fundamental importance for the proof of Theorem \ref{teo:2-nested_caract_bymatrices}.

\section{Admissibility} \label{section:admissibility}

In this section we will define the notion of admissibility for an enriched $(0,1)$-matrix, which will allow us to characterize those enriched matrices for which there is a total block bi-coloring for $A$. In the next chapter, we will see that such a block bi-coloring is a necessary condition to give a circle model.

Notice that the existance of a block bi-coloring for an enriched matrix is a property that can be defined and characterized by subconfigurations and forbidden submatrices. 

Let us consider the matrices defined in \ref{fig:forb_D}. The matrices in this family are all examples of enriched matrices that do not admit a total block bi-coloring as defined in Definition \ref{def:2-nested}.
 
For example, let us consider $D_0$. In order to have a block bi-coloring for every block of $D_0$, it is necessary that $D_0$ admits an LR-ordering of its columns. In particular, in such an ordering every row labeled with L starts in the first column. Hence, if there is indeed an LR-ordering for $D_0$, then the existance of two distinct non-nested rows labeled with L is not possible. The same holds if both rows are labeled with R.
We can use similar arguments to see that $D_2$, $D_3$, $D_7$ and $D_{11}$ do not admit an LR-ordering. 

Let us consider the matrix $D_1$. In this case, we see that condition \ref{item:2nested5} does not hold for the enriched matrix $D_1$.

Consider now the matrix $D_4$. It follows from property \ref{item:2nested8} that if an enriched matrix has two distinct rows labeled with L and colored with distinct colors, then every LR-row has an L-block, and thus $D_4$ does not admit a total block-bi-coloring.
Suppose now that $D_4$ is a submatrix of some enriched matrix and that the LR-row is nonempty. Notice that, if the LR-row has an L-block then it is properly contained in both rows labeled with L. It follows from this and property \ref{item:2nested3} of the definition of $2$-nested that, that the L-block of the LR-row must be colored with a distinct color than the one given to each row labeled with L. However, each of these rows is colored with a distinct color, thus a total block-bi-coloring is not possible in that case. 

If we consider the enriched matrix $D_5$, then it follows from property \ref{item:2nested4} that there is no possible LR-ordering such that the L-block of the LR-row does not intersect the L-row, and the same follows for the R-block of the LR-row and the R-row of $D_5$.

Let us consider an the enriched matrix in which we find $D_6$ as a subconfiguration.If the LR-row has an L-block, then it is contained in the L-row, and the same holds for the R-block of the LR-row and the R-row. By property \ref{item:2nested3}, the L-block must be colored with a distinct color than the L-row, and the R-block must be colored with a distinct color than the R-row. Equivalently, the L-block and the R-block of the LR-row are colored with the same color. However, this is not possible by property \ref{item:2nested1}.
Similarly, we can see that $D_8$, $D_9$, $D_{10}$, $D_{12}$ and $D_{13}$ do not admit a total block bi-coloring, also having in mind that the property \ref{item:2nested9} must hold pairwise for LR-rows.

\vspace{1mm}

\begin{defn} \label{lista:prop_adm}
Let $A$ be an enriched matrix. We define the following list of properties: 
	\begin{enumerate}[${(Adm}_\bgroup1\egroup)$]
	    \item If two rows are labeled both with L or both with R, then they are nested.  \label{item:0_def_adm}
	    
	 	\item If two rows with the same color are labeled one with L and the other with R, then they are disjoint. \label{item:1_def_adm}

	 	\item If two rows with distinct colors are labeled one with L and the other with R, then either they are disjoint or there is no column where both have $0$ entries. \label{item:2_def_adm}

		\item If two rows $r_1$ and $r_2$ have distinct colors and are labeled one with L and the other with R, then any LR-row with at least one non-zero column has nonempty intersection with either $r_1$ or $r_2$. \label{item:3_def_adm}
		
	 	\item If two rows $r_1$ and $r_2$ with distinct colors are labeled both with L or both with R, then for any LR-row $r$, $r_1$ is contained in $r$ or $r_2$ is contained in $r$. \label{item:4_def_adm}		

	
		\item If two non-disjoint rows $r_1$ and $r_2$ with distinct colors, one labeled with L and the other labeled with R, then any LR-row is disjoint with regard to the intersection of $r_1$ and $r_2$. \label{item:5_def_adm}

	
		\item If two rows with the same color are labeled one with L and the other with R, then for any LR-row $r$ one of them is contained in $r$. 
		Moreover, the same holds for any two rows with distinct colors and labeled with the same letter. \label{item:6_def_adm}
		

		\item For each three non-disjoint rows such that two of them are LR-rows and the other is labeled with either L or R, two of them are nested. \label{item:7_def_adm}

		\item If two rows $r_1$ and $r_2$ with distinct colors are labeled one with L and the other with R, and there are two LR-rows $r_3$ and $r_4$ such that $r_1$ is neither disjoint or contained in $r_3$ and $r_2$ is neither disjoint or contained in $r_4$, then $r_3$ is nested in $r_4$ or viceversa. \label{item:8_def_adm}

		\item For each three LR-rows, two of them are nested. \label{item:9_def_adm}
    \end{enumerate} 	
    
\end{defn}
    
\vspace{2mm}
For each of the above properties, we will characterize the set of minimal forbidden subconfigurations with the following Lemma.

\begin{lema}
	For any enriched matrix $A$, all of the following assertions hold:
	\begin{enumerate}
		\item $A$ satisfies \ref{item:0_def_adm} if and only if $A$ contains no $D_0$ or its dual matrix as a subconfiguration.
		\item $A$ satisfies \ref{item:1_def_adm} if and only if $A$ contains no $D_1$ or its dual matrix as a subconfiguration.
		\item $A$ satisfies \ref{item:2_def_adm} if and only if $A$ contains no $D_2$ or its dual matrix as a subconfiguration.
		\item $A$ satisfies \ref{item:3_def_adm} if and only if $A$ contains no $D_2$, $D_3$ or their dual matrices as subconfigurations.
		\item $A$ satisfies \ref{item:4_def_adm} if and only if $A$ contains no $D_0$, $D_4$ or their dual matrices as subconfigurations.
		\item $A$ satisfies \ref{item:5_def_adm} if and only if $A$ contains no $D_5$ or its dual matrix as a subconfiguration.
		\item $A$ satisfies \ref{item:6_def_adm} if and only if $A$ contains no $D_0$, $D_1$, $D_4$, $D_6$ or their dual matrices as subconfigurations.
		\item $A$ satisfies \ref{item:7_def_adm} if and only if $A$ contains no $D_7$, $D_8$, $D_9$ or their dual matrices as subconfigurations.		
		\item $A$ satisfies \ref{item:8_def_adm} if and only if $A$ contains no $D_5$, $D_9$, $D_{10}$ or its dual matrix as a subconfiguration.		
		\item $A$ satisfies \ref{item:9_def_adm} if and only if $A$ contains no $D_{11}$, $D_{12}$, $D_{13}$ or their dual matrices as subconfigurations.	
	\end{enumerate}
\end{lema}

\begin{proof}

\vspace{1mm}
First, we will find every forbidden subconfiguration given by statement \ref{item:0_def_adm}.

Let $f_1$ and $f_2$ be two rows labeled with the same letter, and suppose they are not nested. Thus, there is a column in which $f_1$ has a $1$ and $f_2$ has a $0$, and another column in which $f_2$ has a $1$ and $f_1$ has a $0$. Since the color of each row is irrelevant in the definition, we find $D_0$ as a forbidden subconfiguration in $A$.

Let us find now every forbidden subconfiguration given by statement \ref{item:1_def_adm}.
Let $f_1$ and $f_2$ be rows labeled with distinct letters and colored with the same color. If $f_1$ and $f_2$ are not disjoint, then there is a column in which both rows have a $1$. In this case, we find $D_1$ as a forbidden subconfiguration in $A$.

For statement \ref{item:2_def_adm}, let $f_1$ and $f_2$ be two rows labeled with distinct letters and colored with distinct colors, and suppose they are not disjoint and there is a column $j_1$ such that both rows have a $0$ in column $j_1$. Thus, there is a column $j_2 \neq j_1$ such that both rows have a $1$ in column $j_2$.
If $f_1$ and $f_2$ have the same color, then we find $D_1$ as a subconfiguration. Hence, $D_2$ is a forbidden subconfiguration in $A$.

With regard to statement \ref{item:3_def_adm}, let $f_1$ and $f_2$ be two rows labeled with distinct letters and colored with distinct colors. Let $f_3$ be a non-zero LR-row. Suppose that $f_3$ is disjoint with both $f_1$ and $f_2$. 
Hence, there is a column $l_1$ such that $f_1$ and $f_2$ have a $0$ and $f_3$ has a $1$. Moreover, either there are two distinct columns $j_1$ and $j_2$ such that the column $j_i$ has a $1$ in row $f_i$ and a $0$ in the other rows, for $i=1, 2$, or there is a column $l_2$ such that $f_1$ and $f_2$ both have a $1$ in column $l_2$ and $f_3$ has a $0$.
If the last statement holds, we find $D_2$ as a subconfiguration considering only the submatrix given by the rows $f_1$ and $f_2$.
If instead there are two distinct columns $j_1$ and $j_2$ as described above, then we find $D_3$ as a minimal forbidden subconfiguration in $A$.

For statement \ref{item:4_def_adm}, let $f_1$ and $f_2$ be two rows labeled with L and colored with distinct colors, and let $r$ be an LR-row. If $f_1$ and $f_2$ are not nested, then we find $D_0$. Suppose that $f_1$ and $f_2$ are nested. If neither $f_1$ or $f_2$ are contained in $r$, then there is a column $j$ in which $f_1$ and $f_2$ have a $1$ and $r$ has a $0$.
Thus, $D_4$ is a forbidden subconfiguration in $A$.

For statement \ref{item:5_def_adm}, let $f_1$ and $f_2$ be two non-disjoint rows colored with distinct colors, $f_1$ labeled with L and $f_2$ labeled with R. Since they are non-disjoint, there is at least one column $j$ in which both rows have a $1$. Suppose that for every such column $j$, there is an LR-row $f$ having a $1$ in that column. Then, we find $D_5$ as a subconfiguration in $A$.

For statement \ref{item:6_def_adm}, let $f$ be an LR-row and let $f_1$ and $f_2$ be two rows labeled with L and R respectively, and colored with the same color. If $f_1$ and $f_2$ are not disjoint, then we find $D_1$. Suppose that $f_1$ and $f_2$ are disjoint.
If neither $f_1$ is contained in $f$ nor $f_2$ is contained in $f$, then there are columns $j_1 \neq j_2$ such that $f_i$ has a $1$ and $f$ has a $0$, for $i=1, 2$. Thus, we find $D_6$ as a subconfiguration of $A$. If instead $f_1$ and $f_2$ are both labeled with L and colored with distinct colors, and neither is contained in $f$, then either they are not nested --in which case we find $D_0$-- or we find $D_4$ in $A$.

Suppose that $A$ satisfies \ref{item:7_def_adm}. Let $f_1$ be a row labeled with L, and $f_2$ and $f_3$ two distinct LR-rows such that none of them are nested in the others. Thus, we have three possibilities. If there are three columns $j_i$ $i=1,2,3$ such that $f_i$ has a $1$ and the other rows have a $0$, then we find $D_7$ as a subconfiguration of $A$.
If instead there are three rows $j_i$, $i=1,2,3$ such that $f_i$ and $f_{i+1}$ have a $1$ and $f_{i+2}$ has a $0$ in $j_i$ (mod 3), then we find $D_8$ as a subconfiguration.
The remaining possibility, is that there are 4 columns $j_1, j_2, j_3, j_4$ such that $f_1$ and $f_2$ have a $1$ and $f_3$ has a $0$ in $j_1$, $f_1$ has a $1$ and $f_2$ and $f_3$ have a $0$ in $j_2$, $f_3$ has a $1$ and $f_1$ and $f_2$ have a $0$ in $j_3$, and $f_2$ and $f_3$ have a $1$ and $f_1$ has a $0$ in $j_4$. Moreover, since all three rows are pairwise non-disjoint, either there is a fifth column for which $f_1$ and $f_3$ have a $1$ and $f_2$ has a $0$ (in which case we find $D_8$), or $f_2$ has a $1$ and $f_1$ and $f_3$ have a $0$ (in which case we have $D_7$), or all three rows have a $1$ in such column. In this case, we find $D_9$ has a subconfiguration of $A$.

For statement \ref{item:8_def_adm}, let $f_1$ and $f_2$ be two rows labeled with L and R, respectively, and colored with distinct colors. Let $f_3$ and $f_4$ be two LR-rows such that $f_1$ is neither disjoint or contained in $f_3$ and $f_2$ is neither disjoint or contained in $f_4$. 
If $f_1$ is also not contained in $f_4$ or $f_2$ is not contained in $f_3$, then we find $D_9$. Thus, suppose that $f_1$ is contained in $f_4$ and $f_2$ is contained in $f_3$. Moreover, we may assume that for any column such that $f_1$ and $f_3$ have a $1$, $f_2$ has a $0$, (and analogously for $f_2$ and $f_4$ having a $1$ and $f_1$), for if not we find $D_5$.
Hence, there is a column $j_1$ in $A$ having a $1$ in $f_1$ and $f_4$ and having a $0$ in $f_3$ and $f_2$, and another column $j_2$ having a $1$ in $f_2$ and $f_3$ and having a $0$ in $f_1$ and $f_4$. Moreover, since $f_1$ and $f_3$ are not disjoint and $f_2$ and $f_4$ are not disjoint (and $f_1$ is nested in $f_4$ and $f_2$ is nested in $f_3$), then there are columns $j_3$ and $j_4$ such that $f_1$, $f_3$ and $f_4$ have a $1$ and $f_2$ has a $0$ in $j_3$ and  $f_2$, $f_3$ and $f_4$ have a $1$ and $f_1$ has a $0$. Therefore, we find $D_{10}$ as a subconfiguration of $A$.

It follows by using a similar argument as in the previous statements that, if $A$ satisfies \ref{item:9_def_adm}, then that there are no $D_{11}$, $D_{12}$ or $D_{13}$ in $A$.

\end{proof}

\begin{cor}
Every enriched matrix $A$ that admits a total block bi-coloring contains none of the matrices in $\mathcal{D}$. Equivalently, if $A$ admits a total block bi-coloring, then every property listed in \ref{lista:prop_adm} hold. 
\end{cor}

Another example of families of enriched matrices that do not admit a total block bi-coloring are $\mathcal{S}$ and $\mathcal{P}$, which are the matrices shown in Figures \ref{fig:forb_Si} and \ref{fig:forb_P}, respectively.
Therefore, since the existance of a total block bi-coloring is a property that must hold for every subconfiguration of an enriched matrix, if an enriched matrix $A$ admits a total block bi-coloring, then it is a necessary condition that $A$ contains none of the matrices in $\mathcal{S}$ or $\mathcal{P}$. 
With this in mind, we give the following definition, which is also a characterization by forbidden subconfigurations.

\begin{defn} \label{teo:caract_admissible}
	Let $A$ be an enriched matrix. We say $A$ is \emph{admissible} if and only if $A$ is $ \{ \mathcal{D}, \mathcal{S}, \mathcal{P} \}$-free.
\end{defn}

\section{Partially 2-nested matrices} \label{section:part2nested}

This section is organized as follows. First, we give some definitions and examples that will help us obtain a characterization of LR-orderable matrices by forbidden subconfigurations, which were defined in \ref{def:LR-orderable}. 
Afterwards, we define and characterize partially $2$-nested matrices, which are those enriched matrices that admit an LR-ordering and for which the given pre-coloring of those labeled rows of $A$ induces a partial block bi-coloring.

\begin{defn}
A \emph{tagged matrix} is a matrix $A$, each of whose rows are either uncolored or colored with blue or red, together with a set of at most two distinguished columns of $A$. The distinguished columns will be refered to as \emph{tag columns}.
\end{defn}

\begin{defn}  \label{def:tagged_matrixA}
Let $A$ be an enriched matrix. We define the \emph{tagged matrix of $A$} as a tagged matrix, denoted by $A_\tagg$, whose underlying matrix is obtained from $A$ by adding two columns, $c_L$ and $c_R$, such that:
(1) the column $c_L$ has a $1$ if $f$ is labeled L or LR and $0$ otherwise, 
(2) the column $c_R$ has a $1$ if $f$ is labeled R or LR and $0$ otherwise, and 
(3) the set of distinguished columns of $A_{\tagg}$ is $\{ c_L, c_R\}$. 
We denote $A^*_\tagg$ to the tagged matrix of $A^*$. By simplicity we will consider column $c_L$ as the first and column $c_R$ as the last column of $A_\tagg$ and $A^*_\tagg$.
\end{defn}

\begin{figure}[h!]
\begin{align*}
	A = \bordermatrix{ &  \cr
		\textbf{LR} & 1  0  0  0  1 \cr
		\textbf{LR} & 1  1  0  0  1 \cr
						& 0  1 1  0  0 \cr
		\textbf L & \textcolor{red}{1}  \textcolor{red}{1}  \textcolor{red}{1}  0  0 \cr
		\textbf{LR} & 0  0  0   0  0 \cr
		\textbf R & 0  0  \textcolor{blue}{1} \textcolor{blue}{1} \textcolor{blue}{1} }\,
		\begin{matrix} 
	   \cr  \cr \cr  \cr \textcolor{red}{\bullet} \cr \textcolor{blue}{\bullet} \cr \textcolor{blue}{\bullet} 
	\end{matrix}
	&&
	A_\tagg = \bordermatrix{& c_L \hspace{7mm}  c_R \cr
	& \pmb 1 1 0 0 0 1 \pmb 1 \cr
	& \pmb 1 1 1 0 0 1 \pmb 1 \cr
 	& \pmb 0 0 1 1 0 0 \pmb 0 \cr
	& \pmb 1 \textcolor{red}{1}  \textcolor{red}{1}  \textcolor{red}{1}  0  0  \pmb 0 \cr
	& \pmb 1 0 0 0 0 0 \pmb 1 \cr
	& \pmb 0 0  0  \textcolor{blue}{1} \textcolor{blue}{1} \textcolor{blue}{1} \pmb 1 }\,
	\begin{matrix} 
	   \cr  \cr \cr  \cr \textcolor{red}{\bullet} \cr \textcolor{blue}{\bullet} \cr \textcolor{blue}{\bullet} 
	\end{matrix}
	&&
	A^{*}_\tagg = \bordermatrix{& c_L \hspace{7mm} c_R \cr
	& \pmb 0 0 1 1  1 0  \pmb 0 \cr
	& \pmb 0 0 0 1 1 0  \pmb 0 \cr
 	& \pmb 0 0 1 1 0 0 \pmb 0 \cr
	& \pmb 1 \textcolor{red}{1} \textcolor{red}{1} \textcolor{red}{1} 0 0 \pmb 0 \cr
	& \pmb 0 0 0 \textcolor{blue}{1} \textcolor{blue}{1} \textcolor{blue}{1} \pmb 1  }\, 
	\begin{matrix} 
	   \cr  \cr \cr  \textcolor{red}{\bullet} \cr \textcolor{blue}{\bullet} \cr \textcolor{blue}{\bullet} 
	\end{matrix}
\end{align*}
\caption{Example of a matrix $A$ and the matrices $A_\tagg$ and $A^*_\tagg$} \label{fig:example_tagg}
\end{figure}

The following remarks will allow us to give a simpler proof for the characterization of LR-orderable matrices.

\begin{remark}
	If $A^*_{\tagg}$ has the C$1$P, then the distinguished rows force the tag columns $c_L$ and $c_R$ to be the first and last columns of $A_{\tagg}$, respectively.
\end{remark}

\begin{remark}
	An admissible matrix $A$ is LR-orderable if and only if the tagged matrix $A^*_{\tagg}$ has the C$1$P for the rows.
\end{remark}


\begin{figure}[h!]
\small{
	\centering
	\begin{align*}
			M_2'(k) &= \begin{pmatrix}
				\pmb 0 1 1  \ldots 1 1 1\\
				\pmb 1  1 0  \ldots 0 0 0\\
				\ddots  \\
				\pmb 0 0 0 \ldots 1 1 0 \\
				\pmb 1 1 1 \ldots 1 0 1 \\
			\end{pmatrix}
			&
			M_2''(k) &= \begin{pmatrix}
				\pmb 0  1  1  \ldots  1 1 \pmb 1  \\
				\pmb 1  1  0  \ldots  0 0  \pmb 0 \\
				\ddots \\
				\pmb 0  0  0 \ldots  1 0 \pmb 1  \\
				\pmb 1  1  1  \ldots  1 1  \pmb 0  \\
			\end{pmatrix}	
			&
			M_3'(k) &= \begin{pmatrix}
				\pmb 1   1   0   \ldots   0   0   0 \\
				\pmb 0   1   1   \ldots   0   0   0 \\
				\ddots \\
				\pmb 0   0   0   \ldots   1   1   0 \\
				\pmb 0   1   1   \ldots   1   0   1 \\
			\end{pmatrix}
			\end{align*}
		\begin{align*}
			M_3''(k) &= \begin{pmatrix}
				1   1   0   \ldots   0   0   \pmb 0 \\
				0   1   1   \ldots   0   0   \pmb 0 \\
				\ddots \\
				0   0   0   \ldots   1   1   \pmb 0 \\
				0   1   1   \ldots   1   0   \pmb 1 \\
			\end{pmatrix}
			&
			M_4' &= \begin{pmatrix}
					\pmb 1   1   0   0   0   0 \\ 
					\pmb 0   0   1   1   0   0 \\
					\pmb 0   0   0   0   1   1 \\
					\pmb 0   1   0   1   0   1 \\
			\end{pmatrix}
			&
			M_4'' &= \begin{pmatrix}
					\pmb 1   1   0   0   0  \pmb 0  \\ 
					\pmb 0   0   1   0   0  \pmb 1  \\
					\pmb 0   0   0   1   1 \pmb 0  \\
					\pmb 0   1    1   0   1   \pmb 0 \\
			\end{pmatrix}		
	\end{align*}
	\begin{align*}	
			M_5' &= \begin{pmatrix}
					1   1   0   0   \pmb 0 \\ 
					0   0   1   1   \pmb 0 \\
					1   0   0   1   \pmb 1 \\
					1   1   1   1   \pmb 0 \\
			\end{pmatrix}
			&
			M_5'' &= \begin{pmatrix}
					  \pmb 1  1 0  0  0 \\ 
					  \pmb 0  0 1  1  0 \\
					  \pmb 0  1  0  1  1 \\
					  \pmb 1  1 1  1  0 \\
			\end{pmatrix}
	\end{align*}	
		}
	\caption{The tagged matrices of the family $\mathcal{M}$}  \label{fig:forb_LR-orderable_tags}
\end{figure}

\begin{teo} \label{teo:LR-orderable_caract_bymatrices}
	An admissible matrix $A$ is LR-orderable if and only if the tagged matrix $A^*_{\tagg}$ does not contain any Tucker matrices, nor $M_2'(k)$, $M_2''(k)$, $M_3'(k)$, $M_3''(k)$ for $k\geq 3$, $M_4'$, $M_4''$, $M_5'$, $M_5''$ as subconfigurations.  
\end{teo} 

\begin{proof}
$\Rightarrow )$ This follows from the last remark.

$\Leftarrow )$ Suppose that the tagged matrix $A_{\tagg}$ does not contain any of the above listed submatrices as subconfigurations, and still the C$1$P does not hold for the rows of $A_{\tagg}$. 

Hence, there is a Tucker matrix $M$ such that $M$ is a submatrix of $A_{\tagg}$.

Suppose without loss of generality that, if $M$ intersects only one tag column, then this tag column is $c_L$, since the analysis is symmetric if assumed otherwise and gives as a result in each case the dual matrix.

\vspace{3mm}
\begin{mycases}
\case \textit{Suppose first that $M$ intersects one or both of the distinguished rows.} 
Thus, the underlying matrix of $M$ (i.e., the matrix without the tags) is either \textit{$M_V$, or $M_I(3)$, or $M_{II}(k)$ for some $k \geq 3$.} We consider each case separately. 

\vspace{2mm}
\subcase $M_V = \begin{pmatrix}
				11000\\
				00110\\
				11110\\
				10011\\
\end{pmatrix}$

\vspace{2mm}
In this case, the distinguished row is $(1,1,1,1,0)$ and thus the last column is a tag column. 
Hence $M = M_5'$, which results in a contradiction. 

\vspace{2mm}
\subcase $M_I(3) = \begin{pmatrix}
				110\\
				011\\
				101\\
\end{pmatrix}$

\vspace{2mm}

If $(1,1,0)$ is a distinguished row, then we find $D_0$ as a forbidden submatrix given by the second and third rows. It is symmetric if the distinguished row is either the second or the third row, and therefore this case is not possible.

\vspace{2mm}
\subcase $M_{II}(k) = \begin{pmatrix}
				011...111\\
				110...000\\
				011...000\\
				.   .   .   .   . \\
				.   .   .   .   . \\
				000...110\\
				111...101\\
			\end{pmatrix}$

\vspace{2mm}

In this case, the distinguished rows can be only the first and the last row.

Suppose only the first row $(0, 1, \ldots, 1)$ of $M$ is a distinguished row. Thus, the first column is a tag column. 

Hence, $M_2'(k)$ is a submatrix of $A_{\tagg}$, and this results in a contradiction. The same holds if instead the last row is the sole distinguished row.

Finally, suppose both the first and the last row are distinguished. If this is the case, then the columns $1$ and $k-1$ are tag columns.

Suppose first that $M = M_{II}(4)$. Since every row is a labeled row, then every row is colored. Moreover, the first and second row have distinct colors, for if not we find $D_1$ as a submatrix. The same holds for the second and third row, and also for the third and fourth row.  
However, this implies that the second and third row induce $D_2$, hence this case is not possible.

If instead $M = M_{II}(k)$ for $k \geq 5$, then $M_2''(k)$ is a submatrix of $A_{\tagg}$, and thus we reached a contradiction. 

\vspace{3mm}
\case \textit{Suppose that $M$ does not intersect any distinguished row.}

If $M$ does not have any tag column, then $M$ is a submatrix of $A$. Thus, $A$ does not have the C$1$P and we conclude that $M$ is a Tucker matrix.

Suppose that instead one of the columns in $M$ is a tag column.

\vspace{2mm}
 \subcase $M_I(k) = \begin{pmatrix}
				110...00\\
				011...00\\
				.   .   .   .   . \\
				.   .   .   .   . \\
				.   .   .   .   . \\
				000...11\\
				100...01\\
			\end{pmatrix}$ for some $k \geq 3$.

\vspace{2mm}
Notice that, if any of the columns is a tag column, then we find $D_0$ as a submatrix, which results in $A$ not being admissible and thus reaching a contradiction.

\vspace{2mm}
\subcase $M_{II}(k) = \begin{pmatrix}
				011...111\\
				110...000\\
				011...000\\
				.   .   .   .   . \\
				.   .   .   .   . \\
				000...110\\
				111...101\\
			\end{pmatrix}$ for some $k \geq 4$

\vspace{2mm}
As in the previous case, some of the columns are not elegible for being tag columns. If there is only one tag column, the only remaining possibilities for tag columns are column $1$ or column $k-1$, for in any other case we find $D_0$ as a submatrix. Analogously, if instead $M$ intersects both tag columns, then these columns are also columns $1$ and $k-1$.

However, if $c_L$ is either column $1$ or column $k-1$, then $M_2''(k)$ is a submatrix of $A_{\tagg}$. Notice that we can reorder the columns of $M_{II}(k)$ to have the same disposition of the rows by taking column $k-1$ as the first column. Analogously, if $c_R$ is either column $1$ or $k-1$, then we find the dual matrix of $M_2'(k)$ as a submatrix.

Finally, suppose that both columns are tag columns.
Notice that the first and second rows are colored with distinct colors, for if not we find $D_1$ as a submatrix. The same holds for the last two rows of $M$.
Hence, if $k = 4$, then we find $D_2$ as a submatrix given by the second and third rows. If instead $k > 5$, then $M_2''(k)$ is a submatrix of $A_{\tagg}$, which results once more in a contradiction.
 
\vspace{2mm}
\subcase $M_{III}(k)= \begin{pmatrix}
				110...000\\
				011...000\\
				.   .   .   .   . \\
				.   .   .   .   . \\
				000...110\\
				011...101\\
			\end{pmatrix}$ for some $k \geq 3$.
		
\vspace{2mm}
In this case, the only possibilities for tag columns are column 1, column $k-1$ and column $k$, for if not we find $D_0$ as a submatrix. Once more, it is easy to see that we can reorder the columns in such a way to have the same disposition of the rows with column $k-1$ or column $k$ replacing column 1.

Suppose first that the tag column is the first column. In that case, we find $M_3'(k)$ as a submatrix of $M$, which also results in a contradiction since $M$ is admissible.

\vspace{1mm}
If instead the tag column is column $k$, then we use an analogous reasoning to find $M_3''(k)$ as a submatrix and thus reaching a contradiction.

\vspace{2mm}
Suppose now that both the first column and the last column of $M$ are tag columns.

Since $M$ is admissible, this case is not possible for the first and last row induce $D_1$ or $D_2$ as submatrices, depending on whether the rows are colored with the same color or with distinct colors, respectively.

\vspace{2mm}
\subcase $M_{IV}= \begin{pmatrix}
				110000\\
				001100\\
				000011\\
				010101\\
			\end{pmatrix}$

\vspace{2mm}

In this case, the only elegible columns for being tag columns are column 1, column 3 and column 5, since if any other column is a tag column, we find $D_0$ as a submatrix, thus contradicting the hypothesis of pre-admissibility for $M$ and thus for $A$. 
Further\-more, the election of the tag column is symmetric since there is a reordering of the rows that allows us to obtain the same matrix if the tag column is either column 1, column 3 or column 5, disregarding the election of the column. 
Hence, we have two possibilities: when column $1$ is the sole tag column of $M$, and when the two tag columns are columns $1$ and 3.
If column $1$ is the only tag column, then we find $M_4'$ as a submatrix.
If instead the columns $1$ and 3 are both tag columns, then the first row and the second row are colored with the same color, for if not there is $S_3(3)$ as a submatrix and this is not possible since $M$ is admissible. Thus, in this case we find $M_4''$ as a submatrix. 

\vspace{2mm}
\subcase $M_{V}= \begin{pmatrix}
				11000\\
				00110\\
				11110\\
				10011\\
			\end{pmatrix}$
\vspace{2mm}

Once more and using the same argument, the only elegible columns for being tag columns are columns 2, 3 or 5. Moreover, if the second column is the sole tag column, then there is a reordering of the rows such that the matrix obtained is the same as the matrix when the third column is the tag column.
If column 5 is the only tag column, then we find $M_5'$ as in Case 1.\ 1.
If instead column 2 is the only tag column, then the first and second rows have the same color, for if not we find $S_2(3)$ as a submatrix of $M$, and thus we have $M= M_5''$.
Finally, if columns 2 and 5 are both tag columns, then the first and last row induce $D_2$ as a submatrix, disregarding the coloring of the rows and thus this case is also not possible. 

\end{mycases}

\vspace{2mm}
This finishes every possible case, and therefore we have reached a contradiction by assuming that $A_{\tagg}$ does not contain any of the listed submatrices and still the C$1$P does not hold for $A_{\tagg}$. 
\end{proof}


When giving the guidelines to draw a circle model for any split graph $G=(K,S)$, not only is important that the adjacency matrix of each partition of $K$ results admissible and LR-orderable. We also need to ensure that there is an LR-ordering that satisfies a certain property when considering how to split every LR-row into its L-block and its R-block. 
The following definition states necessary conditions for the LR-ordering that we need to consider to obtain a circle model. We will call this a \emph{suitable LR-ordering}. The lemma that follows ensures that, if a matrix $A$ is admissible and LR-orderable, then we can always find a suitable LR-ordering for the columns of $A$.

\begin{defn} \label{def:suitable_ordering}
	An LR-ordering $\Pi$ is \emph{suitable} if the L-blocks of those LR-rows with exactly two blocks are disjoint with every R-block, the R-blocks of those LR-rows with exactly two blocks are disjoint with the L-blocks and for each LR-row the intersection with any U-block is empty with either its L-block or its R-block.
\end{defn} 

\begin{teo} \label{teo:hay_suitable_ordering}
If $A$ is admissible, LR-orderable and contains no $M_0$, $M_{II}(4)$, $M_V$ or $S_0(k)$ for every even $k \geq 4$, then there is at least one suitable LR-ordering. 
\end{teo} 

\begin{proof}
	Let $A$ be an admissible LR-orderable matrix. Toward a contradiction, suppose that every LR-ordering is non-suitable.
	If $\Pi$ is an LR-ordering of $A$, since $\Pi$ is non-suitable, then either (1) there is a U-block $u$ such that $u$ is not disjoint with the L-block and the R-block of $f_1$, or (2) there is an LR-row $f_1$ such that its L-block is not disjoint with some R-block. In both cases, there is no possible reordering of the columns to obtain a suitable LR-ordering. 	
	
	\vspace{1mm}
	Notice that, since $A$ is admissible, the LR-rows can be split into a two set partition such that the LR-rows in each set are totally ordered. Moreover, any two LR-rows for which the L-block of one intersects the R-block of the other are in distinct sets of the partition and thus the columns may be reordered by moving the portion of the block that one of the rows has in common with the other all the way to the right (or left). Hence, if two such blocks intersect and there is no possible LR-reordering of the columns, then there is at least one non-LR row blocking the reordering.
	Throughout the proof and by simplicity, we will say that a row or block \emph{$a$ is chained to the left (resp.\ to the right) of another row or block $b$ }if $a$ and $b$ overlap and $a$ intersects $b$ in column $l(b)$ (resp.\ $r(b)$).
	
\begin{mycases}	
	\vspace{1mm}
	\case
	\textit{Let $a_1$ be the L-block of $f_1$ and $b_1$ be the R-block of $f_1$.
	Suppose first there is a U-block $u$ such that $u$ intersects both $a_1$ and $b_1$.}
	
	Let $j_1 = r(a_1)+1$, this is, the first column in which $f_1$ has a $0$, $j_2= r(a_1)$ and $j_3= l(b_1)$ in which both rows $f_1$ and $u$ have a $1$.
	Since it is not possible to rearrange the columns to obtain a suitable LR-ordering, in particular, there are two columns $j_4 < j_2$ and $j_5 > j_3$ in which $u$ has $0$, one before and one after the string of $1$'s of $u$. Moreover, there is at least one row $f_2$ distinct to $f_1$ and $u$ blocking the reordering of the columns $j_1$, $j_2$ and $j_3$. 
	\subcase Suppose $f_2$ is the only row blocking the reordering. Notice that $f_2$ is neither disjoint nor nested with $u$ and there is at least one column in which $f_1$ has a $0$ and $f_2$ has a $1$. We may assume without loss of generality that this is column $j_1$. Suppose $f_2$ is unlabeled.
	The only possibility is that $f_2$ overlaps with $u$, $a_1$ and $b_1$, for if not we can reorder the columns to obtain a suitable LR-ordering. In that case, we find $M_0$ in $A$.
	If instead $f_2$ is labeled with either L or R, then we find $S_6'(3)$ in $A$ considering columns $j_4$, $j_2$, $j_1$, $j_3$, $j_5$ and both tag columns. If $f_2$ is an LR-row and $f_2$ is the only row blocking the reordering, then either the L-block of $f_2$ is nested in the L-block of $f_1$ and the R-block of $f_2$ contains the R-block of $f_1$, or viceversa. However, in that case we can move the portion of the L-block of $f_1$ that intersects $u$ to the right and thus we find a suitable LR-ordering, therefore this case is not possible.
	
	\subcase Suppose now there is a sequence of rows $f_2, \ldots, f_k$ for some $k \geq 3$ blocking the reordering such that $f_i$ and $f_{i+1}$ overlap for each $i \in \{2, \ldots, k\}$. 
	Moreover, there is either --at least-- one row that overlaps $a_1$ or $b_1$. We may assume without loss of generality that $f_2$ is such a row and that $f_2$ and $b_1$ overlap.
	Suppose that $f_2$ and $f_3$ are unlabeled rows. Notice that, either all the rows are chained to the left of $f_2$ or to the right. Furthermore, since $A$ contains no $M_0$ and we assumed that $b_1$ and $f_2$ overlap, if $f_i$ is chained to the left of $f_2$, then $f_i$ is contained in $b_1$ for every $i \geq 3$, and if $f_i$ is chained to the right of $f_2$, then $f_i$ is contained in $u$ for every $3 \leq i <k$.
	In either case, we find $M_{II}(4)$ considering the columns $j_2$, $j_1$, $j_3$ and $j_5$. 
	Suppose that $f_2$ is the only labeled row in the sequence and that $f_2$ is labeled with R. If $u$ and $f_2$ overlap, then we find $S_6'(3)$ as in the previous paragraphs.
	Thus, we assume $u$ is nested in $f_2$. Since the sequence of rows is blocking the reordering, the rows $f_3, \ldots, f_k$ are chained one to one to the right and $f_k = u$, therefore we find $S_6(k)$ as a subconfiguration.
	The only remaining possibility is that there are two labeled rows in the sequence blocking the reordering. Since there are no $D_1$ or $S_3(3)$, then either these two rows are labeled with the same letter and nested, or they are labeled one with L and the other with R and are disjoint. We may assume without loss of generality that $f_2$ and $f_k$ are such labeled rows.
	
	If $f_2$ and $f_k$ are both labeled with L, then necessarily one is nested in the other, for $\Pi$ is an LR-ordering. In that case, one has a $0$ in column $j_1$ and the other has a $1$, for if not we can reorder the columns moving $j_1$ --and maybe some other columns in which $f_1$ has a $0$-- to the right. Hence, in this case we find $S_5(k)$ as a subconfiguration of the submatrix given by considering the rows $f_1, f_2, \ldots, f_k$. It is analogous if $f_2$ and $f_k$ are labeled with R.

	If instead $f_2$ and $f_k$ are labeled one with L and the other with R, then we have two possibilities. Either $f_2, \ldots, f_{k-1}$ are nested in $a_1$, or $f_2$ is chained to the right of $u$ and $f_3$ is chained to the left. In either case, if $f_2$ or $f_3$ have a $1$ in some column in which $f_1$ has a $0$ and $u$ has a $1$, then we find $S_6'(3)$. If instead $f_3$ is nested in $a_1$ and $f_2$ is nested in $b_1$, then we find $M_V$ as a subconfiguration considering the columns $j_4$, $j_2$, $j_1$, $j_3$ and $j_5$.

	\vspace{2mm}	
	\case
	\textit{Suppose now that there is a row $f_2$ such that the L-block $a_1$ of $f_1$ and the R-block $b_2$ of $f_2$ are not disjoint.} Notice that, by definition of R-block, $f_2$ is either labeled with R or LR.
	Once more, we consider $j_1=r(a_1) + 1$ the first column in which $f_1$ has a $0$.

	Since $a_1$ and $b_2$ intersect, there is a column $j_2 < j_1$ such that $a_1$ and $b_2$ both have a $1$ in column $j_2$.
	\subcase Suppose first that there is exactly one row $f_3$ blocking the possibility of reordering the columns to obtain a suitable LR-ordering.
	Notice that, for a row to block the reordering of the columns, such row must have a $1$ in $j_2$ and at least one column with a $0$. We have three possible cases:

 	\subsubcase Suppose first that $f_3$ is unlabeled. If $f_2$ is labeled with LR and $f_3$ does not intersect the L-block of $f_2$, then we can move to the R-block of $f_1$ those columns in which $f_3$ has $0$ and $a_1$ has $1$. If $f_3$ intersects the L-block of $f_2$, then this is precisely as in the previous case. Thus, we assume $f_2$ is labeled with R. 
	If $f_3$ is not nested in either $f_1$ nor $f_2$, then there is a column $j_3$ in which $f_3$ and $f_2$ have a $1$ and $f_1$ has a $0$, and a column $j_4$ in which $f_3$ and $f_1$ have a $1$ and $f_2$ has a $0$. 
	In that case, we find $S_6(3)$ considering the columns $j_1$, $j_2$, $j_3$, $j_4$ and both tag columns.
	 If $f_3$ is nested in $f_2$, then we can reorder by moving to the right all the columns in which $a_1$ and $f_2$ both have $1$ and mantaining those columns in which $f_3$ has a $1$ together.
	 If instead $f_3$ is nested in $f_1$, then we find $S_6'(3)$ as a subconfiguration. 
	
	\vspace{1mm}	
	\subsubcase Suppose now that $f_3$ is labeled with L. If $f_2$ is labeled with R, then $f_2$ and $f_3$ are colored with distinct colors, for if not we find $D_1$. Thus, we find $D_5$ as a subconfiguration in the submatrix given by $f_1$, $f_2$, $f_3$. Moreover, notice that, if $f_3$ is also labeled with R, then it is possible to move all those columns of $a_1$ that have a $1$ and intersect $f_2$ (and $f_3$) in order to obtain a suitable LR-ordering and thus $f_3$ did not block the reordering.
	If instead $f_2$ is an LR-row, then we find either $D_7$, $D_8$ or $D_9$, depending on where is the string of $0$'s in row $f_3$. Also notice that it is indistinct in this case if $f_3$ is labeled with R.
	\vspace{1mm} 
	 \subsubcase Suppose $f_3$ is labeled with LR. If $f_2$ is an LR-row, since $A$ is admissible, then either $f_3$ is nested in $f_1$ or $f_3$ is nested in $f_2$ (we may assume this since it is analogous if $f_3$ contains $f_1$ or $f_2$: we will see that $f_3$ is not blocking the reordering). If $f_3$ is nested in $f_2$, then we can move the part of the L-block $a_1$ that intersects $b_2$ all the way to the right and then we have a suitable reordering. It is analogous if $f_3$ is nested in $f_1$.
	 If $f_2$ is labeled with R, then we may assume that $f_2$ is not nested in $f_3$, for if not we have a similar situation as in the previous paragraphs. The same holds if $f_1$ and $f_3$ are nested LR-rows. 
	 We know that the L-block $a_3$ of $f_3$ intersects the R-block $b_2 =r_2$. Hence, in the column $j_3=r(a_3)+1$ the row $f_3$ has a $0$ and $f_2$ has a $1$, and in the column $j_4=l(b_2)-1$ the row $f_3$ has a $1$ and $f_2$ has a $0$. Moreover, since $f_1$ and $f_3$ are not nested, then there is a column greater than $j_2$ in which $f_1$ has a $0$ and $f_2$ and $f_3$ have a $1$. In this case, we find $D_8$ as a subconfiguration.

	\vspace{1mm}
	\subcase Suppose now that it is not possible to reorder the columns to obtain a suitable LR-ordering, since there is a sequence of rows $f_3, \ldots, f_k$, with $k > 3$, blocking --in particular-- the reordering of the columns $j_1=r(a_1)+1$ and $j_2=r(a_1)$.
	
	
	We may assume that the sequence of rows is either chained to the right --and thus $f_k$ is labeled with R-- or to the left --and thus $f_k$ is labeled with L, for if not we find $M_V$ as in the first case.
	Suppose that $f_2$ is labeled with R. If the sequence $f_3, \ldots, f_k$  is chained to the left, then we find $S_4(k)$ as a subconfiguration. If instead the sequence $f_3, \ldots, f_k$ is chained to the right, then we find $S_1(k)$.
	Suppose now that $f_2$ is an LR-row. 
	Since the L-block of $f_1$ and the R-block of $f_2$ intersect, then these rows are not nested. Whether the sequence is chained to the right or to the left, we may assume that $f_3$ is nested in $a_1$ and is disjoint with $a_2$. Let $k$ be the number of $0$'s between the L-block and the R-block of $f_2$. Depending on whether $k$ is odd or even, we find $S_0(k)$ or $S_8(k)$, respectively, as a subconfiguration of the submatrix given by considering the rows $f_1, f_2, \ldots, f_{k+3}$. 
\end{mycases}
	This finishes the proof.

\end{proof}

\begin{defn} \label{def:partially_2-nested}
	Let $A$ be an enriched matrix. We say $A$ is \emph{partially $2$-nested} if the following conditions hold:
	 \begin{itemize}
	 	\item $A$ is admissible, LR-orderable and contains no $M_0$, $M_{II}(4)$, $M_V$ or $S_0(k)$ for every even $k \geq 4$.
	 
	 	\item Each pair of non-LR-rows colored with the same color are either disjoint or nested in $A$.
	 	\item If an L-block (resp.\ R-block) of an LR-row is colored, then any non-LR row colored with the same color is either disjoint or contained in such L-block (resp.\ R-block).
	  	\item If an L-block (resp.\ R-block) of an LR-row $f_1$ is colored and there is a distinct LR-row $f_2$ for which its L-block (resp.\ R-block) is also colored with the same color, then $f_1$ and $f_2$ are nested in $A$.
	 \end{itemize}
\end{defn}

\begin{remark} \label{obs:partially2nested_hasnogems}
Notice that the second statement of the definition of partially $2$-nested implies that there are no monochromatic gems or monochromatic weak gems in $A$, since $A$ is admissible and thus any two labeled non-LR-rows do not contain $D_1$ as a subconfiguration. 
Moreover, the third statement implies that there are no monochromatic weak gems in $A$.
Furthermore, the last statement implies that there are no badly-colored doubly-weak gems in $A$.
\end{remark}

The following Corollary is a consequence of the previous remark and Theorem \ref{teo:LR-orderable_caract_bymatrices}.

\begin{cor} \label{cor:partially_2-nested_caract}
	An admissible matrix $A$ is partially $2$-nested if and only if $A$ contains no $M_0$, $M_{II}(4)$, $M_V$, monochromatic gems nor monochromatic weak gems nor badly-colored doubly-weak gems and the tagged matrix $A^*_{\tagg}$ does not contain any Tucker matrices, $M_2'(k)$, $M_2''(k)$, $M_3'(k)$, $M_3''(k)$ for $k\geq 3$, $M_4'$, $M_4''$, $M_5'$, $M_5''$.
\end{cor}

\section{A characterization of $2$-nested matrices} \label{section:2nestedmatrices}

We begin this section stating and proving a lemma that characterizes when a partial $2$-coloring can be extended to a total proper $2$-coloring, for every partially $2$-colored connected graph $G$.
Then, we give the definition and some properties of the auxiliary matrix $A+$, which will help us throughout the proof of Theorem \ref{teo:2-nested_caract_bymatrices} at the end of the section.

\begin{lema} \label{lema:2-color-extension}
	Let $G$ be a connected graph with a partial proper $2$-coloring of the vertices. Then, the partial $2$-coloring can be extended to a total proper $2$-coloring of the vertices of $G$ if and only if all of the following conditions hold:
	\begin{itemize}
		\item There are no even induced paths such that the only colored vertices of the path are its endpoints, and they are colored with the same color
		\item There are no odd induced paths such that the only colored vertices of the path are its endpoints, and they are colored with distinct colors
		\item There are no induced uncolored odd cycles
		\item There are no induced odd cycles with exactly one colored vertex
		\item There are no induced cycles of length 3 with exactly on uncolored vertex
	\end{itemize}	  
\end{lema}

\begin{proof}
The if case is trivial.

On the other hand, for the only if part, suppose all of the given statements hold. Notice that, since $G$ has a given proper partial $2$-coloring of the rows, then there are no adjacent vertices pre-colored with the same color. 

Let $H$ be the induced uncolored subgraph of $G$. We will prove this by induction on the number of vertices of $H$.

For the base case, this is to say when $|H| = 1$, let $v$ in $H$. If $v$ cannot be colored, since $v$ is the only uncolored vertex in $G$, then there are two vertices $x_1$ and $x_2$ such that $x_1$ and $x_2$ have distinct colors. Thus, the set $\{ x_1, v, x_2 \}$ either induces an odd path in $G$ of length 3 with the endpoints colored with distinct colors, or an induced $C_3$ with exactly one uncolored vertex, which results in a contradiction.

For the inductive step, suppose that we can extend the partial $2$-coloring of $G$ to a proper $2$-coloring if $|V(H)| \leq n$. 

Suppose that $|V(H)| = n+1$. If $H$ is not connected, then for any isolated vertex we have the same situation as in the base case. Hence, we assume $H$ is connected.
Let $v$ in $H$ such that $N(v) \cap V(G-H)$ is maximum. Every vertex $w$ in $N(v) \cap V(G-H)$ must be colored with the same color, for if not we find either a $C_3$ with exactly on uncolored vertex or an odd induced path with its endpoints colored with distinct colors. Suppose that such a color is red. Thus, we can color $v$ with blue.
We will see that the graph $G'$ defined as $G' = (G-H) \cup \{v \}$ fullfils every listed property.
It is straightforward that there are no uncolored odd cycles, for there were no odd uncolored cycles in $H$. Furthermore, using the same argument, we see that there are no induced odd cycles with exactly one colored vertex, for this would imply that there is an odd uncolored cycle $C$ in $H$ such that $v$ is a vertex of $C$.

Since every statement of the list holds for $G$ when $H$ is uncolored, if there was an even induced path $P$ such that the only colored vertices are its endpoints and they are colored with the same color, then the only possibility is that one of the endpoints of $P$ is $v$. 
Let $v_1$ be the uncolored vertex of $P$ such that $v_1$ is adjacent to $v$. Since $N(v) \cap V(G-H)$ is maximum, then there is a vertex $w$ in $N(v) \cap V(G-H)$ such that $w$ is nonadjacent to $v_1$. 
Hence, there is an odd induced path $P'$ in the pre-colored $G$ given by $< P, w >$ such that the only colored vertices of $P'$ are its endpoints and they are colored with the same color, which results in a contradiction.

The same argument holds if there is an odd induced path in $H - \{v\}$.

Finally, there are no $C_3$ with exactly one uncolored vertex, for in that case we would have an odd cycle in the pre-colored $G$ with exactly one colored vertex, and this results once more in a contradiction.
	
\end{proof}


Let $A$ be an enriched matrix, and let $A_{LR}$ be the enriched submatrix of $A$ given by considering every LR-row of $A$. We now give a useful property for this enriched submatrix when $A$ is admissible.

\begin{lema} \label{lema:A_LR_2nested}
If $A$ is admissible, then $A_{LR}$ contains no $F_1(k)$ or $F_2(k)$, for every odd $k \geq 5$.
\end{lema} 

\begin{proof}
Toward a contradiction, suppose that $A_{LR}$ contains either $F_1(k)$ or $F_2(k)$ in $A_{LR}$  as subconfiguration, for some odd $k \geq 5$. Moreover, since $k\geq5$, we find the following enriched submatrix in $A$ as a subconfiguration:
\begin{align*}
	\bordermatrix{ & \cr
		\textbf{LR} & 1 1 0 0 \cr
		\textbf{LR} & 0 1 1 0 \cr
		\textbf{LR} & 0 0 1 1 }\
\end{align*}
Since these three rows induce $D_{13}$, this is not possible. 
It follows from the same argument that there is no $F_2(k)$ in $A_{LR}$. 
Therefore, if $A_{LR}$ contains no $D_{13}$, then $A_{LR}$ contains no $F_1(k)$ or $F_2(k)$, for all odd $k \geq 5$. 

\end{proof}


\begin{remark}
	It follows from Lemma \ref{lema:A_LR_2nested} that, if $A$ is admissible, then there is a partition of the LR-rows of $A$ into two subsets $S_1$ and $S_2$ such that every pair of rows in each subset are either nested or disjoint. Moreover, since $A$ contains no $D_{11}$ as a subconfiguration, every pair of LR-rows that lie in the same subset $S_i$ are nested, for each $i=1,2$. Equivalently, the LR-rows in each subset $S_i$ are totally ordered by inclusion, for each $i=1,2$.
\end{remark} \label{rem:A_LR_2nested}
	 
	 Let $A$ be an admissible matrix, let $S_1$ and $S_2$ be a partition of the LR-rows of $A$ such that every pair of rows in $S_i$ is nested, for each $i=1,2$. 
	 Since there is no $D_0$, there is a row $m_L$ such that $m_L$ is labeled with L and contains every L-block of those rows in $A$ that are labeled with L. Analogously, we find a row $m_R$ such that every R-block of a row in $A$ labeled with R is contained in $m_R$. Moreover, there are two rows $m_1$ in $S_1$ and $m_2$ in $S_2$ such that every row in $S_i$ is contained in $m_i$, for each $i=1,2$.
	 This property allows us to well define the following auxiliary matrix, which will be helpful throughout the proof of Theorem \ref{teo:2-nested_caract_bymatrices}.

\begin{defn} \label{def:A+}
Let $A$ be an enriched matrix and let $\Pi$ be a suitable LR-ordering of $A$. The enriched matrix $A+$ is the result of applying the following rules to $A$:
\begin{itemize}
	\item Every empty row is deleted.

	\item Each LR-row $f$ with exactly one block is replaced by a row labeled with either L or R, depending on whether it has an L-block or an R-block.

	\item Each LR-row $f$ with exactly two blocks, is replaced by two uncolored rows, one having a $1$ in precisely the columns of its L-block and labeled with L, and another having a $1$ in precisely the columns of its R-block and labeled with R. We add a column $c_f$ with $1$ in precisely these two rows and $0$ otherwise. 

	\item If there is at least one row labeled with L or R in $A$, then each LR-row $f$ whose entries are all $1$'s is replaced by two uncolored rows, one having a $1$ in precisely the columns of the maximum L-block and labeled with L, and another having a $1$ in precisely the complement of the maximum L-block and labeled with R. We add a column $c_f$ with $1$ in precisely these two rows and $0$ otherwise.  

\end{itemize}
\end{defn} 

Notice that every non-LR-row remains the same.

\begin{figure}
\begin{align*}
	B = \bordermatrix{ ~ &  ~ \cr
	\textbf{LR} & 1 0 0 0 0  \cr
	\textbf{LR} & 1 1 0 0 1 \cr
 	 \textbf{LR} & 1 1 1 1 1  \cr
					 & 0 1  1  0  0 \cr
	\textbf L & 1 1 1 0 0 \cr
	\textbf{LR} & 0 0 0 0 0 \cr
	\textbf R & 0 0 0 1 1  }\,
	\begin{matrix}
	\\ \\ \\ \\ \\ \textcolor{red}{\bullet} \\ \textcolor{blue}{\bullet} \\ \textcolor{blue}{\bullet} 
	\end{matrix}
	&&
B+ = \bordermatrix{ & \cr
	\textbf{L} & 1 0 0 0 0 0 0 \cr
	\textbf{L} & 1 1 0 0 0 1 0  \cr
	\textbf{R} & 0 0 0 0 1 1 0 \cr
	\textbf{L} & 1 1 1 0 0 0 1 \cr
	\textbf{R} & 0 0 0 1 1  0  1\cr
					 & 0 1  1  0  0  0  0 \cr
	\textbf L & 1 1 1 0 0 0 0 \cr
	\textbf R & 0 0 0 1 1 0 0 }\,
	\begin{matrix}
	\\ \\ \\ \\ \\ \\ \\ \textcolor{red}{\bullet} \\ \textcolor{blue}{\bullet} \\ 
	\end{matrix}
\end{align*}
\caption{Example of an enriched admissible matrix $B$ and $B+$. The last two columns of $B+$ are $c_{r_2}$ and $c_{r_3}$.} \label{fig:example_B+}
\end{figure}


\begin{remark}\label{obs:props_de_A+_suitable}
Let $A$ be a partially $2$-nested matrix. Since $A$ is admissible, LR-orderable and contains no $M_0$, $M_{II}(4)$, $M_V$ or $S_0(k)$ for every even $k \geq 4$, then by Theorem \ref{teo:hay_suitable_ordering} we know that there exists a suitable LR-ordering $\Pi$. Hence, whenever we consider defining the matrix $A+$ for such a matrix $A$, we will always use a suitable LR-ordering $\Pi$ to do so.

Let us consider $A+$ as defined in \ref{def:A+} according to a suitable LR-ordering $\Pi$. Suppose there is at least one LR-row in $A$. Recall that, since $A$ is admissible, the LR-rows may be split into two disjoint subsets $S_1$ and $S_2$ such that the LR-rows in each subset are totally ordered by inclusion. This implies that there is an inclusion-wise maximal LR-row $m_i$ for each $S_i$, $i=1,2$. If we assume that $m_1$ and $m_2$ overlap, then either the L-block of $m_1$ is contained in the L-block of $m_2$ and the R-block of $m_1$ contains the R-block of $m_2$, or viceversa. Hence, if there is at least one LR-row in $A$, since $\Pi$ is suitable and $A$ contains no $D_1$, $D_4$ or uncolored rows labeled with either L or R, then the following holds:
	\begin{itemize}
		\item There is an inclusion-wise maximal L-block $b_L$ in $A+$ such that every R-block in $A+$ is disjoint with $b_L$. 
		\item There is an inclusion-wise maximal R-block $b_R$ in $A+$ such that every L-block in $A+$ is disjoint with $b_L$. 
	\end{itemize}
	
Therefore, when defining $A+$ we replace each LR-row having two strings of $1$'s by two distinct rows, one labeled with L and the other labeled with R, such that the new row labeled with L does not intersect with any row labeled with R and the new row labeled with R does not intersect with any row labeled with L. 
\end{remark}

We denote $A+ \setminus C_f$ to the submatrix induced by considering every non-$c_f$ column of $A+$. 
Notice that $A$ differs from $A+$ only in its LR-rows, which are either deleted or replaced in $A+$ by labeled uncolored rows. The following is a straightforward consequence of this.

\begin{lema} \label{lema:A+_tb_es_adm} 
	If $A$ is admissible and LR-orderable, then $A+ \setminus C_f$ is admissible and LR-orderable.
\end{lema}

	Let us consider an enriched $(0,1)$-matrix $A$. From now on, for each row $f$ in $A$ that is colored, we consider its blocks colored with the same color as $f$ in $A$.

\begin{defn} \label{def:proper2coloring}
	A $2$-color assignment for the blocks of an enriched matrix $A$ is a \emph{proper $2$-coloring} if $A$ is admissible, the L-block and R-block of each LR-row of $A$ are colored with distinct colors, and $A$ contains no monochromatic gems, weak monochromatic gems or badly-colored doubly-weak gems as subconfigurations.
	
	Given a $2$-color assignment for the blocks an enriched matrix $A$, we say it is a \emph{proper $2$-coloring of $A+$} if it is a proper $2$-coloring of $A$. 
	
\end{defn}  

\begin{remark} \label{obs:admisible_espartialproper2color}
Let $A$ be an enriched matrix. If $A$ is admissible, then the given pre-coloring of the blocks is a (partial) proper $2$-coloring. This follows from the fact that every pre-colored row is either labeled with L or R, of is an empty LR-row, thus there are no monochromatic gems, monochromatic weak gems or badly-colored weak gems in $A$ for they would induce $D_1$.
\end{remark}

In Figure \ref{fig:example_B+} we give an example of the matrix $B$ with a pre-coloring that is a proper $2$-coloring, since $B$ is admissible and contains no monochromatic gems, monochromatic weak gems or badly-colored doubly-weak gems (there is no pre-colored nonempty LR-row). 

In Figure \ref{fig:example_B_extensioncoloring}, we show two distinct coloring extensions for the pre-coloring of $B$, and how each of these colorings induce a coloring for $B+$. The first one --represented by $B^{(1)}$-- is a proper $2$-coloring of $A$, whereas the second one represented by $B^{(2)}$ is not. This follows from the fact that the first LR-row and the first L-row of $B^{(2)}$ induce a monochromatic weak gem.
 

\begin{figure}
\begin{align*}
	B = \bordermatrix{ ~ &  ~ \cr
	\textbf{LR} & 1 0 0 0 0  \cr
	\textbf{LR} & 1 1 0 0 1 \cr
 	 \textbf{LR} & 1 1 1 1 1  \cr
					 & 0 1  1  0  0 \cr
	\textbf L & 1 1 1 0 0 \cr
	\textbf{LR} & 0 0 0 0 0 \cr
	\textbf R & 0 0 0 1 1  }\,
	\begin{matrix}
	\\ \\ \\ \\ \\ \textcolor{red}{\bullet} \\ \textcolor{blue}{\bullet} \\ \textcolor{blue}{\bullet} 
	\end{matrix}
	&&
	B^{(1)} = \bordermatrix{ ~ & ~ \cr
	\textbf{LR} & \textcolor{blue}{1} 0 0 0 0  \cr
	\textbf{LR} & \textcolor{blue}{1} \textcolor{blue}{1} 0 0 \textcolor{red}{1} \cr
 	 \textbf{LR} & \textcolor{red}{1} \textcolor{red}{1} \textcolor{red}{1} \textcolor{blue}{1} \textcolor{blue}{1}  \cr
		& 0 \textcolor{red}{1}  \textcolor{red}{1} 0 0 \cr
	\textbf L & \textcolor{red}{1} \textcolor{red}{1} \textcolor{red}{1} 0 0 \cr
	\textbf{LR} & 0 0 0 0 0 \cr
	\textbf R & 0 0 0 \textcolor{blue}{1} \textcolor{blue}{1}  }
	\begin{matrix}
	\\ \\ \\ \\ \\ \textcolor{red}{\bullet} \\ \textcolor{blue}{\bullet} \\ \textcolor{blue}{\bullet}
	\end{matrix}
	&&
	B^{(2)} = \bordermatrix{ ~ &  ~ \cr
	\textbf{LR} & \textcolor{red}{1} 0 0 0 0 \cr
	\textbf{LR} & \textcolor{blue}{1} \textcolor{blue}{1} 0 0 \textcolor{red}{1} \cr
 	 \textbf{LR} & \textcolor{red}{1} \textcolor{red}{1} \textcolor{red}{1} \textcolor{blue}{1} \textcolor{blue}{1}  \cr
		& 0 \textcolor{red}{1} \textcolor{red}{1} 0 0  \cr
	\textbf L & \textcolor{red}{1} \textcolor{red}{1} \textcolor{red}{1} 0 0 \cr
	\textbf{LR} & 0 0 0 0 0 \cr
	\textbf R & 0 0 0 \textcolor{blue}{1} \textcolor{blue}{1}  }\
	\begin{matrix}
	\\ \\ \\ \\ \\ \textcolor{red}{\bullet} \\ \textcolor{blue}{\bullet} \\ \textcolor{blue}{\bullet}
	\end{matrix}
\end{align*}
\begin{align*}
	B+^{(1)} =  \bordermatrix{ & \cr
	\textbf{L} & \textcolor{blue}{1} 0 0 0 0 0 0 \cr
	\textbf{L} & \textcolor{blue}{1} \textcolor{blue}{1} 0 0 0 1 0 \cr
	\textbf{R} & 0 0 0 0 \textcolor{red}{1} 1 0 \cr
	\textbf{L} & \textcolor{red}{1} \textcolor{red}{1} \textcolor{red}{1} 0 0 0 1 \cr
	\textbf{R} & 0 0 0\textcolor{blue}{1}  \textcolor{blue}{1} 0 1 \cr
					 & 0 \textcolor{red}{1} \textcolor{red}{1} 0 0 0 0 \cr
	\textbf L & \textcolor{red}{1} \textcolor{red}{1} \textcolor{red}{1}0 0 0 0  \cr
	\textbf R & 0 0 0 \textcolor{blue}{1} \textcolor{blue}{1} 0 0 }\
	\begin{matrix}
	\\ \\ \\ \\ \\ \\ \\ \textcolor{red}{\bullet} \\ \textcolor{blue}{\bullet} \\ 
	\end{matrix}
	&&
	B+^{(2)} = \bordermatrix{ & \cr
	\textbf{L} & \textcolor{red}{1} 0 0 0 0 0 0 \cr
	\textbf{L} & \textcolor{blue}{1} \textcolor{blue}{1} 0 0 0 1 0 \cr
	\textbf{R} & 0 0 0 0 \textcolor{red}{1} 1 0 \cr
	\textbf{L} & \textcolor{red}{1} \textcolor{red}{1} \textcolor{red}{1} 0 0 0 1 \cr
	\textbf{R} & 0 0 0\textcolor{blue}{1}  \textcolor{blue}{1} 0 1 \cr
					 & 0 \textcolor{red}{1} \textcolor{red}{1} 0 0 0 0 \cr
	\textbf L & \textcolor{red}{1} \textcolor{red}{1} \textcolor{red}{1}0 0 0 0  \cr
	\textbf R & 0 0 0 \textcolor{blue}{1} \textcolor{blue}{1} 0 0 }\
	\begin{matrix}
	\\ \\ \\ \\ \\ \\ \\ \textcolor{red}{\bullet} \\ \textcolor{blue}{\bullet} \\ 
	\end{matrix}
\end{align*}
\caption{Example of a proper and a non-proper $2$-coloring extension for the admissible matrix $B$ and the respective induced colorings for $B+$. The last two colums of $B+$ are $c_{r_2}$ and $c_{r_3}$.} \label{fig:example_B_extensioncoloring}
\end{figure}

The following is a straightforward consequence of Remark \ref{obs:partially2nested_hasnogems}.

\begin{lema}\label{lema:part2nested_is_2_colored}
	Let $A$ be an enriched matrix. If $A$ is partially $2$-nested, then the given pre-coloring of $A$ is a proper partial $2$-coloring. 
	Moreover, if $A$ is partially $2$-nested and admits a total $2$-coloring, then $A$ with such $2$-coloring is partially $2$-nested.
\end{lema}

\begin{lema}  \label{lema:2-nested_if}
	Let $A$ be an enriched matrix. Then, $A$ is \emph{$2$-nested} if $A$ is partially $2$-nested and the given partial block bi-coloring of $A$ can be extended to a total proper $2$-coloring of $A$. 
\end{lema}

\begin{proof}
Let $A$ be an enriched matrix that is partially $2$-nested and for which the given pre-coloring of the blocks can be extended to a total proper $2$-coloring of $A$. In particular, this induces a total block bi-coloring for $A$. Indeed, we want to see that a proper $2$-coloring induces a total block bi-coloring for $A$. 
Notice that the only pre-colored rows may be those labeled with L or R and those empty LR-rows.

Let us see that each of the properties that define $2$-nested hold.

\begin{enumerate}
	\item Since $A$ is an enriched matrix and the only rows that are not pre-colored are the nonempty LR-rows and those that correspond to U-blocks, then there is no ambiguity when considering the coloring of the blocks of a pre-colored row (Prop.\ 2 of $2$-nested).
	
	\item If $A$ is partially $2$-nested, then in particular is admissible, LR-orderable and contains no $M_0$, $M_{II}(4)$ or $M_V$. Thus, by Theorem \ref{teo:hay_suitable_ordering}, there is a suitable LR-ordering $\Pi$ for the columns of $A$. We consider $A$ ordered according to $\Pi$ from now on. Since $\Pi$ is suitable, then every L-block of an LR-row and an R-block of a non-LR-row are disjoint, and the same holds for every R-block of an LR-row and an L-block of a non-LR-row (Prop.\ 4 of $2$-nested). 

	\item Since $A$ is admissible, thus there are no subconfigurations as in $\mathcal{D}$. Moreover, since $A$ is partially $2$-nested, by Corollary \ref{cor:partially_2-nested_caract} there are no monochromatic gems or weak gems and no badly-colored doubly-weak gems induced by pre-colored rows. It follows from this and the fact that the LR-ordering is suitable, that Prop.\ 8 of $2$-nested holds.

	\item The pre-coloring of the blocks of $A$ can be extended to a total proper $2$-coloring of $A$. This induces a total block bi-coloring for $A$, for which we can deduce the following assertions:
	\begin{itemize}
	\item Since there is a total proper $2$-coloring of $A$, in particular the L-block and R-block of each LR-row are colored with distinct colors. (Prop.\ 1 of $2$-nested).
	
	\item Each L-block and R-block corresponding to distinct LR-rows with nonempty intersection are also colored with distinct colors since there are no badly-colored doubly-weak gems in $A$ (Prop.\ 9 of $2$-nested).
	
	\item Since $A$ is admissible, every L-block and R-block corresponding to distinct non-LR-rows are colored with different colors since there is no $D_1$ in $A$ (Prop.\ 5 of $2$-nested)
	
	\item Since there are no monochromatic weak gems in $A$, an L-block of an LR-row and an L-block of a non-LR row that contains the L-block must be colored with distinct colors. Furthermore, if any L-block and a U-block are not disjoint and are colored with the same color, then the U-block is contained in the L-block. (Prop.\ 3 and 7 of $2$-nested)
	
	\item There is no monochromatic gem in $A$, then each two U-blocks colored with the same color are either disjoint or nested. (Prop.\ 6 of $2$-nested)
	\end{itemize}
\end{enumerate}

\end{proof}

\begin{lema} \label{lema:if_suitable_noM}
Let $A$ be an enriched matrix. If $A$ admits a suitable LR-ordering, then $A$ contains no $M_0$, $M_{II}(4)$, $M_V$ or $S_0(k)$ for every even $k \geq 4$. 
\end{lema}

\begin{proof}
The result follows trivially if $A$ contains no LR-rows, since $A$ admits an LR-ordering, thus if we consider $A$ without its LR-rows, that submatrix has the C$1$P and hence it contains no Tucker matrices. Toward a contradiction, suppose that $A$ contains either $M_0$, $M_{II}(4)$, $M_V$ or $S_0(k)$ for some even $k \geq 4$.
Since there is no $M_I(k)$ for every $k \geq 3$, then in particular there is no $M_0$ or $S_0(k)$ where at most one of the rows is an LR-row. Moreover, it is easy to see that, if we reorder the columns of $M_0$, then there is no possible LR-ordering in which every L-block and every R-block are disjoint.
Similarly, consider $S_0(4)$, whose first row has a $1$ in every column. We may assume that the last row is an LR-row for any other reordering of the columns yields an analogous situation with one of the rows. However, whether the first row is unlabeled or not, the first and the last row prevent a suitable LR-ordering. The reasoning is analogous for any even $k > 4$.

Suppose that $A$ contains $M_V$, and let $f_1$,$ f_2$, $f_3$ and $f_4$ be the rows of $M_V$ depicted as follows:
\[ M_V= \bordermatrix{ &  \cr
f_1 & 1 1 0 0 0 \cr
f_2 & 0 0 1 1 0 \cr
f_3 & 1 1 1 1 0 \cr
f_4 & 1 0 0 1 1  }
\]
If the first row is an LR-row, then either $f_3$ or $f_4$ is an LR-row, for if not we find $M_I(3)$ in $A^*$, which is not possible since there is an LR-ordering in $A$. The same holds if the second row is an LR-row. 
If $f_3$ is an LR-row, then $f_4$ is an LR-row, for if not $f_4$ must have a consecutive string of $1$'s, thus, if $f_4$ is an unlabeled row, then it intersect both blocks of $f_3$, and if $f_4$ is an R-row, then its R-block intersects the L-block of $f_3$. 
However, if we move the columns so that the L-block of $f_3$ does not intersect the R-block of $f_4$, then we either cannot split $f_1$ into two blocks such that one starts one the left and the other ends on the right, of we cannot maintain a consecutive string of $1$'s in $f_2$. It follows analogously if we assume that $f_4$ is an LR-row, thus $f_1$ is not an LR-row. By symmetry, we assume that $f_2$ is also a non-LR-row, and thus the proof is analogous if only $f_3$ and $f_4$ may be LR-rows. 

Suppose $A$ contains $M_{II}(4)$. Let us denote $f_1$, $f_2$, $f_3$ and $f_4$ to the rows of $M_{II}(4)$ depicted as follows:
\[ M_{II}(4) = \bordermatrix{ &  \cr
f_1 & 0 1 1 1 \cr
f_2 & 1 1 0 0 \cr
f_3 & 0 1 1 0 \cr
f_4 & 1 1 0 1 }
\]

If $f_2$ is an LR-row, then necessarily $f_3$ or $f_4$ are LR-rows, for if not we find $M_I(3)$ in $A^*$.
If only $f_2$ and $f_3$ are LR-rows, then we find $M_{II}(4)$ in $A^*$.
If instead only $f_2$ and $f_4$ are LR-rows, then --as it is-- whether $f_1$ is an R-row or an unlabeled row, the block of $f_1$ intersects the L-block and the R-block of $f_4$ (and also the L-block of $f_2$). The only possibility is to move the second column all the way to the right and split $f_2$ into two blocks and give the R-block of $f_4$ length $2$. However in this case, it is not possible to move another column and obtain an ordering that keeps all the $1$'s consecutive for $f_3$ and $f_1$ not intersecting both blocks of $f_4$ simultaneously. Thus, $f_1$ is also an LR-row. However, for any ordering of the columns, either it is not possible to simultaneously split the string of $1$'s in $f_1$ and keep the L-block of $f_2$ starting on the left, or it is not possible to simultaneously maintain the string of $1$'s in $f_3$ consecutive and the L-block of $f_1$ disjoint with the R-block of $f_4$. 
It follows analogously if both $f_3$ and $f_4$ are LR-rows. 
Hence, $f_2$ is a non-LR-row, and by symmetry, we may assume that $f_3$ is also a non-LR-row. Suppose now that $f_1$ is an LR-row. If $f_4$ is not an LR-row, then there is no possible way to reorder the columns and having a consecutive string of $1$'s for the rows $f_2$, $f_3$ and $f_4$ simultaneously, unless we move the fourth column all the way to the left. However in that case, either $f_4$ is an L-row and its L-block intersects the R-block of $f_1$ of it is an unlabeled row that intersects both blocks of $f_1$. Moreover, the same holds if $f_4$ is an LR-row, with the difference that in this case the R-block of $f_4$ intersects the L-block of $f_1$ or the string of $1$'s in $f_2$ and $f_3$ is not consecutive. 

\end{proof}

\begin{lema} \label{lema:2-nested_onlyif}
	Let $A$ be an enriched matrix. If $A$ is $2$-nested, then $A$ is partially $2$-nested and the total block bi-coloring induces a proper total $2$-coloring of $A$.
\end{lema}

\begin{proof}

If $A$ is $2$-nested, then in particular there is an LR-ordering $\Pi$ for the columns. Moreover, by properties 4 and 7, such an ordering is suitable.

Suppose first there is a monochromatic gem in $A$. Such a gem is not induced by two unlabeled rows since in that case property 6 of the definition of $2$-nested matrix would not hold. Hence, such a gem is induced by at least one labeled row. Moreover, if one is a labeled row and the other is an unlabeled row, then property 7 would not hold. Thus, both rows are labeled. By property 5, if the gem is induced by two non-disjoint L-block and R-block, then it is not monochromatic, disregarding on whether they correspond to LR-rows or non-LR-rows. Hence, exactly one of the rows is an LR-row. However, by property 4, an L-block of an LR-row and an R-block of a non-LR-row are disjoint, thus they cannot induce a gem. 

Suppose there is a monochromatic weak gem in $A$, thus at least one of its rows is a labeled row. It is not possible that exactly one of its rows is a labeled row and the other is an unlabeled row, since property 7 holds. Moreover, these rows do not correspond to rows labeled with L and R, respectively, for properties 4 and 5 hold. Furthermore, both rows of the weak gem are LR-rows, since if exactly one is an LR-row, then properties 3, 4 and 7 hold and thus it is not possible to have a weak gem. However, in that case, property 5 guarantees that this is also not possible.

Finally, there is no badly-colored doubly-weak gem since properties 4, 5 and 9 hold. 

\vspace{1mm}
Now, let us see that $A$ is admissible. Since there is an LR-ordering of the columns, there are no $D_0$, $D_2$, $D_3$, $D_6$, $D_7$, $D_8$ or $D_{11}$ in $A$. Moreover, by property 5, there is no $D_1$.  	
As we have previously seen, there are no monochromatic gems or monochromatic weak gems. Hence, it is easy to see that if there is a total block bi-coloring, then $A$ contains none of the matrices in $\mathcal{S}$ or $\mathcal{P}$ as a subconfiguration.
Suppose there is $D_4$. By property 8, if there are two L-blocks of non-LR-rows colored with distinct colors, then every LR-row has a nonempty L-block, and in this case such an L-block is contained in both rows labeled with L. However, by property 3, the L-block of the LR-row is properly contained in the L-blocks of the non-LR-rows, thus it must be colored with a distinct color than the color assigned to each L-block of a non-LR-row, and this leads to a contradiction.
By property 4, there is no $D_5$. 
Let us suppose there is $D_9$ given by the rows $f_1$, $f_2$ and $f_3$, were $f_1$ is labeled with L and $f_2$ and $f_3$ are LR-rows. Suppose that $f_1$ is colored with red. Since the L-block of $f_2$ is contained in $f_1$, by property 3, then the L-block of $f_2$ is colored with blue. The same holds for the L-block of $f_3$. However, $f_2$ and $f_3$ are not nested, thus by property 9 the L-blocks of $f_2$ and $f_3$ are colored with distinct colors, which results in a contradiction.

Let us suppose there is $D_{10}$ given by the rows $f_1$, $f_2$, $f_3$ and $f_4$, were $f_1$ is labeled with L and colored with red, $f_2$ is labeled with R and colored with blue, and $f_3$ and $f_4$ are LR-rows. 
Since the L-block of $f_3$ is properly contained in $f_1$, then by property 3, it is colored with blue. By property 1, the R-block of $f_3$ is colored with red. Using a similar argument, we assert that the R-block of $f_4$ is colored with red and the L-block of $f_4$ is colored with blue. However, $f_3$ and $f_4$ are non-disjoint and non-nested, thus the L-block of $f_3$ and the R-block of $f_4$ are colored with distinct colors, which results in a contradiction.

By Lemma \ref{lema:if_suitable_noM}, since there is a suitable LR-ordering, then $A$ contains  no $M_0$, $M_{II}(4)$, $M_V$ or $S_0(k)$ for every even $k \geq 4$. 

Finally, by property 9 and the fact that there is an LR-ordering, there are no $D_{12}$ nor $D_{13}$.

Therefore $A$ is partially $2$-nested.

\vspace{1mm}
Finally, we will see that the total block bi-coloring for $A$ induces a proper total $2$-coloring of $A$. 
Since every property of $2$-nested holds, then it is straightforward that there are no monochromatic gems or monochromatic weak gems or badly-colored weak gems in $A$. For more details on this, see Remark \ref{obs:partially2nested_hasnogems} and Lemma \ref{lema:2-nested_if} since the same arguments are detailed there. 
Moreover, since property 1 of $2$-nested holds, the L-block and R-block of the same LR-row are colored with distinct colors. Therefore, it follows that a total block bi-coloring of $A$ induces a proper total $2$-coloring of $A$. 

\end{proof}

\vspace{2mm}
The following corollary is a straightforward consequence of the previous.

\begin{cor} \label{lema:B_ext_2-nested} 
	Let $A$ be an enriched matrix. If $A$ is partially $2$-nested and $B$ is obtained from $A$ by extending its partial coloring to a total proper $2$-coloring, then $B$ is $2$-nested if and only if for each LR-row its L-block and R-block are colored with distinct colors and $B$ contains no monochromatic gems, monochromatic weak gems or badly-colored doubly-weak gems as subconfigurations.
	
	
\end{cor}


We are now ready to give the proof for the main result of this chapter.

\begin{teo}
	Let $A$ be an enriched matrix. Then, $A$ is $2$-nested if and only if $A$ contains none of the following listed matrices or their dual matrices as subconfigurations: 
	\begin{itemize}		
	\item $M_0$, $M_{II}(4)$, $M_V$ or $S_0(k)$ for every even $k$ (See Figure \ref{fig:forb_M_chiquitas})
	\item Every enriched matrix in the family $\mathcal{D}$ (See Figure \ref{fig:forb_D}) 
	\item Every enriched matrix in the family $\mathcal{F}$ (See Figure \ref{fig:forb_F}) 
	\item Every enriched matrix in the family $\mathcal{S}$ (See Figure \ref{fig:forb_Si})
	\item Every enriched matrix in the family $\mathcal{P}$ (See Figure \ref{fig:forb_P})
	\item Monochromatic gems, monochromatic weak gems, badly-colored doubly-weak gems 
	\end{itemize}
and $A^*$ contains no Tucker matrices and none of the enriched matrices in $\mathcal{M}$ or their dual matrices as subconfigurations (See Figure \ref{fig:forb_LR-orderable}).

\end{teo} 

The proof is organized as follows. The if case follows immediately using Lemma \ref{lema:2-nested_onlyif} and the characterizations of admissibility, LR-orderable and partially $2$-nested given in the previous sections. 
For the only if case, we have two possible cases: (1) either there are no labeled rows in $A$, or (2) there is at least one labeled row in $A$ (either L, R or LR).
	In each case, we define an auxiliary graph $H(A)$ that is partially $2$-colored according to the pre-coloring of the blocks of $A$. Toward a contradiction, we suppose that $H(A)$ is not bipartite. Using the characterization given in Lemma \ref{lema:2-color-extension}, we know there is one of the $5$ possible kinds of paths or cycles, we analyse each case and reach a contradiction.
	A complete proof of case (1) has been published in \cite{PDGS19}.

\begin{proof}
	Suppose $A$ is $2$-nested. In particular, $A$ is partially $2$-nested with the given pre-coloring and the block bi-coloring induces a total proper $2$-coloring of $A$. 
	Thus, by Corollary \ref{cor:partially_2-nested_caract}, $A$ is admissible and contains no $M_0$, $M_{II}(4)$, $M_V$, $S_0(k)$ for every even $k \geq 4$, monochromatic gems, monochromatic weak gems or badly-colored doubly-weak gems as subconfigurations, and $A^*_\tagg$ contains no Tucker matrices, $M_4'$, $M_4''$, $M_5'$, $M_5''$, $M'_2(k)$, $M''_2(k)$, $M_3'(k)$, $M_3''(k)$, $M_3'''(k)$, for any $k \geq 4$ as subconfigurations. 
	In particular, since $A$ is admissible, there is no $D_{13}$ induced by any three LR-rows.
	
	Moreover, notice that every pair of consecutive rows of any of the matrices $F_0$, $F_1(k)$, and $F_2(k)$ for all odd $k \geq 5$ induces a gem, and there is an odd number of rows in each matrix. Thus, if one of these matrices is a submatrix of $A_\tagg$, then there is no proper $2$-coloring of the blocks. Therefore, $A$ contains no $F_0$, $F_1(k)$, and $F_2(k)$ for any odd $k \geq 5$ as submatrices. A similar argument holds for $F'_0$, $F_1'(k)$, $F_2'(k)$, changing 'gem' for 'weak gem' whenever one of the two rows considered is a labeled row. 
	
	\vspace{1mm}	
	Conversely, suppose $A$ is not $2$-nested. 
	Henceforth, we assume that $A$ is admissible. 
	
	If $A$ is not partially $2$-nested, then either $A$ contains $M_0$, $M_{II}(4)$, $M_V$, $S_0(k)$ for some even $k \geq 4$, or there is a submatrix $M$ in $A^*_\tagg$ such that $M$ represents the same configuration as one of the forbidden submatrices for partially $2$-nested stated above, and thus $M$ is a subconfiguration of $A^*_\tagg$.
	
	Henceforth, we assume that $A$ is partially $2$-nested. 
.
	If $A$ is partially $2$-nested but is not $2$-nested, then the pre-coloring of the rows of $A$ (which is a proper partial $2$-coloring of $A$ since $A$ is admissible) cannot be extended to a total proper $2$-coloring of $A$. 
	
\begin{mycases}
\case \label{case:teo2nested_case1} \textit{There are no labeled rows in $A$.}

We define the auxiliary graph $H(A)=(V,E)$ where the vertex set $V= \{ w_1, \ldots, w_n \}$ has one vertex for each row in $A$, and two vertices $w_i$ and $w_k$ in $V$ are adjacent if and only if the rows $a_{i.}$ and $a_{k.}$ are neither disjoint nor nested. By abuse of language, $w_i$ will refer to both the vertex $w_i$ in $H(A)$ and the row $a_{i.}$ of $A$. In particular, the definitions given in the introduction apply to the vertices in $H(A)$; i.e., we say two vertices $w_i$ and $w_k$ in $H(A)$ are \emph{nested} (resp.\ \emph{disjoint}) if the corresponding rows $a_{i.}$ and $a_{k.}$ are nested (resp.\ disjoint). And two vertices $w_i$ and $w_k$ in $H(A)$ \emph{start} (resp.\ \emph{end}) \emph{in the same column} if the corresponding rows $a_{i.}$ and $a_{k.}$ start (resp.\ end) in the same column. 
It follows from the definition of $2$-nested matrices that $A$ is a $2$-nested matrix if and only if there is a bicoloring of the auxiliary graph $H(A)$ or, equivalently, if $H(A)$ is bipartite (i.e., $H(A)$ does not have contain cycles of odd length), since there are no labeled rows in $A$,and thus there are no pre-colored vertices in $H$.

	Let $\Pi$ be a linear ordering of the columns such that the matrix $A$ does not contain any $F_0$, $F_1(k)$ and $F_2(k)$ for every odd $k\geq 5$ or Tucker matrices as subconfigurations.
	Due to Tucker's Theorem, since there are no Tucker submatrices in $A$, the matrix $A$ has the C$1$P. 
	
	Toward a contradiction, suppose that the auxiliary graph $H(A)$ is not bipartite. Hence there is an induced odd cycle $C$ in $H(A)$.

	Suppose first that $H(A)$ has an induced odd cycle $C = w_1, w_2, w_3, w_1 $ of length 3, and suppose without loss of generality that the first rows of $A$ are those corresponding to the cycle $C$. 
	Since $w_1$ and $w_2$ are adjacent, both begin and end in different columns. The same holds for $w_2$ and $w_3$, and $w_1$ and $w_3$. We assume without loss of generality that the vertices start in the order of the cycle, in other words, that $l_1 < l_2 < l_3$. 
	
	Since $w_1$ starts first, it is clear that $a_{2 l_1} = a_{3 l_1} = 0$, thus the column $a_{. l_1}$ of $A$ is the same as the first column of the matrix $F_0$.

	Since $A$ has the C$1$P and $w_1$ and $w_2$ are adjacent, then $a_{1 l_2} = 1$. 
	As stated before, $w_2$ starts before $w_3$ and thus $a_{3 l_2} = 0$. Hence, column $a_{. l_2}$ is equal to the second column of $F_0$.
	
	The third column of $F_0$ is $a_{. l_3}$, for $w_3$ is adjacent to $w_1$ and $w_2$, hence it is straight\-forward that $a_{1 l_3} = a_{2 l_3} = a_{3 l_3} = 1$.
	
	To find the next column of $F_0$, let us look at column $a_{. (r_1+1)}$. Notice that $r_1 + 1 > l_3$. Since $w_1$ is adjacent to $w_2$ and $w_3$, and $w_2$ and $w_3$ both start after $w_1$, then necessarily $a_{2 (r_1+1)} = a_{3 (r_1+1)} = 1$, and thus $a_{. (r_1+1)}$ is equal to the fourth column of $F_0$.
	
	Finally, we look at the column $a_{. (r_2 +1)}$. Notice that $r_2 +1 > r_1 + 1$. 
	Since $A$ has the C$1$P, $a_{1 (r_2 +1)} = 0$ and $r_2 +1 > r_1+1$, then $a_{1 (r_2+1)} =0$ and $a_{3 (r_2+1)} = 1$, which is equal to last column of $F_0$. Therefore we reached a contradiction that came from assuming that there is a cycle of length 3 in $H(A)$.
	
	Suppose now that $H(A)$ has an induced odd cycle $C = w_1, \ldots, w_k, w_1$ of length $k \geq 5$. We assume without loss of generality that the first $k$ rows of $A$ are those in $C$ and that $A$ is ordered according to the C$1$P.
	
	\begin{remark} \label{rem:2N_1}
		Let $w_i, w_j$ be vertices in $H(A)$. If $w_i$ and $w_j$ are adjacent and $w_i$ starts before $w_j$, then $a_{i r_i} = a_{j r_i} = 1$ and $a_{i (r_i+1)} = 0$, $a_{j (r_i +1)} = 1$.
	\end{remark}

\begin{remark} \label{rem:2N_2}
	If $l_{i-1} > l_i$ and $l_{i+1} > l_i$ for some $i = 3, \ldots, k-1$, then for all $j\geq i+1$, $w_j$ is nested in $w_{i-1}$. The same holds if $l_{i-1} < l_i$ and $l_{i+1} < l_i$.
	Since $l_{i-1} > l_i$ and $l_{i+1} > l_i$, then $w_{i-1}$ and $w_{i+1}$ are not disjoint, thus necessarily $w_{i+1}$ is nested in $w_{i-1}$. It follows from this argument that this holds for $j \geq i+1$.
\end{remark}

Notice that $w_2$ and $w_k$ are nonadjacent, hence they are either disjoint or nested. Using this fact and Remark \ref{rem:2N_1}, we split the proof into two cases.

\subcase \textit{$w_2$ and $w_k$ are nested } 
\vspace{-1mm}
We may assume without loss of generality that $w_k$ is nested in $w_2$, for if not, we can rearrange the cycle backwards as $w_1$, $w_k$, $w_{k-1}, \ldots, w_2, w_1$. Moreover, we will assume without loss of generality that both $w_2$ and $w_k$ start before $w_1$. First, we need the following Claim.

\begin{claim} \label{claim:2N_1-1}
	If $w_2$ and $w_k$ are nested, then $w_i$ is nested in $w_2$, for $i = 4, \ldots, k-1$.
\end{claim}

Suppose first that $w_1$ and $w_3$ are disjoint, and toward a contradiction suppose that $w_2$ and $w_4$ are disjoint. In this case, $l_4 < l_3 < r_4 < l_2 < r_3 < r_2$. 
The contradiction is clear if $k=5$.
If instead $k>5$ and $w_5$ starts before $w_4$, then $r_i < l_3$ for all $i > 5$, which contradicts the assumption that $w_k$ is nested in $w_2$. Hence, necessarily $w_5$ is nested in $w_3$ and $w_5$ and $w_2$ are disjoint. This implies that $l_3 < l_5 < r_4 < r_5 < l_2$ and once more, $r_i < l_2$ for all $i >5$, which contradicts the fact that $w_k$ is nested in $w_2$.

Suppose now that $w_3$ is nested in $w_1$. Toward a contradiction, suppose that $w_4$ is not nested in $w_2$. Thus, $w_2$ and $w_4$ are disjoint since they are nonadjacent vertices in $H(A)$.
Notice that, if $w_3$ is nested in $w_1$, then $l_2 < l_3$ and $r_2 < r_3$.
Furthermore, since $w_4$ is adjacent to $w_3$ and nonadjacent to $w_2$, then $l_3 < r_2 < l_4 < r_3 < r_4$. This holds for every odd $k\geq 5$.

If $k=5$, since $w_5$ is nested in $w_2$, then $r_5 < r_2 < l_4$, which results in a contradiction for $w_4$ and $w_5$ are adjacent.

Suppose that $k>5$. If $w_2$ and $w_i$ are disjoint for all $i= 5, \ldots, k-1$, then $w_{k-1}$ and $w_k$ are nonadjacent for $w_k$ is nested in $w_2$, which results in a contradiction. Conversely, if $w_i$ and $w_2$ are not disjoint for some $i > 3$, then they are adjacent, which also results in a contradiction that came from assuming that $w_2$ and $w_4$ are disjoint. Therefore, since $w_4$ is nested in $w_2$, $w_2$ and $w_i$ are nonadjacent and $w_i$ is adjacent to $w_{i+1}$ for all $i >4$, then necessarily $w_i$ is nested in $w_2$, which finishes the proof of the Claim. 

\begin{claim} \label{claim:2N_1-2}
	Suppose that $w_2$ and $w_k$ are nested. Then, if $w_3$ is nested in $w_1$, then $l_i > l_{i+1}$ for all $i=3, \ldots, k-1$. If instead $w_1$ and $w_3$ are disjoint, then $l_i < l_{i+1}$ for all $i=3, \ldots, k-1$.
\end{claim}

\vspace{-0.5mm}
Recall that, by the previous Claim, since $w_i$ is nested in $w_2$ for all $i = 4, \ldots, k $, in particular $w_4$ is nested in $w_2$. Moreover, since $w_3$ and $w_4$ are adjacent, notice that, if $w_3$ is nested in $w_1$, then $l_3 > l_4$, and if $w_1$ and $w_3$ are disjoint, then $l_3 < l_4$.

It follows from Remark \ref{rem:2N_2} that, if $l_5 > l_4$, then $w_i$ is nested in $w_3$ for all $i = 5, \ldots, k $, which contradicts the fact that $w_1$ and $w_{k-1}$ are adjacent. The proof of the first statement follows from applying this argument successively.

The second statement is proven analogously by applying Remark \ref{rem:2N_2} if $l_5 < l_4$, and afterwards successively for all $i > 4$. 

If $w_1$ and $w_3$ are disjoint, then we obtain $F_2(k)$ first, by putting the first row as the last row, and considering the submatrix given by columns $j_1 = l_1 -1$, $j_2 = l_3$, $\ldots$, $j_i = l_{i+1}$, $\ldots$, $j_k = r_1+1$ (using the new ordering of the rows).
If instead $w_3$ is nested in $w_1$, then we obtain $F_1(k)$ by taking the submatrix given by the columns $j_1 = l_1 -1$, $j_2 = r_k$, $\ldots$, $j_i = l_{k -i +2}$, $\ldots$, $j_{k-1} = r_3$.

\vspace{1mm}
\subcase \textit{$w_2$ and $w_k$ are disjoint }

We assume without loss of generality that $l_2 < l_1$ and $l_k > l_1$.

\begin{claim} \label{claim:2N_2-1}
	If $w_2$ and $w_k$ are disjoint, then $l_i < l_{i+1}$ for all $i=2, \ldots, k-1$.
\end{claim}

Notice first that, in this case, $w_i$ is nested in $w_1$, for all $i= 3, \ldots, k-1$. If not, then using Remark \ref{rem:2N_2}, we notice that it is not possible for the vertices $w_1, \ldots, w_k$ to induce a cycle. This implies, in particular, that $w_3$ is nested in $w_1$ and thus $l_2 < l_3$.
Furthermore, using this and the same remark, we conclude that $l_i < l_{i+1}$ for all $i = 2, \ldots, k-1$, therefore proving Claim \ref{claim:2N_2-1}. 
 
In this case, we obtain $F_2(k)$ by considering the submatrix given by the columns $j_1 = l_1 -1$, $j_2 = l_3$, $\ldots$, $j_i = l_{i+1}$, $\ldots$, $j_k = r_1+1$.

\vspace{1mm}

\case \textit{There is at least one labeled row in $A$.}
	
	We wish to extend the partial pre-coloring given for $A$. By Corollary \ref{lema:B_ext_2-nested}, if $B$ is obtained by extending the pre-coloring of $A$ and $B$ is $2$-nested, then neither two blocks corresponding to the same LR-row are colored with the same color, nor there are monochromatic gems, monochromatic weak gems or badly-colored doubly-weak gems in $B$. 
	Let us consider the auxiliary matrix $A+$, defined from a suitable LR-ordering $\Pi$ of the columns of $A$.
	Notice that, if there is at least one labeled row in $A$, then there is at least one labeled row in $A+$ and these labeled rows in $A+$ correspond to rows of $A$ that are labeled with either L, R, or LR. 
	
\vspace{1mm}
	Let $H = H(A+)$ be the graph whose vertices are the rows of $A+$. We say a vertex is an \emph{LR-vertex (resp.\ non-LR vertex)} if it corresponds to a block of an LR-row (resp.\ non-LR row) of $A$. The adjacencies in $H$ are as follows:
	\begin{itemize}
		\item Two non-LR vertices are adjacent in $H$ if the underlying uncolored submatrix of $A$ determined by these two rows contains a gem or a weak gem as a subconfiguration. 
		\item Two LR-vertices corresponding to the same LR-row in $A$ are adjacent in $H$.
		\item Two LR-vertices $v_1$ and $v_2$ corresponding to distinct LR-rows are adjacent if $v_1$ and $v_2$ are labeled with the same letter in $A+$ and the LR-rows corresponding to $v_1$ and $v_2$ overlap in $A$. 
		\item An LR-vertex $v_1$ and a non-LR vertex $v_2$ are adjacent in $H$ if the rows corresponding to $v_1$ and $v_2$ are not disjoint and $v_2$ is not contained in $v_1$. 
	\end{itemize}	  
	The vertices of $H$ are partially colored with the pre-coloring given for the rows of $A$.
		
	Notice that every pair of vertices corresponding to the same LR-row $f$ induces a gem in $A+$ that contains the column $c_f$, and two adjacent LR-vertices $v_1$ and $v_2$ in $H$ do not induce a any kind of gem in $A+$, except when considering both columns $c_{r_1}$ and $c_{r_2}$.
	
The following Claims will be useful throughout the proof.


\begin{claim}  \label{claim:1_teo2nested} 
	Let $C$ be a cycle in $H= H(A+)$. Then, there are at most 3 consecutive LR-vertices labeled with the same letter. The same holds for any path $P$ in $H$. 
\end{claim}

Let $v_1$, $v_2$ and $v_3$ be 3 consecutive LR-vertices in $H$, all labeled with the same letter. Notice that any subset in $H$ of LR-vertices labeled with the same letter in $A+$ corresponds to a subset of the same size of distinct LR-rows in $A$. By definition, two LR-vertices are adjacent in $H$ only if they are labeled with the same letter and the corresponding rows in $A$ contain a gem, or equivalently, if they are not nested. 
Moreover, notice that once the columns of $A$ are ordered according to $\Pi$, these rows have a $1$ in the first non-tag column and a $1$ in the last non-tag column.
Hence, if there were 4 consecutive LR-vertices $v_1$, $v_2$, $v_3$ and $v_4$ in the cycle $C$ of $H$ and all of them are labeled with the same letter, then $v_1$ and $v_2$ are not nested, $v_2$ and $v_3$ are not nested and $v_1$ and $v_3$ must be nested. Thus, since $v_2$ and $v_4$ and $v_1$ and $v_4$ are also nested, then $v_4$ either contains $v_1$ and $v_2$ or is nested in both. In either case, since $v_3$ and $v_4$ are not nested, then $v_1$ and $v_3$ are not nested and this results in a contradiction. \QED

\begin{claim} \label{claim:2_teo2nested}
	There are at most 6 uncolored labeled consecutive vertices in $C$. The same holds for any path $P$ in $H$.
\end{claim} 

This follows from the previous claim and the fact that every pair of uncolored labeled vertices labeled with distinct rows are adjacent only if they correspond to the same LR-row in $A$. \QED

\vspace{2mm}
If $A$ is not $2$-nested, then the partial $2$-coloring given for $H$ cannot be extended to a total proper $2$-coloring of the vertices. Notice that the only pre-colored vertices are those labeled with either L or R, and those LR vertices corresponding to an empty row, which we are no longer considering when defining $A+$. 
According to Lemma \ref{lema:2-color-extension} we have 5 possible cases. 

	\vspace{1mm} 
	\subcase  \textit{There is an even induced path $P =  v_1, v_2, \ldots, v_k $ such that the only colored vertices are $v_1$ and $v_k$, and they are colored with the same color.
}	

We assume without loss of generality througout the proof that $v_1$ is labeled with L, since it is analogous otherwise by symmetry.

If $v_2, \ldots, v_{k-1}$ are unlabeled rows, then we find either $S_2(k)$ or $S_3(k)$ which is not possible since $A$ is admissible.  	


Suppose there is at least one LR-vertex in $P$.	Recall that, an LR-vertex and a non-LR-vertex are adjcent in $H$ only if the rows in $A+$ are both labeled with the same letter and the LR-row is properly contained in the non-LR-row. 

Suppose that every LR-vertex in $P$ is nonadjacent with each other. Let $v_i$ be the first LR-vertex in $P$, and suppose first that $i = 2$. Since $v_2$ is an LR-vertex and is adjacent to $v_1$, then $v_2$ is labeled with L and $v_2 \subsetneq v_1$. Hence, since we are assuming there are no adjacent LR-vertices in $P$ and $k \geq 4$, then $v_3$ is not an LR-vertex, thus it is unlabeled since we are considering a suitable LR-ordering to define $A+$. Let $v_3, \ldots, v_j$ be the maximal sequence of consecutive unlabeled vertices in $P$ that starts in $v_3$. Thus, $v_l \subseteq v_1$ for every $3 \leq l \leq j$. 
	
	Notice that there are no other LR-vertices in $P$: toward a contradiction, let $v_j$ be the next LR-vertex in $P$. If $v_j$ is labeled with L, since $v_3$ is nested in $v_1$, then $v_j$ is adjacent to $v_1$, which is not possible. It is analogous if $v_j$ is labeled with R. Thus, $v_l$ is unlabeled for every $3 \leq l \leq k-1$.
	Moreover, the vertex $v_k$ is labeled with L, for if not we find $D_1$ in $A$ induced by $v_1$ and $v_k$ and this is not possible since $A$ is admissible. However, in that case we find $S_5(k)$.
	
	Hence, if $v_i$ is an isolated LR-vertex (i.e., nonadjacent to other LR-vertices), then $i >2$. It follows that $v_2$ is an unlabeled vertex.
	Notice that a similar argument as in the previous paragraph proves that there are no more LR-vertices in $P$: since $v_{i+1}$ is nested in $v_{i-1}$, it follows that any other LR-vertex is adjacent to $v_{i-1}$. 
	Suppose first that $v_i$ is labeled with L and let $v_2, \ldots, v_{i-1}$ be the maximal sequence of unlabeled vertices in $P$ that starts in $v_2$.
	
	 Since $v_i$ is the only LR-vertex in $P$, if $v_k$ is labeled with L, then necessarily $i=k-1$ for if not $v_k$ is adjacent to $v_{i-1}$. However, since in that case $v_k \supsetneq v_{k-1}= v_i$ and $v_k$ is nonadjacent to every other vertex in $P$, then we find $S_5(k)$. Analogously, if $v_k$ is labeled with R, since $v_j \subseteq v_{i-1}$ for every $j > i$, then $v_k$ is adjacent to $v_{i-1}$ which leads to a contradiction. 
	
	Suppose now that $v_i$ is labeled with R and remember that $i>2$. Furthermore, $v_j$ is unlabeled for every $j > i$. Moreover, $v_j$ is nested in $v_{i-1}$ for every $j>i$, for if not $v_k$ would be adjacent to $v_i$. However, in that case $v_k$ is adjacent to $v_{i-1}$, whether labeled with R or L, and this results in a contradiction. 
	
	Notice that we have also proven that, when considering an admissible matrix and a suitable LR-ordering to define $H$, there cannot be an isolated LR-vertex in such a path $P$, disregarding of the parity of the length of $P$. This last part follows from the previous and the fact that, if the length is 3 and $P$ has one LR-vertex, since the endpoints are colored with distinct colors, then we find $D_4$ if the endpoints are labeled with the same letter and $D_5$ if the endpoints are labeled one with L and the other with R. Moreover, the ordering would not be suitable, which is a necessary condition for the well definition of $A+$, and thus of $H$. If the length of $P$ is odd and greater than 3, then the arguments are analogous as in the even case.
	The following Claim is a straightforward consequence of the previous.
	
\begin{claim}   \label{claim:3_teo2nested}
If there is an isolated LR-vertex in $P$, then it is the only LR-vertex in $P$. Moreover, there are no two nonadjacent LR-vertices in $P$. Equivalently, every LR-vertex in $P$ lies in a sequence of consecutive LR-vertices.
\end{claim}

We say a subpath $Q$ of $P$ is an \emph{LR-subpath} if every vertex in $Q$ is an LR-vertex. We say an LR-subpath $Q$ in $P$ is \emph{maximal} if $Q$ is not properly contained in any other LR-subpath of $P$.

We say that two LR-vertices $v_i$ and $v_j$ are \emph{consecutive }in the path $P$ (resp.\ in the cycle $C$) if either $j=i+1$ or $v_l$ is unlabeled for every $l=i+1, \ldots, j-1$. 
 	
\vspace{1mm} 	
It follows from Claims \ref{claim:1_teo2nested}, \ref{claim:2_teo2nested} and \ref{claim:3_teo2nested} that there is one and only one maximal LR-subpath in $P$. Thus, we have one subcase for each possible length of such maximal LR-subpath of $P$, which may be any integer between $2$ and $6$, inclusive.

	
	\subsubcase Let $v_i$ and $v_{i+1}$ be the two adjacent LR-vertices that induce the maximal LR-subpath. Suppose first that both are labeled with L and that $i=2$. Since $v_2$ is an LR-vertex, $v_2$ is nested in $v_1$ and $v_3$ contains $v_1$.
	Moreover, $v_4$ is labeled with R, for if not $v_4$ is also adjacent to $v_2$. This implies that the R-block of the LR-row corresponding to $v_2$ contains $v_4$ in $A$, for if not we find $D_6$. However, either the R-block of $v_2$ intersects the L-block of $v_3$ --which is not possible since we are considering a suitable LR-ordering--, or $v_3$ is disjoint with $v_4$ since the LR-rows corresponding to $v_2$ and $v_3$ are nested, and thus we find $D_6$. Hence, $k>4$. 
	
	By Claim \ref{claim:3_teo2nested} and since there is no other LR-vertex in the maximal LR-subpath, there are no other LR-vertices in $P$. Equivalently, $v_4, \ldots, v_{k-1}$ are unlabeled vertices. Moreover, this sequence of unlabeled vertices is chained to the right, since if it was chained to the left, then every left endpoint of $v_j$ for $j=4, \ldots, k-1$ would be greater than $r(v_1)$ and thus $v_k$ results adjacent to $v_2$. Hence, we find $P_0(k-1,0)$ in $A$ as a subconfiguration of the submatrix given by considering the rows corresponding to $v_1$, $v_2$, $v_4, \ldots, v_k$, which is not possible since $A$ is admissible. 
	The proof is analogous if $i>2$, with the difference that we find $P_0(k-1,i)$ in $A$.
	Furthermore, the proof is analogous if $v_i$ and $v_{i+1}$ are labeled with distinct letters.
	
\vspace{1mm}

\subsubcase Let $Q= < v_i, v_{i+1},v_{i+2}>$ be the maximal LR-subpath of $P$. Suppose first that not every vertex in $Q$ is labeled with the same letter. 

If $v_i$ is labeled with R, since there is a sequence of unlabeled vertices between $v_1$ and $v_i$, then $v_{i+1}$ is labeled with R. This follows from the fact that if not, $v_{i+1}$ would be adjacent to either $v_1$ or some vertex in the unlabeled chain. The same holds for $v_{i+1}$ and thus we are in the previous situation. Hence, $v_{i}$ is labeled with L and we have the following claim. 

\begin{claim} \label{claim:1_case2-1_teo2nested}
For every maximal LR-subpath of $P$, the first vertex is labeled with L.
\end{claim}

Suppose $v_i$ and $v_{i+1}$ are both labeled with L and $v_{i+2}$ is labeled with R. Notice that, if $i=2$, then $v_2$ is labeled with L, $v_4$ is labeled with R and $v_3$ may be labeled with either L or R.

Since $v_{i+1}$ and $v_{i+2}$ are labeled with distinct letters, then they correspond to the same LR-row in $A$. Notice that $v_i$ is contained in $v_{i+1}$. Thus, since $v_i$ and $v_{i+1}$ are adjacent, the R-block corresponding to $v_i$ in $A$ contains $v_{i+2}$. Therefore, we find $P_0(k,i)$ or $P_1(k,i)$ as a subconfiguration of the submatrix induced by considering all the rows of $P$.

If instead $v_{i+1}$ and $v_{i+2}$ are both labeled with R, then $v_i$ and $v_{i+1}$ are the two blocks of the same LR-row in $A$. Hence, since $v_{i+1}$ and $v_{i+2}$ are adjacent and $v_k$ is nonadjacent to $v_{i+1}$, then $v_{i+1}$ contains $v_{i+2}$ and thus the L-block of the LR-row corresponding to $v_{i+2}$ contains $v_i$. Once again, we find either $P_0(k,i)$ or $P_1(k,i)$.

\vspace{1mm}
Suppose now that all vertices in $Q$ are labeled with the same letter and suppose first that $i=2$. Since $v_1$ and $v_2$ are adjacent, then every vertex in $Q$ is labeled with L.
Notice that $k>4$ since $v_5$ is uncolored and the endpoints of $P$ are colored with the same color.
Since $v_2$ is adjacent to $v_1$, then $v_1 \subsetneq v_3$ and $v_4 \subsetneq v_3$.
Since $k$ is even and $k>4$, then $v_5$ is an unlabeled vertex. Moreover, for every unlabeled vertex $v_j$ such that $j >4$, $l(v_j) > r(v_1)$ and $r(v_j) \leq r(v_3)$, for if not $v_j$ and $v_3$ would be adjacent. However, $v_k$ is not labeled with L for in that case it would be adjacent to $v_3$. Furthermore, if $v_k$ is labeled with R, then we find $D_8$, which is not possible since we assumed $A$ to be admissible.

Suppose now that $i>2$. In this case, there is a sequence of unlabeled vertices between $v_1$ and $v_i$. 
If every vertex in $Q$ is labeled with L, since $v_1$ and $v_i$ are nonadjacent (and thus $v_1$ is nested in $v_i$) and $v_{i+1}$ is nonadjacent with $v_{i-1}$, then $v_i \subsetneq v_{i+1}$, $v_{i+2} \subsetneq v_{i+1}$. It follows that $v_j$ is contained between $r(v_i)$ and $r(v_{i+1})$ for every $j>i+2$ and therefore $v_k$ is adjacent either to $v_{i+1}$ or $v_i$, which results in a contradiction. 

If every vertex in $Q$ is labeled with R, then $v_{i+1} \subsetneq v_{i}$ and $v_{i+1} \subsetneq v_{i+2}$ for if not $v_{i+1}$ would be adjacent to $v_{i-1}$ and $v_{i+2}$. Hence, if $i+2=k-1$, then $v_k$ would be adjacent also to $v_{i+1}$. Hence, there is at least one unlabeled vertex $v_j$ with $j>i+2$. Moreover, for every such vertex $v_j$ holds that $l(v_j) < l(v_i)$ and $r(v_j) > l(v_{i+1})$. Hence, if $v_k$ is labeled with R, then $v_k$ is adjacent to $v_{i+1}$. If instead $v_k$ is labeled with L, then we find $D_8$ as a subconfiguration in the submatrix of $A$ induced by $v_k$ and the LR-rows corresponding to $v_{i}$ and $v_{i+1}$.


\subsubcase Let $Q= <v_i, v_{i+1}, v_{i+2}, v_{i+3}>$ be the maximal LR-subpath of $P$. Notice that either 2 are labeled with L and 2 are labeled with R, or $1$ is labeled with L and 3 are labeled with R, or viceversa. Moreover, by Claim \ref{claim:1_case2-1_teo2nested} we know that $v_i$ is labeled with L. Every vertex $v_j$ such that $1<j<i$ or $i+3<j<k$ is an unlabeled vertex. 

Suppose first that $v_i$ is the only vertex in $Q$ labeled with L. Thus, $v_{i+1}$ is the R-block of the LR-row in $A$ corresponding to $v_i$. Hence, either $v_{i+1} \subsetneq v_{i+2}$ or viceversa.
Notice that there is at least one unlabeled vertex $v_j$ between $v_{i+3}$ and $v_k$, for if not $v_k$ is adjacent to $v_{i+1}$ or $v_{i+2}$. Moreover, either $v_j$ is contained in $v_{i+2} \setminus v_{i+3}$ or in $v_{i+3} \setminus v_{i+2}$ for every $j>i+4$. In any case, $v_k$ results adjacent to either $v_{i+2}$ or $v_{i+3}$, which results in a contradiction. 

Hence, at least $v_i$ and $v_{i+1}$ are labeled with L. Suppose that $v_{i+2}$ is labeled with R --and thus $v_{i+3}$ is labeled with R. Notice that, if $v_{i+3} \supsetneq v_{i+2}$, then there is no possibility for $v_k$ for, if $v_k$ is labeled with R, then $v_k$ is adjacent to $v_{i+2}$ and if $v_k$ is labeled with L, then $v_k$ is adjacent to $v_i$ and $v_{i+1}$.
However, the same holds if $v_{i+2} \supsetneq v_{i+3}$ since there is at least one unlabeled vertex $v_j$ with $j>i+3$ and thus for every such vertex holds $l(v_j)>l(v_{i+2})$ and therefore this case is not possible.


Finally, suppose that $v_i$, $v_{i+1}$ and $v_{i+2}$ are labeled with L and thus $v_{i+3}$ is labeled with R. Thus, $v_k$ is labeled with R and is nested in $v_{i+3}$. Moreover, there is a chain of unlabeled vertices $v_j$ between $v_{i+3}$ and $v_k$ such that $v_j$ is nested in $v_{i+3}$ for every $j>i+4$. 
Furthermore, $v_i \subsetneq v_{i+1}$ and $v_i \subseteq v_{i+2} \subsetneq v_{i+1}$: if $i=2$, then $v_2 \subsetneq v_1$ and since $v_3$ and $v_4$ are nonadjacent to $v_1$, then $v_3, v_4 \supseteq v_1$. If instead $i>2$, then for every unlabeled vertex $v_j$ between $v_1$ and $v_i$, $r(v_j) < r (v_i)$, except for $j=i-1$ for which holds $r(v_{i-1}) > r(v_i)$. Hence, since $v_{i+1}$ and $v_{i+2}$ are nonadjacent to every such vertex, then $v_j \subset v_{i+1}, v_{i+2}$ for $1<j<i$. We find $P_0(k-3,i)$ in $A$ 
since the R-block corresponding to $v_i$ is contained in $v_{i+3}$ and thus the R-block intersects the chain of vertices between $v_{i+3}$ and $v_k$. 

We have the following as a consequence of the previous arguments.
\begin{claim} \label{claim:2_case2-1_teo2nested}
Let $v_i$ and $v_{i+1}$ be the first LR-vertices that appear in $P$. If $v_{i+1}$ is also labeled with L, then $v_i \subsetneq v_{i+1}$. Moreover, if $v_{i+2}$ is also an LR-vertex that is labeled with L, then $v_{i+2} \subsetneq v_{i+1}$.
\end{claim} 

\subsubcase 
Let $Q=<v_i, \ldots, v_{i+4}>$ be the maximal LR-subpath of $P$. By Claim \ref{claim:1_case2-1_teo2nested}, $v_i$ is labeled with L. Moreover, either (1) $v_i$ and $v_{i+1}$ are labeled with L and $v_{i+2}$, $v_{i+3}$ and $v_{i+4}$ are labeled with R, or (2) $v_i$, $v_{i+1}$ and $v_{i+2}$ are labeled with L and $v_{i+3}$ and $v_{i+4}$ are labeled with R. It follows from Claim \ref{claim:2_case2-1_teo2nested} that $v_i \subsetneq v_{i+1}$. 

Let us suppose the first statement. If $v_{i+3} \subsetneq v_{i+4}$, then there is at least one unlabeled vertex in $P$ between $v_{i+4}$ and $v_k$, for if not $v_k$ would be adjacent to $v_{i+2}$. Since every vertex $v_j$  for $i+5<j \leq k$ is contained in $v_{i+4} \setminus v_{i+3}$, it follows that $v_k$ is adjacent to $v_{i+2}$ and thus this is not possible. Hence, $v_{i+3} \supsetneq v_{i+4}$. 
Furthermore, $v_{i+3} \supsetneq v_{i+2}$, and since $v_k$ is nonadjacent to $v_{i+2}$, then $v_{i+2} \supsetneq v_{i+4}$. Since there is a sequence of unlabeled vertices between $v_{i+4}$ and $v_k$, then we find $P_2(k,i-2)$ if $v_{i+4}$ is nested in the R-block of $v_{i+2}$, or $P_0(k-3,i-2)$ otherwise.

Suppose now (2), this is, $v_i$, $v_{i+1}$ and $v_{i+2}$ are labeled with L and $v_{i+3}$ and $v_{i+4}$ are labeled with R. By Claim \ref{claim:2_case2-1_teo2nested}, $v_i \subsetneq v_{i+1}$ and $v_{i+2} \subsetneq v_{i+1}$. Furthermore, since $v_k$ is nonadjacent to $v_{i+3}$, it follows that $v_{i+3} \supsetneq v_{i+4}$. In this case, we find $P_2(k,i-2)$ if $v_{i+4}$ is nested in the R-block of $v_{i+2}$, or $P_0(k-3,i-2)$ otherwise. 

\vspace{1mm}
\subsubcase  
Suppose by simplicity that the length of $P$ is 8 (the proof is analogous if $k>8$), and thus let $Q=<v_2, \ldots, v_7>$ be the maximal LR-subpath of $P$ of length 6. Notice that $v_8$ is labeled with R and colored with the same color as $v_1$.
Hence, $v_2$, $v_3$ and $v_4$ are labeled with L and $v_5$, $v_6$ and $v_7$ are labeled with R. By Claim \ref{claim:2_case2-1_teo2nested}, $v_2 \subsetneq v_3$ and $v_4 \subsetneq v_3$. It follows that $v_2 \subsetneq v_4$, since $v_1$ and $v_4$ are nonadjacent. Using an analogous argument, we see that $v_5 \subsetneq v_6$, $v_6 \supsetneq v_5, v_7$ and $v_7 \subsetneq v_5$ for if not it would be adjacent to $v_8$.
Since consecutive LR-vertices are adjacent, the LR-rows corresponding to $v_{i+3}$ and $v_{i+4}$ are not nested, and the same holds for the LR-rows in $A$ of $v_3$ and $v_2$. Since $A$ is admissible, the LR-rows of $v_6$ and $v_3$ are nested. This implies that the L-block of the LR-row corresponding to $v_6$ contains the L-block of $v_4$ and $v_2$.
Moreover, since the LR-rows of $v_7$ and $v_5$ are nested, the LR-rows of $v_6$ and $v_7$ are not and $v_7$ is contained in  $v_6$, then the L-block of $v_7$ contains the L-block of $v_6$. Hence, $v_7$ contains $v_5$ and thus $v_8$ results adjacent to $v_5$, which is a contradiction. 
	
	\vspace{2mm}
 	\subcase \textit{There is an odd induced path $P =  < v_1, v_2, \ldots, v_k >$ such that the only colored vertices are $v_1$ and $v_k$, and they are colored with distinct colors.}
	
	Throughout the previous case proof we did not take under special consideration the parity of $k$, with one exception: when $k=5$ and the maximal LR-subpath has length $2$. In other words, notice that for every other case, we find the same forbidden submatrices of admissibility with the appropriate coloring for those colored labeled rows.

	Suppose that $k=5$, the maximal LR-subpath has length $2$, and suppose without loss of generality that $v_2$ and $v_3$ are the LR-vertices (it is analogous otherwise by symmetry). If both are labeled with L, then $v_2$ is contained in $v_3$ and thus the R-block of $v_2$ properly contains the R-block of $v_3$. Moreover, since $v_4$ is unlabeled and adjacent to $v_5$ --which should be labeled with R since the LR-ordering is suitable--, it follows that there is at least one column in which the R-block of the LR-row corresponding to $v_3$ has a $0$ and $v_5$ has a $1$. Furthermore, there exists such a column in which also the R-block of $v_2$ has a $1$. Since $v_1$ and $v_2$ are adjacent, then $v_2 \subsetneq v_1$ and thus there is also a column in which $v_2$ has a $0$, $v_3$ has a $1$ and $v_1$ has a $1$. Moreover, there is a column in which $v_1$, $v_2$ and $v_3$ have a $1$ and $v_5$ and the R-blocks of $v_2$ and $v_3$ all have a $0$, and an analogous column in which $v_1$, $v_2$ and $v_3$ have a $0$ and $v_5$ and the R-blocks of $v_2$ and $v_3$ have a $1$. It follows that there is $D_{10}$ in $A$ which is not possible since $A$ is admissible. 
	If instead $v_2$ is labeled with L and $v_3$ is labeled with R, then $v_2$ and $v_3$ are the L-block and R-block of the same LR-row $r$ in $A$, respectively. We can find a column in $A$ in which $v_1$ and $r$ have a $1$ and the other rows have a $0$, a column in which only $v_1$ has a $1$, a column in which only $v_4$ has a $1$ (notice that $v_4$ is unlabeled), and a column in which $r$, $v_4$ and $v_5$ have a $1$ and $v_1$ has a $0$. It follows that there is $P_0(4,0)$ in $A$, which results in a contradiction.

 	\vspace{2mm}
 	\subcase \textit{There is an induced uncolored odd cycle $C$ of length $k$.} 

	If every vertex in $C$ is unlabeled, then the proof is analogous as in case \ref{case:teo2nested_case1}, where we considered that there are no labeled vertices of any kind.
	
	Suppose there is at least one LR-vertex in $C$. Notice that there no labeled vertices in $C$ corresponding to rows in $A$ labeled with L or R, which are the only colored rows in $A+$.
		
	Suppose $k = 3$. If 2 or 3 vertices in $C$ are LR-vertices, then there is either $D_7$, $D_8$, $D_9$, $D_{11}$, $D_{12}$, $D_{13}$ or $S_7(3)$. If instead there is exactly one LR-vertex and since every uncolored vertex corresponds either to an unlabeled row or to an LR-row, then we find $F'_0$ in $A$.
		
	Suppose that $k\geq 5$ and let $C = v_1, v_2, \ldots, v_k $ be an uncolored odd cycle of length $k$. Suppose first that there is exactly one LR-vertex in $C$. We assume without loss of generality by symmetry that $v_1$ is such LR-vertex and that $v_1$ is labeled with L in $A+$.
	
 	Hence, either $v_j$ is nested in $v_1$, or $v_j$ is disjoint with $v_1$, for every $j=3, \ldots, k-1$.
 	If $v_j$ is nested in $v_1$ for every $j=3, \ldots, k-1$, since $v_k$ is adjacent to $v_1$ and nonadjacent to $v_j$ for every $j=3, \ldots, k-1$, then $l(v_k) < l(v_{k-2}) < l(v_{k-3}) < \ldots < l(v_2)<r(v_1)$ and $r(v_k) > r(v_1)$. Hence, we either find $F_1(k)$ or $F'_1(k)$ induced by the columns $l(v_{k-1}), \ldots, r(v_k)$. 
 	 	
 	If instead $v_j$ is disjoint with $v_1$ for all $j=3, \ldots, k-1$, then $v_j$ is nested in $v_k$ for every $j=3, \ldots, k-2$. In this case, we find $F_2(k)$ or $F'_2(k)$ induced by the columns $l(v_k)-1, \ldots, r(v_{k-1})$. 
 	
 	Now we will see what happens if there is more than one LR-vertex in $C$. First we need the following Claim.
 	
\begin{claim} \label{claim:5_claim_2-nested}
	If $v$ and $w$ in $C$ are two nonadjacent consecutive LR-vertices, then there is one sense of the cycle for which there is exactly one unlabeled vertex between $v$ and $w$.
\end{claim} 	

If $k=5$, then we have to see what happens if $v_1$ and $v_4$ are such vertices and $v_5$ is an LR-vertex. We are assuming that $v_2$ and $v_3$ are unlabeled since by hypothesis $v_1$ and $v_4$ are consecutive LR-vertices in $C$. Suppose that $v_1$ and $v_4$ are labeled with L and for simplicity assume that $v_1 \subsetneq v_4$. Thus, $v_5$ is labeled with L, for if not $v_5$ can only be adjacent to $v_1$ or $v_4$ and not both. Moreover, since $v_5$ is nonadjacent to $v_2$, then $v_5$ is contained in $v_1$ and $v_4$. In this case, we find $F_2(5)$ as a subconfiguration in $A$.

If instead $v_1$ is labeled with L and $v_4$ is labeled with R, then $v_5$ is the L-block of the LR-row corresponding to $v_4$. In this case, we find $S_7(4)$ as a subconfiguration of $A_\tagg$.

Let $k>5$, and suppose without loss of generality that $v_1$ and $v_4$ are such LR-vertices. Thus, by hypothesis, $v_2$ and $v_3$ are unlabeled vertices. 
Suppose first that $v_1$ and $v_4$ are labeled with L and $v_1 \subsetneq v_4$. Then $l(v_2) < l(v_3)$. 
If $v_j$ is unlabeled for every $j>4$, then $v_j$ is nested in $v_3$ and thus $v_k$ cannot be adjacent to $v_1$. 
Moreover, for every $j>4$, $v_j$ is not an LR-vertex labeled with L either. Suppose to the contrary that $v_5$ is an LR-vertex labeled with L. Since $v_5$ is adjacent to $v_4$ and the LR-rows corresponding to $v_1$ and $v_4$ are nested, then $v_5$ is also adjacent to $v_1$, which is not possible since we are assuming that $k>5$. If instead $j>5$, since there is a sequence of unlabeled vertices between $v_4$ and $v_j$, then $r(v_j)>l(v_3)$ and thus it is adjacent to $v_3$.
By an analogous argument, we may assert that $v_j$ is not an LR-vertex for every $j>4$. The proof is analogous if $v_1 \supsetneq v_4$.

Thus, let us suppose now that $v_1$ is labeled with L and $v_4$ is labeled with R. If $v_5$ is the L-block of the LR-row corresponding to $v_4$, since $v_2$ and $v_5$ are nonadjacent, then $r(v_5)< l(v_2)$ and hence $v_5 \subsetneq v_1$. 
Moreover, $v_6$ is not an LR-vertex for in that case $v_6$ must be labeled with L and thus $v_6$ is also adjacent to $v_1$. Furthermore, since at least $v_6$ is an unlabeled vertex, then every LR-vertex $v_j$ in $C$ with $j>4$ is labeled with L, for if not $v_j$ is either adjacent to $v_4$ or nonadjacent to $v_6$ (or the maximal sequence of unlabeled vertices in $C$ that contains $v_6$).
Thus, we may assume that there no other LR-vertices in $C$, perhaps with the exception of $v_k$. 
However, if $v_k$ is an LR-vertex labeled with L, since it is adjacent to $v_1$, then it is also adjacent to $v_5$. And if $v_k$ is unlabeled, then $v_k$ is adjacent to $v_2$, $v_3$ or $v_4$ ($v_k$ must contain this vertices so that it results nonadjacent to them, but $v_4$ is the limit since $v_4$ is labeled with R and thus it ends in the last column).

Analogously, if $v_5$ is unlabeled, then $v_k$ is nonadjacent to $v_1$ since it must be contained in $v_3$. Finally, if $v_5$ is an LR-vertex labeled with R, then it is contained in $v_4$. Thus, the only possibility is that $v_{k-1}$ is an LR-vertex labeled with R and $v_7$ is the L-block of the corresponding LR-row. However, since $A$ is admissible, either $v_6$ is nested in $v_5$ or $v_6$ is nested in $v_4$. In the first case, it results also adjacent to $v_4$ and in the second case it results nonadjacent to $v_5$, which is a contradiction. 
Notice that the arguments are analogous if the number of unlabeled vertices in both senses of the cycle is more than $2$. Therefore, there is one sense of the cycle in which there is exactly one unlabeled vertex between any two nonadjacent consecutive LR-vertices of $C$. \QED

This Claim follows from the previous proof.

\begin{claim} \label{cor:5_claim_2-nested}
	If $C$ is an odd uncolored cycle in $H$, then there are at most two nonadjacent LR-vertices.
\end{claim} 
	
 	Suppose that $v_1$ and $v_i$ are consecutive nonadjacent LR-vertices, where $i>2$. It follows from Claim \ref{claim:5_claim_2-nested} that $i=3$ or $i=k-1$. We assume the first without loss of generality, and suppose that $v_1$ is labeled with L.
 	Suppose there is at least one more LR-vertex nonadjacent to both $v_1$ and $v_3$, and let $v_j$ be the first LR-vertex that appears in $C$ after $v_3$. It follows from Claim \ref{claim:5_claim_2-nested} that $j=5$. 	
 	If $v_1$ and and $v_3$ are labeled with distinct letters, since $v_4$ is an unlabeled vertex, then $v_4$ is contained in $v_2$, and thus $v_5$ cannot be labeled with L or R for, in either case, it would be adjacent to $v_2$.
 	Thus, every LR-vertex in $C$ must be labeled with the same letter. Let us assume for simplicity that $k=5$ (the proof is analogous for every odd $k>5$) and that $v_1 \subset v_3$. Since $v_5$ is nonadjacent to $v_3$, then the corresponding LR-rows are nested. The same holds for $v_1$ and $v_3$. Moreover, $v_5$ contains both $v_1$ and $v_3$, and the R-block of $v_3$ contains the R-block of $v_1$. Furthermore, since $v_1$ and $v_5$ are adjacent, the R-block of the LR-row corresponding to $v_1$ contains the R-block of the LR-row corresponding to $v_5$ and thus the R-block of $v_3$ also contains the R-block of $v_5$, which results in $v_3$ and $v_5$ being adjacent and thus, in a contradiction that came from assuming that there is are at least three nonadjacent LR-vertices in $C$. \QED
 	
\vspace{2mm}
	We now continue with the proof of the case. 
	Notice first that, as a consequence of the previous claim and Claim \ref{claim:2_teo2nested}, either there are exactly two nonadjacent LR-vertices in $C$ or every LR-vertex is contained in a maximal LR-subpath of length at most $6$.
	 
	\subsubcase Suppose there are exactly two LR-vertices in $C$ and that they are nonadjcent. Let $v_1$ and $v_3$ be such LR-vertices. Suppose without loss of generality that $v_1 \subset v_3$. Hence, every vertex that lies between $v_3$ and $v_1$ is nested in $v_2$, since they are all unlabeled vertices by hypothesis. 	Thus, if $v_1$ and $v_3$ are both labeled with L, then we find $F_1(k)$ contaned in the submatrix induced by the columns $r(v_1), \ldots, r(v_2)$.
	If instead $v_1$ is labeled with L and $v_3$ is labeled with R, then we find $F_2(k)$ contained in the same submatrix. 

	
	\subsubcase Suppose instead that $v_1$ and $v_2$ are the only LR-vertices in $C$. If $v_1$ and $v_2$ are the L-block and R-block of the same LR-row, then we find $S_8(k-1)$ in $A$. 
	If instead they are both labeled with L, then every other vertex $v_j$ in $C$ is unlabeled and $v_j$ is nested in $v_1$ or $v_2$ for every $j>3$, depending on whether $v_1 \subsetneq v_2$ or viceversa. Suppose that $v_1 \subsetneq v_2$. If there is a column in which both $v_3$ and the R-block of $v_1$ have a $1$, then we find $S_8(k-1)$ in $A$. If there is not such a column, then we find $F_2(k)$ in $A$.
	
 	\vspace{1mm}
	\subsubcase Suppose that the maximal LR-subpath $Q$ in $C$ has length $3$, and suppose $Q= <v_1, v_2, v_3>$.
	If $v_1$, $v_2$ and $v_3$ are labeled with the same letter, then either $v_2 \subsetneq v_1, v_3$ or $v_2 \supsetneq v_1, v_3$, and since $v_1$ and $v_3$ are nonadjacent if $k>3$, either $v_3 \subsetneq v_1$ or $v_1 \subsetneq v_3$. Suppose without loss of generality that all three LR-vertices are labeled with L, $v_2 \subsetneq v_1, v_3$ and $v_1 \subsetneq v_3$. In this case, there is a sequence of unlabeled vertices between $v_3$ and $v_1$ such that the column index of the left endpoints of the vertices decreases as the vertex path index increases. As in the previous case, if there is a column such that the R-block of $v_2$ and $v_4$ have a $1$, then we find $S_8(k-1)$ in $A$ contained in the submatrix induced by the columns $r(v_1), \ldots, l(v_1)$. If instead there is not such column, then we find $F_2(k)$ contained in the same submatrix.
	
	If $v_1$ and $v_2$ are labeled with L and $v_3$ is labeled with R, then there is a sequence of unlabeled vertices $v_4, \ldots, v_k$ such that the column index of the left endpoints of such vertices decreases as the path index increases. Moreover, since $v_k$ is adjacent to $v_1$ and nonadjacent to $v_2$, then $v_1 \supsetneq v_2$. Hence, we find $S_7(k-1)$ contained in the submatrix of $A$ induced by the columns $r(v_1), \ldots, l(v_1)$.	 
	 	
 	\vspace{1mm}
 	\subsubcase Suppose that the maximal LR-subpath $Q$ in $C$ has length $4$ and that $Q= <v_1$, $v_2$, $v_3$, $v_4>$.
 	Suppose that $v_1$ and $v_2$ are labeled with L and $v_3$ and $v_4$ are labeled with R. If $v_1 \subsetneq v_2$, then $v_k$ cannot be adjacent to $v_1$. 
 	Thus $v_2 \subsetneq v_1$ and $v_4 \supsetneq v_3$. Since there is a chain of unlabeled vertices and its left endpoints decrease as the cycle index increases, then we find $S_7(k-1)$ considering the submatrix induced by every row in $A$.
	If instead $v_1$ is labeled with L and the other three LR-vertices are labeled with R, then first let us notice that $v_2$ is the R-block of $v_1$, the LR-rows of $v_2$ and $v_4$ are nested and $v_3 \subsetneq v_2,v_4$. Moreover, $v_2 \subsetneq v_4$, for if not $v_k$ would not be adjacent to $v_1$. Thus, the left endpoint of the chain of unlabeled vertices between $v_4$ and $v_1$ decreases as the cycle index increases. Hence, if $k=5$, then we find $S_7(3)$ induced by the LR-rows corresponding to $v_3$ and $v_4$ and the unlabeled row corresponding to $v_5$. Suppose that $k>5$. Since $v_3 \subsetneq v_2$ and $v_2$ is the R-block of $v_1$, then the L-block of the LR-row corresponding to $v_3$ contains both $v_1$ and the L-block of $v_4$. We find $S_7(k-3)$ in $A$ as a subconfiguration of the submatrix induced by the rows $v_3, v_4, \ldots, v_{k-1}$.	
 	
 	\subsubcase Suppose now that $Q=<v_1, \ldots, v_5>$ is the longest LR-subpath in $C$, and suppose that $v_1$ and $v_2$ are labeled with L and that the remaining rows in $Q$ are labeled with R. Since $v_1$ is adjacent to $v_k$, then $v_1 \supsetneq v_2$ and $v_5 \supsetneq v_4,v_3$. Since the LR-rows corresponding to $v_3$ and $v_5$ are nested, then $v_2$ is contained in the L-block corresponding to $v_5$, and since $v_4 \subsetneq v_5$, the R-block of $v_1$ is also contained in $v_5$. Thus, we find $S_7(k-3)$ in $A$ as a subconfiguration considering the LR-rows corresponding to $v_1$ and $v_5$ and $v_6, \ldots, v_k$.
 	The proof is analogous if $Q$ has length $6$, and thus this case is finished. 
 	
 	\vspace{1mm}
 	\subcase \textit{There is an induced odd cycle $C = v_1, v_2, \ldots, v_k, v_1$ with exactly one colored vertex. }
 	We assume without loss of generality that $v_1$ is the only colored vertex in the cycle $C$, and that $v_1$ is labeled with L.
 	Notice that, if there are no LR-vertices in $C$, then the proof is analogous as in the case in which we considered that there are no labeled vertices of any kind. Hence, we assume there is at least one LR-vertex in $C$.

\begin{claim} \label{claim:6_claim_2-nested}
	If there is at least one LR-vertex $v_i$ in $C$ and $i \neq 2$, then $v_i$ is the only LR-vertex in $C$.
\end{claim} 

	Let $v_i$ be the LR-vertex in $C$ with the minimum index, and suppose first that $v_i$ is labeled with L. Since $i \neq 2$ and $v_1$ is a non-LR row in $A$, then $v_i \supseteq v_1$, for if not they would be adjacent. Moreover, $v_l \subset v_i$ for every $l < i-1$. 
	Toward a contradiction, let $v_j$ the first LR-vertex in $C$ with $j >i$ and suppose $v_j$ is labeled with L. Notice that the only possibility for such vertex is $j=i+1$. This follows from the fact that, if $v_{i+1}$ is unlabeled, then $v_{i+1}$ is contained in $v_{i-1}$, and the same holds for every unlabeled vertex between $v_i$ and $v_j$. Hence, if there was other LR-vertex $v_j$ labeled with L such that $j>i+1$, then it would be adjacent to $v_{i+1}$ which is not possible.		
	Then, $j= i+1$ and thus $v_j$ contains $v_l$ for every $l \leq i$. However, $v_k$ and $v_1$ are adjacent, and since $v_k$ must be an unlabeled vertex, then $v_k$ is not disjoint with $v_i$, which results in a contradiction. 
	
	Suppose that instead $v_j$ is labeled with R. Using the same argument, we see that, if $j>i+1$, then every unlabeled vertex between $v_i$ and $v_j$ is contained in $v_{i-1}$ and thus it is not possible that $v_j$ results adjacent to $v_{j-1}$ if it is unlabeled. Hence, $j=i+1$. Moreover, there must be at least one more LR-vertex labeled with R since if not, it is not possible for $v_1$ and $v_k$ to be adjacent. Thus, $v_{k-1}$ must be labeled with R and $v_k$ is the L-block of the LR-row corresponding to $v_{k-1}$. Furthermore, $v_{k-1}$ is contained in $v_1$. We find $F_2(k)$ in $A$ as a subconfiguration in the submatrix induced by considering every row. 
	Therefore, $v_i$ is the only LR-vertex in $C$. \QED 	

	The following is a straightforward consequence of the previous proof and the fact that, if $v_i$ is the first LR-vertex in $C$ and $i >2$, then every unlabeled vertex that follows $v_i$ is nested in $v_{i-1}$, thus if $v_1$ is adjacent to $v_k$ then $v_k$ must be nested in $v_2$.
	
\begin{claim} \label{claim:6_claim_2-nested_2}
	If $v_i$ in $C$ is an LR-vertex and $i \neq 2$, then $i=3$.
\end{claim}	

	 It follows from Claim \ref{claim:2_teo2nested} that there are at most 6 consecutive LR-vertices in such a cycle $C$.
	Let $Q=<v_i, \ldots, v_j>$ be the maximal LR-subpath and suppose that $|Q|=5$ and $v_1$ is labeled with L. 
	Notice that, if $v_i$ is labeled with R,  then $v_{j-1}$ and $v_j$ are labeled with L. Moreover, since there is a sequence of unlabeled vertices between $v_1$ and $v_i$ and $v_{j-1}$ is nonadjacent to $v_2$, then $v_{j-1}$ is contained in $v_1$ and thus it results adjacent to $v_1$, which is not possible. 
	Then, necessarily $v_i$ is labeled with L and thus $v_j$ is labeled with R. Moreover, if $i>2$, then $v_i$ contains $v_1$ and every unlabeled vertex between $v_1$ and $v_{i-1}$, and if $i=2$, then $v_2 \subsetneq v_1$. In either case, $v_{i+1}$ contains $v_i$. Hence, at most $v_{i+2}$ is labeled with L and there are no other LR-vertices labeled with L for they would be adjacent to $v_i$ or $v_{i+1}$. In particular, the last vertex of the cycle $v_k$ is not labeled with L, thus, since it is uncolored, $v_k$ is an unlabeled vertex. 
	However, $v_k$ is adjacent to $v_1$, and this results in a contradiction. Therefore, it is easy to see that it is not possible to have more than 4 consecutive LR-vertices in $C$. Furthermore, in the case of $|Q|=4$, either $v_i$ and $v_{i+3}$ are labeled with L and $v_{i+1}$ and $v_{i+2}$ are labeled with R, or $v_i$ and $v_{i+1}$ are labeled with R and $v_{i+3}$ is labeled with L.

\begin{claim} \label{claim:7_claim_2-nested}
Suppose $v_2$ is an LR-vertex and let $v_i$ be another LR-vertex in $C$. Then, either $i=k$ or $i \in \{3,4,5\}$. Moreover, in this last case, $v_j$ is an LR-vertex for every $2 \leq j \leq i$.
\end{claim}	

Notice first that, if $v_2$ is an LR-vertex, then by definition of $H$, $v_2$ is labeled with L and $v_2 \subsetneq v_1$. 
If $i = 3$ or $i=k$, then we are done. Suppose that $i \neq k$ and there is a sequence of unlabeled vertices $v_j$ between $v_2$ and $v_i$, where $j=3, \ldots, i-1$. Hence, since $v_2 \subsetneq v_1$, then $v_j \subseteq v_1$ for $j=3, \ldots, i-1$. In that case, $v_i$ is labeled with the same letter than $v_1$ and $v_2$. Moreover, since $i \neq k$, $v_1$ and $v_i$ are nonadjacent and thus $v_i \supseteq v_1$ which is not possible since $v_{i-1} \subseteq v_1$. The contradiction came for assuming that there is a sequence of unlabeled vertices between $v_2$ and $v_i$ and that $v_i \neq v_k$. 
Hence, if $i \neq 3,k$, then every vertex between $v_2$ and $v_i$ is an LR-vertex. Since we know that the maximal LR-subpath in $C$ has length at most $4$ and $v_2$ is an LR-vertex, then necessarily $v_i$ must be either $v_3$,$v_4$ or $v_5$. \QED

\vspace{1mm}
 We now split the proof in two cases.
 \subsubcase \textit{$v_2$ is an LR-vertex.}
 
 	Suppose first that $v_2$ is the only LR-vertex in $C$. By definition of $H$, $v_2$ is labeled with L and $v_2 \subsetneq v_1$. Since there are no other LR-vertices in $C$, then $v_j \subseteq v_1$ for every $j<k$. 
 	In this case, we find $F'_1(k)$ as a subconfiguration contained in the submatrix of $A$ induced by the columns $r(v_2), \ldots, r(v_k)$. 
 	
 	Suppose now that there is exactly one more LR-vertex $v_i$ with $i > 2$. If $i \neq 3$, then by the previous claim we know that $i=k$. If $v_k$ is labeled with L, then we find $F'_1(k)$ contained in the submatrix of $A$ induced by the columns $r(v_2), \ldots, r(v_k)$. 
 	If instead $v_k$ is labeled with R, then $r(v_1) > l(v_k)$ and this is not possible since the LR-ordering used to define $A+$ is suitable. 
 	Suppose that $i =3$. If $v_3$ is labeled with L, then $v_3 \supseteq v_1$, and if $v_3$ is labeled with R, then $v_3$ is the R-block corresponding to the same LR-row of $v_2$ in $A$. In either case, since every other vertex $v_j$ in $C$ is unlabeled, then $l(v_j) > r(v_1)$ for every $j<k$. Thus, if $v_3$ is labeled with L, then we find $F'_2(k)$ contained in the submatrix of $A$ induced by the columns $r(v_1), r(v_2), r(v_3), \ldots, r(v_k)$. 
 	If instead $v_3$ is labeled with R, then we find $S_1(k)$ in $A$ contained in the submatrix induced by the columns $r(v_1), r_(v_{k-1}), \ldots, r(v_3)$. 
 	
 	Suppose that there are exactly two LR-vertices distinct than $v_2$. As a consequence of Claim \ref{claim:7_claim_2-nested}, we see that these vertices are necessarily $v_3$ and $v_4$.
 If $v_3$ and $v_4$ are LR-vertices and are both labeled with L, then $v_3$ and $v_4$ correspond to two distinct LR-rows that are not nested. Moreover, since $v_2 \subsetneq v_1$, then $v_3 \supseteq v_1$ and thus $v_1 \subseteq v_4 \subsetneq v_3$. Thus, since $v_5$ is unlabeled and there is at least one column for which the R-blocks of $v_2$, $v_3$ and $v_5$ have $1$, $0$ and $1$, respectively, we find $F_1(k)$ contained in the submatrix of $A$ induced by columns $1$ to $k-1$.
 
 If instead $v_3$ or $v_4$ (or both) are labeled with R, then $v_3$ corresponds to the same LR-row in $A$ as $v_2$. This follows from the fact that, if $v_3$ and $v_4$ correspond to the same LR-row in $A$, then $v_3$ is labeled with L and $v_4$ is labeled with R. Hence, since $v_3 \subseteq v_1$, $v_k$ cannot be adjacent to $v_1$ and thus this is not possible. However, if $v_3$ is the R-block of the LR-row corresponding to $v_2$, then we find $D_9$ in $A$ induced by the three rows corresponding to $v_1$, $v_2$ and $v_3$, and $v_4$.
 
 \vspace{1mm}
  	Suppose that there are exactly three LR-vertices other than $v_2$. Hence, these vertices are $v_3$, $v_4$ and $v_5$. Recall that $v_1$ and $v_2$ are labeled with L, and that two LR-vertices labeled with distinct letters are adjacent only if they correspond to the same LR-row in $A$. In any case, $v_5$ is labeled with R. However, since $v_1$ is labeled with L and $v_3 \supseteq v_1$, then $v_k$ results either adjacent to $v_3$, $v_4$ or $v_5$, which is a contradiction.
 
	 \subsubcase \textit{$v_2$ is not an LR-vertex.} 
	 
	By Claim \ref{claim:6_claim_2-nested_2}, if there is an LR-vertex $v$, then there are no other LR-vertices and $v = v_3$.

	In either case, since there is a exactly one LR-vertex in $C$ (we are assuming that there is at least one LR-vertex for if not the proof is as in Case 1.), then $v_2$ contains $v_j$ for every $j>3$. If $v_3$ is labeled with L, then there is $F'_2(k)$ as a subconfiguration in $A$ of the submatrix given by columns $r(v_1), \ldots, l(v_2)$. 
	If instead $v_3$ is labeled with R, then we find $S_1(k)$ as a subconfiguration in the same submatrix.

 	\vspace{1mm}
 	\subcase \textit{There is an induced $3$-cycle with exactly one uncolored vertex. }
 	
 	Let $C_3 =  v_1, v_2, v_3, v_1 $. We assume without loss of generality that $v_1$ and $v_3$ are the colored vertices. Since $A+$ is defined by considering a suitable LR-ordering and $v_1$ and $v_3$ are adjacent colored vertices, then $v_1$ and $v_3$ are labeled with distinct letters, for if not, the underlying uncolored matrix induced by these rows either induce $D_0$ or not induce any kind of gem. Moreover, $v_1$ and $v_3$ are colored with distinct colors since $A$ is admissible and thus there is no $D_1$. 
 	Furthermore, $v_2$ is unlabeled for if not it cannot be adjacent to both $v_1$ and $v_3$, since in that case $v_2$ should be nested in both $v_1$ and $v_3$. However, we find $F''_0$ as a submatrix of $A$, and this is a contradiction.

\end{mycases} 	
 	
 	\vspace{1mm}
 	This finishes the proof, since we have reached a contradiction by assuming that $A$ is partially $2$-nested but not $2$-nested.

\end{proof}


\selectlanguage{spanish}%
\chapter*{Caracterización por subgrafos prohibidos para grafos split circle}

El resultado más importante de este capítulo es el Teorema \ref{teo:circle_split_caract}, en el cual se utiliza la teoría matricial desarrollada en el capítulo anterior. 

Denotamos $\mathcal{T}$ a la familia de grafos obtenidos al considerar el tent${}\vee{}K_1$, todos los soles impares con centro y aquellos grafos cuya matriz de adyacencia $A(S,K)$ representa la misma configuración que alguna de las matrices de Tucker, con la excepción de $M_I(k)$ para todo $k\geq 3$ impar o $M_{III}(k)$ para todo $k\geq 5$ impar.
Denotamos $\mathcal{F}$ a la familia de grafos obtenida al considerar todos los grafos cuya matriz de adyacencia $A(S,K)$ representa la misma configuración que $F_0$, $F_1(k)$ o $F_2(k)$ para algún $k \geq 5$ impar.
Una representación gráfica de estos grafos puede verse en las Figuras \ref{fig:forb_T_graphs_esp} y \ref{fig:forb_F_graphs_esp}. Ninguno de estos grafos es un grafo círculo (ver Apéndice).xs

\begin{figure}[h!]
\centering
\includegraphics[scale=.35]{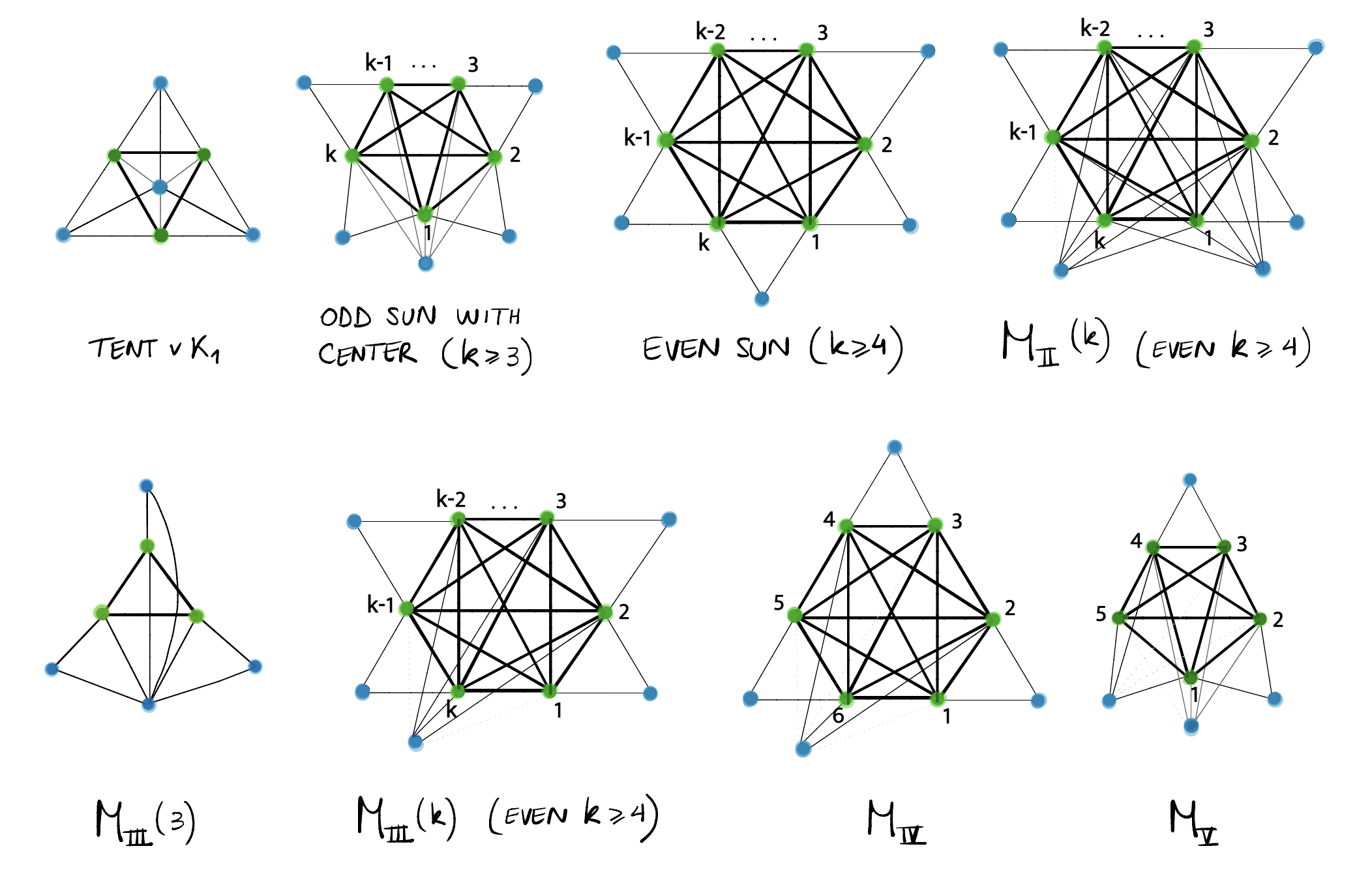} 
\caption{Los grafos de la familia $\mathcal{T}$.} \label{fig:forb_T_graphs_esp}
\end{figure}

\begin{figure}[h!]
\centering
\includegraphics[scale=.33]{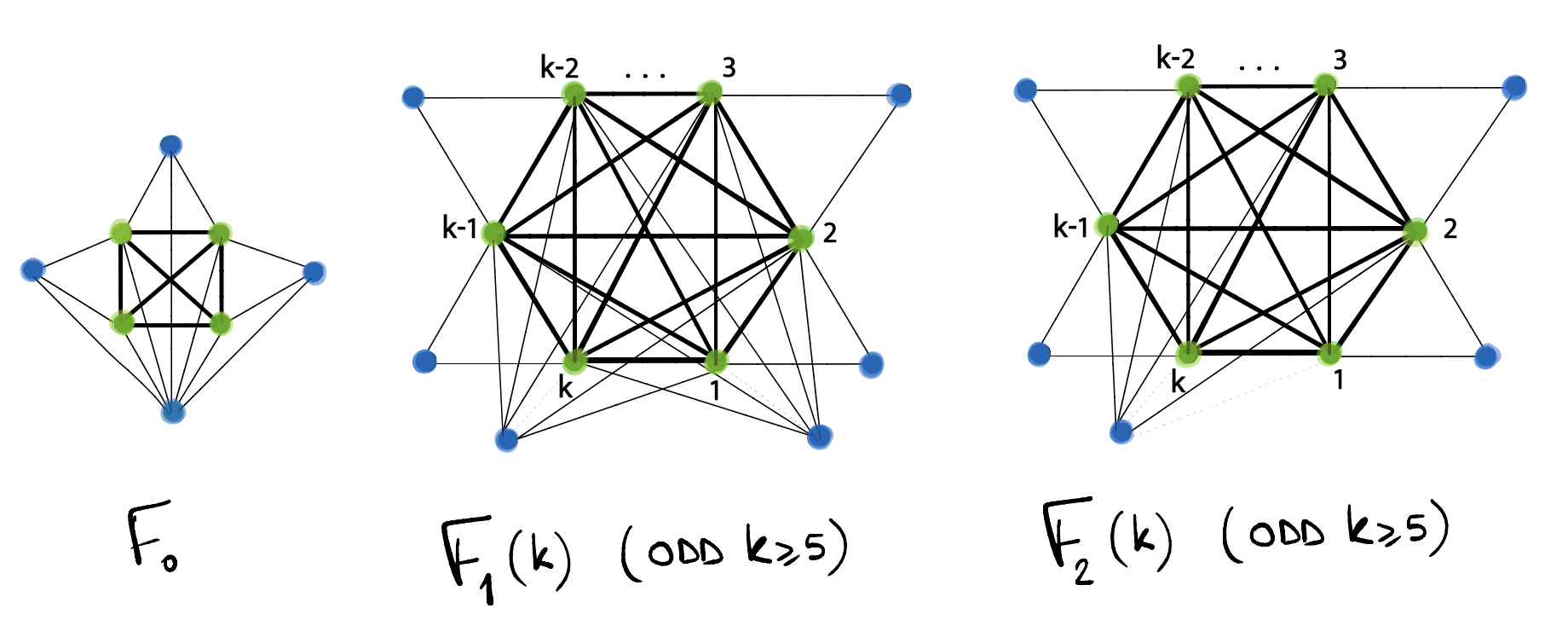} 
\caption{Los grafos de la familia  $\mathcal{F}$.} \label{fig:forb_F_graphs_esp}
\end{figure}

\begin{teo}[label={teo:circle_split_caract_esp}] 
	Sea $G=(K,S)$ un grafo split. Entonces, $G$ es un grafo circle si y sólo si $G$ no contiene ningún grafo de las familias $\mathcal{T}$ y $\mathcal{F}$ (Ver Figuras \ref{fig:forb_T_graphs_esp} y \ref{fig:forb_F_graphs_esp}).
\end{teo}

En las Secciones $4.1$ a $4.4$ se estudia el problema de caracterizar aquellos grafos split que además son circle. En cada una de estas secciones, consideramos un grafo split $G$ que contiene un cierto subgrafo $H$, donde $H$ puede ser un tent, un $4$-tent o un co-$4$-tent. Cada sección es un caso de la prueba del Teorema \ref{teo:circle_split_caract_esp}. 
Usando las particiones de $K$ y $S$ descriptas en el Capítulo $2$, definimos una matriz enriquecida para cada partición $K_i$ de $K$ y otras cuatro matrices auxiliares (no enriquecidas) que nos ayudarán a con los detalles a la hora de dar un modelo circle para el grafo $G$. Al final de cada sección, se demuestra que $G$ es circle si y sólo si las matrices enriquecidas definidas en cada sección son $2$-nested y las cuatro matrices no enriquecidas son nested, dando las pautas necesarias para encontrar un modelo circle en cada caso.

El primer caso, tratado en la Sección $4.1$, consiste en considerar un grafo split $G$ que contiene un tent como subgrafo inducido. Este es el caso más simple de todos, dada la simetría entre la mayoría de los conjuntos de la partición de $K$ y $S$, y ya que las matrices enriquecidas $\mathbb A_1, \ldots, \mathbb A_6$ definidas en la Sección $4.1.1$ no tienen filas LR en su definición.
En el segundo caso, el cual se trata en la Sección $4.2$, consideramos un grafo split $G$ que no contiene subgrafos isomorfos al tent pero sí contiene un $4$-tent como subgrafo inducido. La mayor diferencia con el caso anterior, es que la matriz enriquecida $\mathbb B_6$ definida en la Sección $4.2.1$ sí puede tener filas LR.
En la Sección $4.3$ consideramos un grafo split $G$ que no contiene subgrafos isomorfos al tent ni al $4$-tent, pero sí contiene co-$4$-tent como subgrafo inducido. En este caso, el mayor obstáculo es que, a diferencia del tent y del $4$-tent, el co-$4$-tent no es un grafo primo. Además, la matriz enriquecida $\mathbb C_7$ definida en la Sección $4.3.1$ puede tener filas LR.
Finalmente, en la Sección $4.4$ se explica en detalle cómo el caso en que $G$ contiene un net se reduce utilizando los tres casos previos. 


\selectlanguage{english}%
\chapter{Characterization by forbidden subgraphs for split circle graphs} \label{chapter:split_circle_graphs}

The main result of this chapter is Theorem \ref{teo:circle_split_caract}, which uses the matrix theory de\-vel\-oped in the previous chapter.

We denote $\mathcal{T}$ to the family of graphs obtained by considering the tent${}\vee{}K_1$, all the odd-suns with center and those graphs whose $A(S,K)$ matrix represents the same configuration as a Tucker matrix distinct to $M_I(k)$ for every odd $k\geq 3$ or $M_{III}(k)$ for every odd $k\geq 5$.
We denote $\mathcal{F}$ to the family of graphs obtained by considering those graphs whose $A(S,K)$ matrix represents the same configuration as either $F_0$, $F_1(k)$ or $F_2(k)$ for some odd $k \geq 5$.
For a representation of these graphs, see Figures \ref{fig:forb_T_graphs} and \ref{fig:forb_F_graphs}.

\begin{teo}[label={teo:circle_split_caract}]
	Let $G=(K,S)$ be a split graph. Then, $G$ is a circle graph if and only if $G$ is $\{ \mathcal{T}, \mathcal{F}\}$-free (See Figures \ref{fig:forb_T_graphs2} and \ref{fig:forb_F_graphs2}).
\end{teo}

\begin{figure}[h]
\centering
\includegraphics[scale=.35]{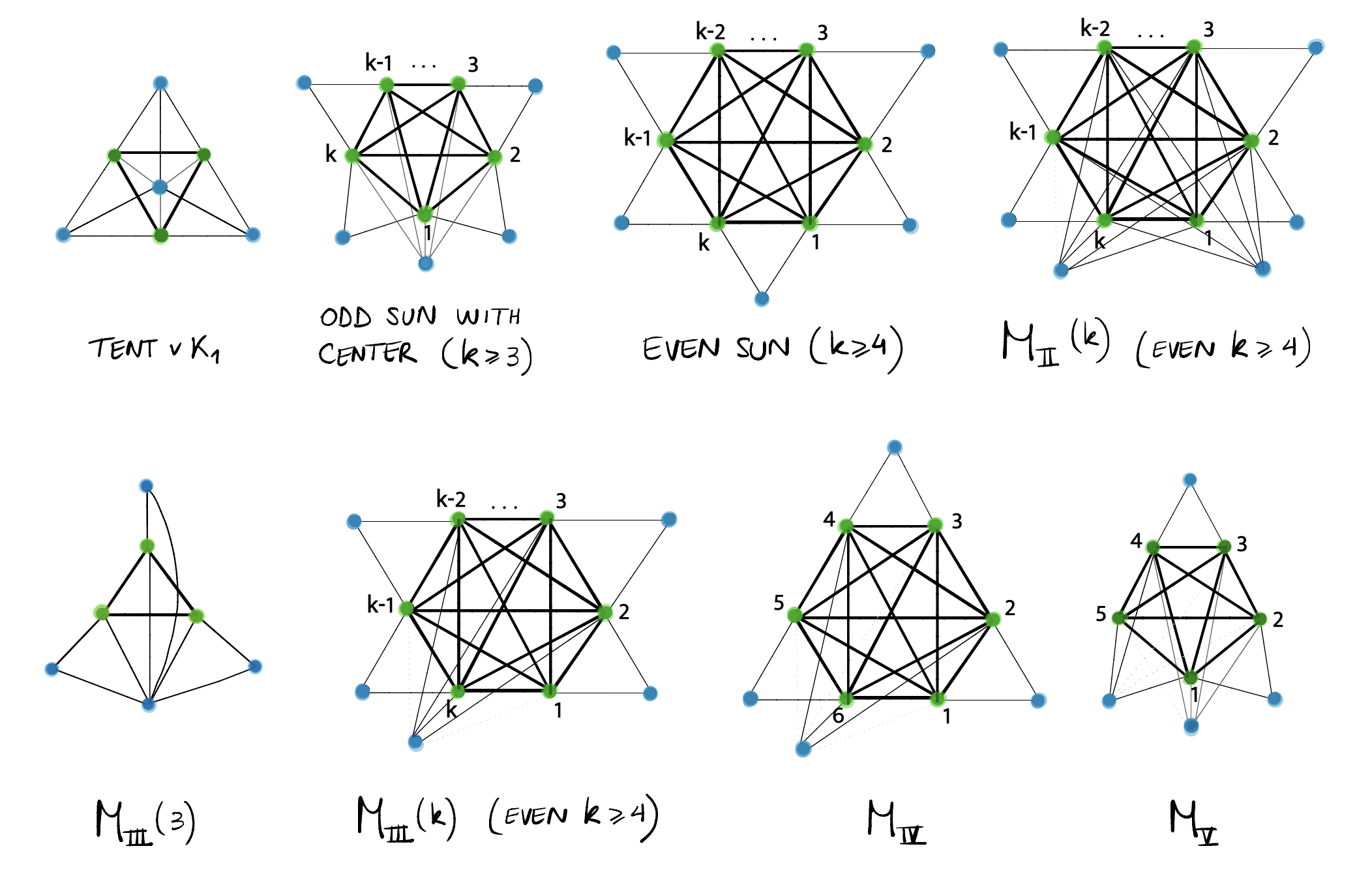} 
\caption{The graphs in the family $\mathcal{T}$.} \label{fig:forb_T_graphs2}
\end{figure}

\begin{figure}[h]
\centering
\includegraphics[scale=.35]{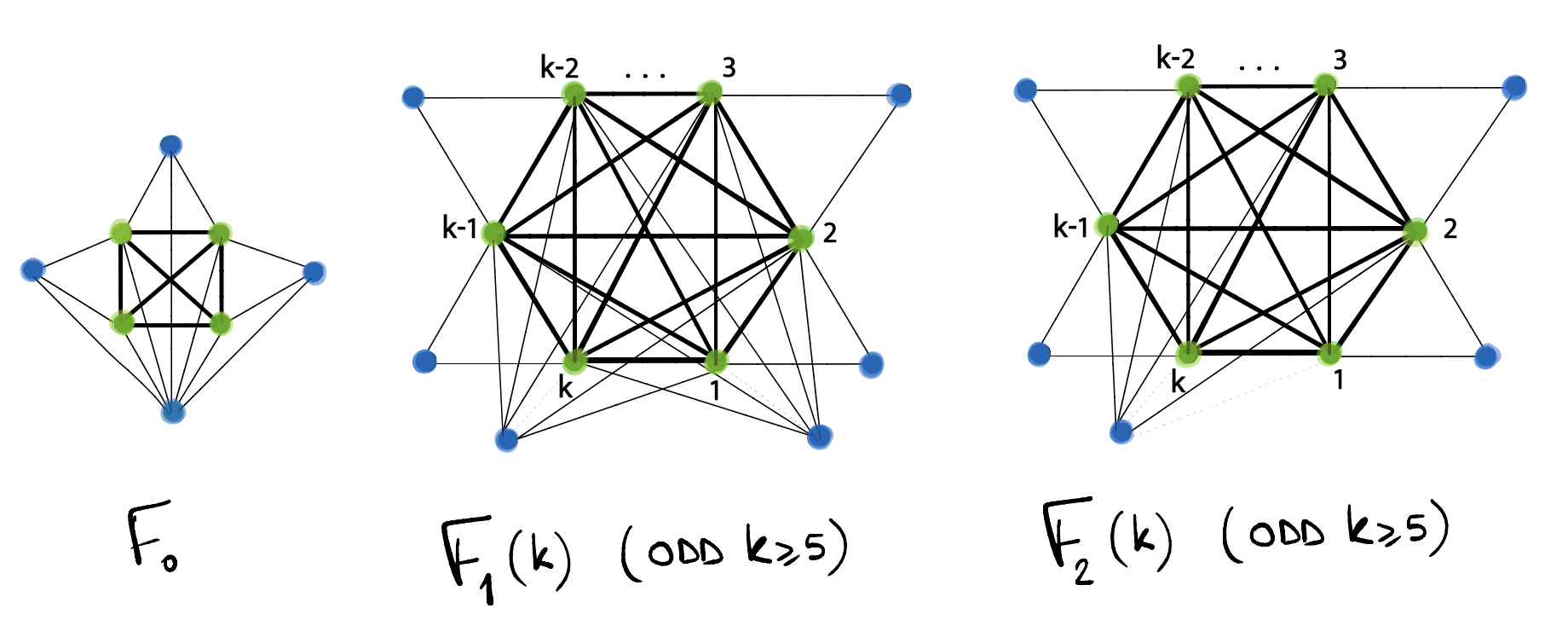} 
\caption{The graphs in the family $\mathcal{F}$.} \label{fig:forb_F_graphs2}
\end{figure}

All the graphs depicted in Figures~\ref{fig:forb_T_graphs2} and~\ref{fig:forb_F_graphs2} are non-circle graphs (see Appendix). It follows that if $G$ is a circle graph, then it contains none the graphs in Figures~\ref{fig:forb_T_graphs2} and~\ref{fig:forb_F_graphs2} as induced subgraph. 

This chapter is organized as follows. 
 In Sections \ref{sec:circle2}, \ref{sec:circle3}, \ref{sec:circle4} and \ref{sec:circle5} we address the problem of characterizing those split graphs that are also circle. In each of these sections, we consider a split graph $G$ that contains an induced subgraph $H$, where $H$ is either a tent, a $4$-tent or a co-$4$-tent, and each of these is a case of the proof of Theorem \ref{teo:circle_split_caract}. 
Using the partitions of $K$ and $S$ described in Chapter 2, we define one enriched $(0,1)$-matrix for each partition $K_i$ of $K$ and four auxiliary non-enriched $(0,1)$-matrices that will help us give a circle model for $G$. At the end of each section, we prove that $G$ is circle if and only if these enriched matrices are $2$-nested and the four non-enriched matrices are nested, giving the guidelines for a circle model in each case.

The first case, adressed in Section \ref{sec:circle2}, consists of considering a split graph $G$ that contains a tent as an induced subgraph. This is the simplest case, given the symmetry between most of the partitions of $K$ and $S$ and since the enriched matrices $\mathbb A_1, \ldots, \mathbb A_6$ that are defined in Section \ref{subsec:tent2} do not have any LR-rows.
In the second case, adressed in Section \ref{sec:circle3}, we consider a split graph $G$ that contains no  tent but contains a $4$-tent as an induced subgraph. The main difference with the previous section is that the enriched matrix $\mathbb B_6$ defined in Section \ref{subsec:4tent2} may have some LR-rows.
In Section \ref{sec:circle4} we consider a split graph $G$ that contains no tent or $4$-tent, but contains a co-$4$-tent as an induced subgraph. In this case, the main obstacles are that the co-$4$-tent is not a prime graph and that the enriched matrix $\mathbb C_7$ defined in Section \ref{subsec:co4tent1} may have some LR-rows.
Finally, in Section \ref{sec:circle5} we explain in detail how to reduce the case in which $G$ contains a net as an induced subgraph using the previous cases.

\section{Split circle graphs containing an induced tent} \label{sec:circle2}

In this section we will address the first case of the proof of Theorem \ref{teo:circle_split_caract}, which is the case where $G$ contains an induced tent.
This section is subdivided as follows. In Section \ref{subsec:tent2}, we use the partitions of $K$ and $S$ given in Section \ref{sec:tent_partition} to define the matrices $\mathbb{A}_i$ for each $i=1, 2, \ldots, 6$ and prove some properties that will be useful further on. 
In Subsection \ref{subsec:tent3}, the main results are the necessity of the $2$-nestedness of each $\mathbb A_i$ for $G$ to be $\{\mathcal{T}, \mathcal{F}\}$-free and the guidelines to give a circle model for a split graph $G$ containing an induced tent in Theorem \ref{teo:finalteo_tent}.

Notice that the net${}\vee{}K_1$, the $4$-tent${}\vee{}K_1$ and the co-$4$-tent${}\vee{} K_1$ are the graphs whose $A(S,K)$ matrix represents the same configuration as $M_{III}(3)$, $F_0$ and $M_{II}(4)$, respectively.

\subsection{Matrices $\mathbb A_1,\mathbb A_2,\ldots,\mathbb A_6$} \label{subsec:tent2}

Let $G=(K,S)$ and $H$ as in Section \ref{sec:tent_partition}. 
For each $i\in\{1,2,\ldots,6\}$, let $\mathbb A_i$ be an enriched $(0,1)$-matrix having one row for each vertex $s\in S$ such that $s$ belongs to $S_{ij}$ or $S_{ji}$ for some $j\in\{1,2,\ldots,6\}$, and one column for each vertex $k\in K_i$ and such that the entry corresponding to the row $s$ and the column $k$ is $1$ if and only if $s$ is adjacent to $k$ in $G$. For each $j\in\{1,2,\ldots,6\}-\{i\}$, we mark those rows corresponding to vertices of $S_{ji}$ with L and those corresponding to vertices of $S_{ij}$ with R.

Moreover, we color some of the rows of $\mathbb A_i$ as follows.
\begin{itemize}
 \item If $i\in\{1,3,5\}$, then we color each row corresponding to a vertex $s\in S_{ij}$ for some $j\in\{1,2,\ldots,6\}-\{i\}$ with color red and each row corresponding to a vertex $s\in S_{ji}$ for some $j\in\{1,2,\ldots,6\}-\{i\}$ with color blue.
 \item If $i\in\{2,4,6\}$, then we color each row corresponding to a vertex $s\in S_{ij}\cup S_{ji}$ for some $j\in\{1,2,\ldots,6\}$ with color red if $j=i+1$ or $j=i-1$ (modulo $6$) and with color blue otherwise.
\end{itemize}

Example:
\[ \mathbb A_3 = \bordermatrix{ & K_3\cr
		S_{34}\ \textbf{R} & \cdots \cr
                S_{35}\ \textbf{R} & \cdots \cr
                S_{33}\            & \cdots \cr
                S_{13}\ \textbf{L} & \cdots \cr
                S_{23}\ \textbf{L} & \cdots}\
                \begin{matrix}
                \textcolor{red}{\bullet} \\ \textcolor{red}{\bullet} \\ \\ \textcolor{blue}{\bullet} \\ \textcolor{blue}{\bullet}
                \end{matrix} \qquad\qquad
   \mathbb A_4 = \bordermatrix{ & K_4\cr
		S_{34}\ \textbf{L} & \cdots \cr
                S_{45}\ \textbf{R} & \cdots\cr
                S_{44}\            & \cdots \cr
                S_{14}\ \textbf{L} & \cdots \cr
                S_{64}\ \textbf{L} & \cdots \cr
                S_{41}\ \textbf{R} & \cdots \cr
                S_{42}\ \textbf{R} & \cdots }\ 
                \begin{matrix}
                \textcolor{red}{\bullet} \\ \textcolor{red}{\bullet} \\ \\ \textcolor{blue}{\bullet} \\ \textcolor{blue}{\bullet} \\ \textcolor{blue}{\bullet} \\ \textcolor{blue}{\bullet}
                \end{matrix} \]

\begin{figure}[h!]
  \begin{subfigure}[b]{0.5\textwidth}
    \includegraphics[width=\textwidth]{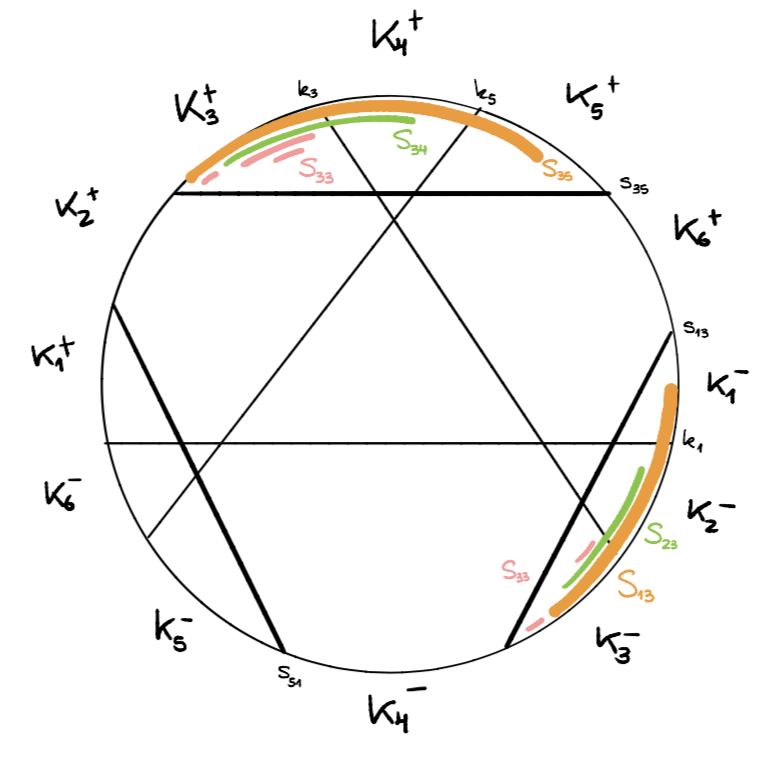}
    \caption{$\mathbb{A}_3$}
    \label{fig:modelA3}
  \end{subfigure}
  \hfill
  \begin{subfigure}[b]{0.48\textwidth}
    \includegraphics[width=\textwidth]{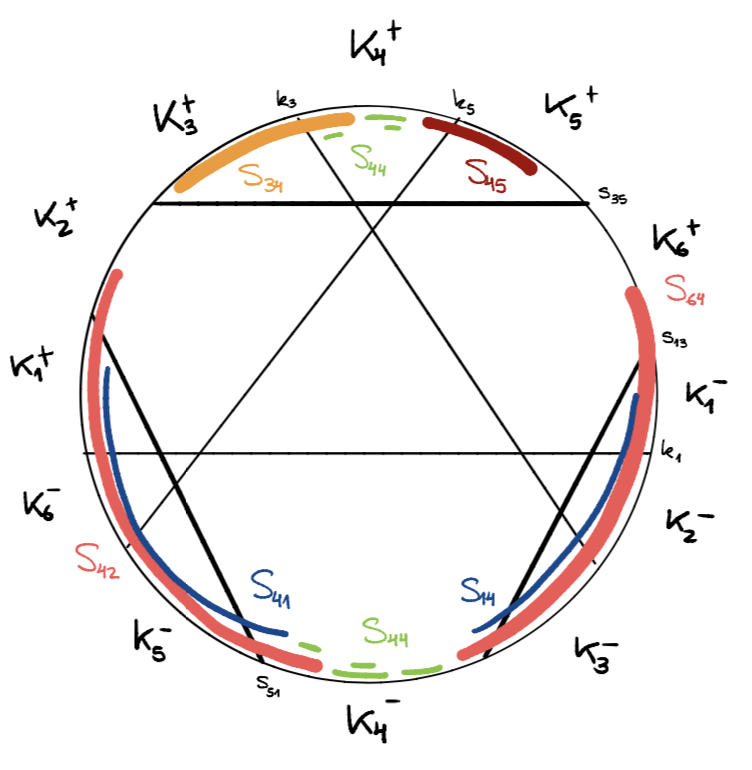}
    \caption{$\mathbb{A}_4$}    
    \label{fig:f2}
  \end{subfigure}
  \caption{Sketch model of $G$ with some of the chords associated to rows in $\mathbb{A}_3$ and $\mathbb{A}_4$.}
\end{figure}

The following results are useful in the sequel.

\begin{claim} \label{claim:tent_0}
Let $v_1$ in $S_{ij}$ and $v_2$ in $S_{ik}$, for $i,j,k \in \{1, 2, \ldots, 6 \}$ such that $i \neq j,k$. If $\mathbb A_i$ is admissible for each $i \in \{1, 2, \ldots, 6\}$, then the following assertions hold:
	\begin{itemize}
		\item If $j \neq k$, then $v_1$ and $v_2$ are nested in $K_i$. Moreover, if $j=k$, then $v_1$ and $v_2$ are nested in both $K_i$ and $K_j$.
		\item For each $i \in \{1, 2, \ldots, 6 \}$, there is a vertex $v^*_i$ in $K_i$ such that for every $j \in \{1, 2, \ldots, 6 \} - \{i\}$ and every $s$ in $S_{ij}$, the vertex $s$ is adjacent to $v^*_i$.
	\end{itemize}
\end{claim}

Let $v_1$, $v_2$ in $S_{ij}$, for some $i, j \in \{1, \ldots, 6\}$. Toward a contradiction, suppose without loss of generality that $v_1$ and $v_2$ are not nested in $K_i$, since by symmetry the proof is analogous in $K_j$. Since $v_1$ and $v_2$ are both adjacent to at least one vertex in $K_i$, then there are vertices $w_1$, $w_2$ in $K_i$ such that $w_1$ is adjacent to $v_1$ and nonadjacent to $v_2$, and $w_2$ is adjacent to $v_2$ and nonadjacent to $v_1$. 
Moreover, since $v_1$ and $v_2$ lie in $S_{ij}$ and $i \neq j$, then by definition of $\mathbb A_i$ the corresponding rows are labeled with the same letter and colored with the same color. 
Therefore, we find $D_0$ induced by the rows corresponding to $v_1$ and $v_2$, and the columns $w_1$ and $w_2$, which results in a contradiction for $\mathbb A_i$ is admissible. 
The proof is analogous if $j\neq k$. Moreover, the second statement of the claim follows from the previous argument and the fact that there is a C1P for the columns of $\mathbb A_i$. \QED

                
\subsection{Split circle equivalence} \label{subsec:tent3}

In this subsection, we will use the matrix theory developed in Chapter~\ref{chapter:2nested_matrices} to characterize the forbidden induced subgraphs that arise in a split graph that contains an induced tent when this graph is not a circle graph. We will start by proving that, given a split graph $G$ that contains an induced tent, if $G$ is $\{ \mathcal{T}, \mathcal{F} \}$-free, then the matrices $\mathbb A_i$ for each $i=1, 2, \ldots, 6$ are $2$-nested. 

\begin{lema} \label{lema:equiv_circle_2nested_tent}
	If $\mathbb A_i$ is not $2$-nested, for some $i \in \{ 1, \ldots, 6 \}$, then $G$ contains an induced subgraph of the families $\mathcal{T}$ or $\mathcal{F}$.
\end{lema}

\begin{proof}

We will prove each case assuming that either $i = 3$ or $i= 4$, since the matrices $\mathbb A_i$ are analogous when $i$ is odd or even, and thus the proof depends solely on the parity of $i$.

The proof is organized as follows. First, we will assume that $\mathbb A_i$ is not admissible. In that case, $\mathbb A_i$ contains one of the forbidden subconfigurations stated in Theorem \ref{teo:caract_admissible}. Once we reach a contradiction, we will assume that $\mathbb A_i$ is admissible but not LR-orderable, thus $\mathbb{A}_i$ contains one of the forbidden subconfigurations in Theorem \ref{teo:LR-orderable_caract_bymatrices}, once again reaching a contradiction. The following steps are to assume that $\mathbb A_i$ is LR-orderable but not partially $2$-nested, and finally that $\mathbb A_i$ is partially $2$-nested but not $2$-nested. We will use the characterizations given in Corollary \ref{cor:partially_2-nested_caract} and Theorem $\ref{teo:2-nested_caract_bymatrices}$ for each case, respectively.

Recall that for each vertex $k_i$ of the tent, $k_i$ lies in $K_i$ by definition and thus $K_i \neq \emptyset$ for every $i= 1, 3, 5$.
Notice that, if $G$ is $\{ \mathcal{T}, \mathcal{F} \}$-free, then in particular, for each $i=1, \ldots, 6$, $\mathbb A_i$ contains no $M_0$, $M_{II}(4)$, $M_V$ or $S_0(k)$ for every even $k \geq 4$ since these matrices are the adjacency matrices of non-circle graphs.

\vspace{1mm}
\begin{mycases}
\case \textit{Suppose first that $\mathbb A_i$ is not admissible.} By Theorem \ref{teo:caract_admissible} and since $\mathbb{A}_i$ contains no LR-rows, then $\mathbb A_i$ contains either $D_0$, $D_1$, $D_2$ or $S_2(k)$, $S_3(k)$ for some $k \geq 3$. 

\subcase \textit{$\mathbb A_i$ contains $D_0$. }

Let $v_0$ and $v_1$ in $S$ be the vertices whose adjacency is represented by the first and second row of $D_0$, respectively, and let $k_{i1}$ and $k_{i2}$ in $K_i$ be the vertices whose adjacency is represented by the first and second column of $D_0$, respectively.

Notice that both rows of $D_0$ are labeled with the same letter, and the coloring given to each row is indistinct. We assume without loss of generality that both rows are labeled with L, due to the symmetry of the problem.

\subsubcase \textit{Suppose first that $i = 3$.} In this case, $v_1$ and $v_2$ lie in $S_{34}$ or $S_{35}$. Hence, there are vertices $k_{31}$ and $k_{32}$ in $K_3$ such that $v_j$ is adjacent to $k_{3j}$ and is nonadjacent to $k_{3(j+1)}$ (induces modulo $2$). By Claim \ref{claim:tent_0} there is a vertex $k_4$ in $K_4$ (resp.\ $k_5$ in $K_5$) adjacent to every vertex in $S_{34}$ (resp.\ $S_{35}$).
Thus, if both $v_1$ and $v_2$ lie in $S_{35}$, since $s_{51}$ is adjacent to every vertex in $K_5$ by definition, then we find a net${}\vee{} K_1$ induced by $\{ k_5$, $k_{31}$, $k_{32}$, $v_1$, $v_2$, $s_{51}, k_1 \}$. If instead both $v_1$ and $v_2$ lie in $S_{34}$, then we find a tent with center induced by $\{ k_4$, $k_{31}$, $k_{32}$, $v_1$, $v_2$, $s_{35}, s_{13} \}$.

Thus, let us suppose that $v_1$ in $S_{34}$ and $v_2$ in $S_{35}$. Let $k_4$ in $K_4$ such that $v_1$ is adjacent to $k_4$. Recall that $v_2$ is complete to $K_4$. Let $k_5$ in $K_5$ such that $v_2$ is adjacent to $k_5$, and let $k_1$ be any vertex in $K_1$. 
Since $v_2$ in $S_{34}$ and $v_1$ in $S_{35}$, then $v_1$ and $v_2$ are nonadjacent to $k_1$, and also $v_1$ is nonadjacent to $k_5$. Hence, we find a $4$-sun induced by the set $\{ s_{13}$, $s_{51}$, $ v_1$, $v_2$, $k_1$, $k_{31}$, $k_4$, $k_5 \}$.

\subsubcase 
\textit{Suppose now that $i = 4$. }Thus, the vertices $v_1$ and $v_2$ belong to either $S_{34}$, $S_{14}$ or $S_{64}$.
Suppose $v_1$ in $S_{34}$ and $v_2$ in $S_{14}$, and let $k_1$ in $K_1$ and $k_3$ in $K_3$ such that $v_1$ is adjacent to $k_3$. Since $v_2$ is complete to $K_3$, then $v_2$ is adjacent to $k_3$, and both $v_1$ and $v_2$ are nonadjacent to $k_1$. Hence, we find co-$4$-tent${}\vee{} K_1$ induced by $\{ s_{13}$, $s_{35}$, $v_1$, $v_2$, $k_3$, $k_{41}$, $k_{42}$, $k_1 \}$.
The same holds if $v_2$ lies in $S_{64}$.

If instead $v_1$ and $v_2$ lie in $S_{34}$, then we find a net${}\vee{} K_1$ induced by the set $\{ k_3$, $k_{41}$, $k_{42}$, $v_1$, $v_2$, $s_{13}, k_1 \}$. 

\vspace{1mm}
Finally, if $v_1$ and $v_2$ lie in $S_{14} \cup S_{64}$, then we find a tent${}\vee{} K_1$ induced by $\{ k_3$, $k_1$, $k_{41}$, $k_{42}$, $v_1$, $v_2$, $s_{35} \}$, where $k_1$ is a vertex in $K_1$ adjacent to $v_1$ and $v_2$.

\subcase \textit{$\mathbb A_i$ contains $D_1$. }

As in the previous case, let $v_1$ and $v_2$ in $S$ be the vertices whose adjacency is rep\-re\-sent\-ed by the first and second row of $D_1$, respectively, and let $k_i$ in $K_i$ be the vertex whose adjacency is represented by the column of $D_1$.

Notice that both rows of $D_1$ are labeled with distinct letters and are colored with the same color. We assume without loss of generality that $v_1$ is labeled with L and $v_2$ is labeled with R.
Moreover, if $i$ is odd, then it is not possible to have two such vertices corresponding to rows in $\mathbb A_i$ labeled with distinct letters and colored with the same color. 

\subsubcase \textit{Let us suppose that $i = 4$. }
In this case, either $v_1$ in $S_{34}$ and $v_2$ in $S_{45}$, or $v_1$ in $S_{14} \cup S_{64}$ and $v_2$ in $S_{41} \cup S_{42}$.

If $v_1$ in $S_{34}$ and $v_2$ in $S_{45}$, then we find a $4$-sun induced by $\{ v_1$, $v_2$, $s_{13}$, $s_{51}$, $k_1$, $k_3$, $k_4$, $k_5 \}$, where $k_3$ in $K_3$ is adjacent to $v_1$ and nonadjacent to $v_2$, $k_4$ in $K_4$ is adjacent to both $v_1$ and $v_2$, $k_5$ in $K_5$ is adjacent to $v_2$ and nonadjacent to $v_1$, and $k_1$ in $K_1$ is nonadjacent to both $v_1$ and $v_2$.

Suppose that $v_1$ lies in $S_{14}$ and $v_2$ lies in $S_{41}$. In this case, we find a tent${}\vee{}K_1$ induced by $\{ v_1$, $v_2$, $s_{35}$, $k_1$, $k_3$, $k_4$, $k_5 \}$, where $k_1$, $k_3$, $k_4$ and $k_5$ are vertices analogous as those described in the previous paragraph. The same holds if $v_1$ in $S_{64}$ or $v_2$ in $S_{42}$.

\subcase \textit{$D_2$ in $\mathbb A_i$. }

Let $v_1$ and $v_2$ in $S$ be the vertices whose adjacency is represented by the first and second row of $D_2$, respectively, and let $k_{i1}$ and $k_{i2}$ in $K_i$ be the vertices whose adjacency is represented by the first and second column of $D_2$, respectively.

Both rows of $D_2$ are labeled with distinct letters and colored with distinct colors, for the `'same color'' case is covered since we proved that there is no $D_1$ as a submatrix of $\mathbb A_i$. We assume without loss of generality that $v_1$ is labeled with L and $v_2$ is labeled with R.

\subsubcase \textit{Suppose that $i= 4$.} Thus, $v_1$ in $S_{34}$ and $v_2$ in $S_{41} \cup S_{42}$.
In this case we find a tent with center induced by $\{ v_1$, $v_2$, $s_{13}$, $k_1$, $k_3$, $k_{41}$, $k_{42} \}$, where $k_1$ in $K_1$ is adjacent to $v_2$ and nonadjacent to $v_1$ and $k_3$ in $K_3$ is adjacent to $v_1$ and nonadjacent to $v_2$. We find the same forbidden subgraph if $v_2$ in $S_{41}$ or $S_{42}$.

\subsubcase \textit{Suppose that $i=3$.} In this case, $v_1$ in $S_{13} \cup S_{23}$, and $v_2$ in $S_{34} \cup S_{35}$.

Suppose first that $K_2 \neq \emptyset$. If $v_1$ in $S_{23}$ and $v_2$ in $S_{34}$, then we find co-$4$-tent${}\vee{}K_1$ induced by $\{ v_1$, $v_2$, $s_{13}$, $s_{35}$, $k_2$, $k_4$, $k_{31}$, $k_{32} \}$, where $k_2$ in $K_2$ is adjacent to $v_1$ and nonadjacent to $v_2$, and $k_4$ in $K_4$ is adjacent to $v_2$ and nonadjacent to $v_1$. If instead $v_2$ in $S_{35}$, then we find once more a co-$4$-tent${}\vee{}K_1$ induced by the same set of vertices with the exception of $k_4$ and adding a vertex $k_5$ in $K_5$ adjacent to $v_2$ and nonadjacent to $v_1$.
The same forbidden subgraph can be found if $v_1$ in $S_{13}$, if $K_2 \neq \emptyset$.

If instead $K_2 = \emptyset$, then necesarily $v_1$ in $S_{13}$. If $v_2$ in $S_{35}$, then we find a tent with center induced by the subset $\{ v_1$, $v_2$, $s_{51}$, $k_1$, $k_5$, $k_{31}$, $k_{32} \}$, where $k_1$ in $K_1$ is adjacent to $v_1$ and nonadjacent to $v_2$, and $k_5$ in $K_5$ is adjacent to $v_2$ and nonadjacent to $v_1$.
If $v_2$ in $S_{34}$, then we find $M_{III}(4)$ induced by $\{ v_1$, $v_2$, $s_{51}$, $s_{13}$, $k_1$, $k_4$, $k_5$, $k_{31}$, $k_{32} \}$, where $k_1$ in $K_1$ is adjacent to $v_1$ and nonadjacent to $v_2$, $k_4$ in $K_4$ is adjacent to $v_2$ and nonadjacent to $v_1$, and $k_5$ in $K_5$ is nonadjacent to both $v_1$ and $v_2$.

\subcase \textit{There is $S_2(j)$ as a submatrix of $\mathbb A_i$, with $j \geq 3$.} Let $v_1, v_2, \ldots, v_j$ be the vertices in $S$ represented by the rows of $S_2(j)$, and let $k_{i1}, \ldots, k_{i,j-1}$ be the vertices in $K_i$ that are represented by columns $1$ to $j-1$ of $S_2(j)$. Notice that $v_1$ and $v_j$ are labeled with the same letter, and depending on whether $j$ is odd or even, then $v_1$ and $v_j$ are colored with distinct colors or with the same color, respectively. We assume without loss of generality that $v_1$ and $v_j$ are both labeled with L.

\subsubcase \textit{Suppose first that $j$ is odd. }If $i =3$, then there are no vertices $v_1$ and $v_j$ labeled with the same letter and colored with distinct colors as in $S_2(j)$. 
Hence, suppose that $i= 4$. In this case, $v_1$ in $S_{34}$ and $v_j$ in $S_{14} \cup S_{64}$.
Let $k_3$ in $K_3$ be a vertex adjacent to both $v_1$ and $v_j$, and let $k_1$ in $K_1$ adjacent to $v_j$. Thus, we find $F_1(j+2)$ induced by the subset $\{ s_{13}$, $s_{35}$, $v_1$, $\ldots$, $v_j$, $k_1$, $k_3$, $k_{i1}$, $\ldots$, $k_{i,j-1} \}$.

\subsubcase  \textit{Suppose $j$ is even. }We split this in two cases, depending on the parity of $i$. 
If $i=3$, then $v_1$ and $v_j$ lie in $S_{13} \cup S_{23}$.
Suppose that $v_1$ in $S_{13}$ and $v_j$ in $S_{23}$. Let $k_2$ in $K_2$ adjacent to $v_1$ and $v_j$. Hence, we find $F_1(j+2)$ induced by the subset $\{ v_1$, $\ldots$, $v_j$, $k_2$, $k_{i2}$, $\ldots$, $k_{i,j-1}$, $s_{35} \}$. The same holds if both $v_1$ and $v_j$ lie in $S_{23}$.
If instead $v_1$ and $v_j$ both lie in $S_{13}$, then we find $F_1(j+2)$ induced by the same subset but replacing $k_2$ for a vertex $k_1$ in $K_1$ adjacent to both $v_1$ and $v_j$.

Suppose now that $i=4$. In this case, $v_1$ and $v_j$ lie in $S_{14} \cup S_{64}$. In either case, there is a vertex $k_1$ in $K_1$ that is adjacent to both $v_1$ and $v_j$. Hence, we find $F_1(j+1)$ induced by $\{ v_1$, $\ldots$, $v_j$, $k_1$, $k_{i1}$, $\ldots$, $k_{i,j-1}$, $ s_{35} \}$.

\subcase \textit{There is $S_3(j)$ as a submatrix of $\mathbb A_i$, for some $j \geq 3$. }
Let $v_1, v_2, \ldots, v_j$ be the vertices in $S$ represented by the rows of $S_3(j)$, and let $k_{i1}, \ldots, k_{i(j-1)}$ be the vertices in $K_i$ that are represented by columns $1$ to $j-1$ of $S_3(j)$. Notice that $v_1$ and $v_j$ are labeled with distinct letters, and as in the previous case, depending on whether $j$ is odd or even, $v_1$ and $v_j$ are either colored with distinct colors or with the same color, respectively. We assume without loss of generality that $v_1$ is labeled with L and $v_j$ is labeled with R.

\subsubcase \textit{Suppose first that $j$ is odd. }If $i=3$, then $v_1$ lies in $S_{34} \cup S_{35}$, and $v_j$ lies in $S_{13} \cup S_{23}$.
If $v_1$ lies in $S_{34}$ and $v_j$ lies in $S_{23}$, then we find $F_1(j+2)$ induced by $\{ v_1$, $\ldots$, $v_j$, $k_2$, $k_4$, $k_{i1}$, $\ldots$, $k_{i(j-1)}$, $s_{35}$, $s_{13} \}$, where $k_4$ in $K_4$ is adjacent to $v_1$ and nonadjacent to $v_j$, and $k_2$ in $K_2$ adjacent to $v_j$ and nonadjacent to $v_1$. 
If $v_1$ lies in $S_{34}$ and $v_j$ lies in $S_{13}$, then we find $F_1(j+2)$ induced by $\{ v_1$, $\ldots$, $v_j$, $k_1$, $k_4$, $k_{i1}$, $\ldots$, $k_{i(j-1)}$, $s_{35}$, $s_{13} \}$, with $k_1$ in $K_1$ adjacent to $v_j$ and nonadjacent to $v_1$. If instead $v_1$ lies in $S_{35}$ and $v_j$ lies in $S_{23}$, then we find $F_1(j+2)$ induced by $\{ v_1$, $\ldots$, $v_j$, $k_2$, $k_5$, $k_{i1}$, $\ldots$, $k_{i(j-1)}$, $s_{35}$, $s_{13} \}$, with $k_5$ in $K_5$ adjacent to $v_1$ and nonadjacent to $v_j$.

Suppose that $i=4$. In this case, $v_1$ in $S_{34}$ and $v_j$ in $S_{41} \cup S_{42}$. In either case, we find a $j+1$-sun induced by $\{ v_1$, $\ldots$, $v_j$, $k_{i1}$, $\ldots$, $k_{i(j-1)}$, $k_1$, $k_3$, $s_{13} \}$, with $k_1$ in $K_1$ adjacent to $v_j$ and nonadjacent to $v_1$, and $k_3$ in $K_3$ adjacent to $v_1$ and nonadjacent to $v_j$.

\subsubcase \textit{Suppose now that $j$ is even.} If $i = 3$, then there no two rows in $\mathbb A_3$ labeled with distinct letters and colored with the same color. Hence, let $i=4$. In this case, either $v_1$ in $S_{34}$ and $v_j$ in $S_{45}$, or $v_1$ in $S_{14} \cup S_{64}$ and $v_j$ in $S_{41} \cup S_{42}$.

If $v_1$ in $S_{34}$ and $v_j$ in $S_{45}$, then we find a $(j+2)$-sun induced by $\{ v_1$, $\ldots$, $v_j$, $k_1$, $k_3$, $k_5$, $k_{i1}$, $\ldots$, $k_{i(j-1)}$, $s_{13}$, $s_{51} \}$, where $k_1$ in $K_1$ is nonadjacent to both $v_1$ and $v_j$, $k_3$ in $K_3$ is adjacent to $v_1$ and nonadjacent to $v_j$, and $k_5$ in $K_5$ is adjacent to $v_j$ and nonadjacent to $v_1$.

If instead $v_1$ in $S_{14} \cup S_{64}$ and $v_j$ in $S_{41} \cup S_{42}$, then we find a $j$-sun induced by $\{ v_1$, $\ldots$, $v_j$, $k_1$, $k_{i1}$, $\ldots$, $k_{i(j-1)} \}$, with $k_1$ in $K_1$ adjacent to both $v_1$ and $v_j$.

\vspace{1mm}
Since we have reached a contradiction for every forbidden submatrix of admissibility, then the matrix $\mathbb A_i$ is admissible.


\case \textit{$\mathbb A_i$ is admissible but not LR-orderable. }

Then it contains a Tucker matrix, or one of the following submatrices: $M_4'$, $M_4''$, $M_5'$, $M_5''$, $M'_2(k)$, $M''_2(k)$, $M_3'(k)$, $M_3''(k)$, or their corresponding dual matrices, for any $k \geq 4$.

We will assume throughout the rest of the proof that, for each pair of vertices $x$ and $y$ that lie in the same subset $S_{ij}$ of $S$, there are vertices $k_i$ in $K_i$ and $k_j$ in $K_j$ such that both $x$ and $y$ are adjacent to $k_i$ and $k_j$. This is given by Claim \ref{claim:tent_0}.

Suppose there is $M_I(j)$ as a submatrix of $\mathbb A_i$. Let $v_1, \ldots, v_j$ be the vertices of $S$ represented by rows $1$ to $j$ of $M_I(k)$, and let $k_{i1}, \ldots, k_{ij}$ be the vertices in $K$ represented by colums $1$ to $j$. Thus, if $j$ is even, then we find either a $j$-sun induced by $\{v_1$, $\ldots$, $v_j$, $k_{i1}$, $\ldots$, $k_{ij} \}$, and if $j$ is odd, then we find a $j$-sun with center induced by the subset $\{v_1$, $\ldots$, $v_j$, $k_{i1}$, $\ldots$, $k_{ij}$, $s_{i,i+2} \}$.

For any other Tucker matrix, we find the homonym forbidden subgraph induced by the subset $\{v_1$, $\ldots$, $v_j$, $k_{i1}$, $\ldots$, $k_{ij} \}$.

Suppose that $\mathbb A_i$ contains one of the following submatrices: $M_4'$, $M_4''$, $M_5'$, $M_5''$, $M'_2(k)$, $M''_2(k)$, $M_3'(k)$, $M_3''(k)$, or their corresponding dual matrices, for any $k \geq 4$. Let $M$ be such a submatrix. In this case, we have the following remark.

Notice that, for any tag column $c$ of $M$ that denoted which vertices are labeled with L, there is a vertex $k'$ in either $K_{i-1}$ or $K_{i-2}$ such that the vertices represented by a labeled row in $c$ are adjacent in $G$ to $k'$. 
	If instead the tag column $c$ denoted which vertices are labeled with R, then we find an analogous vertex $k''$ in either $K_{i+1}$ or $K_{i+2}$.

Depending on whether there is one or two tag columns in $M$, we find the homonym forbidden subgraph induced by the vertices in $S$ and $K$ represented by the rows and non-tagged columns of $M$ plus one or two vertices $k'$ and $k''$ as described in the previous remark.

\case \textit{ $\mathbb A_i$ is LR-orderable but not partially $2$-nested. }Thus, since there are no LR-rows in $\mathbb A_i$, then there is either a monochromatic gem or a monochromatic weak gem in $\mathbb A_i$.

Let $v_1$ and $v_2$ in $S$ the independent vertices represented by the rows of the monochromatic gem. Notice that both rows are labeled rows, since every unlabeled row in $\mathbb A_i$ is uncolored. It follows from this that a monochromatic gem or a monochromatic weak gem is induced only by two rows labeled with L or R, and thus both are the same case.

\subcase If $i=3$, since both vertices need to be colored with the same color, then $v_1$ in $S_{34}$ and $v_2$ in $S_{35}$. 
In that case, we find $D_0$ in $\mathbb A_i$ since both rows are labeled with the same letter, which results in a contradiction for we assumed that $\mathbb A_i$ is admissible. The same holds if both vertices belong to either $S_{34}$ or $S_{35}$.

\subcase If instead $i= 4$, then we have three possibilities. Either $v_1$ in $S_{14}$ and $v_2$ in $S_{64}$, or $v_1$ in $S_{34}$ and $v_2$ in $S_{45}$, or $v_1$ in $S_{14}$ and $v_2$ in $S_{41}$.
The first case is analogous to the $i=3$ case stated above.
For the second and third case, since both rows are labeled with distinct letters, then we find $D_1$ as a submatrix of $\mathbb A_i$. This results once more in a contradiction, for $\mathbb A_i$ is admissible.

Therefore, $\mathbb A_i$ is partially $2$-nested. 

\case \textit{ $\mathbb A_i$ is partially $2$-nested but not $2$-nested.}

Hence, for every proper $2$-coloring of the rows of $\mathbb A_i$, there is either a monochromatic gem or a monochromatic weak gem. Notice that, in such a gem, there is at least one unlabeled row for there are no LR-rows in $\mathbb A_i$ and we have just proven that $\mathbb A_i$ is partially $2$-nested. We consider the columns of the matrix $\mathbb A_i$ ordered according to an LR-ordering. 
Let us suppose without loss of generality that there is a monochromatic gem, since the case in which one of the rows is labeled with L or R and the other is unlabeled is analogous.
Let $v_j$ and $v_{j+1}$ be the rows that induce such a gem, and suppose that the gem induced by $v_j$ and $v_{j+1}$ is colored with red. 

Since there is no possible $2$-coloring for which these two rows are colored with distinct colors, then there is at least one distinct row $v_{j-1}$ colored with blue that forces $v_j$ to be colored with red. If $v_{j-1}$ is unlabeled, then $v_{j-1}$ and $v_j$ are neither disjoint or nested. If $v_{j-1}$ is labeled with L or R, then $v_j$ and $v_{j-1}$ induce a weak gem.

If $v_{j-1}$ forces the coloring only on $v_j$, let $v_{j+2}$ be a row such that $v_{j+2}$ forces $v_{j+1}$ to be colored with red. Suppose first that $v_{j+2}$ forces the coloring only to the row $v_{j+1}$. Hence, there is a submatrix as the following in $\mathbb A_i$:

\vspace{-5mm}
\[ \bordermatrix{ & \cr
		v_{j-1}\  & 1 1 0 0 0 \cr
		v_{j}\  & 0 1 1 0 0  \cr
		v_{j+1}\  & 0 0 1 1 0  \cr	
		v_{j+2}\  &  0 0 0 1 1  }\ 
	\begin{matrix}
	\textcolor{blue}{\bullet} \\ \textcolor{red}{\bullet} \\ \textcolor{red}{\bullet} \\ \textcolor{blue}{\bullet}
	\end{matrix}			
		\]

If there are no more rows forcing the coloring of $v_{j-1}$ and $v_{j+2}$, then this submatrix can be colored blue-red-blue-red. Since this is not possible, there are rows $v_l, \ldots, v_{j-2}$ and $v_{j+3}, \ldots, v_k$ such that every row forces the coloring of the next one -and only that row- including $v_{j-1}$, $v_j$, $v_{j+1}$ and $v_{j+2}$. Moreover, if this is the longest chain of vertices with this property, then $v_l$ and $v_k$ are labeled rows, for if not, we can proper color again the rows and thus extending the pre-coloring, which would be a contradiction.
Hence, we find either $S_2(k-l+1)$ or $S_3(k-l+1)$ in $\mathbb A_i$, and this also results in a contradiction, for $\mathbb A_i$ is admissible.

Suppose now that $v_{j-1}$ forces the red color on both $v_j$ and $v_{j+1}$. Thus, if $v_{j-1}$ is unlabeled, then $v_{j-1}$ is neither nested nor disjoint with both $v_j$ and $v_{j+1}$. Since $v_j$ and $v_{j+1}$ are neither disjoint nor nested, either $v_j[r_j] = v_{j+1}[r_j] = 1$ or $v_j[l_j] = v_{j+1}[l_j] = 1$. Suppose without loss of generality that $v_j[r_j] = v_{j+1}[r_j] = 1$. 
Since $v_{j-1}$ is neither disjoint or nested with $v_j$, then either $v_{j-1}[l_j] = 1$ or $v_{j-1}[r_j] = 1$, and the same holds for $v_{j-1}[l_{j+1}] = 1$ or $v_{j-1}[r_{j+1}] = 1$.

If $v_{j-1}[l_j] = 1$, then $v_{j-1}[l_{j+1}] = 1$ and $v_j[l_{j+1}] = 1$, and thus we find $F_0$ induced by $\{ v_{j-1}$, $v_j$, $v_{j+1}$, $l_{j-1}$, $l_{j+1}-1$, $l_{j+1}$, $r_j$, $r_j + 1 \}$, which results in a contradiction.

Analogously, if $v_{j-1}[r_j] = 1$, then $v_{j-1}[l_{j+1}] = 1$ and $v_{j-1}[l_j] = 1$, and thus we find $F_0$ induced by $\{ v_{j-1}$, $v_j$, $v_{j+1}$, $l_j$, $l_{j+1}$, $r_j$, $r_j +1$, $r_{j-1} \}$.

If instead $v_{j-1}$ is labeled with L or R, then the proof is analogous except that we find $F'_0$ instead of $F_0$ as a subconfiguration in $\mathbb A_i$.
\end{mycases}
Therefore, we have reached a contradiction in every case and thus $\mathbb A_i$ is $2$-nested.

\end{proof}

Let $G= (K,S)$ and $H$ as in Section \ref{sec:tent_partition}, and the matrices $\mathbb A_i$ for each $i=1, 2, \ldots, 6$ as in the previous subsection.

Suppose $\mathbb A_i$ is $2$-nested for each $i =1, 2, \ldots, 6$. Let $\chi_i$ be a coloring for every matrix $\mathbb A_i$. Hence, every row in each matrix  $\mathbb A_i$ is colored with either red or blue, and this is a proper $2$-coloring extension of the given precoloring (or equivalently, a block bi-coloring), and there is an LR-ordering $\Pi_i$ for each $i =1, 2, \ldots, 6$.

Let $\Pi$ be the ordering of the vertices of $K$ given by concatenating the LR-orderings $\Pi_1$, $\Pi_2$, $\ldots$, $\Pi_6$.
Let $A=A(S,K)$ and consider the columns of $A$ ordered according to $\Pi$.

For each vertex $s$ in $S_{ij}$, if $i \leq j$, then the R-block corresponding to $s$ in $\mathbb A_i$ and the L-block corresponding to $s$ in $\mathbb A_j$ are colored with the same color. Thus, we consider the row corresponding to $s$ in $A$ colored with that color. Notice that, if $i < l < j$, then $v$ is complete to each $K_l$. Thus, when defining $\mathbb A_l$ we did not consider such vertices since they do not interfere with the possibility of having an LR-ordering of the columns, for such a vertex would have a $1$ in each column of $\mathbb A_l$. 

If instead $i >j$, then the R-block corresponding to $s$ in $\mathbb A_i$ and the L-block corresponding to $s$ in $\mathbb A_j$ are colored with distinct colors. Moreover, notice that the row corresponding to $s$ in $A$ has both an L-block and an R-block. Thus, we consider its L-block colored with the same color assigned to $s$ in $\mathbb A_j$ and the R-block colored with the same color assigned to $s$ in $\mathbb A_j$.
Notice that the distinct coloring in $\mathbb A_i$ and $\mathbb A_j$ makes sense, since we are describing vertices whose chords must have one of its endpoints drawn in the $K_i^+$ portion of the circle and the other endpoint in the $K_j^-$ portion of the circle. 
Throughout the following, we will denote $s_i$ to the row corresponding to $s$ in $\mathbb A_i$.

Let $s \in S$. Hence, $s$ lies in $S_{ij}$ for some $i,j\in \{1,2,\ldots,6\}$. Notice that, a row representing a vertex $s$ in $S_{ii}$ is entirely colored with the same color. Moreover, this is also true for a row representing $s$ in $S_{ij}$ such that $i<j$. However, if $s$ in $S_{ij}$ and $i>j$, then $s_i$ and $s_j$ are colored with distinct colors.

\begin{defn} \label{def:matrices_A_por_colores}
We define the $(0,1)$-matrix $\mathbb A_r$ as the matrix obtained by considering only those rows representing vertices in $S \setminus \bigcup_{i=1}^6 S_{ii}$ and adding two distinct columns $c_L$ and $c_R$ such that the entry $\mathbb A_r (s, k)$ is defined as follows:
\begin{itemize}
	\item If $i<j$ and $s_i$ is colored with red, then the entry $\mathbb A_r (s, k)$ has a $1$ if $s$ is adjacent to $k$ and a $0$ otherwise, for every $k$ in $K$, and $\mathbb A_r (s, c_R)$ = $\mathbb A_r (s, c_L) = 0$. 
	\item If $i>j$ and $s_i$ is colored with red, then the entry $\mathbb A_r (s, k)$ has a $1$ if $s$ is adjacent to $k$ and a $0$ otherwise, for every $k$ in $K_i \cup \ldots K_6$, and $\mathbb A_r (s, c_R) = 1$, $\mathbb A_r (s, c_L) = 0$. Analogously, if $i>j$ and instead $s_j$ is colored with red, then the entry $\mathbb A_r (s, k)$ has a $1$ if $s$ is adjacent to $k$ and a $0$ otherwise, for every $k$ in $K_1 \cup \ldots K_j$, and $\mathbb A_r (s, c_R) = 0$, $\mathbb A_r (s, c_L) = 1$.
\end{itemize}

The matrix $\mathbb A_b$ is defined in an entirely analogous way, changing red for blue in the definition. 


\vspace{2mm}
We define the $(0,1)$-matrix $\mathbb A_{r-b}$ as the submatrix of $A$ obtained by considering only those rows corresponding to vertices $s$ in $S_{ij}$ with $i>j$ for which $s_i$ is colored with red.
The matrix $\mathbb A_{b-r}$ is defined as the submatrix of $A$ obtained by considering those rows corresponding to vertices $s$ in $S_{ij}$ with $i>j$ for which $s_i$ is colored with blue.

\end{defn}


\begin{lema} \label{lema:matrices_union_son_nested}
	Suppose that $\mathbb A_i$ is $2$-nested for every $=1, 2, \ldots, 6$. If $\mathbb A_r$, $\mathbb A_b$, $\mathbb A_{r-b}$ or $\mathbb A_{b-r}$ are not nested, then $G$ contains $F_0$ as a minimal forbidden induced subgraph for the class of circle graphs.
\end{lema}

\begin{proof}
	Suppose first that $\mathbb A_r$ is not nested. Then, there is a $0$-gem.
	Since $\mathbb A_i$ is $2$-nested for every $=1, 2, \ldots, 6$, in particular there are no monochromatic gems in each $\mathbb A_i$. Let $f_1$ and $f_2$ be two rows that induce a $0$-gem in $\mathbb A_r$ and let $v_1$ in $S_{ij}$ and $v_2$ in $S_{lm}$ be the vertices corresponding to such rows in $G$. Notice that, in each case the proof will be analogous whenever two rows overlap and the corresponding two vertices lie in the same subset.
	
	\begin{mycases}
	The rows in $\mathbb A_r$ represent vertices in the following subsets of $S$: $S_{34}$, $S_{45}$, $S_{35}$, $S_{36}$, $S_{25}$, $S_{26}$, $S_{42}$, $S_{52}$, $S_{51}$, $S_{61}$, $S_{64}$ or $S_{63}$.
	Notice that $S_{36} = S_{[36}$, $S_{25} = S_{25]}$.
	
	\case $v_1$ in $S_{34}$. Thus, $v_2$ in $S_{35}$ since $\mathbb A_4$ is admissible. We find $F_0$ induced by $\{ v_1$, $v_2$, $s_{13}$, $k_{1}$, $k_{31}$, $k_{32}$, $k_{4}$, $k_{5} \}$.
	It follows analogously if $v_1$ in $S_{45}$, for in this case the only possibility is $v_2$ in $S_{35}$ since $S_{25}$ is complete to $K_5$. 
	\case $v_1$ in $S_{35} \cup S_{36}$. Since $S_{36}$ is complete to $K_3$, $S_{25}$ is complete to $K_5$ and $\mathbb A_6$ is admissible, the only possibility is $v_1$ in $S_{36}$ and $v_2$ in $S_{25} \cup S_{26}$. We find $F_0$ induced by $\{ v_1$, $v_2$, $s_{13}$, $k_{1}$, $k_{2}$, $k_{3}$, $k_{5}$, $k_{6} \}$ if $v_2$ in $S_{25}$ or $\{ v_1$, $v_2$, $s_{13}$, $k_{1}$, $k_{2}$, $k_{3}$, $k_{61}$, $k_{62} \}$ if $v_2$ in $S_{26}$.
	
	\case $v_1$ in $S_{25} \cup S_{26}$. In this case, the only possibility is that $v_1$ in $S_{25}$ and $v_2$ in $S_{26}$, since $\mathbb A_2$ and $\mathbb A_6$ are admissible.
	We find $F_0$ induced by $\{ v_1$, $v_2$, $s_{13}$, $k_{1}$, $k_{21}$, $k_{22}$, $k_{5}$, $k_{6} \}$.
	
	\end{mycases}	
	
	Thus, $\mathbb A_r$ is nested. Let us suppose that $\mathbb A_b$ is not nested.
		\begin{mycases}
	The rows in $\mathbb A_b$ represent vertices in the following subsets of $S$: $S_{12}$, $S_{13}$, $S_{23}$, $S_{14}$, $S_{42}$, $S_{52}$, $S_{51}$, $S_{61}$, $S_{64}$ or $S_{63}$.
	Notice that $S_{14} = S_{[14}$.
	
	\case $v_1$ in $S_{13}$. Thus, $v_2$ in $S_{12} \cup S_{23}$. We find $F_0$ induced by $\{ v_1$, $v_2$, $s_{51}$, $k_{5}$, $k_{11}$, $k_{12}$, $k_{2}$, $k_{3} \}$.
	The proof is analogous by symmetry if $v_1$ in $S_{23}$. Notice that there is no $0$-gem induced by $S_{12}$ and $S_{23}$ since $\mathbb A_2$ is admissible.
	
	\case $v_1$ in $S_{23}$. Since $S_{14}$ is complete to $K_1$, the only possibility is $v_2$ in $S_{63}$. We find $F_0$ induced by $\{ v_1$, $v_2$, $s_{35}$, $k_{6}$, $k_{2}$, $k_{31}$, $k_{32}$, $k_{5} \}$.
	
	\case $v_1$ in $S_{14}$. In this case, the only possibility is that $v_1$ in $S_{63} \cup S_{64}$, since $\mathbb A_4$ is admissible.
	We find $F_0$ induced by $\{ v_1$, $v_2$, $s_{35}$, $k_{6}$, $k_{1}$, $k_{3}$, $k_{4}$, $k_{5} \}$ if $v_2$ in $S_{63}$ and induced by $\{ v_1$, $v_2$, $s_{35}$, $k_{6}$, $k_{1}$, $k_{3}$, $k_{41}$, $k_{42} \}$ if $v_2$ in $S_{64}$.
	
	\end{mycases}	
	
	
	Suppose now that $\mathbb A_{b-r}$ is not nested. The rows in $\mathbb A_{b-r}$ represent vertices in the following subsets of $S$: $S_{41}$, $S_{42}$, $S_{51}$, $S_{52}$ or $S_{61}$.
	Suppose that $v_1$ in $S_{41}$ and $v_2$ in $S_{42}$. Thus, we find $F_0$ induced by $\{ v_1$, $v_2$, $s_{13}$, $k_{41}$, $k_{42}$, $k_{1}$, $k_{2}$, $k_{3} \}$. The proof is analogous if the vertices lie in $S_{51} \cup S_{52}$.
	Suppose that $v_1$ in $S_{61}$, thus $v_2$ in $S_{51} \cup S_{41}$. We find $F_0$ induced by $\{ v_1$, $v_2$, $s_{13}$, $k_{11}$, $k_{12}$, $k_{3}$, $k_{5}$, $k_{6} \}$ and therefore $\mathbb A_{b-r}$ is nested.
	
	 Suppose that $\mathbb A_{r-b}$ is not nested. The rows in $\mathbb A_{r-b}$ represent vertices in $S_{63}$ or $S_{64}$. If $v_1$ in $S_{63}$ and $v_2$ in $S_{64}$, then we find $F_0$ induced by $\{ v_1$, $v_2$, $s_{51}$, $k_{5}$, $k_{61}$, $k_{62}$, $k_{3}$, $k_{4} \}$. It follows analogously if one of both lie in $S_{63}$ or one or both lie in $S_{64}$ changing $k_3$ and $k_4$ for some analogous $k_{31}$, $k_{32}$ in $K_3$ or $k_{41}$, $k_{42}$ in $K_4$, respectively.
	
	
	This finishes the proof and therefore the four matrices are nested.
\end{proof}

\begin{teo} \label{teo:finalteo_tent}
	Let $G=(K,S)$ be a split graph containing an induced tent. Then, $G$ is a circle graph if and only if $\mathbb A_1,\mathbb A_2,\ldots,\mathbb A_6$ are $2$-nested and $\mathbb A_r$, $\mathbb A_b$, $\mathbb A_{b-r}$ and $\mathbb A_{r-b}$ are nested.
\end{teo}

\begin{proof} 
Necessity is clear by the previous lemmas and the fact that the graphs in families $\mathcal{T}$ and $\mathcal{F}$ are all non-circle. Suppose now that each of the matrices $\mathbb A_1,\mathbb A_2,\ldots,\mathbb A_6$ is $2$-nested, and that the matrices $\mathbb A_r$, $\mathbb A_b$, $\mathbb A_{b-r}$ and $\mathbb A_{r-b}$ are nested.
Let $\Pi_i$ be an LR-ordering for the columns of $\mathbb A_i$ for each $i=1,2,\ldots,6$, and let $\Pi$ be the ordering obtained by concatenation of $\Pi_i$ for all the vertices in $K$.
Consider the circle divided into twelve pieces as in Figure \ref{fig:modelA3}. For each $i\in\{1,2,\ldots,6\}$ and for each vertex $k_i \in K_i$ we place a chord having one end in $K_i^+$ and the other end in $K_i^-$, in such a way that the ordering of the endpoints of the chords in $K_i^+$ and $K_i^-$ is $\Pi_i$. 

Let us see how to place the chords for every subset $S_{ij}$ of $S$.

Notice that, by Lemma \ref{lema:matrices_union_son_nested} for every subset $S_{ij}$ such that $i \neq j$, all the vertices in $S_{ij}$ are nested according to the ordering $\Pi$. In other words, the vertices in each $S_{ij}$ are totally ordered by inclusion.
Moreover, it is also a consequence of Lemma \ref{lema:matrices_union_son_nested} and Claim \ref{claim:tent_0}, that if $i \geq k$ and $j\leq l$, then every vertex in $S_{ij}$ is contained in every vertex of $S_{kl}$.

Furthermore, let $i \in \{1, 3, 5\}$. Notice that, since $S_{i-1,i}$ is labeled with L in $\mathbb A_i$, $S_{i,i+1}$ is labeled with R in $\mathbb A_i$, any row in each of these subsets is colored with red and $\mathbb A_i$ is admissible and LR-orderable, then there is no vertex in $K_i$ such that the corresponding column has value $1$ in two distinct vertices in $S_{i-1,i}$ and $S_{i,i+1}$, respectively. Equivalently, the vertex set $N_{K_i}(S_{i-1,i}) \cap N_{K_i}(S_{i,i+1})$ is empty. 

A similar situation occurs with the vertices in $S_{i-2,i+1}$ and $S_{i+1,i-2}$ for each $i\in\{2,4,6\}$, for the vertices in each subset are labeled with R and L respectively, and since $\mathbb A_{i-2}$ is $2$-nested, then the rows corresponding to vertices in $S_{i-2,i+1}$ end in the last column of $\mathbb A_{i-2}$ and the vertices corresponding to $S_{i+1,i-2}$ start in the first column of $\mathbb A_{i-2}$. Furthermore, this implies that the sets $N_{K_{i-2}}(S_{i-2,i+1})$ and $N_{K_{i-2}}(S_{i+1,i-2})$ are disjoint. The same holds for $N_{K_{i+1}}(S_{i-2,i+1})$ and $N_{K_{i+1}}(S_{i+1,i-2})$.

\vspace{2mm}
We will place the chords according to the ordering $\Pi$ given for every vertex in $K$. For each subset $S_{ij}$, we order its vertices with the inclusion ordering of the neighbourhoods in $K$ and the ordering $\Pi$. When placing the chords corresponding to the vertices of each subset, we do it from lowest to highest according to the previously stated ordering given for each subset.

Hence, we first place the chords of every subset $S_{i,i+1}$.
\begin{itemize}
	\item If $i =1, 2, 5$, then we place one endpoint in $K_i^-$ and the other endpoint in $K_{i+1}^-$. 
	\item If $i =3, 4$, then we place one endpoint in $K_i^+$ and the other endpoint in $K_{i+1}^+$. 
	\item If $i = 6$, then we place one endpoint in $K_6^-$ and the other endpoint in $K_1^+$.  
\end{itemize}
Afterwards, we place the chords that represent vertices in $S_{i-1, i+1}$.
\begin{itemize}
	\item If $i = 2$, then we place one endpoint in $K_1^-$ and the other endpoint in $K_3^-$. 
	\item If $i =4$, then we place one endpoint in $K_3^+$ and the other endpoint in $K_5^+$. 
	\item If $i =6$, then we place one endpoint in $K_5^-$ and the other endpoint in $K_1^+$. 
\end{itemize}

We denote $a_i^-$ and $a_i^+$ to the placement in the circle given to the chords of $K_i$ 
corresponding to the first and last column of $\mathbb A_i$, respectively.
We denote $s_{i,i+2}^+$ to the placement of the chord corresponding to the vertex $s_{i,i+2}$ of the tent $H$, which lies between $a_{i-1}^+$ and $a_i^-$, and $s_{i,i+2}^+$ to the placement of the chord of the vertex $s_{i, i+2}$ that lies between $a_{i+1}^+$ and $a_{i+2}^-$.


\vspace{1mm}
For each $i \in \{ 1, 2, \ldots, 6\}$, we give the placement of the chords corresponding to the vertices in $S_{i-1, i+2}$:
\begin{itemize}
	\item For $i = 1$, we place one endpoint in $K_6^+$, and the other endpoint between $s_{13}^-$ and the chord corresponding to $a_4^-$ in $K_4^-$. 
	\item For $i = 2$, we place one endpoint between the chord corresponding to $a_6^+$ in $K_6^+$ and $s_{13}^+$, and the other endpoint in $K_4^-$. 
	\item For $i = 3$, we place one endpoint in $K_2^+$, and the other endpoint between $s_{35}^-$ and the chord corresponding to $a_6^-$ in $K_6^+$. 
	\item For $i = 4$, we place one endpoint between the chord corresponding to $a_2^+$ in $K_2^+$ and $s_{35}^+$, and the other endpoint in $K_6^+$.
	\item For $i = 5$, we place one endpoint in $K_4^-$, and the other endpoint between $s_{51}^-$ and the chord corresponding to $a_2^-$ in $K_2^+$. 
	\item For $i = 6$, we place one endpoint between the chord corresponding to $a_4^+$ in $K_4^-$ and $s_{51}^+$, and the other endpoint in $K_2^+$.  
\end{itemize}

Finally, for the vertices in $S_{i-2,i+2}$, we place the chords as follows:
\begin{itemize}
	\item For $i=2$, we place one endpoint in $K_6^+$ and the other endpoint in $K_4^-$.
	\item For $i=4$, we place one endpoint in $K_2^+$ and the other endpoint in $K_6^+$.
	\item For $i=6$, we place one endpoint in $K_4^-$ and the other endpoint in $K_2^+$.
\end{itemize}

This gives a circle model for the given split graph $G$.
\end{proof}

\section{Split circle graphs containing an induced $4$-tent} \label{sec:circle3}

In this section we will address the second case of the proof of Theorem \ref{teo:circle_split_caract}, which is the case where $G$ contains an induced $4$-tent. The difference between this case and the tent case, is that one of the matrices that we need to define contains LR-rows, which does not happen in the tent case.
This section is subdivided as follows. 
In Subsection \ref{subsec:4tent2}, we define the matrices $\mathbb{B}_i$ for each $i=1, 2, \ldots, 6$ and demonstrate some properties that will be useful further on. 
In subsections \ref{subsec:4tent3} and \ref{subsec:4tent4}, we prove the necessity of the $2$-nestedness of each $\mathbb B_i$ for $G$ to be a circle graph, and give the guidelines to draw a circle model for a split graph $G$ containing an induced $4$-tent in Theorem \ref{teo:finalteo_4tent}.

\subsection{Matrices $\mathbb B_1,\mathbb B_2,\ldots,\mathbb B_6$} \label{subsec:4tent2}

Let $G=(K,S)$ and $H$ as in Section \ref{sec:4tent_partition}.
For each $i\in\{1,2,\ldots,6\}$, let $\mathbb B_i$ be an enriched $(0,1)$-matrix having one row for each vertex $s\in S$ such that $s$ belongs to $S_{ij}$ or $S_{ji}$ for some $j\in\{1,2,\ldots,6\}$ and one column for each vertex $k\in K_i$ and such that such that the entry corresponding to row $s$ and column $k$ is $1$ if and only if $s$ is adjacent to $k$ in $G$. For each $j\in\{1,2,\ldots,6\}-\{i\}$, we mark those rows corresponding to vertices of $S_{ji}$ with L and those corresponding to vertices of $S_{ij}$ with R. Those vertices in $S_{[15]}$ and $S_{[16}$ are labeled with LR.

As in the previous section, some of the rows of $\mathbb B_i$ are colored. However, since we do not have the same symmetry as in the tent case, we will give a description of every matrix separately, for each $i \in \{1, \ldots, 6\}$ (See Figure \ref{fig:matricesB}). 

\begin{figure}
\centering
\[ \mathbb B_1 = \bordermatrix{ & K_1\cr
		S_{12}\ \textbf{R} &\cdots \cr
		         S_{11}\            & \cdots \cr
                S_{14}\ \textbf{R} & \cdots \cr
                S_{15}\ \textbf{R} & \cdots \cr
                S_{16}\ \textbf{R} & \cdots \cr
                S_{61}\ \textbf{L} & \cdots }\
                \begin{matrix}
                \textcolor{red}{\bullet} \\ \\ \textcolor{blue}{\bullet} \\ \textcolor{blue}{\bullet} \\ \textcolor{blue}{\bullet} \\ \textcolor{blue}{\bullet}
                \end{matrix} \qquad
   \mathbb B_2 = \bordermatrix{ & K_2\cr
		S_{12}\ \textbf{L} & \cdots \cr
                S_{22}\            & \cdots \cr
                S_{23}\ \textbf{R} & \cdots \cr
                S_{24}\ \textbf{R} &  \cdots }\
                \begin{matrix}
                \textcolor{red}{\bullet} \\ \\ \textcolor{blue}{\bullet} \\ \textcolor{blue}{\bullet} 
                \end{matrix} \qquad
   \mathbb B_3 = \bordermatrix{ & K_3\cr
                S_{35}\ \textbf{R} & \cdots \cr
                S_{36}\ \textbf{R} & \cdots \cr
                S_{13}\ \textbf{L} & \cdots \cr
                S_{33}\            & \cdots \cr
                S_{34}\ \textbf{R} & \cdots \cr
                S_{23}\ \textbf{L} & \cdots }\
                \begin{matrix}
                 \textcolor{red}{\bullet} \\ \textcolor{red}{\bullet} \\ \textcolor{red}{\bullet} \\ \\ \textcolor{blue}{\bullet} \\ \textcolor{blue}{\bullet}
                \end{matrix} \]
                
\[   \mathbb B_4 = \bordermatrix{ & K_4\cr
		S_{45}\ \textbf{R} & \cdots \cr
                S_{44}\            & \cdots \cr
                S_{24}\ \textbf{L} & \cdots \cr
                S_{34}\ \textbf{L} & \cdots }\       
                \begin{matrix}
                \textcolor{red}{\bullet} \\ \\ \textcolor{blue}{\bullet} \\ \textcolor{blue}{\bullet} 
                \end{matrix} \qquad
   \mathbb B_5 = \bordermatrix{ & K_5\cr
                S_{45}\ \textbf{L} & \cdots \cr
                S_{55}\            & \cdots \cr
                S_{56}\ \textbf{R} & \cdots \cr
                S_{25}\ \textbf{L} & \cdots \cr
                S_{15}\ \textbf{L} & \cdots \cr
                S_{65}\ \textbf{L} & \cdots }\
                \begin{matrix}
                \textcolor{red}{\bullet} \\ \\ \textcolor{blue}{\bullet} \\ \textcolor{blue}{\bullet} \\ \textcolor{blue}{\bullet} \\ \textcolor{blue}{\bullet}
                \end{matrix} \qquad
    \mathbb B_6 = \bordermatrix{ & K_6\cr
		S_{61}\ \textbf{R} & \cdots \cr
				S_{64}\ \textbf{R} & \cdots \cr
				S_{65}\ \textbf{R} & \cdots \cr
               S_{36}\ \textbf{L} & \cdots \cr
                S_{46}\ \textbf{L} & \cdots \cr
                S_{66}\            & \cdots \cr
                S_{62}\ \textbf{R} & \cdots \cr
                S_{63}\ \textbf{R} & \cdots \cr
                S_{26}\ \textbf{L} & \cdots \cr
                S_{56}\ \textbf{L} & \cdots \cr
                S_{16}\ \textbf{L} & \cdots \cr
                S_{[15]}\ \textbf{LR} & \cdots \cr
                S_{[16}\ \textbf{LR} & \cdots }\
                \begin{matrix}
                \textcolor{red}{\bullet} \\ \textcolor{red}{\bullet} \\ \textcolor{red}{\bullet} \\ \textcolor{red}{\bullet} \\ \textcolor{red}{\bullet} \\ \\ \textcolor{blue}{\bullet} \\ \textcolor{blue}{\bullet} \\ \textcolor{blue}{\bullet} \\ \textcolor{blue}{\bullet} \\ \textcolor{blue}{\bullet} \\ \\ \\
                \end{matrix} \]

\caption{The matrices $\mathbb{B}_1$, $\mathbb{B}_2$, $\mathbb{B}_3$, $\mathbb{B}_4$, $\mathbb{B}_5$ and $\mathbb{B}_6$.} \label{fig:matricesB}
\end{figure}
                
Notice that, since $S_{25}$, $S_{26}$, $S_{52}$ and $S_{62}$ are complete to $K_2$, then they are not considered for the definition of the matrix $\mathbb B_2$. The same holds for $S_{13}$ with regard to $\mathbb B_1$, $S_{63}$ with regard to $\mathbb B_3$, $S_{41}$, $S_{46}$, $S_{14}$ and $S_{64}$ with regard to $\mathbb B_4$, and $S_{35}$ with regard to $\mathbb B_5$.
Also notice that we considered $S_{16}$ and $S_{[16}$ as two distinct subsets of $S$. Moreover, every vertex in $S_{[16}$ is labeled with LR and every vertex in $S_{16}$ is labeled with L. 
Furthermore, every row that represents a vertex in $S_{[15]}$ is an empty LR-row in $\mathbb B_6$. Since we need $\mathbb B_6$ to be an enriched matrix, by definition of enriched matrix every row corresponding to a vertex in $S_{[15]}$ must be colored with the same color.
We will give more details on this in Subsection \ref{subsec:4tent4}.


\begin{remark}
	Claim \ref{claim:tent_0} remains true if $G$ contains an induced $4$-tent.
	The proof is anal\-o\-gous as in the tent case.
\end{remark}

\subsection{Split circle equivalence} \label{subsec:4tent3}

In this subsection, we will prove a result analogous to Lemma \ref{lema:equiv_circle_2nested_tent}. In this case, the matrices $\mathbb B_i$ contain no LR-rows, for each $i \in \{1, \ldots, 5 \}$, hence the proof is very similar to the one given in Subsection \ref{subsec:tent3} for the tent case.

\begin{lema}  \label{lema:equiv_circle_2nested_4tent_sinLR} 
	If $\mathbb B_i$ is not $2$-nested, for some $i \in \{ 1, \ldots, 5 \}$, then $G$ contains one of the forbidden subgraphs in $\mathcal{T}$ or $\mathcal{F}$. 
\end{lema}

\begin{proof}
	Relying on the symmetry between some of the sets $K_1, \ldots, K_5$, we will only prove the statement for $i = 1, 2, 3$.
	The proof is organized analogously as in Lemma \ref{lema:equiv_circle_2nested_tent}. 
	As in Lemma \ref{lema:equiv_circle_2nested_tent}, notice that, if $G$ is $\{ \mathcal{T}, \mathcal{F} \}$-free, then in particular, for each $i=1, \ldots, 6$, $\mathbb B_i$ contains no $M_0$, $M_{II}(4)$, $M_V$ or $S_0(k)$ for every even $k \geq 4$ since these matrices are the adjacency matrices of non-circle graphs.

\begin{mycases}
	
	\case \textit{$\mathbb B_i$ is not admissible} 
	
	It follows from Theorem \ref{teo:caract_admissible} and the fact that $\mathbb B_i$ contains no LR-rows that $\mathbb B_i$ contains some submatrix $D_0$, $D_1$, $D_2$, $S_2(k)$ or $S_3(k)$ for some $k \geq 3$. 

Let $v_1$ and $v_2$ in $S$ be the vertices whose adjacency is represented by the first and second row of $D_j$, for each $j=0, 1, 2$, and let $k_{i1}$ and $k_{i2}$ in $K_i$ be the vertices whose adjacency is represented by the first and second column of $D_j$ respectively, for each $j= 0, 2$, and $k_i$ in $K_i$ is the vertex whose adjacency is represented by the column of $D_1$.

\subcase \textit{$\mathbb B_i$ contains $D_0$}
 
We assume without loss of generality that both rows are labeled with L.

\subsubcase \textit{Suppose that $i=1$.} 
Since the coloring is indistinct, the vertices $v_1$ and $v_2$ may belong to one or two of the following subclasses: $S_{61}$, $S_{12}$, $S_{14}$, $S_{15}$, $S_{16}$.
Suppose first that $v_1$ and $v_2$ both lie in $S_{61}$, thus $K_6 \neq \emptyset$. If $N_{K_6}(v_1)$ and $N_{K_6}(v_2)$ are non-disjoint, then we find a tent induced by $\{v_1$, $v_2$, $s_{12}$, $k_{11}$, $k_{12}$, $k_6\}$, where $k_6$ in $K_6$ is adjacent to both $v_1$ and $v_2$. If instead there is no such vertex $k_6$, then there are vertices $k_{61}$ and $k_{62}$ in $K_6$ such that $k_{61}$ is adjacent to $v_1$ and nonadjacent to $v_2$, and $k_{62}$ is adjacent to $v_2$ and nonadjacent to $v_1$. Then, we find $M_{IV}$ induced by $\{ v_1$, $v_2$, $s_{12}$, $s_{24}$, $k_{11}$, $k_2$, $k_4$, $k_{61}$, $k_{62}$, $k_{12} \}$.

Suppose that $v_1$ and $v_2$ lie in $S_{12}$.  
If $N_{K_2}(v_1)$ and $N_{K_2}(v_2)$ are non-disjoint, then we find net${}\vee{}K_1$ induced by $\{ v_1$, $v_2$, $s_{24}$, $k_{11}$, $k_2$, $k_4$, $ k_{12}\}$. We find the same subgraph very similarly if $v_1$ and $v_2$ lie both in $S_{14}$ or in $S_{15}$ and neither $v_1$ nor $v_2$ is complete to $K_5$. 
If $v_1$ in $S_{12}$ and $v_2$ in $S_{14]} \cup S_{15} \cup S_{16}$, then we find $M_{II}(4)$ induced by $\{ k_{11}$, $k_{12}$, $k_2$, $k_4$, $v_1$, $v_2$, $s_{12}$, $s_{24} \}$.

If $v_1$ in $S_{14]}$ and $v_2$ in $S_{15} \cup S_{16}$, then we find tent with center induced by $\{ k_{11}$, $k_{12}$, $k_{2}$,$k_4$, $v_1$, $v_2$, $s_{12} \}$. Moreover, we find the same subgraph if $v_1$ and $v_2$ in $S_{15}$ and only $v_1$ is complete to $K_5$ and if $v_1$ and $v_2$ in $S_{16}$ or in $S_{15}$ and are both complete to $K_5$. 

If instead $N_{K_2}(v_1)$ and $N_{K_2}(v_2)$ are disjoint, then we find $M_{IV}$ induced by $\{ k_{11}$, $k_{22}$, $k_{12}$, $k_{21}$, $k_5$, $k_4$, $v_1$, $v_2$, $s_{45}$, $s_{24} \}$.

\subsubcase
\textit{Suppose that $i=2$.} If $v_1$ and $v_2$ lie in $S_{12}$, and $N_{K_1}(v_1)$ and $N_{K_1}(v_2)$ are disjoint, then we find $M_{IV}$ as in the previous case, induced by $\{ v_1$, $v_2$, $s_{24}$, $s_{45}$, $k_{11}$, $k_{21}$, $k_{12}$, $k_{22}$, $k_5$, $k_4 \}$. If instead $N_{K_1}(v_1)$ and $N_{K_1}(v_2)$ are non-disjoint, then we find net${}vee{}K_1$ induced by $\{ v_1$, $v_2$, $s_{24}$, $k_{21}$, $k_{22}$, $k_{1}$, $k_4 \}$. Similarly, we find the same subgraphs if $v_1$ and $v_2$ lie in $S_{23} \cup S_{24}$.

\subsubcase
\textit{Suppose that $i=3$.}
If $v_1$ and $v_2$ lie in $S_{34}$ and $N_{K_4}(v_1)$ and $N_{K_4}(v_2)$ are disjoint, then we find $M_V$ induced by $\{ v_1$, $v_2$, $s_{24}$, $s_{45}$, $k_{41}$, $k_{31}$, $k_{32}$, $k_{42}$, $k_5 \}$.
If $N_{K_4}(v_1)$ and $N_{K_4}(v_2)$ are non-disjoint, then we find a net $\vee K_1$ induced by $\{ v_1$, $v_2$, $s_{45}$, $k_{31}$, $k_{32}$, $k_4$, $k_5 \}$. 
Similarly, we find the same subgraphs if $v_1$ and $v_2$ in $S_{23}$.

Suppose that $v_1$ and $v_2$ lie in $S_{36} \cup S_{35}$. If $N_{K_5}(v_1)$ and $N_{K_5}(v_2)$ are not disjoint, then we find a tent induced by $\{v_1$, $v_2$, $s_{24}$, $k_5$, $k_{31}$, $k_{32}\}$. If $N_{K_5}(v_1)$ and $N_{K_5}(v_2)$ are disjoint, then we find a $4$-sun induced by $\{ v_1$, $v_2$, $s_{45}$, $s_{24}$, $k_{31}$, $k_{32}$, $k_{51}$, $k_{52} \}$.

If $v_1$ in $S_{34}$ and $v_2$ in $S_{35} \cup S_{36}$, then we find $M_{II}(4)$ induced by $\{v_1$, $v_2$, $s_{24}$, $s_{45}$, $k_{31}$,$k_{32}$,$k_4$, $k_5 \}$. The proof is analogous if $v_1$ and $v_2$ in $S_{23} \cup S_{13}$ or $S_{34} \cup S_{35}$.

\subcase \textit{$\mathbb B_i$ contains $D_1$} 

\subsubcase
\textit{Suppose that $i=1$.} In this case, $v_1$ lies in $S_{61}$ and $v_2$ lies in $S_{14}\cup S_{15} \cup S_{16}$. 
If $v_1$ in $S_{61}$ and $v_2$ in $S_{14} \cup S_{15}$ is not complete to $K_5$, then we find $F_2(5)$ induced by $\{ v_1$, $s_{12}$, $s_{24}$, $s_{45}$, $v_2$, $k_6$, $k_1$, $k_2$, $k_4$, $k_5\}$.
If $v_2$ lies in $S_{15}$ but is complete to $K_5$, then by definition of $\mathbb B_1$, $v_2$ is not complete to $K_1$. Let $k_{11}$ in $K_1$ be a vertex nonadjacent to $v_2$ and let $k_{12}$ in $K_1$ be the vertex represented by the column of $D_1$. Thus, $v_1$ and $v_2$ are adjcent to $k_{12}$. If $v_1$ is also adjacent to $k_{11}$, then we find $F_0$ induced by $\{ v_1$, $v_2$, $s_{12}$, $k_6$, $k_{11}$, $k_{12}$, $k_2$, $k_4 \}$. If instead $v_1$ is nonadjacent to $k_{11}$, then we find a net${}\vee{}K_1$ induced by $\{ v_1$, $v_2$, $s_{12}$, $k_6$, $k_{11}$, $k_{12}$, $k_4 \}$.
The same forbidden subgraph arises when considering a vertex $v_2$ in $S_{16}$ such that there is a vertex $k_6$ in $K_6$ adjacent to $v_1$ and nonadjacent to $v_2$. Suppose now that $v_2$ in $S_{16}$ and $v_2$ is nested in $v_1$ with regard to $K_6$. If $v_1$ is adjacent to $k_{11}$ and $k_{12}$, then we find a tent with center induced by $\{ v_1$, $v_2$, $s_{12}$, $s_{24}$, $s_{45}$, $k_6$, $k_{11}$, $k_{12}$, $k_2$, $k_4$, $k_5\}$. If instead $v_1$ is nonadjacent to $k_{11}$, then we find $M_V$ induced by $\{ s_{24}$, $v_1$, $v_2$, $s_{12}$, $k_2$, $k_4$, $k_6$, $k_{12}$, $k_{11}\}$.

\subsubcase If $i=2$, then there are no vertices labeled with distinct letters and colored with the same color.

\subsubcase
\textit{Suppose that $i=3$}, We have two possibilities: either $v_1$ lies in $S_{35} \cup S_{36}$ and $v_2$ lies in $S_{13}$, or $v_1$ lies in $S_{23}$ and $v_2$ lies in $S_{34}$.
If $v_1$ lies in $S_{35} \cup S_{36}$ and $v_2$ lies in $S_{13}$, then we find $F_0$ induced by $\{ v_1$, $v_2$, $s_{24}$, $k_1$, $k_2$, $k_3$, $k_4$, $k_5\}$. 
If $v_1$ lies in $S_{34}$ and $v_2$ lies in $S_{23}$, then we find $F_2(5)$ induced by $\{ v_1$, $v_2$, $s_{12}$, $s_{45}$, $s_{24}$, $k_1$, $k_2$, $k_3$, $k_4$, $k_5\}$.

\subcase \textit{$\mathbb B_i$ contains $D_2$} 

\subsubcase \textit{Let $i=1$. }In this case, $v_1$ in $S_{12}$ and $v_2$ in $S_{61}$, hence we find $M_{III}(4)$ induced by $\{ s_{24}$, $v_1$, $v_2$, $s_{12}$, $k_4$, $k_2$, $k_{11}$, $k_6$, $ k_{12}\}$.

\subsubcase
\textit{Suppose that $i=2$. }In this case, $v_1$ in $S_{12}$ and $v_2$ lies in $S_{23} \cup S_{24}$. We find $M_{II}(4)$ induced by $\{ v_1$, $v_2$, $s_{24}$, $s_{12}$, $k_1$, $k_{21}$, $k_4$, $k_{22} \}$.

\subsubcase
Finally,\textit{ let $i=3$.} We have two possibilities. If $v_1$ lies in $S_{35} \cup S_{36}$ and $v_2$ lies in $S_{23}$, then we find $M_{III}(4)$ induced by $\{ v_1$, $v_2$, $s_{12}$ $s_{24}$, $k_1$, $k_2$, $k_{31}$, $k_5$, $k_{32}\}$.
If $v_1$ in $S_{13}$ and $v_2$ lies in $S_{34}$, then we find $M_{III}(4)$ induced by $\{ v_1$, $v_2$, $s_{24}$, $s_{45}$, $k_1$, $k_{31}$, $k_4$, $k_5$, $k_{32}\}$.

\subcase \textit{Suppose there is $S_2(j)$ in $\mathbb B_i$ for some $j \geq 3$.} Let $v_1, v_2, \ldots, v_j$ be the vertices corresponding to the rows in $S_2(j)$ and $k_{i1}, k_{i2}, \ldots, k_{i(j-1)} $ be the vertices corresponding to the columns in $S_3(j)$. Thus, $v_1$ and $v_j$ are labeled with the same letter.

\subsubcase
\textit{Let $i=1$}, and suppose first that $j$ is odd. Hence, $v_1$ and $v_j$ are colored with distinct colors. If $v_1$ in $S_{12}$ and $v_j$ in $S_{14} \cup S_{15} \cup S_{16}$, then we find $F_1(j+2)$ induced by $\{ v_1$, $\ldots$, $v_j$, $s_{12}$, $s_{24}$, $k_4$, $k_2$, $k_{11}$, $\ldots$, $ k_{1j}\}$.
Conversely, if $v_j$ in $S_{12}$ and $v_1$ in $S_{14} \cup S_{15} \cup S_{16}$, then we find $F_2(j)$ induced by $\{ v_1$, $\ldots$, $v_j$, $k_4$, $k_{11}$, $\ldots$, $ k_{1j}\}$.

Suppose instead that $j$ is even, hence $v_1$ and $v_j$ are colored with the same color. If $v_1$ and $v_j$ lie in $S_{14} \cup S_{15} \cup S_{16}$, since there is no $D_0$, then Claim \ref{claim:tent_0} and Claim \ref{claim:tent_0} hold and thus $v_1$ and $v_j$ are nested in $K_4$. Hence, we find $F_1(j+1)$ induced by $\{ v_1$, $\ldots$, $v_j$, $s_{12}$, $k_4$, $k_{11}$, $\ldots$, $ k_{1j}\}$. We find the same forbidden subgraph if $v_1$ and $v_j$ lie both in $S_{61}$ by changing $k_6$ for $k_4$.

\subsubcase
\textit{Let $i=2$.} Since there are no vertices labeled with the same letter and colored with distinct colors, then it is not possible to find $S_2(j)$ for any odd $j$. If instead $j$ is even, then either $v_1$ and $v_j$ lie in $S_{12}$ or $v_1$ and $v_j$ lie in $S_{23}$. If $v_1$ and $v_j$ lie in $S_{12}$, then we find $F_2(j+1)$ induced by $\{ v_1$, $\ldots$, $v_j$, $s_{24}$, $k_1$, $k_{21}$, $\ldots$, $ k_{2j} \}$. We find the same forbidden subgraph if $v_1$ and $v_j$ lie in $S_{23}$ or $S_{24}$ by changing $k_1$ for $k_4$ and $S_{24}$ for $s_{12}$.

\subsubcase
\textit{Suppose that $i=3$}, and suppose first that $j \geq 3$ is odd. 
If $v_1$ in $S_{35} \cup S_{36}$ and $v_j$ in $S_{34}$, then we find $F_2(j)$ induced by $\{ v_1$, $\ldots$, $v_j$, $k_5$, $k_{31}$, $\ldots$, $ k_{3j} \}$. If instead $v_1$ in $S_{34}$ and $v_j$ in $S_{35} \cup S_{36}$, then we find $F_1(j+2)$ induced by $\{ v_1$, $\ldots$, $v_j$, $s_{45}$, $s_{24}$, $k_5$, $k_4$, $k_{31}$, $\ldots$, $ k_{3j} \}$.
We find the same forbidden subgraphs if $v_1$ in $S_{13}$ and $v_j$ in $S_{23}$ by changing $k_1$ for $k_5$, and if $v_1$ in $S_{23}$ and $v_j$ in $S_{13}$ by changing $k_4$ for $k_2$ and $k_5$ for $k_1$.

Suppose that $j$ is even. If $v_1$ and $v_j$ lie in $S_{35} \cup S_{36}$, then it follows from Claim \ref{claim:tent_0} that they are nested in $K_5$, hence we find $F_1(j+1)$ induced by $\{ v_1$, $\ldots$, $v_j$, $s_{24}$, $k_5$, $k_{31}$, $\ldots$, $ k_{3j} \}$. If $v_1$ and $v_j$ lie in $S_{13}$ we find the same forbidden subgraph by changing $k_5$ for $k_1$. It follows analogously for $v_1$ and $v_j$ lying both in $S_{34}$ or $S_{23}$.

\subcase \textit{Suppose there is $S_3(j)$ in $\mathbb B_i$ for some $j \geq 3$.} Let $v_1, v_2, \ldots, v_j$ be the vertices corresponding to the rows in $S_3(j)$ and $k_{i1}, k_{i2}, \ldots, k_{i(j-1)} $ be the vertices corresponding to the columns in $S_3(j)$. Thus, $v_1$ and $v_j$ are labeled with the distinct letters.

\subsubcase
\textit{Let $i=1$}, and suppose that $j$ is odd. In this case, $v_1$ in $S_{12}$ and $v_j$ in $S_{61}$, and we find $F_2(j+2)$ induced by $\{ v_1$, $\ldots$, $v_j$, $s_{12}$, $s_{24}$, $k_4$, $k_2$, $k_{11}$, $\ldots$, $ k_{1(j-1)}$, $k_6\}$.
If instead $j$ is even, then $v_1$ in $S_{14} \cup S_{15} \cup S_{16}$ and $v_j$ in $S_{61}$, and we find $F_2(j+1)$ induced by $\{ v_1$, $\ldots$, $v_j$, $s_{12}$, $k_4$, $k_{11}$, $\ldots$, $ k_{1(j-1)}$, $k_6\}$.

\subsubcase
\textit{Let $i = 2$}. If $j$ is even, then there are no vertices labeled with the same letter and colored with distinct colors in $S_3(j)$.

If instead $j$ is odd, then $v_1$ in $S_{12}$ and $v_j$ in $S_{23} \cup S_{24}$. In this case, we find $F_1(j+2)$ induced by $\{ v_1$, $\ldots$, $v_j$, $s_{12}$, $s_{24}$, $k_1$, $k_{21}$, $\ldots$, $ k_{2(j-1)}$, $k_4\}$.

\subsubcase
\textit{Suppose that $i=3$}. Let $j$ be odd. If $v_1$ lies in $S_{35} \cup S_{36}$ and $v_j$ in $S_{23}$, then we find $F_2(j+2)$ induced by $\{ v_1$, $\ldots$, $v_j$, $s_{12}$, $s_{24}$, $k_5$, $k_{31}$, $\ldots$, $ k_{3(j-1)}$, $k_2$, $k_1\}$. 
If instead $v_1$ in $S_{13}$ and $v_j$ in $S_{34}$, then we find $F_2(j+2)$ induced by $\{ v_1$, $\ldots$, $v_j$, $s_{45}$, $s_{24}$, $k_1$, $k_{31}$, $\ldots$, $ k_{3(j-1)}$, $k_4$, $k_5\}$.

If instead $j$ is even, then $v_1$ in $S_{35} \cup S_{36}$ and $v_j$ in $S_{13}$. In this case we find $F_2(j+1)$ induced by $\{ v_1$, $\ldots$, $v_j$, $s_{24}$, $k_5$, $k_{31}$, $\ldots$, $ k_{3(j-1)}$, $k_1\}$.

\vspace{1mm}
Notice that $\mathbb B_i$ has no LR-rows, thus there are no $S_1(j)$, $S_4(j)$, $S_5(j)$, $S_6(j)$, $S_7(j)$, $S_8(j)$, $P_0(k,l)$, $P_1(k,l)$ or $P_2(k,l)$ as subconfigurations. Hence, $\mathbb B_i$ is admissible for each $i=1,2,3$, and thus it follows for $i=4,5$ for symmetry.

Furthermore, it follows by the same argument as in the tent case that it is not possible that $\mathbb B_i$ is admissible but not LR-orderable.

\case \textit{Suppose that $\mathbb B_i$ is LR-orderable and is not partially $2$-nested.} 

Since there are no LR-rows in $\mathbb B_i$ for each $i=1,2,3$, if $\mathbb B_i$ is not partially $2$-nested, then there is either a monochromatic gem or a monochromatic weak gem in $\mathbb B_i$ as a subconfiguration. 
Remember that every colored row in $\mathbb B_i$ is a row labeled with L or R, hence both rows of a monochromatic gem or weak gem are labeled rows. However, this is not possible since in each case we find either $D_0$ or $D_1$, and this results in a contradiction for we showed that $\mathbb B_i$ is admissible and therefore $\mathbb B_i$ is partially $2$-nested.

\case \textit{Suppose that $\mathbb B_i$ is partially $2$-nested and is not $2$-nested}. 

If $\mathbb B_i$ is partially $2$-nested and is not $2$-nested, then, for every proper $2$-coloring of the rows of $\mathbb B_i$, there is a monochromatic gem or a monochromatic wek gem indued by at least one unlabeled row. This proof is also analogous as in the tent case (See Lemma \ref{lema:equiv_circle_2nested_tent} for details).

\end{mycases}
\end{proof}


\subsection{The matrix $\mathbb B_6$} \label{subsec:4tent4}

In this subsection we will demostrate a lemma analogous to Lemma \ref{lema:equiv_circle_2nested_4tent_sinLR} but for the matrix $\mathbb B_6$. In other words, we will use the matrix theory developed in Chapter \ref{chapter:2nested_matrices} in order to characterize the $\mathbb B_6$ matrix when the split graph $G$ that contains an induced $4$-tent is also a circle graph. Although the result is the same --we will find all the forbidden subgraphs for the class of circle graphs given when $\mathbb B_6$ is not $2$-nested--, the most important difference between this matrix and the matrices $\mathbb B_i$ for each $i=1, 2, \ldots, 5$, is that $\mathbb B_6$ contains LR-rows. 

First, we will define how to color those rows that correspond to vertices in $S_{[15]}$, since we defined $\mathbb B_6$ as an enriched matrix and these rows are the only empty LR-rows in $\mathbb B_6$.
Remember that all the empty LR-rows must be colored with the same color. Hence, if there is at least one red row labeled with L or one blue row labeled with R (resp.\ blue row labeled with L or red row labeled with R), then we color every LR-row in $S_{[15]}$ with blue (resp.\ with red).
This will give a $1$-color assignment to each empty LR-row only If $G$ is $\{ \mathcal{T}, \mathcal{F} \}$-free.

\begin{lema} \label{lema:4tent_coloreoLRvacias}
	Let $G$ be a split graph that contains an induced $4$-tent and such that $G$ contains no induced tent, and let $\mathbb B_6$ as defined in the previous section. If $S_{[15]} \neq \emptyset$ and one of the following holds:
	\begin{itemize}
	\item There is at least one red row $f_1$ and one blue row $f_2$, both labeled with L (resp.\ R)
	\item There is at least one row $f_1$ labeled with L and one row $f_2$ labeled with R, both colored with red (resp.\ blue).
	\end{itemize}
Then, we find either $F_1(5)$ or $4$-sun as an induced subgraph of $G$.
\end{lema}

\begin{proof}
	We assume that $\mathbb B_6$ contains no $D_0$, for we will prove this in Lemma \ref{lema:B6_2nested_4tent}.
	
	Let $v_1$ be a vertex corresponding to a red row labeled with L, $v_2$ be the vertex corresponding to a blue row labeled with L, and $w$ in $S_{[15]}$.
	Thus, $v_1$ in $S_{36} \cup S_{46}$ and $v_2$ in $S_{56} \cup S_{26} \cup S_{16}$. 
	In either case, we find $F_1(5)$ induced by $\{ k_2$, $k_4$, $k_5$, $k_6$, $v_1$, $v_2$, $w$, $s_{24}$, $s_{45} \}$ or $\{ k_1$, $k_2$, $k_4$, $k_6$, $v_1$, $v_2$, $w$, $s_{12}$, $s_{24} \}$, depending on whether $v_2$ in $S_{56}$ or in $S_{26} \cup S_{16}$.
	Suppose now that $v_1$ is a vertex corresponding to a red row labeled with L and $v_2$ is a vertex corresponding to a red row labeled with R. Thus, $v_1$ in $S_{36} \cup S_{46}$ and $v_2$ in $S_{61} \cup S_{64} \cup S_{65}$.
	 If $v_2$ in $S_{61}$, then there is a $4$-sun induced by $\{ k_1$, $k_2$, $k_4$, $k_6$, $v_1$, $v_2$, $s_{12}$, $s_{24} \}$. If instead $v_2$ in $S_{64} \cup S_{65}$, then we find a tent with center induced by $\{ k_6$, $k_1$, $k_4$, $k_5$, $v_1$, $v_2$, $w \}$.
	This finished the proof since the other cases are analogous by symmetry.
\end{proof}

In order to prove the following lemma, we will assume without loss of generality that $S_{[15]} = \emptyset$.

\begin{lema}\label{lema:B6_2nested_4tent}
	Let $G=(K,S)$ be a split graph containing an induced $4$-tent such that $G$ contains no induced tent, and let $B = \mathbb B_6$.
	 If $B$ is not $2$-nested, then $G$ contains an induced subgraph of the families $\mathcal{T}$ or $\mathcal{F}$.
\end{lema}

\begin{proof}
 	We will assume proven Lemma \ref{lema:equiv_circle_2nested_4tent_sinLR}. This is, we assume that the matrices $\mathbb B_1, \ldots,$ $\mathbb B_5$ are $2$-nested. In particular, it follows that any pair of vertices $v_1$ in $S_{ij}$ and $v_2$ in $S_{ik}$ such that $i\neq 6$ and $j \neq k$ are nested in $K_i$. Moreover, there is a vertex $v*_i$ in $K_i$ adjacent to both $v_1$ and $v_2$. 
 
 	Throughout the proof, we will refer indistinctly to a row $r$ (resp.\ a column $c$) and the vertex in the independent (resp.\ complete) partition of $G$ whose adjacency is represented by the row (resp.\ column). The structure of the proof is analogous as in Lemmas \ref{lema:equiv_circle_2nested_tent} and \ref{lema:equiv_circle_2nested_4tent_sinLR}. The only difference is that, in this case $B$ admits LR-rows by definition, and thus we have to consider all the forbidden subconfigurations for every characterization in each case.

\begin{mycases} 	
 	\case \textit{Suppose that $B$ is not admissible. }
 	
 	Hence, $B$ contains at least one of the matrices $D_0, D_1, \ldots, D_{13}$, $S_1(j), S_2(j), \ldots, S_8(j)$ for some $j \geq 3$ or $P_0(j,l)$, $P_1(j,l)$ for some $l \geq 0, j \geq 5$ or $P_2(j,l)$, for some $l \geq 0, j \geq 7$.
	
	\subcase \textit{$B$ contains $D_0$}. 
	Let $v_1$ and $v_2$ be the vertices represented by the first and second row of $D_0$ respectively, and $k_{61}$, $k_{62}$ in $K_6$ represented by the first and second column of $D_0$, respectively. 
	\subsubcase Suppose first that both vertices are colored with the same color. Since the case is symmetric with regard of the coloring, we may assume that both rows are colored with red, hence either $v_1$ and $v_2$ lie in $S_{61} \cup S_{64} \cup S_{65}$, or $v_1$ and $v_2$ lie in $S_{36} \cup S_{46}$.
	If $v_1$ and $v_2$ lie in $S_{61}$ and $k_1$ in $K_1$ is adjacent to both $v_1$ and $v_2$, then we find a net${}\vee{}K_1$ induced by $\{ k_{61}$, $k_{62}$, $k_1$, $k_2$, $v_1$, $v_2$, $ s_{12} \}$. We find the same forbidden subgraph if either $v_1$ and $v_2$ lie in $S_{64} \cup S_{65}$ changing $k_1$ for some $k_4$ in $K_4$ adjacent to both $v_1$ and $v_2$, $k_2$ for some $k_5$ in $K_5$ nonadjacent to both $v_1$ and $v_2$ and $s_{12}$ for $s_{45}$. We also find the same subgraph if $v_1$ and $v_2$ lie in $S_{36} \cup S_{46}$, changing $k_1$ for some $k_4$ in $K_4$ adjacent to both $v_1$ and $v_2$ and $s_{12}$ for $s_{24}$.
 	If instead $v_1$ in $S_{61}$ and $v_2$ in $S_{64} \cup S_{65}$, since by definition every vertex in $S_{65}$ is adjacent but not complete to $K_5$, then there are vertices $k_4$ in $K_4$ and $k_5$ in $K_5$ such that $v_1$ is nonadjacent to both, and $v_2$ is adjacent to $k_4$ and is nonadjacent to $k_5$. Thus, we find $F_{2}(5)$ induced by $\{ k_{62}$, $k_1$, $k_2$, $k_4$, $k_5$, $v_1$, $v_2$, $s_{45}$, $s_{12}$, $s_{24} \}$.
	
	\subsubcase Suppose now that both rows are colored with distinct colors. By symmetry, assume without loss of generality that $v_1$ is colored with red and $v_2$ is colored with blue. Hence, $v_1$ lies in $S_{62} \cup S_{63}$, and $v_2$ lies in $S_{61} \cup S_{64} \cup S_{65}$.	
	If $v_2$ in $S_{61}$, then there is a vertex $k_4$ in $K_4$ nonadjacent to $v_1$ and $v_2$. Hence, we find $M_{III}(4)$ induced by $\{ k_{61}$, $k_{62}$, $k_1$, $k_2$, $k_4$, $v_1$, $v_2$, $s_{12}$ $s_{24} \}$. If instead $v_2$ in $S_{64}$ or $S_{65}$, then we find $M_{III}(4)$ induced by $\{ k_{61}$, $k_{62}$, $k_2$, $k_4$, $k_5$, $v_1$, $v_2$, $s_{24}$, $s_{45} \}$.
	
	\subcase \textit{$B$ contains $D_1$. }
	Let $v_1$ and $v_2$ be the vertices that represent the rows of $D_1$, and let $k_6$ in $K_6$ be the vertex that represents the column of $D_1$. Suppose without loss of generality that both rows are colored with red, hence $v_1$ in  $S_{36} \cup S_{46}$ and $v_1$ in $S_{61} \cup S_{64} \cup S_{65}$. Notice that we are assuming there is no $D_1$ in $\mathbb B_4$, thus, if $v_2$ is not complete to $K_4$, then there is a vertex $k_4$ in $K_4$ adjacent to $v_1$ and nonadjacent to $v_2$. If $v_2$ in $S_{61}$, then we find a $4$-sun induced by $\{ k_6$, $k_1$, $k_2$, $k_4$, $v_1$, $v_2$, $s_{12}$, $s_{24} \}$. If $v_2$ in $S_{64}$ is not complete to $K_4$, then we find a tent induced by $\{ k_6$, $k_2$, $k_4$, $v_1$, $v_2$, $s_{24} \}$.
	 If instead $v_2$ in $S_{64} \cup S_{65}$ is complete to $K_4$, then we find a $M_{II}(4)$ induced by $\{ k_2$, $k_4$, $k_5$, $k_6$, $v_1$, $v_2$, $s_{24}$, $s_{45} \}$.
		
	\subcase \textit{$B$ contains $D_2$. } 
	Let $v_1$ and $v_2$ be the first and second row of $D_2$, and let $k_{61}$ and $k_{62}$ be the vertices corresponding to first and second column of $D_2$, respectively. By symmetry we suppose without loss of generality that $v_1$ is colored with blue and $v_2$ is colored with red. Thus, $v_1$ lies in $S_{56} \cup S_{26} \cup S_{16}$ and $v_2$ lies in $S_{61} \cup S_{64} \cup S_{65}$.
	If $v_1$ in $S_{56}$ and $v_2$ in $S_{61}$, then we find a $5$-sun with center induced by $\{ k_{61}$, $k_{62}$, $k_1$, $k_2$, $k_4$, $k_5$, $v_1$, $v_2$, $s_{12}$, $s_{24}$, $s_{45} \}$.
	If instead $v_1$ in $S_{26} \cup S_{16}$, since $v_1$ is not complete to $K_1$ and we assume that $\mathbb B_1$ is admissible, then there is a vertex $k_1$ in $K_1$ adjacent to $v_2$ and nonadjacent to $v_1$, for if not we find $D_1$ in $\mathbb B_1$. We find a tent induced by $\{ k_{61}$, $k_1$, $k_2$, $v_1$, $v_2$, $s_{12} \}$. The same holds if $v_1$ in $S_{56}$ and $v_2$ in $S_{65}$, for $\mathbb B_5$ is admissible and $v_2$ is adjacent but not complete to $K_5$.
	Moreover, if $v_1$ in $S_{56}$ and $v_2$ in $S_{64}$, then we find a tent induced by $\{ k_{61}$, $k_4$, $k_5$, $v_1$, $v_2$, $s_{45} \}$.
	Finally, if $v_1$ in $S_{26} \cup S_{16}$ and $v_2$ in $S_{64} \cup S_{65}$, then there are vertices $k_1$ in $K_1$ and $k_5$ in $K_5$ such that $k_1$ is nonadjacent to $v_1$ and adjacent to $v_2$, and $k_5$ is nonadjacent to $v_2$ and adjacent to $v_1$. Hence, we find $F_1(5)$ induced by $\{ k_1$, $k_2$, $k_4$, $k_5$, $v_1$, $v_2$, $s_{12}$, $s_{24}$, $s_{45} \}$.
	
\begin{remark} \label{obs:4tent_modelo0}
	If $G$ is $\{ \mathcal{T}, \mathcal{F} \}$-free, then $S_{26} \cup S_{16} \neq \emptyset$ if $S_{64} \cup S_{65} = \emptyset$, and viceversa.
\end{remark}	

	
	\subcase \textit{$B$ contains $D_3$.} 
	Let $v_1$ and $v_2$ be the vertices corresponding to the rows of $D_3$ labeled with L and R, respectively, $w$ be the vertex corresponding to the LR-row, and $k_{61}$, $k_{62}$ and $k_{63}$ in $K_6$ be the vertices corresponding to the columns of $D_3$. An uncolored LR-row in $B$ represents a vertex in $S_{[16}$. 
	Notice that there is no vertex $k_i$ in $K_i$ for some $i \in \{1, \ldots, 5\}$ adjacent to both $v_1$ and $v_2$, since $w$ is complete to $K_i$, thus we find $M_{III}(3)$ induced by $\{ k_{61}$, $k_{62}$, $k_{63}$, $k_i$, $v_1$, $v_2$, $w \}$.
	
	If $v_1$ is colored with red and $v_2$ is colored with blue, then $v_1$ in $S_{36} \cup S_{46}$ and $v_2$ in $S_{62} \cup S_{63}$. We find $M_{III}(4)$ induced by $\{ k_{62}$, $k_2$, $k_4$, $k_{61}$, $k_{63}$, $v_1$, $v_2$, $w$, $s_{24} \}$.
	
	Conversely, if $v_1$ is colored with blue and $v_2$ is colored with red, then $v_1$ in $S_{56} \cup S_{26} \cup S_{16}$ and $v_2$ in $S_{61} \cup S_{64} \cup S_{65}$.
	It follows by symmetry that it suffices to see what happens if $v_2$ in $S_{61}$ and $v_1$ in either $S_{56}$ or $S_{26}$.
	If $v_1$ in $S_{56}$, then we find $M_{III}(6)$ induced by $\{ k_{61}$, $k_1$, $k_2$, $k_4$, $k_5$, $k_{62}$, $k_{63}$, $v_1$, $v_2$, $s_{12}$, $s_{24}$, $s_{45}$, $w \}$. If instead $v_1$ in $S_{26}$, then we find $M_{III}(4)$ induced by $\{ k_{63}$, $k_{61}$, $k_1$, $k_2$, $k_{62}$, $v_1$, $v_2$, $s_{12}$, $w \}$.

	
	
	
	\subcase \textit{$B$ contains $D_4$.} 
	Let $v_1$ and $v_2$ be the vertices represented by the rows labeled with L, $w$ be the vertex represented by the LR-row and $k_6$ in $K_6$ corresponding to the only column of $D_4$. Suppose without loss of generality that $v_1$ is colored with red and $v_2$ is colored with blue. Thus, $v_1$ lies in $S_{61} \cup S_{64} \cup S_{65}$ and $v_2$ lies in $S_{62} \cup S_{63}$.
	In either case, we find $F_1(5)$: if $v_1$ in $S_{61}$, then it is induced by $\{ k_6$, $k_1$, $k_2$, $k_4$, $v_1$, $v_2$, $w$, $s_{12}$, $s_{24} \}$, and if $v_1$ in $S_{64}$ or $S_{65}$, then it is induced by $\{ k_6$, $k_2$, $k_4$, $k_5$, $v_1$, $v_2$, $w$, $s_{24}$, $s_{45} \}$. 
	
	\subcase \textit{$B$ contains $D_5$.} Let $v_1$ and $v_2$ be the vertices representing the rows labeled with L and R, respectively, $w$ be the vertex corresponding to the LR-row, and $k_6$ in $K_6$ corresponding to the column of $D_5$. Suppose without loss of generality that $v_1$ is colored with blue and $v_2$ is colored with red. 

	\begin{remark} \label{obs:2_demo_equiv_circle_hayLR}
		If $x_1$ in $S_{ij}$ and $x_2$ in $S_{jk}$, then we may assume that there are vertices $k_{j1}$ and $k_{j2}$ in $K_j$ such that $x_1$ is adjacent to $k_{j1}$ and is nonadjacent to $k_{j2}$ and $x_2$ is adjacent to $k_{j2}$ and is nonadjacent to $k_{j1}$, for if not $\mathbb B_i$ is not admissible, for $i \in \{1, \ldots, 5\}$.  
	\end{remark} 
	
	By the previous remark, notice that, if $v_1$ in $S_{26} \cup S_{16}$ and $v_2$ in $S_{61}$, then there is a tent induced by $\{k_6$, $k_1$, $k_2$, $v_1$, $v_2$, $s_{12} \}$, where $k_1$ is a vertex nonadjacent to $v_1$. The same holds if $v_1$ in $S_{56}$ and $v_2$ in $S_{65}$, where the tent is induced by $\{ k_6$, $k_4$, $k_5$, $v_1$, $v_2$, $s_{45} \}$, with $k_5$ in $K_5$ adjacent to $v_1$ and nonadjacent to $v_2$.
	Finally, if $v_1$ in $S_{56}$ and $v_2$ in $S_{61}$, then we find a $5$-sun with center induced by $\{ k_5$, $k_6$, $k_1$, $k_2$, $k_4$, $v_1$, $v_2$, $w$, $s_{12}$, $s_{24}$, $s_{45} \}$.
	
\begin{remark} \label{obs:4tent_modelo2}
	If $G$ is $\{ \mathcal{T}, \mathcal{F} \}$-free and contains no induced tent, then any two vertices $v_1$ in $S_{56}$ and $v_2$ in $S_{65}$ are disjoint in $K_5$. The same holds for any two vertices $v_1$ in $S_{16}$ and $v_2$ in $S_{61}$ in $K_1$. 
	\end{remark}

	\subcase \textit{$B$ contains $D_6$}. Let $v_1$ and $v_2$ be the vertices represented by the rows labeled with L and R, respectively, $w$ be the vertex corresponding to the LR-row, and $k_{61}$ and $k_{62}$ in $K_6$ corresponding to the first and second column of $D_6$, respectively. Suppose without loss of generality that $v_1$ and $v_2$ are both colored with red, thus $v_1$ lies in $S_{36} \cup S_{46}$ and $v_2$ lies in $S_{61} \cup S_{64} \cup S_{65}$.
	Since $\mathbb B_i$ is admissible for every $i \in \{1, \ldots, 5\}$, then there is no vertex in $K_i$ adjacent to both $v_1$ and $v_2$. It follows that, since $v_1$ is complete to $K_4$, then $v_2$ in $S_{61}$. However, we find $F_2(5)$ induced by $\{ k_{62}$, $k_1$, $k_2$, $k_4$, $k_{61}$, $v_1$, $v_2$, $w$, $s_{12}$, $s_{24} \}$. 
		
	The following remark is a consequence of the previous statement.
	\begin{remark} \label{obs:4tent_modelo1}
	If $G$ is $\{ \mathcal{T}, \mathcal{F} \}$-free, $v_1$ in $S_{36}\cup S_{46}$ and $v_2$ in $S_{61} \cup S_{64} \cup S_{65}$, then for every vertex $w$ in $S_{[16}$ either $N_{K_6}(v_1) \subseteq N_{K_6}(w)$ or $N_{K_6}(v_2) \subseteq N_{K_6}(w)$. The same holds for $v_1$ in $S_{56} \cup S_{26} \cup S_{16}$ and $v_2$ in $S_{62} \cup S_{63}$.
	\end{remark}

	\subcase \textit{$B$ contains $D_7$ or $D_{11}$}. 
	Thus, there is a vertex $k_i$ in some $K_i$ with $i\neq 6$ such that $k_i$ is adjacent to the three vertices corresponding to every row of $D_7$, thus we find a net ${}\vee{}K_1$.	The same holds if there is $D_{11}$.
	\subcase \textit{$B$ contains $D_8$ or $D_{12}$}.
	 In that case, there is an induced tent.
	\subcase \textit{$B$ contains $D_9$ or $D_{13}$.}
	It is straightforward that in this case we find $F_0$.
	
	\subcase \textit{$B$ contains $D_{10}$}. 
	Let $v_1$ and $v_2$ be the vertices represented by the rows labeled with L and R, respectively, $w_1$ and $w_2$ be the vertices represented by the LR-rows and $k_{61}, \ldots, k_{64}$ in $K_6$ be the vertices corresponding to the columns of $D_{10}$. 
	Suppose without loss of generality that $v_1$ is colored with red and $v_2$ is colored with blue. Hence, $v_1$ lies in $S_{36} \cup S_{46}$ and $v_2$ lies in $S_{62} \cup S_{63}$.
	Let $k_2$ in $K_2$ adjacent to $v_2$ and nonadjacent to $v_1$ and let $k_4$ in $K_4$ adjacent to $v_1$ and nonadjacent to $v_2$. Hence, we find $F_1(5)$ induced by $\{ v_1$, $v_2$, $w_1$, $w_2$, $s_{24}$, $k_2$, $k_4$, $k_{62}$, $k_{63} \}$.

	\subcase \textit{Suppose that $B$ contains $S_1(j)$}
	\subsubcase If $j \geq 4$ is even, let $v_1, v_2, \ldots, v_j$ be the vertices represented by the rows of $S_1(j)$, where $v_1$ and $v_j$ are labeled both with L or both with R, $v_{j-1}$ is a vertex corresponding to the LR-row, and $k_{61}, \ldots, k_{6(j-1)}$ in $K_6$ the vertices corresponding to the columns. Suppose without loss of generality that $v_1$ and $v_j$ are labeled with L. It follows that either $v_1$ and $v_j$ lie in $S_{36} \cup S_{46}$, or $v_1$ and $v_j$ lie in $S_{62} \cup S_{63}$ or $v_1$ lies in $S_{56} \cup S_{26} \cup S_{16}$ and $v_j$ lies in $S_{36} \cup S_{46}$.
In either case, there is $k_5$ in $K_5$ adjacent to both $v_1$ and $v_j$. Moreover, $k_5$ is also adjacent to $v_{j-1}$. Thus, this vertex set induces a $j-1$-sun with center.

\subsubcase If $j$ is odd, since $S_1(j)$ has $j-2$ rows (thus there are $v_1, \ldots, v_{j-2}$ vertices), then the subset of vertices given by $\{ v_1, \ldots, v_{j-2}$, $k_{61}, \ldots, k_{6(j-2)}$, $k_5 \}$ induces an even $j-1$-sun.


 \subcase \textit{Suppose that $B$ contains $S_2(j)$}. 
 
 Let $v_1$ and $v_j$ be the vertices corresponding to the labeled rows, $k_{61}, \ldots, k_{6 (j-1})$ in $K_6$ be the vertices corresponding to the columns of $S_2(j)$, and suppose without loss of generality that $v_1$ and $v_j$ are labeled with R. 
 
 \subsubcase Suppose first that $j$ is odd, $v_1$ is colored with red and $v_j$ is colored with blue. Thus, $v_1$ in $S_{61} \cup S_{64} \cup S_{65}$ and $v_j$ in $S_{62} \cup S_{63}$, or viceversa. If $v_1$ in $S_{61}$, then let $k_i$ in $K_i$ for $i = 1, 2, 4$ such that $k_1$ is adjacent to $v_1$ and $v_j$, $k_2$ is adjacent to $v_j$ and nonadjacent to $v_1$, and $k_4$ is nonadjacent to both $v_1$ and $v_j$. We find $F_2(j+2)$ induced by $\{ k_4$, $k_2$, $k_1$, $k_{61}, \ldots, k_{6 (j-1})$, $v_1, \ldots, v_j$, $s_{12}$, $s_{24} \}$. 
	If $v_1$ in $S_{64} \cup S_{65}$, then we find $F_2(j)$ induced by $\{ k_5$, $k_{61}, \ldots, k_{6 (j-1)}$, $v_1, \ldots, v_j \}$, with $k_5$ in $K_5$ adjacent to $v_1$ and nonadjacent to $v_j$. 

	 Conversely, suppose $v_1$ in $S_{62}\cup S_{63}$ and $v_j$ in $S_{61} \cup S_{64} \cup S_{65}$. If $v_j$ lies in $S_{64} \cup S_{65}$, then $F_2(j+2)$ is induced by $\{ k_2$, $k_4$, $k_5$, $k_{61}, \ldots, k_{6 (j-1)}$, $v_1, \ldots, v_j$, $s_{24}$, $s_{45} \}$, with $k_i$ in $K_i$ for $i=2, 4, 5$ such that $k_2$ is adjacent to $v_1$ and $v_k$, $k_4$ is adjacent to $v_j$ and nonadjacent to $v_1$, and $k_5$ is nonadjacent to both $v_1$ and $v_j$. If instead $v_j$ in $S_{61}$, then it is induced by $\{ k_4$, $k_2$, $k_1$, $k_{61}, \ldots, k_{6 (j-1)}$, $v_1, \ldots, v_j$, $s_{12}$, $s_{24} \}$. 
	\subsubcase Suppose now that $j$ is even, and thus both $v_1$ and $v_j$ are colored with the same color. Suppose without loss of generality that are both colored with red, and thus $v_1$ and $v_j$ lie in $S_{61} \cup S_{64} \cup S_{65}$.
	If $v_1$ and $v_j$ in $S_{61}$, then we find $F_2(j+1)$ induced by $\{ k_2$, $k_1$, $k_{61}, \ldots, k_{6 (j-1)}$, $v_1, \ldots, v_j$, $s_{12} \}$.
	We find the same forbidden subgraph if $v_1$ and $v_j$ lie in $S_{64}$ or $S_{65}$, by changing $s_{12}$ for $s_{45}$, and $k_1$ and $k_2$ for $k_4$ and $k_5$, where $k_5$ is nonadjacent to both $v_1$ and $v_j$ and $k_4$ is adjacent to both.
	If only $v_1$ lies in $S_{61}$, then we find $F_2(j+3)$ induced by $\{ k_1$, $k_2$, $k_4$, $k_5$, $k_{61}, \ldots, k_{6 (j-1)}$, $v_1, \ldots, v_j$, $s_{12}$,$s_{24}$, $s_{45} \}$, with $k_i$ in $K_i$ for $i=1, 2, 4, 5$. 
	If only $v_j$ lies in $S_{61}$, then we find $F_2(5)$ induced by $\{ k_1$, $k_2$, $k_4$, $k_5$, $k_{62}$, $v_1$, $v_j$, $s_{12}$, $s_{24}$, $s_{45} \}$, with $k_i$ in $K_i$ for $i=1, 2, 4, 5$.


	\subcase \textit{Suppose that $B$ contains $S_3(j)$}. 
Let $v_1$ and $v_j$ be the vertices corresponding to the labeled rows, $k_{61}, \ldots, k_{6 (j-1})$ in $K_6$ be the vertices corresponding to the columns of $S_3(j)$.

	\subsubcase Suppose first that $j$ is odd, and suppose that $v_1$ is labeled with L and colored with blue and $v_j$ is labeled with R and colored with red. In this case, $v_1$ in  $S_{56} \cup S_{26} \cup S_{16}$ and $v_j$ in $S_{61} \cup S_{64} \cup S_{65}$. 
	If $v_1$ in $S_{56}$, then we find a $(j+3)$-sun if $v_j$ in $S_{61}$, induced by $\{ k_1$, $k_2$,  $k_4$, $k_5$, $k_{61}, \ldots, k_{6 (j-1)}$, $v_1, \ldots, v_j$, $s_{45}$, $s_{12}$, $s_{24} \}$. 
	If $v_j$ in $S_{64} \cup S_{65}$, then we find a $(j+1)$-sun induced by $\{ k_4$, $k_5$, $k_{61}, \ldots, k_{6 (j-1)}$, $v_1, \ldots, v_j$, $s_{45} \}$. 
	Moreover, if $v_j$ in $S_{61}$ and $v_1$ in $S_{26} \cup S_{16}$, then we find a $(j+1)$-sun induced by $\{ k_1$, $k_{61}, \ldots, k_{6 (j-1)}$, $k_2$, $v_1, \ldots, v_j$, $s_{12} \}$. Finally, if $v_j$ in $S_{64} \cup S_{65}$ and $v_1$ in $S_{26} \cup S_{16}$, then we find $F_1(5)$ induced by $\{ k_1$, $k_2$, $k_4$, $k_5$, $v_1$, $v_j$, $s_{24}$, $s_{45}$, $s_{12} \}$.
		
	\subsubcase Suppose now that $j$ is even, and suppose without loss of generality that $v_1$ and $v_j$ are both colored with red. Thus, $v_1$ in $S_{61} \cup S_{64} \cup S_{65}$ and $v_j$ in $S_{36} \cup S_{46}$.
	If $v_1$ in $S_{61}$, then we find $(j+2)$-sun induced by $\{ k_4$, $k_2$, $k_1$, $k_{61}, \ldots, k_{6 (j-1)}$, $v_1, \ldots, v_j$, $s_{12}$, $s_{24} \}$. If instead $v_1$ in $S_{64} \cup S_{65}$, then we find $j$-sun induced by $\{ k_4$, $k_{61}, \ldots, k_{6 (j-1)}$, $v_1, \ldots, v_j \}$.
	
	\subcase \textit{$B$ contains $S_4(j)$.}
	
	Let $v_1$, $v_2$ and $v_j$ be the labeled rows and $k_{61}, \ldots, k_{6 (j-1})$ in $K_6$ be the vertices corresponding to the columns of $S_4(j)$. Suppose without loss of generality that $v_1$ is the vertex corresponding to the row labeled with LR, $v_2$ corresponding to the row labeled with L, $v_j$ labeled with R. Notice that $v_1$ lies in $S_{[16}$.
	
	\subsubcase Suppose $j$ is even, thus $v_2$ and $v_j$ are colored with the same color. Suppose without loss of generality that they are both colored with red. Hence, $v_2$ in $S_{36} \cup S_{46}$ and $v_j$ in $S_{61} \cup S_{64} \cup S_{65}$.
	If $v_j$ lies in $S_{64} \cup S_{65}$, then we find a $(j-1)$-sun with center induced by $\{ k_4$, $k_{61}, \ldots, k_{6(j-1)}$, $v_1$, $2, \ldots, v_j \}$. If instead $v_j$ in $S_{61}$, then we find a $(j+1)$-sun with center induced by $\{ k_{61}, \ldots, k_{6(j-1)}$, $k_1$, $k_2$, $k_4$, $v_1$, $2, \ldots, v_j$, $s_{12}$, $s_{24} \}$.
	
	\subsubcase Suppose $j$ is odd, thus assume without loss of generality that $v_2$ is colored with red and $v_j$ is colored with blue. Hence, $v_2$ in $S_{36} \cup S_{46}$ and $v_j$ in $S_{62} \cup S_{63}$. We find a $j$-sun with center induced by $\{ k_4$, $k_{61}, \ldots, k_{6(j-1)}$, $k_2$,  $v_1$, $2, \ldots, v_j$, $s_{24} \}$.

	\subcase \textit{$B$ contains $S_5(j)$.}	
	
	Let $v_1$, $v_{j-1}$ and $v_j$ be the labeled rows and $k_{61}, \ldots, k_{6 (j-2})$ in $K_6$ be the vertices corresponding to the columns of $S_4(j)$. Suppose without loss of generality $v_2$ and $v_j$ are labeled with L and that $v_{j-1}$ is the vertex corresponding to the row labeled with LR.
	
	\subsubcase Suppose $j$ is even, hence $v_1$ and $v_j$ lie in $S_{36} \cup S_{46}$. In this case we find $F_1(j+1)$ induced by $\{ k_2$, $k_4$, $k_{61}, \ldots, k_{6(j-2)}$, $v_1, \ldots, v_{j-1}, v_j$, $s_{24} \}$.
	
	\subsubcase Suppose $j$ is odd, and suppose that $v_1$ is colored with red and $v_j$ is colored with blue. Thus, $v_1$ in $S_{36} \cup S_{46}$ and $v_j$ in $S_{56} \cup S_{26} \cup S_{16}$.
	If $v_j$ in $S_{56}$, then we find $F_1(j)$ induced by $\{ k_4$, $k_5$, $k_{61}, \ldots, k_{6(j-2)}$, $v_1, \ldots, v_{j-1}, v_j$, $s_{45} \}$. If instead $v_j$ lies in $S_{26} \cup S_{16}$, then we find $F_1(j+2)$ induced by $\{ k_1, $, $k_2$, $k_4$, $k_{61}, \ldots, k_{6(j-2)}$, $v_1, \ldots, v_{j-1}, v_j$, $s_{24}$, $s_{12} \}$.
	
	\subcase \textit{$B$ contains $S_6(j)$.}
	
	\subsubcase Suppose first that $B$ contains $S_6(3)$ or $S'_6(3)$, and let $v_1$, $v_2$ and $v_3$ be the vertices that respresent the LR-row, the R-row and the unlabeled row, respectively. Independently on where does $v_2$ lie in, there is vertex $v$ in $K \setminus K_6$ such that $v$ is adjacent to $v_1$ and $v_2$ and nonadjacent to $v_3$, then we find an induced tent with center. 
	
	\subsubcase If $B$ contains $S_6(j)$ for some even $j$, then we find $F_1(j)$ induced by every row and column of $S_6(j)$. If instead $j$ is odd, then we find $M_{II}(j)$ induced by every row and column of $S_6(j)$ and a vertex $k_i$ in some $K_i$ with $i \neq 6$. We choose such a vertex $k_i$ adjacent to $v_2$, and thus since $v_1$ in $S_{[16}$, $v_1$ is also adjacent to $k_i$ and $v_3, \ldots, v_j$ are nonadjacent to $k_i$ for they represent vertices in $S_{66}$.
	
	\subcase \textit{$B$ contains $S_7(j)$.}
	
	Suppose $B$ contains $S_7(3)$. It is straightforward that the rows and columns induce a co-$4$-tent${}\vee{}K_1$.
	Furthermore, if $j >3$, then $j$ is even. The rows and columns of $S_7(j)$ induce a $j$-sun.

	\subcase \textit{$B$ contains $S_8(2j)$.} 
	
	If $j=2$, then we can find an tent induced by the last three columns and the last three rows.
	If instead $j>2$, then we find a $(2j-1)$-sun with center induced by every unlabeled row, every column but the first and one more column --which will be the center-- representing any vertex in $K_1$, since $K_1 \neq \emptyset$.
	
	\subcase \textit{$B$ contains $P_0(j,l)$.}
	
	Let $v_1, \ldots, v_j$ be the vertices represented by the rows of $P_0(j,l)$ and $k_{61}, \ldots, k_{6j}$	be the vertices in $K_6$ represented by the columns. The rows corresponding to $v_1$ and $v_j$ are labeled with L and R, respectively, and the row corresponding to $v_{l+2}$ is an LR-row.
	
	\subsubcase Suppose first that $l=0$. If $j$ is even, then $v_1$ and $v_j$ are colored with the same color. Suppose without loss of generality that both are colored with red, thus $v_1$ lies in $S_{36}\cup S_{46}$ and $v_j$ lies in $S_{62} \cup S_{63}$. In that case, there are vertices $k_i$ in $K_i$ for $i=2,4$ such that $k_2$ is adjacent to $v_j$ and nonadjacent to $v_1$ and $k_4$ is adjacent to $v_1$ and nonadjacent to $v_j$. We find $F_2(j+1)$ induced by $\{ k_2, $, $k_4$, $k_{62}, \ldots, k_{6j}$, $v_1, \ldots, v_j$, $s_{24} \}$
	
	If instead $j$ is odd, then $v_1$ and $v_j$ are colored with the same colors. Suppose without loss of generality that they are both colored with red. Hence, $v_1$ lies in $S_{36} \cup S_{46}$ and $v_j$ lies in $S_{61} \cup S_{64} \cup S_{65}$. In either case, we find $F_2(j+2)$ induced by $\{ k_1, $, $k_2$, $k_4$, $k_{62}, \ldots, k_{6j}$, $v_1, \ldots, v_j$, $s_{24}$, $s_{12} \}$ if $v_j$ lies in $S_{61}$, and induced by $\{ k_2, $, $k_4$, $k_5$, $k_{62}, \ldots, k_{6j}$, $v_1, \ldots, v_j$, $s_{24}$, $s_{45} \}$ if $v_j$ lies in $S_{64} \cup S_{65}$.
	
	\subsubcase Suppose that $l>0$.	 The proof is very similar to the case $l=0$. If $j$ is odd, then $v_1$ and $v_j$ are colored with the same color. If it is red, then we find $F_2(j+2)$ induced by $\{ k_1, $, $k_2$, $k_4$, $k_{61}, \ldots, k_{6(j-1)}$, $v_1, \ldots, v_j$, $s_{24}$, $s_{12} \}$ if $v_j$ lies in $S_{61}$, and we find $F_2(j)$ induced by $\{ k_4$, $k_{61}, \ldots, k_{6(j-1)}$, $v_1, \ldots, v_j \}$ if $v_j$ lies in $S_{64} \cup S_{65}$. 
	
	If instead $j$ is even, then $v_1$ and $v_j$ are colored with distinct colors. Then, we find $F_2(j+1)$ induced by $\{ k_2$, $k_4$, $k_{61}, \ldots, k_{6(j-1)}$, $v_1, \ldots, v_j$, $s_{24} \}$.
	
	\subcase \textit{$B$ contains $P_1(j,l)$.}
	
		Let $v_1, \ldots, v_j$ be the vertices represented by the rows of $P_1(j,l)$ and $k_{61}, \ldots, k_{6(j-1)}$ be the vertices in $K_6$ represented by the columns. The rows corresponding to $v_1$ and $v_j$ are labeled with L and R, respectively, and the rows corresponding to $v_{l+2}$ and $v_{l+3}$ are LR-rows.
	
	\subsubcase Suppose first that $l=0$. 
	If $j$ is odd, then $v_1$ and $v_j$ are colored with the same color. We assume without loss of generality that they are colored with red. Thus, $v_1$ lies in $S_{36} \cup S_{46}$ and $v_j$ lies in $S_{61} \cup S_{64} \cup S_{65}$. In either case, $v_1$ is anticomplete to $K_1$. Hence, we find $F_1(j)$ induced by every row and column of $P_1(j,0)$ and an extra column that represents a vertex in $K_1$ adjacent to $v_j$, $v_2$ and $v_3$ and nonadjacent to $v_i$, for $1 \leq i \leq j-1$, $i \neq j,2,3$.
	If instead $j$ is even, then we assume that $v_1$ and $v_j$ are colored with red and blue, respectively. Thus, $v_1$ lies in $S_{36} \cup S_{46}$ and $v_j$ lies in $S_{62} \cup S_{63}$. We find $F_1(j+1)$ induced by every row and every column of $P_1(j,0)$, the row corresponding to $s_{24}$ and two columns corresponding to vertices $k_2$ in $K_2$ and $k_4$ in $K_4$ such that $k_2$ is adjacent to $v_j$, $v_2$ and $v_3$ and is nonadjacent to $v_i$, and $k_4$ is adjacent to $v_1$, $v_2$ and $v_3$ and is nonadjacent to $v_i$, for each $1 \leq i \leq j-1$, $i \neq j, 2, 3$. 
	
	\subsubcase Suppose $l>0$. The proof is analogous to the previous case if $j$ is even. If instead $j$ is odd, then $v_1$ lies in $S_{36} \cup S_{46}$ and $v_j$ lies in $S_{61} \cup S_{64} \cup S_{65}$. If $v_j$ in $S_{61}$, then we find $F_1(j+2)$ induced by $\{ k_4$, $k_{61}, \ldots, k_{6(j-2)}$, $k_1$, $k_2$, $v_1, \ldots, v_j$, $s_{12}$, $s_{24} \}$.	
	If instead $v_j \not\in S_{61}$, then we find $F_1(j)$ induced by every row and every column of $P_1(j,l)$ and one more column representing a vertex in $K_4$ adjacent to every vertex represented by a labeled row.

	\subcase \textit{$B$ contains $P_2(j,l)$.}
	
		Let $v_1, \ldots, v_j$ be the vertices represented by the rows of $P_2(j,l)$ and $k_{61}, \ldots, k_{6(j-1)}$ be the vertices in $K_6$ represented by the columns. The rows corresponding to $v_1$ and $v_j$ are labeled with L and R, respectively, and the rows corresponding to $v_{l+2}$, $v_{l+3}$, $v_{l+4}$ and $v_{l+5}$ are LR-rows.
	
	Suppose $l=0$. If $j$ is even, then we find $F_1(j-1)$ induced by $\{ k_{62}$, $k_{65}, \ldots, k_{6(j-1)}$, $v_1$, $v_2$, $v_5, \ldots, v_j$, $s_{24} \}$.	The same subgraph arises if $l>0$.
	
	 Suppose now that $j$ is odd, thus $v_1$ and $v_j$ are colored with the same color. We can assume without loss of generality that $v_1$ lies in $S_{36} \cup S_{46}$ and $v_j$ lies in $S_{61} \cup S_{64} \cup S_{65}$.
	 If $v_j \not\in S_{61}$, then we find $F_1(j-2)$ induced by $\{ k_{61}$, $k_{62}$, $k_{65}, \ldots, k_{6(j-1)}$, $v_1$, $v_2$, $v_5, \ldots, v_j$, $k_4 \}$, where $k_4$ in $K_4$ is adjacent to $v_1$, $v_2$, $v_5$ and $v_j$. The same subgraph arises if $l>0$.
	 If $v_j$ in $S_{61}$, then there are vertices $k_i$ in $K_i$, for $i=1,2,4$ such that $k_1$ is adjacent to $v_j$ and is nonadjacent to $v_1$, $k_2$ is nonadjacent to both and $k_4$
 is adjacent to $v_1$ and nonajcent to $v_j$. If $l=0$, we find $M_{II}(j)$ induced by $\{ k_{62}$, $k_{63}$, $k_{65}, \ldots, k_{6(j-1)}$, $v_1$, $v_2$, $v_5, \ldots, v_j$, $k_1$, $k_2$, $k_4$, $s_{12}$, $s_{24} \}$.
	  If instead $l>0$, then we find $F_1(j)$ induced by  $\{ k_{61}$, $k_{62}$, $k_{64}, \ldots, k_{6(j-1)}$, $k_1$, $k_2$, $k_4$, $v_1$, $v_2$, $v_3$, $v_6, \ldots, v_j$, $s_{12}$, $s_{24} \}$.

Therefore, $B$ is admissible.	

	\vspace{3mm}	
	\case Suppose now that $B$ is admissible but not LR-orderable, thus $B^*_\tagg$ contains either a Tucker matrix, or $M_4'$, $M_4''$, $M_5'$, $M_5''$, $M'_2(k)$, $M''_2(k)$, $M_3'(k)$, $M_3''(k)$, $M_3'''(k)$ for some $k \geq 4$.
	
	Toward a contradiction, it suffices to see that $B^*_\tagg$ does not contain any Tucker matrix, for in the case of the matrices listed in Figure \ref{fig:forb_LR-orderable_tags}, each labeled column can be replaced by a column that represents a vertex that belongs to the same subclasses considered in the analysis for a Tucker matrix with at least one LR-row, and since some of the rows may be non-LR-rows, then that case can be reduced to a particular case. 

Let $M$ be a Tucker matrix contained in $B^*_\tagg$. Thoughout the proof, when we refer to an LR-row in $M$, we refer to the row in $B$, this is, the complement of the row that appears in $M$. 
	
	 	
	\subcase Suppose first that $B^*_\tagg$ contains $M_I(j)$, for some $j \geq 3$. Let $v_1, \ldots, v_j$ be the vertices corresponding to the rows of $M_I(j)$, and $k_{61}, \ldots, k_{6j}$ in $K_6$ be the vertices corresponding to the columns.
		
	\begin{remark} \label{obs:no_Di_in_MI}
	If two non-LR-rows in $M_I(j)$ are labeled with the same letter, then they induce $D_0$. Moreover, any pair of consecutive non-LR-rows labeled with distinct letters induce $D_1$ or $D_2$. This follows from the fact that $B$ is admissible.
		Hence, there are at most two non-LR-rows in $M_I(j)$ and such rows are non-consecutive and labeled with distinct letters.
	Furthermore, since $B$ is admissible, it is easy to see that there are at most two LR-rows in $M_(j)$, for if not such rows induce $D_{11}$, $D_{12}$ or $D_{13}$.
	\end{remark}
	
	
	\subsubcase Suppose first that $j=3$ and that $v_1$ is the only LR-row in $M_I(j)$.

	If rows $v_2$ and $v_3$ are unlabeled, then we find a net $\vee{}K_1$ induced by $\{v_1$, $v_2$, $v_3$, $k_{61}$, $k_{62}$, $k_{63}$, $ k_l \}$, where $k_l$ is any vertex in $K_l \neq K_6$. The same holds if either $v_2$ or $v_3$ are labeled rows, by accordingly replacing $k_l$ for some $l$ such that $k_l$ is nonadjacent to both $v_2$ and $v_3$ (there are no labeled rows complete to each partition $K_i \neq K_6$ of $K$). 
	By the previous remark, if both $v_2$ and $v_3$ are labeled rows, then they are labeled with distinct letters. Thus, we find $F_0$ induced by $\{ v_1$, $v_2$, $v_3$, $k_{61}$, $k_{62}$, $k_{63}$, $k_1$, $k_5  \}$, where $k_1$ in $K_1$ is adjacent to $v_2$ and nonadjacent to $v_3$ and $k_5$ in $K_5$ is adjacent to $v_3$ and nonadjacent to $v_2$, or viceversa. Such vertices exist since we assumed $\mathbb B_i$ admissible for every $i \in \{1, \ldots, 5\}$. 
	
	
	If instead $v_1$ and $v_2$ are LR-rows, then we find a tent by considering any vertex $k_l$ in $K_l$ for some $l \in \{1, \ldots, 5\}$ such that $v_3$ is nonadjacent to $k_l$. The tent is induced by the set $\{ v_1$, $v_2$, $v_3$, $k_{61}$, $k_{63}$, $ k_l \}$. Every other case is analogous by symmetry.
Moreover, if $v_1$, $v_2$ and $v_3$ are LR-rows, then there is a vertex $k_l$ in $K_l$ with $l \neq 6$ such that $v_1$, $v_2$ and $v_3$ are adjacent to $k_l$, hence we find a net${}\vee{}K_1$ induced by $\{ v_1$, $v_2$, $v_3$, $k_{61}$, $k_{62}$, $k_{63}$, $ k_l \}$.
 
	\subsubcase Suppose now that $j \geq 4$, and let us suppose first that there is exactly one LR-row in $M_I(j)$. Thus, we may assume that $v_1$ is the only LR-row in $M_I(j)$.
	Notice first that, if $j$ is odd, then we find $F_2(j)$ in $B$ induced by the vertices represented by every row and column. Hence, we may assume that $j$ is even.
	By Remark \ref{obs:no_Di_in_MI}, there are at most two labeled rows in $M_I(j)$ and such rows are labeled with distinct letters.
	 
	If either there are no labeled rows or there is exactly one labeled row, then we find $M_{III}(j)$ induced by $\{ v_1, \ldots, v_j$, $k_{61}, \ldots, k_{6j}$, $ k_l \}$, where $k_l$ is any vertex in some $K_l \neq K_6$ that is nonadjacent to the only labeled row. 
	
	Suppose there are two labeled rows in $M_I(j)$. If there are two labeled rows $v_i$ and $v_l$, then it suffices to see what happens if $v_i$ belongs to $S_{36} \cup S_{46}$ and $v_l$ belongs to either $S_{61}$, $S_{64} \cup S_{65}$ or $S_{62} \cup S_{63}$. If $v_l$ belongs to $S_{61}$, then there is a vertex $k_2$ in $K_2$ nonadjacent to both $v_i$ and $v_l$, and thus we also find $M_{III}(j)$ induced by the same vertex set as before.
If instead $v_l$ lies in $S_{64} \cup S_{65}$, then there are vertices $k_2$ in $K_2$ and $k_4$ in $K_4$ such that $k_4$ is adjacent to both $v_l$ and $v_i$. Hence, if $|l-i|$ is even, then we find an $(l-i)$-sun. If instead $|l-i|$ is odd, then we find a $(l-i)$-sun with center, where the center is given by the LR-vertex $v_1$. 
Using a similar argument, if $v_l$ lies in $S_{62} \cup S_{63}$, then we find an even sun or an odd sun with center considering the same vertex set as before plus $s_{24}$.


	
	Suppose now that $v_1$ and $v_2$ are LR-rows. If $j \geq 4$ is even and every row $v_i$ with $i>2$ is unlabeled (or is at most one is a labeled row), then we find $M_{II}(j)$ induced by $\{ v_1, \ldots, v_j$, $k_{61}$, $k_{63}, \ldots, k_{6j}$, $ k_l \}$, where $k_l$ is any vertex in some $K_l \neq K_6$ such that each $v_i$ is nonadjacent to $k_l$ for every $i \geq 3$. 
	Moreover, if $j \geq 4$ is odd, then we find $F_1(j)$ induced by $\{ v_1, \ldots, v_j$, $k_{61}$, $k_{63}, \ldots, k_{6j}\}$. The same holds if there is exactly one labeled row since we can always choose when necessary a vertex in some $K_l$ with $l \neq 6$ that is nonadjacent to such labeled vertex.

	Let us suppose there are exactly two labeled rows $v_i$ and $v_l$. By Remark \ref{obs:no_Di_in_MI}, these rows are non-consecutive and are labeled with distinct letters. As in the previous case, $v_i$ belongs to $S_{36} \cup S_{46}$ and $v_l$ belongs to either $S_{61}$ or $S_{64} \cup S_{65}$. If $v_l$ belongs to $S_{61}$, then there is a vertex $k_2$ in $K_2$ nonadjacent to both $v_i$ and $v_l$, and thus we find $\{ v_1, \ldots, v_j$, $k_{61}$, $k_{63}, \ldots, k_{6j}$, $k_2 \}$. If instead $v_l$ lies in $S_{64} \cup S_{65}$, then we find $k_4$ in $K_4$ adjacent to both $v_i$ and $v_l$ and thus we find either an even sun or an odd sun with center as in the previous case.
	Using a similar argument, if $v_l$ lies in $S_{62} \cup S_{63}$, then we find an even sun if $l-i$ is even or an odd sun with center if $l-i$ is odd.
	
	Finally, suppose $v_1$ and $v_i$ are LR-rows, where $i>2$. If $j=4$, then we find a $4$-sun induced by every row and every column, hence, suppose that $j>5$. In that case, we find a tent contained in the subgraph induced by $\{v_1$, $v_2$, $v_3 \}$ if $i=3$ and $\{v_1$, $v_{j-1}$, $v_j \}$ if $i= j-1$.
	Thus, let $3<i<j-1$. However, in that case we find $M_{II}(i)$ induced by $\{v_1$, $v_2, \ldots, v_i$, $k_{62}, \ldots, k_{6(j-2)}, k_{6j} \}$.
	Therefore, there is no $M_I(j)$ in $B^*_\tagg$.

	\subcase Suppose that $B^*_\tagg$ contains $M_{II}(j)$. Let $v_1, \ldots, v_j$ be the vertices corresponding to the rows, and $k_{61}, \ldots, k_{6j}$ in $K_6$ the vertices representing the columns. If $j$ is odd and there are no labeled rows, then we find $F_1(j)$ by considering $\{ v_1, \ldots, v_j$, $k_{61} \ldots, k_{6(j-1)} \}$. Moreover, if there are no LR-rows and $j$ is odd, then we find $M_{II}(j)$ as an induced subgraph. 
	Hence, we assume from now on that there is at least one LR-row.

\begin{remark} \label{obs:reduce_casos_MII_4tent}
	As in the previous case, there are at most two rows labeled with L or R in $M_{II}(k)$, for any three LR-rows induce an enriched submatrix that contains either $D_0$, $D_1$ or $D_2$. Moreover, since $B$ is admisssible, then there are at most three LR-rows.

	If $v_i$ and $v_l$ with $1<i<l<j$ are two rows labeled with either L or R, then they are labeled with distinct letters for if not we find $D_0$. Moreover, they are not consecutive since in that case we find either $D_1$ or $D_2$. Thus, since $v_i$ belongs to $S_{36} \cup S_{46}$ and $v_l$ belongs to either $S_{61}$ or $S_{64} \cup S_{65}$ or $S_{62} \cup S_{63}$, one of the following holds:
	\begin{itemize}
		\item If $v_l$ in $S_{61}$, then we find a $(l-i+2)$-sun if $l-i$ is even or a $(l-i+2)$-sun with center if $l-i$ is odd (the center is $k_{6j}$) induced by $\{ v_i, \ldots, v_l$, $s_{12}$, $s_{24}$, $k_{6(i+1)} \ldots, k_{6l}$, $k_1$, $k_2$, $k_4$, $k_{6j} \}$.
		\item If $v_l$ in $S_{64} \cup S_{65}$ (resp.\ $S_{62} \cup S_{63}$), then we find a $(l-i)$-sun if $l-i$ is even or a $(l-i)$-sun with center if $l-i$ is odd (the center is $k_{6j}$) induced by $\{ v_i, \ldots, v_l$, $k_{6(i+1)} \ldots, k_{6l}$, $k_4$, $k_{6j} \}$ (resp.\ $k_1$, $k_2$).	
	\end{itemize}
	
	Furthermore, suppose $v_1$ and $v_i$ are rows labeled with either L or R, where $1<i \leq j$. If $i=2,j$, then they are labeled with distinct letters for if not we find $D_0$. Moreover, they are colored with distinct colors for if not we find $D_1$. If instead $2<i<j$, then they are labeled with the same letter for if not we find $D_1$ or $D_2$.

\end{remark}

As a consequence of the previous remark we may assume without loss of generality that, if there are rows labeled with either L or R, then these rows are either $v_j$ and $v_{j-1}$, $v_1$ and $v_j$ or $v_{j-2}$ and $v_j$ for every other case is analogous. Moreover, if $v_j$ and $v_{j-1}$ (resp.\ $v_1$) are labeled rows, then we may assume they are colored with distinct colors.

	\subsubcase Suppose there is exactly one LR-row and suppose first that $v_1$ is the only LR-row. If every non-LR row is unlabeled or $v_{j-2}$ and $v_j$ are labeled rows, since they are labeled with the same letter (for if not we find $D_1$ or $D_5$ considering $v_1$, $v_{j-2}$ and $v_j$), then we find $M_{III}(j)$ induced by $\{ k_l$, $v_1, \ldots, v_j$, $k_{61}$, $\ldots$, $k_{6j} \}$, where $k_l$ is any vertex in $K_l \neq K_6$.
	Moreover, if $v_{j-1}$ is a labeled row, then we find either a $(j-1)$-sun or a $(j-1)$-sun with center, depending on whether $j$ is even or odd, induced by $\{ v_1, \ldots, v_{j-1}$, $k_l$, $k_{61}, \ldots, k_{6(j-2)}$, $k_{6j} \}$, thus we finished this case. 

If $v_2$ is an LR-row, then we find $M_{II}(j-1)$ or $F_1(j-1)$ (depending on whether $j$ is odd or even) induced by every column of $B$ and the rows $v_2$ to $v_j$. It does not depend on whether there are or not rows labeled with L or R.

Suppose $v_i$ is an LR-row for some $2<i<j-1$. Let $r_i$ be the first column in which $v_i$ has a $0$ and $c_i$ be column in which $v_j$ has a $0$, then we find a tent induced by columns $k_{61}$, $k_{6(r_i)}$ and $k_{6(c_i)}$ and the rows $v_1$, $v_i$ and $v_j$.

If $v_{j-1}$ is an LR-row, then we find $M_{II}(j-1)$ induced by $\{ v_1, \ldots, v_{j-1}$, $k_{61}, \ldots, k_{6(j-2)}$, $k_{6j}\}$.

If $v_j$ is an LR-row and either every other row is unlabeled or there is exactly one labeled row, then we find $M_{III}(j)$ induced by $\{ k_l$, $v_1, \ldots, v_j$, $k_{61}$, $\ldots$, $k_{6j} \}$, where $k_l$ is any vertex in $K_l \neq K_6$ such that the vertex representing the only labeled row is nonadjacent to $k_l$.
Suppose there are two labeled rows. It follows from Remark \ref{obs:reduce_casos_MII_4tent} that such rows are either $v_1$ and $v_2$ or $v_1$ and $v_i$ for some $2<i<j$. However, if $v_i$ is a labeled row for some $1<i<j-1$, then we find either an even sun or an odd sun with center analgously as we have in Remark \ref{obs:reduce_casos_MII_4tent}. If instead $v_{j-1}$ and $v_1$ are labeled rows, then they are labeled with the same letter and thus we are in the same situation as if there were no labeled rows in $B$ since we can find a vertex that results nonadjacent to both $v_1$ and $v_{j-1}$.

\subsubcase Suppose there are two LR-rows. If $v_1$ and $v_2$ are LR-rows, then we find $M_{II}(j-1)$ as we have in the case where only $v_2$ is an LR-row. 
Suppose $v_1$ and $v_3$ are LR-rows. If $j=4$, then we find $M_{II}(j)$ induced by $\{v_1, \ldots, v_4$, $k_{61}$, $k_{62}$, $k_{64}$, $ k_l \}$ where $k_l$ in $K_l \neq K_6$. Such a vertex exists, since $v_2$ and $v_4$ are either unlabeled rows or are rows labeled with the same letter, for if they were labeled with distinct letters we would find $D_0$ or $D_1$. Thus, there is a vertex that is nonadjacent to both $v_2$ and $v_4$ and is adjacent to $v_1$ and $v_3$. 
If $j>4$, then we find a tent induced by rows $v_3$, $v_{j-1}$ and $v_j$ and columns $j-2$, $j-1$ and $j$. 
Moreover, if $v_i$ is an LR-row for $1<2<j-1$ and $v_{j-1}$ and $v_j$ are non-LR-rows, then we find a tent induced by the rows $v_i$, $v_{j-1}$ and $v_j$ and the columns $j-2$, $j-1$ and $j$. 

Thus, it remains to see what happens if $v_1$ and $v_{j-1}$ and $v_1$ and $v_j$ are LR-rows.
If $v_1$ and $v_{j-1}$ are LR-rows, then we find $M_{II}(j)$ induced by all the rows of $M_{II}(j)$ and every column except for column $j-1$, which is replaced by some vertex $k_l$ in $K_l \neq K_6$ (since in this case, if there are two labeled rows, then they must be $v_i$ for some $1<i<j-1$ and $v_j$, thus they are labeled with the same letter, hence there is a vertex $k_l$ nonadjacent to both).
Finally, if $v_1$ and $v_j$ are LR-rows, then we find a $j$-sun or a $j$-sun with center, depending on whether $j$ is even or odd, contained in the subgraph induced by $\{ v_1, \ldots, v_j$, $k_{61}, \ldots, k_{6j}$, $k_l \}$, where $k_l$ in $K_l \neq K_6$ is nonadjacent to every non-LR row (same argument as before). Therefore, there is no $M_{II}(j)$ in $B^*_\tagg$.

	\subcase Suppose that $B$ contains $M= M_{III}(j)$, let $v_1, \ldots v_j$ be the rows of $M$ and $k_{61}, \ldots, k_{6(j+1)}$ be the columns of $M$.
	If there are no LR-rows, then we find $M_{III}(j)$, hence we assume there is at least one LR-row. As in the previous cases, since $B$ is admissible, there are at most two LR-rows in $M$.
	
	Notice that every pair of rows $v_i$ and $v_l$ with $1\leq 1<i,l<j-1$ are not labeled with the same letter, since they induce $D_0$. Once more, if such rows are labeled with distinct letters, then they are not consecutive for in that case we would find $D_1$ or $D_2$. Furthermore, if such $v_i$ and $v_l$ are labeled rows, then we find either an even sun or an odd sun with center.
	Moreover, if $i=1,j-1$ and $l=j$, then $v_i$ and $v_l$ are not both labeled rows, for the same arguments holds.
	Hence, if there are two labeled rows, then such rows must be $v_j$ and $v_i$ for some $i$ such that $2<i<j-1$.

\subsubcase There is exactly one LR-row. Suppose first that $v_1$ is an LR-row. In this case, we find $M_{II}(j)$ induced by $\{ v_1, \ldots, v_j$, $k_{62}, \ldots, k_{6(j+1)} \}$. If $v_i$ is an LR-row, for some $1\leq i< j-1$, then we find $M_{II}(j-i+1)$ induced by $\{ v_i, \ldots, v_j$, $k_{6(i+1)}, \ldots, k_{6(j+1)} \}$.

If $v_{j-1}$ is an LR-row, then we also find $M_{II}(j)$, induced by $\{ v_1, \ldots, v_j$, $k_{62}, \ldots,  k_{6(j-1)}$, $k_{6(j+1)} \}$.

If instead $v_j$ is an LR-row, then we find an even $j$-sun or an odd $j$-sun with center $k_{6(j+1)}$.

\subsubcase Suppose now there are two LR-rows $v_i$ and $v_l$. If $1\leq i<l<j-1$ and $v_i$ and $v_l$ are not consecutive rows, then we find a tent induced by the rows $v_i$, $v_l$ and $v_j$, and columns $k_s$ in $K_s \neq K_6$ adjacent to both $v_i$ and $v_l$ and nonadjacent to $v_j$, and $k_{6i}$ (or $k_{6(i+1)}$ if $i=1$) and $k_{6l}$ (or $k_{6(l+1)}$ if $l=j-1$). The same subgraph contains an induced tent if $l=i+1$ and $i>1$.
If instead $i=1$ or $i=j-1$ and $l=i+1$, then we find $F_0$ (or $M_{III}(3)$ if $j=3)$ induced by $\{ v_i$, $v_{i+1}$, $k_{6i}$, $k_{6(i+1)}$, $k_{6(i+2)}$, $k_6(j+1)$, $ k_s\}$ with $k_s$ in $K_s \neq K_6$ adjacent to both $v_i$ and $v_{i+1}$.

Finally, if $v_1$ and $v_j$ are LR-rows, then we find $M_{III}(j)$ induced by every row $v_1, \ldots, v_j$ and column $k_{61}, \ldots, k_6(j+1)$. If instead $v_i$ and $v_j$ are LR-rows with $i>1$, then we find $M_V$ induced by $\{ v_i$, $v_j$, $v_1$, $v_{j-1}$, $k_{61}$, $k_{62}$, $k_{6i}$, $k_{6(i+1)}$, $k_{6j}\}$, therefore there is no $M_{III}(j)$ in $B^*_\tagg$.

	\subcase Suppose that $B$ contains $M=M_{IV}$, let $v_1, \ldots, v_4$ be the rows of $M$ and $k_{61}, \ldots, k_{66}$ be the columns of $M$. If there are no labeled rows, then we find $M_{IV}$ as an induced subgraph, and since $B$ is admissible and any three rows are not pairwise nested, then there are at most two LR-rows, hence we assume there are exactly either one or two LR-rows.

	If the row $v_i$ is an LR-row for $i=1,2,3$, then we find $M_V$ induced by $\{ v_2$, $v_3$, $v_4$, $k_{62}, \ldots, k_{66} \}$. Moreover, if only $v_4$ is an LR-row, then we find $M_{IV}$ induced by all the rows and columns of $M$. Thus, we assume there are exactly two LR-rows. 
	If $v_1$ and $v_4$ are LR-rows, then we find $M_V$ induced by $\{ v_1$, $v_2$, $v_3$, $v_4$, $k_{61}$, $k_{63}, \ldots, k_{66} \}$. The same holds if $v_i$ and $v_4$ are LR-rows, with $i=2,3$.
	Finally, if $v_1$ and $v_2$ are LR-rows, then we find a tent induced by $\{ v_1$, $v_2$, $v_4$, $k_{62}$, $k_{64}$, $k_{65} \}$. It follows analogously by symmetry if $v_1$ and $v_3$ or $v_2$ and $v_3$ are LR-rows, therefore there is no $M_{IV}$ in $B^*_\tagg$. 
	
	\subcase Suppose that $B$ contains $M=M_V$, let $v_1, \ldots, v_4$ be the rows of $M$ and $k_{61}, \ldots, k_{65}$ be the columns of $M$.
	Once more, if there are no LR-rows, then we find $M_V$ as an induced subgraph, thus we assume there is at least one LR-row. Moreover, since any three rows are not pairwise nested, there are at most two LR-rows. 
	
	\subsubcase If $v_1$ is the only LR-row, then we find a tent induced by $\{ v_1$, $v_3$, $v_4$, $k_{61}$, $k_{63}$, $k_{65} \}$. The same holds if $v_2$ is the only LR-row.
	
	If $v_3$ is the only LR-row and every other row is unlabeled or are all labeled with the same letter, then we find $M_{IV}$ induced by$\{ v_1$, $v_2$, $v_3$, $v_4$, $k_{61}, \ldots, k_{65}$, $k_l \}$ where $k_l$ in $K_l \neq K_6$ adjacent only to $v_3$.
	Suppose there are at least two rows labeled with either L or R. Notice that, if $v_1$ and $v_2$ are labeled, then they are labeled with distinct letters for if not they contain $D_0$. Moreover, $v_1$ (resp.\ $v_2$) and $v_4$ cannot be both labeled, for in that case they contain either $D_0$ or $D_1$ or $D_2$. Hence, there are at most two rows labeled with either L or R, and they are necessarily $v_1$ and $v_2$.
	In that case, there is a vertex $k_l$ in some $K_l \neq K_6$ such that $v_2$ and $v_3$ are adjacent to $k_l$ and $v_4$ is nonadjacent to $k_l$, thus we find a tent induced by $v_2$, $v_3$, $v_4$, $k_l$, $k_{64}$ and $k_{65}$.
	
	
	If $v_4$ is the only LR-row and every other row is unlabeled or are (one, two or) all labeled with the same letter, then we find $M_V$ induced by $\{ v_1$, $v_2$, $v_3$, $v_4$, $k_{61}, \ldots, k_{64}$, $k_l \}$ where $k_l$ in $K_l \neq K_6$ adjacent only to $v_4$.
	
	\subsubcase Suppose there are exactly two LR-rows. 
	If $v_1$ and $v_2$ are such LR-rows, then we find a tent induced by $\{ v_1$, $v_2$, $v_3$, $k_{62}$, $k_{63}$, $k_{65} \}$, thus we discard this case.
	If instead $v_1$ and $v_3$ are LR-rows and every other row is unlabeled or (one or) all are labeled with the same letter, then we find $M_V$ induced by every row and column plus a vertex $k_l$ in some $K_l \neq K_6$ such that both $v_2$ and $v_4$ are nonadjacent to $k_l$.
	Moreover, since $v_2$ and $v_4$ are neither disjoint or nested and there is a column in which both rows have a $0$, then they are not labeled with distinct letters, disregarding of the coloring, for in that case we find $D_1$ or $D_2$.

	If exactly $v_1$ and $v_4$ are LR-rows and every other row is unlabeled or are (one or) all labeled with the same letter, then we find a tent induced by every row and column plus a vertex $k_l$ in some $K_l \neq K_6$ such that both $v_2$ and $v_4$ are nonadjacent to $k_l$. Once more, $v_2$ and $v_3$ are not labeled with distinct letters since in that case we find either $D_1$ or $D_2$.
	
	If exactly $v_3$ and $v_4$ are LR-rows and every other row is unlabeled or either $v_1$ or $v_2$ is labeled with L or R, then we find $M_{IV}$ induced by every row and column plus a vertex $k_l$ in some $K_l \neq K_6$ such that both $v_1$ and $v_2$ are nonadjacent to $k_l$. 
	Once more, $v_1$ and $v_2$ are not labeled with the same letter for they would induce $D_0$, neither they are labeled with distinct letters since in that case we find either $D_1$ or $D_2$. 
	
	If $v_1$, $v_2$ and $v_3$ are LR-rows, since there is a vertex $k_l \in K_l$ with $l \neq 6$ such that $v_4$ is nonadjacent to $k_l$, then we find a tent induced by $\{ v_1$, $v_2$, $v_4$, $k_{61}$, $k_{64}$, $k_l \}$. Analogously, if $v_1$, $v_2$ and $v_4$ are LR-rows and $v_3$ is not, then the tent is induced by $\{ v_1$, $v_2$, $v_3$, $k_{61}$, $k_{64}$, $k_{65} \}$. The same holds if all 4 rows are LR-rows, where the tent is induced by $\{ v_1$, $v_2$, $v_4$, $k_{62}$, $k_{63}$, $k_{65} \}$.
	Finally, if $v_2$, $v_3$ and $v_4$ are LR-rows, since there is a vertex $k_l \in K_l$ with $l \neq 6$ such that $v_1$ is nonadjacent to $k_l$, then we find $M_V$ induced by $\{ v_1$, $v_2$, $v_3$, $v_4$, $k_{61}$, $k_{62}$, $k_{63}$, $k_{65}$, $k_l \}$.

	\case Therefore, we may assume that $B$ is admissible and LR-orderable but is not partially $2$-nested. Since there are no uncolored labeled rows and those colored rows are labeled with either L or R and do not induce any of the matrices $\mathcal{D}$, then in particular no pair of pre-colored rows of $B$ induce a monochromatic gem or a monochromatic weak gem, and there are no badly-colored gems since every LR-row is uncolored, therefore $B$ is partially $2$-nested. 

	
	\case Finally, let us suppose that $B$ is partially $2$-nested but is not $2$-nested. As in the previous cases, we consider $B$ ordered with a suitable LR-ordering.
	Let $B'$ be a matrix obtained from $B$ by extending its partial pre-coloring to a total $2$-coloring. It follows from Lemma \ref{lema:B_ext_2-nested} that, if $B'$ is not $2$-nested, then either there is an LR-row for which its L-block and R-block are colored with the same color, or $B'$ contains a monochromatic gem or a monochromatic weak gem or a badly-colored doubly weak gem. 	
	
	If $B'$ contains a monochromatic gem where the rows that induce such a gem are not LR-rows, then the proof is analogous as in the tent case. Thus, we may assume that at least one of the rows is an LR-row.

		
	\subcase \textit{Let us suppose first that there is an LR-row $w$ for which its L-block $w_L$ and R-block $w_R$ are colored with the same color. 	}
	If these two blocks are colored with the same color, then there is either one odd sequence of rows $v_1, \ldots, v_j$ that force the same color on each block, or two distinct sequences, one that forces the same color on each block. 
	
	\subsubcase Let us suppose first that there is one odd sequence $v_1, \ldots, v_j$ that forces the color on both blocks. If $k=1$, then notice this is not possible since we are coloring $B'$ using a suitable LR-ordering. If there is not a suitable LR-ordering, then $B$ is not admissible or LR-orderable, which results in a contradiction. Thus, let $j>1$ and assume without loss of generality that $v_1$ intersects $w_L$ and $v_j$ intersects $w_R$. Moreover, we assume that each of the rows in the sequence $v_1, \ldots, v_j$ is colored with a distinct color and forces the coloring on the previous and the next row in the sequence.
	If $v_1, \ldots, v_j$ are all unlabeled rows, then we find an even $(j+1)$-sun.
	If instead $v_1$ is an L-row, then $w_L$ is properly contained in $v_1$. Thus, $v_2, \ldots, v_{j-1}$ are not contained in $v_1$, since at least $v_j$ intersects $w_R$. If $v_j$ is unlabeled or labeled with R, then we find an even $(j+1)$-sun. If instead $v_j$ is labeled with L, since $j$ is odd, then we find $S_1(j+1)$ in $B$ which is not possible since we are assuming $B$ admissible.

	\subsubcase Suppose now that there are two independent sequences $v_1, \ldots, v_j$ and $x_1, \ldots, x_l$  that force the same color on $w_L$ and $w_R$, respectively. Suppose without loss of generality that $w_L$ and $w_R$ are colored with red.
	If $j=1$ and $l=1$, then we find $D_6$, which is not possible. Hence, we assume that either $j>1$ or $l>1$. 
	Suppose that $j>1$ and $l>1$. In this case, there is a labeled row in each sequence, for if not we can change the coloring for each row in one of the sequences and thus each block of $w$ can be colored with distinct colors. We may assume that $v_j$ is labeled with L and $x_l$ is labeled with R (for the LR-ordering used to color $B'$ is suitable and thus there is no R-row intersecting $w_L$, and the same holds for each L-block and $w_R$). As in the previous paragraphs, we assume that each row in each sequence forces the coloring on both the previous and the next row in its sequence. In that case, $v_2, \ldots, v_j$  is contained in $w_L$ and $x_2, \ldots, x_l$ is contained in $w_R$. 
	Moreover, $w$ represents a vertex in $S_{[16}$, $v_j$ lies in $S_{46} \cup S_{36}$ or $S_{16} \cup S_{26} \cup S_{56}$ and $x_l$ lies in $S_{61} \cup S_{64} \cup S_{65}$ or $S_{62} \cup S_{63}$ (depending on whether they are colored with red or blue, respectively).
	Suppose first that they are both colored with red, thus $v_j$ lies in $S_{46} \cup S_{36}$ and $x_l$ lies in $S_{61} \cup S_{64} \cup S_{65}$. In this case $j$ and $l$ are both even.
	If $x_l$ lies in $S_{64} \cup S_{65}$, since there is a $k_i$ in some $K_i \neq K_6$ adjacent to both $v_j$ and $x_l$, then we find $F_2(j+l+1)$ contained in the submatrix induced by each row and column on which the rows in $w$ and both sequences are not null and the column representing $k_i$. If instead $x_l$ lies in $S_{61}$, we find $F_2(k+l+3)$ contained in the same submatrix but adding three columns representing vertices $k_i$ in $K_i$ for $i=1,2,4$.
	The same holds if $v_j$ and $x_l$ are both blue.
	Suppose now that $v_j$ is colored with red and $v_l$ is colored with blue. Thus, $j$ is even and $l$ is odd. In this case, we find $F_2(j+l+2)$ contained in the submatrix induced by the row that represents $s_{24}$, two columns representing any two vertices in $K_2$ and $K_4$ and each row and column on which the rows in $w$ and both sequences are not null.
	The proof is analogous if either $j=1$ or $l=1$.	
	
	
	Hence, we may assume there is either a monochromatic weak gem in which one of the rows is an LR-row or a badly-colored doubly-weak gem in $B'$, for the case of a monochromatic gem or a monochromatic weak gem where one of the rows is an L-row (resp.\ R-row) and the other is unlabeled is analogous to the tent case.  
	
	\subcase \textit{Let us suppose there is a monochromatic weak gem in $B'$}, and let $v_1$ and $v_2$ be the rows that induce such gem, where $v_2$ is an LR-row. 
	Suppose first that $v_1$ is a pre-colored row. Suppose without loss of generality that the monochromatic weak gem is induced by $v_1$ and the L-block of $v_2$ and that $v_1$ and $v_2$ are both colored with red. We denote $v_{2L}$ to the L-block of $v_2$.
	If $v_1$ is labeled with R, then $v_2$ is the L-block of some LR-row $r$ in $B$ and $v_1$ is the R-block of itself. However, since the LR-ordering we are considering to color $B'$ is suitable, then the L-block of an LR-row has empty intersection with the R-block of a non-LR row and thus this case is not possible. 

	If $v_1$ is labeled with L, since they induce a weak gem, then $v_{2L}$ is properly contained in $v_1$.
	Since $v_1$ is a row labeled with L in $B$, then $v_1$ is a pre-colored row. Moreover, since $v_{2L}$ is colored with the same color as $v_1$, then there is either a blue pre-colored row, or a sequence of rows $v_3, \ldots, v_j$ where $v_j$ forces the red coloring of $v_{2L}$. 
	In either case, there is a pre-colored row in that sequence that forces the color on $v_{2L}$, and such row is either labeled with L or with R. 
	
	Suppose first that such row is labeled with L. If $v_3$ is a the blue pre-colored row that forces the red coloring on $v_{2L}$, then $v_{2L}$ is properly contained in $v_3$. However, in that case we find $D_4$ which is not possible since $B$ is admissible. 
	Hence, we assume $v_3, \ldots, v_{j-1}$  is a sequence of unlabeled rows and that $v_j$ is a labeled row such that this sequence forces $v_{2L}$ to be colored with red, and each row in the sequence forces the color on both its predecesor and its succesor.
	If $j-3$ is even, then $v_j$ is colored with blue, and if $j-3$ is odd then $v_j$ is colored with red. In either case, we find $S_5(j)$ contained in the submatrix induced by rows $v_1, v_2, v_3, \ldots, v_j$.
	
	If instead the row $v$ that forces the coloring on $v_{2L}$ is labeled with R, since the LR-ordering used to color $B$ is suitable, then the intersection between $v_{2L}$ and $v$ is empty. Hence, $v \neq v_3$, thus we assume that $v_3, \ldots, v_{j-1}$ are unlabeled rows and $v_j = v$. If $j-3$ is odd, then $v_j$ is colored with red, and if $j-3$ is even, then $v_j$ is colored with blue. In either case we also find $S_5(j)$, which is not possible since $B$ is admissible.
		
	Suppose now that $v_1$ is an unlabeled row. Notice that, since $v_1$ and $v_2$ induce a weak gem, then $v_1$ is not nested in $v_2$. 

	Hence, either the coloring of both rows is forced by the same sequence of rows or the coloring of $v_1$ and $v_2$ is forced for each by a distinct sequence of rows. As in the previous cases, we assume that the last row of each sequence represents a pre-colored labeled row. 
	 
	Suppose first that both rows are forced to be colored with red by the same row $v_3$. Thus, $v_3$ is a labeled row pre-colored with blue. Moreover, since $v_3$ forces $v_1$ to be colored with red, then $v_1$ is not contained in $v_3$ and thus there is a column $k_{61}$ in which $v_1$ has a $1$ and $v_3$ has a $0$.
	
We may also assume that $v_2$ has a $0$ in such a column since $v_1$ is also not contained in $v_2$.
Moreover, since $v_3$ forces $v_2$ to be colored with red, then $v_3$ is labeled with the same letter than $v_2$ and $v_3$ is not contained in $v_2$, thus we can find a column $k_{62}$ in which $v_2$ has a $0$ and $v_1$ and $v_3$ both have a $1$.
Furthermore, since $v_3$ and $v_2$ are both labeled with the same letter and the three rows have pairwise nonempty intersection, then there is a column $k_{63}$ in which all three rows have a $1$. Since $v_3$ is a row labeled with either L or R in $B$, then there are vertices $k_l \in K_l$, $k_m \in K_m$ with $l \neq m$, $l,m \neq 6$ such that $v_3$ is adjacent to $k_l$ and nonadjacent to $k_m$. Moreover, since $v_2$ is an LR-row, then $v_2$ is adjacent to both $k_l$ and $k_m$ and $v_j$ is nonadjacent to $k_l$ and $k_m$.
Hence, we find $F_0$ induced by $\{v_3$, $v_1$, $v_2$, $k_l$, $k_{61}$, $k_{63}$, $k_{62}$, $k_m \}$.

Suppose instead there is a sequence of rows $v_3, \ldots, v_j$ that force the coloring of both $v_1$ and $v_2$, where $v_3, \ldots, v_{j-1}$ are unlabeled rows and $v_j$ is labeled with either L or R and is pre-colored. 

We have two possibilities: either $v_j$ is labeled with L or with R.

If $v_j$ is labeled with L and $v_j$ forces the coloring of $v_2$, then we have the same situation as in the previous case. Thus we assume $v_j$ is nested in $v_2$. In this case, since $v_j$ and $v_2$ are labeled with L, the vertices $v_3, \ldots, v_{j-1}$ are nested in $v_2$ and thus they are chained from right to left. Moreover, since $v_1$ and $v_2$ are colored with the same color, then there is an odd index $1 \leq l \leq j-1$ such that $v_1$ contains $v_3, \ldots, v_l$ and does not contain $v_{l+1}, \ldots, v_j$. Hence, we find $F_1(l+1)$ considering the rows $v_1, v_2, \ldots, v_{l+1}$.

Suppose now that $v_j$ is labeled with R. Since $B'$ is colored using a suitable LR-ordering, then $v_j$ and $v_2$ have empty intersection, thus there is a sequence of unlabeled rows $v_3, \ldots, v_{j-1}$, chained from left to right. Notice that it is possible that $v_1= v_3$. 
Suppose first that $j$ is even. If $v_1 = v_3$, then there is an odd number of unlabeled rows between $v_1$ and $v_j$. In this case we find a $(j-2)$-sun contained in the subgraph induced by rows $v_2, v_1=v_3, v_4, \ldots, v_j$. 
If instead $v_1 \neq v_3$, then $v_1$ and $v_3$ and $v_1$ and $v_5$ both induce a $0$-gem, and thus we find a $(j-2)$-sun in the same subgraph.
If $j$ is odd, then there is an even number of unlabeled rows between $v_2$ and $v_j$. Once more, we find a $(j-1)$-sun contained in the subgraph induced by rows $v_2, v_3, \ldots, v_j$.


Notice that these are all the possible cases for a weak gem. This follows from the fact that, if there is a pre-colored labeled row that forces the coloring upon $v_1$ then it forces the coloring upon $v_2$ and viceversa. Moreover, if there is a sequence of rows that force the coloring upon $v_2$, then one of these rows of the sequence also forces the coloring upon $v_1$, and viceversa. 
Furthermore, since the label of the pre-colored row of the sequence determines a unique direction in which the rows overlap in chain, then there is only one possibility in each case, as we have seen in the previous paragraphs.
It follows that the case in which there is a sequence forcing the coloring upon each $v_1$ and $v_2$ can be reduced to the previous case.

	\subcase \textit{Suppose there is a badly-colored doubly-weak gem in $B'$. }
	Let $v_1$ and $v_2$ be the LR-rows that induce the doubly-weak gem. Since the suitable LR-ordering determines the blocks of each LR-row, then the L-block of $v_1$ properly contains the L-block of $v_2$ and the R-block of $v_1$ is properly contained in the R-block of $v_2$, or viceversa. Moreover, the R-block of $v_1$ may be empty. Let us denote $v_{1L}$ and $v_{2L}$ (resp.\ $v_{1R}$ and $v_{2R}$) to the L-blocks (resp.\ R-blocks) of $v_1$ and $v_2$.
	
	There is a sequence of rows that forces the coloring on both LR-rows simultaneously or there are two sequences of rows and each forces the coloring upon the blocks of $v_1$ and $v_2$, respectively.
	Whenever we consider a sequence of rows that forces the coloring upon the blocks of $v_1$ and $v_2$, we will consider a sequence in which every row forces the coloring upon its predecessor and its succesor, a pre-colored row is either the first or the last row of the sequence, the first row of the sequence forces the coloring upon the corresponding block of $v_1$ and the last row forces the coloring upon the corresponding block of $v_2$. It follows that, in such a sequence, every pair of consecutive unlabeled rows overlap.
	We can also assume that there are no blocks corresponding to LR-rows in such a sequence, for we can reduce this to one of the cases.
	
		Suppose first there is a sequence of rows $v_3, \ldots, v_j$ that forces the coloring upon both LR-rows simultaneously. We assume that $v_3$ intersects $v_1$ and $v_j$ intersects $v_2$.
	
	If $v_3, \ldots, v_j$ forces the coloring on both L-blocks, then we have four cases: (1) either $v_3, \ldots, v_j$ are all unlabeled rows, (2) $v_3$ is the only pre-colored row, (3) $v_j$ is the only pre-colored row or (4) $v_3$ and $v_j$ are the only pre-colored rows. 
	In either case, if $v_3, \ldots, v_j$ is a minimal sequence that forces the same color upon both $v_{1L}$ and $v_{2L}$, then $j$ is odd.

	\subsubcase Suppose $v_3, \ldots, v_j$ are unlabeled.
	If $j=3$, then we find $S_7(3)$ contained in the submatrix induced by $v_1$, $v_2$ and $v_3$. Suppose $j>3$, thus we have two possibilities. If $v_2 \cap v_3 \neq \emptyset$, since $j$ is odd, then we find a $(j-1)$-sun contained in the submatrix induced by considering all the rows $v_1$, $v_2$, $v_3, \ldots, v_j$.
	If instead $v_2 \cap v_3 = \emptyset$, then we find $F_2(j)$ contained in the same submatrix.

	\subsubcase Suppose $v_3$ is the only pre-colored row.
	Since $v_3$ is a pre-colored row and forces the color red upon the L-block of $v_1$, then $v_3$ contains $v_{1L}$ and $v_3$ is colored with blue. If $v_4 \cap v_{1L} \neq \emptyset$, then we find $F_0$ in the submatrix given by considering the rows $v_1$, $v_3$, $v_4$, having in mind that there is a column representing some $k_i$ in $K_i \neq K_6$ in which the row corresponding to $v_1$ has a $1$ and the rows corresponding to $v_3$ and $v_4$ both have $0$. This follows since $v_4$ is unlabeled and thus represents a vertex that lies in $S_{66}$, and $v_3$ is pre-colored and labeled with L or R and, thus it represents a vertex that is not adjacent to every partition $K_i$ of $K$. 
	If instead $v_4 \cap v_{1L} = \emptyset$, then we find $F_2(j-2)$ contained in the submatrix induced by the rows $v_1, v_2, \ldots, v_{j-2}$ if $v_2 \cap v_{2R} = \emptyset$, and induced by the rows $v_1, v_2, v_5, \ldots, v_j$ if $v_2 \cap v_{2R} \neq \emptyset$.
	
	\subsubcase \label{subsubcase:B6_4.3.3.} Suppose $v_j$ is the only pre-colored row.
	In this case, $v_j$ properly contains $v_{2L}$ and we can assume that the rows $v_4, \ldots, v_{j-1}$ are contained in $v_{1L}$. If $v_3 \cap v_2 \neq \emptyset$, then we find an even $(j-1)$-sun in the submatrix induced by the rows $v_2, v_3, \ldots, v_j$.
	If instead $v_3 \cap v_2 = \emptyset$, then we find $F_2(j)$ in the submatrix given by rows $v_1, \ldots, v_j$.	
	
	\subsubcase Suppose that $v_3$ and $v_j$ are the only pre-colored rows. Thus, we can assume that $v_j$ properly contains $v_{2L}$ and $v_3$ properly contains $v_{2L}$, thus $v_3$ properly contains $v_{2L}$. Hence, we find $D_9$ induced by the rows $v_1$, $v_2$ and $v_3$ which is not possible since $B$ is admissible.
	
	The only case we have left is when $v_3, \ldots, v_j$ forces the coloring upon $v_{1L}$ and $v_{2R}$. This follows from the fact that, if $v_3, \ldots, v_j$ forces the color upon $v_{2L}$ and $v_{1R} \neq \emptyset$, then this case can be reduced to case (4.3.\ref{subsubcase:B6_4.3.3.}).
	
	 Hence, either (1) $v_3, \ldots, v_j$ are unlabeled rows, (2) $v_3$ is the only pre-colored row, or (3) $v_3$ and $v_j$ are the only pre-colored rows.
	Notice that in either case, $j$ is even and thus for (1) we find $S_8(j)$, which results in a contradiction since $B$ is admissible.
	Moreover, in the remaining cases, $v_3$ properly contains $v_{1L}$ and $v_{2L}$. Since $v_1$ and $v_2$ overlap, we find $D_9$ which is not possible for $B$ is admissible.


	This finishes the proof.

\end{mycases} 	

\end{proof}


Let $G= (K,S)$, $H$ as in Section\ref{sec:4tent_partition} and the matrices $\mathbb B_i$ for $i= \{1 \ldots, 6\}$ as defined in the previous subsection.
Suppose $\mathbb B_i$ is $2$-nested for each $i \in \{1, 2, \ldots, 6\}$. Let $\chi_i$ be a proper $2$-coloring for $\mathbb B_i$ for each $i \in \{1, \ldots, 5\}$ and $\chi_6$ be a proper $2$-coloring for $\mathbb B_6$. Moreover, there is a suitable LR-ordering $\Pi_i$ for each $i \in \{1, 2, \ldots, 6\}$.

Let $\Pi$ be the ordering of the vertices of $K$ given by concatenating the orderings $\Pi_1$, $\Pi_2$, $\ldots$, $\Pi_6$, as defined in Subsection \ref{subsec:tent3}.
Let $s \in S$. Hence, $s$ lies in $S_{ij}$ for some $i,j\in \{1,2,\ldots,6\}$. 
Notice that there are at most two rows $r_1$ in $\mathbb B_i$ and $r_2$ in $\mathbb B_j$ both representing $s$. Also notice that the row $r_l$ represents the adjacencies of $s$ with regard to $K_l$ for each $l=i,j$, and if $i>j$, then $r_i$ and $r_j$ are colored with distinct colors.

\begin{defn} \label{def:matrices_B_por_colores}
We define the $(0,1)$-matrices $\mathbb B_r$, $\mathbb{B}_b$, $\mathbb B_{r-b}$ and $\mathbb B_{b-r}$ as in the previous subsection, considering only those independent vertices that are not in $S_{[16}$.
\end{defn}





Notice that the only nonempty subsets $S_{ij}$ with $i>j$ that we are considering are those with $i=6$. Hence, the rows of $\mathbb B_{r-b}$ are those representing vertices in $S_{61} \cup S_{64} \cup S_{65}$ and the rows of $\mathbb B_{b-r}$ are those representing vertices in $S_{62} \cup S_{63}$.



\begin{lema} \label{lema:matrices_union_son_nested_4tent} 
	Suppose that $\mathbb B_i$ is $2$-nested for each $i =1,2 \ldots, 6$. If $\mathbb B_r$, $\mathbb B_b$, $\mathbb B_{r-b}$ or $\mathbb B_{b-r}$ are not nested, then $G$ contains $F_0$, $F_1(5)$ or $F_2(5)$ as forbidden induced subgraphs for the class of circle graphs.
\end{lema}

\begin{proof}
	Notice that the only partial rows considered in $\mathbb B_r$ and $\mathbb B_b$ may be those in $S_{62} \cup S_{63}$ and $S_{61} \cup S_{64} \cup S_{65}$, respectively. Thus, if the partial row coincides with the row in $\mathbb B_6$ or $\mathbb B_1$, then we can consider the matrices $\mathbb B_r$ and $\mathbb B_b$ without these rows since the compatiblity with the rest of the rows was already considered when analysing if $\mathbb B_6$ and $\mathbb B_1$ are $2$-nested or not.
	
	Suppose first that $\mathbb B_r$ is not nested. Thus, there is a $0$-gem.
	Let $f_1$ and $f_2$ be two rows that induce a gem in $\mathbb B_r$ and $v_1$ in $S_{ij}$ with $i<j$ and $v_2$ in $S_{lm}$ with $l<m$ be the corresponding to vertices in $G$. Suppose without loss of generality that $f_1$ starts before $f_2$, thus $i \geq l$. Since $\mathbb B_i$ is $2$-nested for every $i\in \{1, 2, \ldots, 5,6\}$, in particular there are no monochromatic gems in each $\mathbb B_i$. Moreover, if $j=l$, then we find $D_1$ in $K_i$ or $K_j$, respectively.

Notice that every row in $\mathbb B_r$ represents a vertex that belongs to one of the following subsets of $S$: $S_{12}$, $S_{13}$, $S_{35}$, $S_{36}$, $S_{45}$, $S_{62}$ or $S_{63}$.
Analogously, every row in $\mathbb B_b$ represents a vertex belonging to either $S_{23}$, $S_{24}$, $S_{34}$, $S_{14}$, $S_{25}$, $S_{15}$, $S_{16}$, $S_{61}$, $S_{64}$ or $S_{65}$.
	
	\begin{mycases}
	\case Suppose first that $i=l$. We have two cases:
	\subcase $v_1$, $v_2$ in $S_{12} \cup S_{13}$. Suppose without loss of generality that both vertices lie in $S_{12}$ since the proof is analogous otherwise. Let $k_{ii}$ in $K_i$ such that $v_i$ is adjacent to $k_{ii}$ and $v_{i+1}$ is nonadjacent to $k_{ii}$ for $i=1,2$ (mod 2). Notice that $v_1$ and $v_2$ are labeled with R in $\mathbb B_1$ and are labeled with L in $\mathbb B_2$. Moreover, since $\mathbb B_1$ and $\mathbb B_2$ are admissible, then there are vertices $k_{12}$ in $K_1$ and $k_{21}$ in $K_2$ adjacent to both $v_1$ and $v_2$, for if not we find $D_0$ in each matrix. Moreover, there is a vertex $k_4$ in $K_4$ nonadjacent to both. We find $F_0$ induced by $\{ v_1$, $v_2$, $s_{24}$, $k_{11}$, $k_{12}$, $k_{21}$, $k_{22}$, $k_4 \}$.
	
	The proof is analogous if $v_1$ and $v_2$ in $S_{45} \cup S_{46}$, where $F_0$ is induced by $\{ v_1$, $v_2$, $s_{24}$, $k_2$, $k_{41}$, $k_{42}$, $k_{5}$, $k_6 \}$ or $\{ v_1$, $v_2$, $s_{24}$, $k_2$, $k_{41}$, $k_{42}$, $k_{51}$, $k_{52} \}$, depending on whether only one lies in $S_{46}$ or both lie in $S_{46}$.	
	If $v_1$ in $S_{45} \cup S_{46}$ and $v_2$ in $S_{62} \cup S_{63}$ is the vertex represented by a partial row in $\mathbb B_r$, then it is not possible that these rows induce a gem since they do not intersect. 
	Thus, we assume that $v_1$ in $S_{12} \cup S_{13}$. We find $F_0$ induced by $\{ v_1$, $v_2$, $s_{24}$, $k_{11}$, $k_{12}$, $k_{21}$, $k_{22}$, $k_{4} \}$ if $v_1$ in $S_{12}$ (thus necessarily $v_2$ in $S_{62}$ since they induce a $0$-gem). 
	If instead $v_1$ in $S_{13}$, since $v_1$ is complete to $K_1$, then one of the columns of the $0$-gem is induced by the column $c_L$. Thus, there is a vertex $k_6$ in $K_6$ adjacent to $v_2$ and nonadjacent to $v_1$. Hence, we find $F_0$ induced by $\{ v_1$, $v_2$, $s_{24}$, $k_{6}$, $k_{1}$, $k_{2}$, $k_{3}$, $k_{4} \}$.

	\subcase $v_1$, $v_2$ in $S_{35} \cup S_{36}$. Suppose that $v_1$ in $S_{35}$ and $v_2$ in $S_{36}$. Let $k_2$ in $K_2$ nonadjacent to both. There are vertices $k_{31}$, $k_{32}$ in $K_3$ such that $k_{31}$ is adjacent only to $v_1$ and $k_{32}$ is adjacent to both. Moreover, there are vertices $k_5$ in $K_5$ and $k_6$ in $K_6$ such that $k_5$ is adjacent to both and $k_6$ is adjacent only to $v_2$. We find $F_0$ induced by $\{ v_1$, $v_2$, $s_{24}$, $k_{31}$, $k_{32}$, $k_{5}$, $k_{6}$, $k_2 \}$. The proof is analogous if both lie in $S_{35}$ changing $k_6$ for other vertex in $K_5$ adjacent only to $v_2$ (exists since both rows induce a gem), and if both lie in $S_{36}$ we can find two vertices $k_{61}$ and $k_{62}$ in $K_6$ to replace $k_5$ and $k_6$ in the previous subset. Notice that, if instead $v_1$ in $S_{35} \cup S_{36}$ and $v_2$ in $S_{45} \cup S_{46}$ we also find $F_0$ changing $k_{32}$ for some vertex $k_4$ in $K_4$ in the same subset. This is the only case we had to see in which $j=m$.
	Furthermore, the partial rows corresponding to $S_{62} \cup S_{63}$ cannot induce a gem with a row corresponding to a vertex in $S_{35} \cup S_{36}$ since we aer assuming that $\mathbb B_3$ is admissible.

	\case Suppose now that $i<l$. Since $j \neq l$ and both rows induce a gem, then $i<l<j<m$. Thus, the only possibility is $v_1$ in $S_{35}$ and $v_2$ in $S_{46}$.
	In this case we find $F_0$ induced by $\{ v_1$, $v_2$, $s_{24}$, $k_2$, $k_3$, $k_4$, $k_{5}$, $k_6 \}$.
	\end{mycases}
	
	Hence $\mathbb B_r$ is nested.
	Suppose now that $\mathbb B_b$ is not nested, and let $v_1$ in $S_{ij}$ with $i<j$ and $v_2$ in $S_{lm}$ with $l<m$ two vertices for which its rows in $\mathbb B_b$ induce a $0$-gem. Once more, we assume that $i \leq l$.
	
	\begin{mycases}
		\case Suppose that the gem is induced by two rows corresponding to two vertices $v_1$ and $v_2$ such that $v_2$ is a partial row, thus $v_2$ in $S_{64} \cup S_{65}$. Notice that the $0$-gem may be induced by the column $c_L$.
		\subcase $v_2$ in $S_{64}$. 
		\subsubcase $v_1$ in $S_{24} \cup S_{34} \cup S_{14}$. We find $F_0$ induced by $\{ v_1$, $v_2$, $s_{45}$, $k_1$, $k_{2}$, $k_{41}$, $k_{42}$, $k_{5} \}$. Notice that, since $S_{64}$ is complete to $K_4$, the $0$-gem cannot be induced by $v_2$ and a vertex $v_1$ in $S_{14}$ complete to $K_1$, since we are considering that every vertex in $S_{14}$ is also complete to $K_4$ (for if not we have previously shown a forbidden subgraph).
		\subsubcase $v_1$ in $S_{15} \cup S_{25} \cup S_{16}$. In this case we find $F_1(5)$ induced by $\{ v_1$, $v_2$, $s_{12}$, $s_{24}$, $s_{45}$, $k_1$, $k_{2}$, $k_{4}$, $k_{5} \}$ if $v_1$ in $S_{15}$ is not complete to $K_1$. If instead $v_1$ in $S_{15}$ is complete to $K_1$, then it is not complete to $K_5$ (for we split those vertices that are adjacent to $K_1, \ldots, K_5$ into two disjoint subsets, $S_{[15]}$ and $S_{15}$). Moreover, one of the columns that induce the $0$-gem is the column $c_L$. Thus, there are vertices $k_6$ in $K_6$, $k_{51}$ and $k_{52}$ in $K_5$ such that $v_2$ is adjacent to $k_6$ and is nonadjacent to $k_{51}$ and $k_{52}$ and $v_1$ is adjacent to $k_{51}$ and is nonadjacent to $k_6$ and $k_{52}$. Hence, we find $F_0$ induced by $\{ v_1$, $v_2$, $s_{45}$, $k_{6}$, $k_{2}$, $k_{4}$, $k_{51}$, $k_{52} \}$. 
		
		\subcase $v_2$ in $S_{65}$. In this case, $v_1$ in $S_{25} \cup S_{15} \cup S_{16}$. Since these rows induce a gem and $v_2$ has a $1$ in every column corresponding to $K_1, \ldots, K_4$, there are vertices $k_1$ in $K_1$ and $k_5$ in $K_5$ such that $v_1$ is adjacent to $k_1$ and $v_2$ is nonadjacent to $k_1$, and $v_1$ is nonadjacent to $k_5$ and $v_2$ is adjacent to $k_5$. Thus, we find $F_1(5)$ induced by $\{ v_1$, $v_2$, $s_{12}$, $s_{24}$, $s_{45}$, $k_1$, $k_{2}$, $k_{4}$, $k_{5} \}$.
	
		\case Suppose now that $i=l$.
		\subcase \label{subcase:Bb_1.1}$v_1$, $v_2$ in $S_{23} \cup S_{24} \cup S_{25}$. Suppose first that both lie in $S_{24}$. In that case we find $F_0$ induced by $\{ v_1$, $v_2$, $s_{12}$, $k_{21}$, $k_{22}$, $k_{41}$, $k_{42}$, $k_1 \}$. If instead one of both lie in $S_{23}$, then we change $k_{41}$ for some analogous $k_3$ in $K_3$, and if one of both lie in $S_{25}$ we change $k_{42}$ for some analogous $k_5$ in $K_5$.
		
		\subcase $v_1$, $v_2$ in $S_{34}$. In this case, we find $F_0$ induced by $\{ v_1$, $v_2$, $s_{45}$, $k_{31}$, $k_{32}$, $k_{41}$, $k_{42}$, $k_5 \}$.
		
		\subcase $v_1$, $v_2$ in $S_{14} \cup S_{15} \cup S_{16}$. Remember that $S_{15}$ are those independent vertices that are not complete to $K_5$ and $S_{16}$ are those independent vertices that are not complete to $K_1$.

		\subsubcase If both lie in $S_{14}$, then we find $F_0$ induced by $\{ v_1$, $v_2$, $s_{24}$, $k_{11}$, $k_{12}$, $k_{41}$, $k_{42}$, $k_5 \}$. 
		
		\subsubcase If $v_1$ in $S_{14}$ and $v_2$ in $S_{15}$, then we find $F_1(5)$ induced by $\{ v_1$, $v_2$, $s_{12}$, $s_{24}$, $s_{45}$, $k_{1}$, $k_{2}$, $k_{4}$, $k_{5} \}$. The same holds if instead $v_2$ in $S_{16}$ or if both lie in $S_{15}$. Moreover, we find the same subgraph induced by the same subset if $v_1$ in $S_{15}$ and $v_2$ in $S_{16}$, since there is a vertex in $K_5$ that is nonadjacent to $v_1$. 
		
		\subsubcase If both lie in $S_{16}$, then we find $F_0$ induced $\{ v_1$, $v_2$, $s_{12}$, $k_{11}$, $k_{12}$, $k_{2}$, $k_{4}$, $k_6 \}$.
		
		\case Suppose now that $j=m$. The case where $v_1$, $v_2$ in $S_{14} \cup S_{24} \cup S_{34}$ is analogous as Case \ref{subcase:Bb_1.1}. Let $v_1$ in $S_{15}$ and $v_2$ in $S_{25}$. We find $F_1(5)$ induced by $\{ v_1$, $v_2$, $s_{12}$, $s_{24}$, $s_{45}$, $k_1$, $k_{2}$, $k_{4}$, $k_{5} \}$.
		
		\case Suppose that $i<l$, thus $i<l<j<m$. In this case, $v_1$ in $S_{14}$ and $v_2$ in $S_{25}$. We find $F_1(5)$ induced by $\{ v_1$, $v_2$, $s_{12}$, $s_{24}$, $s_{45}$, $k_1$, $k_{2}$, $k_{4}$, $k_{5} \}$. 		
		\end{mycases}
	
	Hence $\mathbb B_b$ is nested.
	Suppose that $\mathbb B_{b-r}$ is not nested, thus let $v_1$ and $v_2$ in $S_{62} \cup S_{63}$ two vertices whose rows induce a $0$-gem. If both lie in $S_{62}$, then we find $F_0$ induced by $\{ v_1$, $v_2$, $s_{24}$, $k_{61}$, $k_{62}$, $k_{21}$, $k_{22}$, $k_4 \}$. If instead one or both lie in $S_{63}$, we find the same subgraph changing $k_{22}$ for some analogous $k_3$ in $K_3$.
	
	Finally, suppose that $\mathbb B_{r-b}$ is not nested, and let $v_1$ and $v_2$ in $S_{61} \cup S_{64} \cup S_{65}$ be two vertices whose rows induce a $0$-gem.
	If both lie in $S_{61}$, then we find $F_0$ induced by $\{ v_1$, $v_2$, $s_{12}$, $k_{61}$, $k_{62}$, $k_{11}$, $k_{12}$, $k_2 \}$. Similarly, we find $F_0$ induced by $\{ v_1$, $v_2$, $s_{45}$, $k_{61}$, $k_{2}$, $k_{4}$, $k_{51}$, $k_{52} \}$ if $v_1$ in $S_{64}$ and $v_2$ in $S_{65}$ or if both lie in $S_{64}$, changing $k_{51}$ for an analogous vertex $k'_4$ in $K_4$.
	If instead $v_1$ in $S_{61}$ and $v_2$ in $S_{64} \cup S_{65}$, then we find $F_2(5)$ induced by $\{ v_1$, $v_2$, $s_{12}$, $s_{24}$, $s_{45}$, $k_{61}$, $k_{1}$, $k_{2}$, $k_{4}$, $k_5 \}$.

\end{proof}

\begin{teo} \label{teo:finalteo_4tent}
	Let $G=(K,S)$ be a split graph containing an induced $4$-tent. Then, $G$ is a circle graph if and only if $\mathbb B_1,\mathbb B_2,\ldots,\mathbb B_6$ are $2$-nested and $\mathbb B_r$, $\mathbb B_b$, $\mathbb B_{r-b}$ and $\mathbb B_{b-r}$ are nested.
\end{teo}

\begin{proof}

Necessity is clear by the previous lemmas and the fact that the graphs in families $\mathcal{T}$ and $\mathcal{F}$ are all non-circle. Suppose now that each of the matrices $\mathbb B_1,\mathbb B_2,\ldots,\mathbb B_6$ is $2$-nested and the matrices $\mathbb B_r$, $\mathbb B_b$, $\mathbb B_{r-b}$ or $\mathbb B_{b-r}$ are nested.
Let $\Pi$ be the ordering for all the vertices in $K$ obtained by concatenating each suitable LR-ordering $\Pi_i$ for $i \in \{1, 2,\ldots, 6\}$.

\begin{figure}[h!]
    \includegraphics[width=\textwidth]{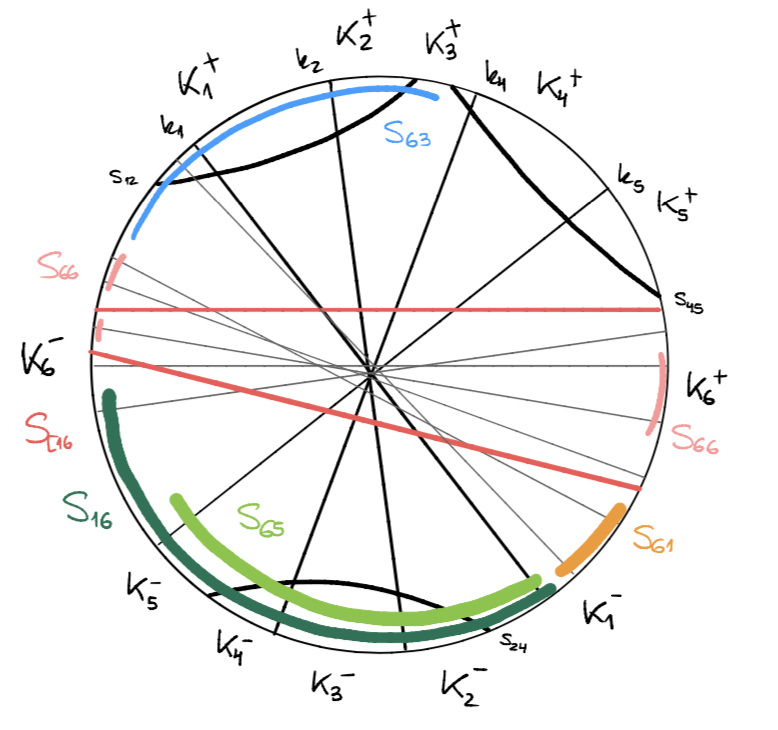}   
	\caption{Sketch model of $G$ with some of the chords associated to the rows in $\mathbb{B}_6$.}
     \label{fig:4tent_model} 
\end{figure}

Consider the circle divided into twelve pieces as in Figure \ref{fig:4tent_model}. For each $i\in\{1,2,\ldots,6\}$ and for each vertex $k_i \in K_i$ we place a chord having one endpoint in $K_i^+$ and the other endpoint in $K_i^-$, in such a way that the ordering of the endpoints of the chords in $K_i^+$ and $K_i^-$ is $\Pi_i$. 

Let us see how to place the chords for each subset $S_{ij}$ of $S$. First, some useful remarks.

\begin{remark} The following assertions hold:
	\begin{itemize}
		\item By Lemma \ref{lema:matrices_union_son_nested_4tent}, all the vertices in $S_{ij}$ are nested, for every pair $i, j = \{1, 2, \ldots, 6\}$, $i \neq j$. This follows since any two vertices in $S_{ij}$ are nondisjoint. Moreover, in each $S_{ij}$, all the vertices are colored with either one color (the same), or they are colored red-blue or blue-red. Hence, these vertices are represented by rows in the matrices $\mathbb B_{r-b}$ and $\mathbb B_{b-r}$ and therefore they must be nested since each of these matrices is a nested matrix.
		\item As a consequence of the previous and Claim \ref{claim:tent_0}, if $i \geq k$ and $j\leq l$, then every vertex in $S_{ij}$ is nested in every vertex of $S_{kl}$.
		\item Also as a consequence of the previous and Lemma \ref{lema:matrices_union_son_nested_4tent}, if we consider only those vertices labeled with the same letter in some $\mathbb B_i$, then there is a total ordering of these vertices. This follows from the fact that, if two vertices $v_1$ and $v_2$ are labeled with the same letter in some $\mathbb B_i$, since $\mathbb B_i$ is --in particular-- admissible, then they are nested in $K_i$. Moreover, if $v_1$ and $v_2$ are labeled with L in $\mathbb B_i$, then they are either complete to $K_{i-1}$ or labeled with R in $\mathbb B_{i-1}$. Thus, there is an index $j_l$ such that $v_i$ is labeled with R  in $\mathbb B_{j_l}$, for $l=1,2$. Therefore, we can find in such a way a total ordering of all these vertices.
		\item If $v_1$ and $v_2$ are labeled with distinct letters in some $\mathbb B_i$, then they are either disjoint in $K_i$ (if they are colored with the same color) or $N_{K_i}(v_1) \cup N_{K_i}(v_2) = K_i$  (if they are colored with distinct colors), for there are no $D_1$ or $D_2$ in $\mathbb B_i$ for all $i \in \{1, 2, \ldots, 6\}$.
	\end{itemize}
\end{remark}


Notice that, when we define the matrix $\mathbb B_6$, we pre-color every vertex in $S_{[15]}$ with the same color. Since, we are assuming $\mathbb B_6$ is $2$-nested and thus in particular is admissible, the subset $S_{[15]} \neq \emptyset$ if and only if the vertices represented in $\mathbb B_6$ are either all vertices in $S_{66} \cup S_{[16}$ and vertices that are represented by labeled rows $r$, all of them colored and labeled with the same color and letter L or R.

Moreover, since $\mathbb B_6$ is admissible, the sets $N_{K_6}(S_{6i}) \cap N_{K_6}(S_{j6})$ are empty, for $i = 1, 4, 5$ and $j =3,4$. The same holds for the sets $N_{K_6}(S_{6i}) \cap N_{K_6}(S_{j6})$, for $i = 2, 3$ and $j =2, 5, 1$.

If $S_{[16}= \emptyset$, then the placing of the chords that represent vertices with one or both endpoints in $K_6$ is very similar as in the tent case. Suppose that $S_{[16} \neq \emptyset$.

\vspace{1mm}
Before proceeding with the guidelines to draw the circle model, we have some remarks on the relationship between the vertices in $S_{ij}$ with either $i=6$ or $j=6$, and those vertices in $S_{[16}$. This follows from the proof of Lemma \ref{lema:B6_2nested_4tent}: 

\begin{remark} \label{obs:4tent_guidelines_model}
Let $G$ be a circle graph that contains no induced tent but contains an induced $4$-tent, and such that each matrix $\mathbb B_i$ is $2$-nested for every $i=1,2, \ldots, 6$. Then, all of the following statements hold:
\begin{itemize}
\item If $S_{26} \cup S_{16} \neq \emptyset$, then $S_{64} \cup S_{65} = \emptyset$, and viceversa.
\item If $S_{36} \cup S_{46} \neq \emptyset$, then $S_{61} \cup S_{64} \cup S_{65} = \emptyset$, and viceversa.
\item If $S_{56} \cup S_{26} \cup S_{16} \neq \emptyset$, then $S_{62} \cup S_{63} = \emptyset$, and viceversa.	
\item If $S_{56} \neq \emptyset$, then $S_{65} = \emptyset$.
\end{itemize}
\end{remark}

Let $v$ in $S_{ij} \neq S_{[16}$ and $w$ in $S_{[16}$, with either $i=6$ or $j=6$.
Suppose first that $i=j=6$. Since $\mathbb B_6$ is $2$-nested, the submatrix induced by the rows that represent $v$ and $w$ in $\mathbb B_6$ contains no monochromatic gems or monochromatic weak gems.
If instead $i<j$, since $\mathbb B_6$ is admissible, then the submatrix induced by the rows that represent $v$ and $w$ in $\mathbb B_6$ contains no monochromatic weak gem, and thus we can place the endpoint of $w$ corresponding to $K_6$ in the arc portion $K^+_6$ and the $K_6$ endpoint of $v$ in $K^-_6$, or viceversa.


Remember that, since we are considering a suitable LR-ordering, there is an L-row $m_L$ such that any L-row and every L-block of an LR-row are contained in $m_L$ and every R-row and R-block of an LR-row are contained in the complement of $m_L$. Moreover, since we have a block bi-coloring for $\mathbb B_6$, then for each LR-row one of its blocks is colored with red and the other is colored with blue. 
Hence, for any LR-row, we can place one endpoint in the arc portion $K^+_6$ using the ordering given for the block that colored with red, and the other endpoint in the arc portion $K^-_6$ using the ordering given for the block that is colored with blue.

Notice that, if $\mathbb B_6$ is $2$-nested, then all the rows labeled with L (resp.\ R) and colored with the same color and those L-blocks (resp.\ R-blocks) of LR-rows are nested. In particular, the L-block (resp.\ R-block) of every LR-row contains all the L-blocks of those rows labeled with L (resp.\ R) that are colored with the same color. 
Equivalently, let $r$ be an LR-row in $\mathbb B_6$ with its L-block $r_L$ colored with red and its R-block $r_R$ colored with blue, $r'$ be a row labeled with L and $r''$ be a row labeled with R. Hence, if $r_L$, $r'$ and $r''$ are colored with the same color, then $r$ contains $r'$ and $r \cap r'' = \emptyset$. This holds since we are considering a suitable LR-ordering and a total block bi-coloring of the matrix $\mathbb B_6$, thus it contains no $D_0$, $D_1$, $D_2$, $D_8$ or $D_9$.

Since every matrix $\mathbb B_r$, $\mathbb B_b$, $\mathbb B_{r-b}$ and $\mathbb B_{b-r}$ are nested, there is a total ordering for the rows in each of these matrices. Hence, there is a total ordering for all the rows that intersect that are colored with the same color, or with red-blue or with blue-red, respectively. Moreover, if $v$ and $w$ are two vertices in $S$ such that they both have rows representing them in one of these matrices --hence, they are colored with the same color or sequence of colors--, then either $v$ and $w$ are disjoint or they are nested. 


With this in mind, we give guidelines to build a circle model for $G$.

\vspace{2mm}
We place first the chords corresponding to every vertex in $K$, using the ordering $\Pi$. 
For each subset $S_{ij}$, we order its vertices with the inclusion ordering of the neighbourhoods in $K$ and the ordering $\Pi$. When placing the chords corresponding to the vertices of each subset, we do it from lowest to highest according to the previously stated ordering given for each subset.

Notice that there are no other conditions besides being disjoint or nested outside each of the following subsets: $S_{11}$, $S_{22}$, $S_{33}$, $S_{44}$, $S_{55}$, $S_{66}$.
For the subset $S_{12}$, we only need to consider if every vertex in $S_{12} \cup S_{11} \cup S_{22}$ are disjoint or nested. The same holds for the subsets $S_{24}$, $S_{45}$, considering $S_{22} \cup S_{44}$ and $S_{44} \cup S_{55}$, respectively. 

Since each matrix $\mathbb B_i$ is $2$-nested for every $i=1, 2, \ldots, 6$, if there are vertices in both $S_{23}$ and $S_{34}$, then they are disjoint in $K_3$. The same holds for vertices in $S_{62}$ and $S_{63}$, and $S_{61}$ and $S_{14} \cup S_{15} \cup S_{16}$.
This is in addition to every property seen in Remark \ref{obs:4tent_guidelines_model}.

First, we place those vertices in $S_{ii}$ for each $i=1,2, \ldots, 6$, considering the ordering given by inclusion. If $v$ in $S_{ii}$ and the row that represents $v$ is colored with red, then both endpoints of the chord corresponding to $v$ are placed in $K_i^+$. If instead the row is colored with blue, then both endpoints are placed in $K_i^-$. 

For each $v$ in $S_{ij} \neq S_{[16}$, if the row that represents $v$ in $\mathbb B_i$ is colored with red (resp.\ blue) , then we place the endpoint corresponding to $K_i$ in the portion $K_i^+$ (resp.\ $K_i^-$) . We apply the same rule for the endpoint corresponding to $K_j$.

Let us consider now the vertices in $S_{[15]}$. If $G$ is $\{ \mathcal{T}, \mathcal{F} \}$-free, then all the rows in $\mathbb B_6$ are colored with the same color. Moreover, if $S_{[15]} \neq \emptyset$, then either every row labeled with L or R in $\mathbb B_6$ is labeled with L and colored with red or labeled with R and colored with blue, or viceversa. 
Suppose first that every row labeled with L or R in $\mathbb B_6$ is labeled with L and colored with red or labeled with R and colored with blue. In that case, every row representing a vertex $v$ in $S_{[15]}$ is colored with blue, hence we place one endpoint of the chord corresponding to $v$ in $K_6^+$ and the other endpoint in $K_6^-$. In both cases, the endpoint of the chord corresponding to $v$ is the last chord of an independent vertex that appears in the portion of $K_6^+$ and is the first chord of an independent vertex that appears in the portion of $K_6^-$. We place all the vertices in $S_{[15]}$ in such a manner.
If instead every row labeled with L or R in $\mathbb B_6$ is labeled with L and colored with blue or labeled with R and colored with red, then every row representing a vertex in $S_{[15]}$ is colored with red. We place the endpoints of the chord in $K_6^-$ and $K_6^+$, as the last and first chord that appears in that portion, respectively.

Finally, let us consider now a vertex $v$ in $S_{[16}$. Here we have two possibilities: (1) the row that represents $v$ has only one block, (2) the row that represents the row that represents $v$ has two blocks of $1$'s.
Let us consider the first case. If the row that represents $v$ has only one block, then it is either an L-block or an R-block. Suppose that it is an L-block. If the row in $\mathbb B_6$ is colored with red, then we place one endpoint of the chord as the last of $K_6^-$ and the other endpoint in $K_6^+$, considering in this case the partial ordering given for every row that has an L-block colored with red in $\mathbb B_6$. 
If instead the row in $\mathbb B_6$ is colored with blue, then we place one endpoint of the chord as the first of $K_6^+$ and the other endpoint in $K_6^-$, considering in this case the partial ordering given for every row that has an L-block colored with blue in $\mathbb B_6$. The placement is analogous for those LR-rows that are an R-block.

Suppose now that the row that represents $v$ has an L-block $v_L$ and an R-block $v_R$. If $v_L$ is colored with red, then $v_R$ is colored with blue. We place one endpoint of the chord in $K_6^+$, considering the partial ordering given by every row that has an L-block colored with red in $\mathbb B_6$, and the other enpoint of the chord in $K_6^-$, considering the partial ordering given by every row that has an R-block colored with blue in $\mathbb B_6$. The placement is analogous if $v_L$ is colored with blue.

This gives a circle model for the given split graph $G$.

\end{proof}

\section{Split circle graphs containing an induced co-4-tent} \label{sec:circle4}


In this section we will address the last case of the proof of Theorem \ref{teo:circle_split_caract}, which is the case where $G$ contains an induced co-$4$-tent. This case is mostly similar to the $4$-tent case, with one particular difference: the co-$4$-tent is not a prime graph, and thus there is more than one possible circle model for this graph. 

This section is subdivided as follows. In Subsection \ref{subsec:co4tent1}, we define the matrices $\mathbb{C}_i$ for each $i=1, 2, \ldots, 8$ and prove some properties that will be useful further on. 
In Subsection \ref{subsec:co4tent2} we prove the necessity of the $2$-nestedness of each $\mathbb C_i$ for $G$ to be a circle graph and give the guidelines to draw a circle model for a split graph $G$ containing an induced co-$4$-tent in Theorem \ref{teo:finalteo_co4tent}.

\subsection{Matrices $\mathbb C_1,\mathbb C_2,\ldots,\mathbb C_8$}\label{subsec:co4tent1}

Let $G=(K,S)$ and $H$ as in Section \ref{sec:co4tent_partition}. For each $i\in\{1,2,\ldots,8\}$, let $\mathbb C_i$ be a $(0,1)$-matrix having one row for each vertex $s\in S$ such that $s$ belongs to $S_{ij}$ or $S_{ji}$ for some $j\in\{1,2,\ldots,8\}$ and one column for each vertex $k\in K_i$ and such that such that the entry corresponding to row $s$ and column $k$ is $1$ if and only if $s$ is adjacent to $k$ in $G$. 
For each $j\in\{1,2,\ldots,8\}-\{i\}$, we label those rows corresponding to vertices of $S_{ji}$ with L and those corresponding to vertices of $S_{ij}$ with R, with the exception of those rows in $\mathbb{C}_7$ that represent vertices in $S_{76]}$ and $S_{[86]}$which are labeled with LR.
Notice that we have considered those vertices that are complete to $K_1, \ldots, K_5$ and $K_8$ and are also adjacent to $K_6$ and $K_7$ divided into two distinct subsets. Thus, $S_{76}$ are those vertices that are not complete to $K_6$ and therefore the corresponding rows are labeled with R in $\mathbb C_6$ and with L in $\mathbb C_7$.
As in the $4$-tent case, there are LR-rows in $\mathbb C_7$. Moreover, there may be some empty LR-rows, which represent those independent vertices that are complete to $K_1, \ldots, K_6$ and $K_8$ and are anticomplete to $K_7$. These vertices are all pre-colored with the same color, and that color is assigned depending on whether $S_{74} \cup S_{75} \cup S_{76} \neq \emptyset$ or $S_{17} \cup S_{27} \neq\emptyset$. 

We color some of the remaining rows of $\mathbb C_i$ as we did in the previous sections, to denote in which portion of the circle model the chords have to be drawn. 
In order to characterize the forbidden induced subgraphs of $G$ and using an argument of symmetry, we will only analyse the properties of the matrices $\mathbb C_1$, $\mathbb C_2$, $\mathbb C_3$, $\mathbb C_6$ and $\mathbb C_7$, since the matrices $\mathbb C_i$ $i=4,5,8$ are symmetric to $\mathbb C_2$, $\mathbb C_3$ and $\mathbb C_6$, respectively.


We will consider 5 distinct cases, according to whether the subsets $K_6$, $K_7$ and $K_8$ are empty or not, for the matrices we need to define may be different in each case.

Using the symmetry of the subclasses $K_6$ and $K_8$, the cases we need to study are the following: (1) $K_6, K_7, K_8 \neq \emptyset$, (2) $K_6, K_7 \neq \emptyset$, $K_8 = \emptyset$, (3) $K_6, K_8 \neq \emptyset$, $K_7 = \emptyset$, (4) $K_6 \neq \emptyset$, $K_7, K_8 = \emptyset$, (5) $K_7 \neq \emptyset$, $K_6, K_8 = \emptyset$

In (1), the subsets are given as described in Table \ref{fig:tabla_co4tent_1}, and thus the matrices we need to analyse are as follows:
\[ \mathbb C_1 = \bordermatrix{ & K_1\cr
				S_{12}\ \textbf{L} & \cdots \cr
		         S_{11}\            & \cdots \cr
                S_{16]}\ \textbf{L} & \cdots \cr
                S_{17}\ \textbf{L} & \cdots }\
                \begin{matrix}
                \textcolor{red}{\bullet} \\ \\ \textcolor{blue}{\bullet} \\ \textcolor{blue}{\bullet} 
                \end{matrix} \qquad
   \mathbb C_2 = \bordermatrix{ & K_2\cr
				S_{12}\ \textbf{R} & \cdots \cr
                S_{22}\            & \cdots \cr
                S_{23}\ \textbf{L} & \cdots \cr
                S_{25]}\ \textbf{L} & \cdots \cr
                S_{26}\ \textbf{L} & \cdots }\
                \begin{matrix}
                \textcolor{red}{\bullet} \\ \\ \textcolor{blue}{\bullet} \\ \textcolor{blue}{\bullet} \\ \textcolor{blue}{\bullet} 
                \end{matrix} \qquad
   \mathbb C_3 = \bordermatrix{ & K_3\cr
                S_{13}\ \textbf{R} & \cdots \cr
                S_{34}\ \textbf{L} & \cdots \cr
                S_{33}\            & \cdots \cr
                S_{35}\ \textbf{L} & \cdots \cr
                S_{36}\ \textbf{L} & \cdots \cr
                S_{23}\ \textbf{R} & \cdots }\
                \begin{matrix}
                \textcolor{red}{\bullet} \\  \textcolor{red}{\bullet} \\ \\ \textcolor{blue}{\bullet} \\ \textcolor{blue}{\bullet} \\ \textcolor{blue}{\bullet} 
                \end{matrix}  \]
                
\[   \mathbb C_6 = \bordermatrix{ & K_6\cr
                S_{66}\            & \cdots \cr
                S_{26}\ \textbf{R} & \cdots \cr
                S_{36}\ \textbf{R} & \cdots \cr
                S_{[46}\ \textbf{R} & \cdots \cr
                S_{76}\ \textbf{R} & \cdots \cr
                S_{[86}\ \textbf{R} & \cdots }\
                \begin{matrix}
                \\ \\  \textcolor{blue}{\bullet} \\  \textcolor{blue}{\bullet} \\ \textcolor{blue}{\bullet} \\ \textcolor{red}{\bullet} \\ \textcolor{red}{\bullet} \\ \\ 
                \end{matrix}  \qquad
    \mathbb C_7 = \bordermatrix{ & K_7\cr
				S_{17}\ \textbf{R} & \cdots \cr
				S_{[27}\ \textbf{R} & \cdots \cr
                S_{77}\            & \cdots \cr
                S_{74]}\ \textbf{L} & \cdots \cr
                S_{75}\ \textbf{L} & \cdots \cr
                S_{76}\ \textbf{L} & \cdots \cr
                S_{87}\ \textbf{R} & \cdots \cr
                S_{[86]}\ \textbf{LR} & \cdots \cr
                S_{76]}\ \textbf{LR} & \cdots }\
                \begin{matrix}
                \textcolor{blue}{\bullet} \\  \textcolor{blue}{\bullet} \\ \\ \textcolor{blue}{\bullet} \\ \textcolor{blue}{\bullet} \\ \textcolor{blue}{\bullet}  \\ \textcolor{blue}{\bullet} \\ \\ \\
                \end{matrix}  \]

In (2), the matrices $\mathbb C_1$, $\mathbb C_2$ and $\mathbb C_3$ are analogous. The subclasses $S_{[15}$ and $S_{[16}$ may be nonempty and are analogous to the subclasses $S_{[85}$ and $S_{[86}$, respectively. Moreover, the vertices in $S_{[16]}$ are analogous to those vertices in $S_{[86]}$, which are represented as empty LR-rows in $\mathbb C_7$. 

\begin{figure}[h!]	 
\begin{center}
	\begin{tabular}{ c | c c c c c c c c} 
		 \hline
		 $i\setminus j$ & 1 & 2 & 3 & 4 & 5 & 6 & 7 \\ 
		  \hline
		 1 & \checkmark & \checkmark & \checkmark & \checkmark & \checkmark & \checkmark & \checkmark \\ 
		 2 & $\emptyset$ & \checkmark & \checkmark & $\emptyset$ & \checkmark & \checkmark & \checkmark  \\
 		 3 & $\emptyset$ & $\emptyset$ & \checkmark & \checkmark & \checkmark & \checkmark & $\emptyset$  \\
		 4 & $\emptyset$ & $\emptyset$ & $\emptyset$ & \checkmark & \checkmark & \checkmark & $\emptyset$ \\
		 5 & $\emptyset$ & $\emptyset$ & $\emptyset$ & $\emptyset$ & \checkmark & $\emptyset$ & $\emptyset$ \\
		 6 & $\emptyset$  & $\emptyset$  & $\emptyset$  & $\emptyset$  & $\emptyset$ & \checkmark & $\emptyset$ \\
		7 & $\emptyset$ & $\emptyset$ & $\emptyset$ & \checkmark & \checkmark & \checkmark & \checkmark \\
	\end{tabular}
\end{center} 
\caption{The nonempty subsets $S_{ij}$ in case (2) $K_6, K_7 \neq \emptyset$, $K_8 = \emptyset$.} \label{fig:tabla_co4tent_2}
\end{figure}

For its part, the matrices $\mathbb C_6$ and $\mathbb C_7$ are as follows:
                
\[   \mathbb C_6 = \bordermatrix{ & K_6\cr
                S_{66}\            & \cdots \cr
                S_{26}\ \textbf{R} & \cdots \cr
                S_{36}\ \textbf{R} & \cdots \cr
                S_{[46}\ \textbf{R} & \cdots \cr
                S_{76}\ \textbf{R} & \cdots \cr
                S_{[16}\ \textbf{R} & \cdots }\
                \begin{matrix}
                \\ \\  \textcolor{blue}{\bullet} \\  \textcolor{blue}{\bullet} \\ \textcolor{blue}{\bullet} \\ \textcolor{red}{\bullet} \\ \textcolor{red}{\bullet} \\ \\ 
                \end{matrix}  \qquad
    \mathbb C_7 = \bordermatrix{ & K_7\cr
				S_{17}\ \textbf{R} & \cdots \cr
				S_{[27}\ \textbf{R} & \cdots \cr
                S_{77}\            & \cdots \cr
                S_{74]}\ \textbf{L} & \cdots \cr
                S_{75}\ \textbf{L} & \cdots \cr
                S_{76}\ \textbf{L} & \cdots \cr
                S_{[16]}\ \textbf{LR} & \cdots \cr
                S_{76]}\ \textbf{LR} & \cdots }\
                \begin{matrix}
                \textcolor{blue}{\bullet} \\  \textcolor{blue}{\bullet} \\ \\ \textcolor{blue}{\bullet} \\ \textcolor{blue}{\bullet} \\ \textcolor{blue}{\bullet}  \\ \\ \\
                \end{matrix}  \]
                
Therefore this case can be considered as a particular case of case (1).

If instead we are in case (3), then the matrices $\mathbb C_2$ and $\mathbb{C}_3$ are analogous as in (1). In this case there are no LR-vertices in any of the matrices. 

\begin{figure}[h!]	 
\begin{center}
	\begin{tabular}{ c | c c c c c c c c} 
		 \hline
		 $i\setminus j$ & 1 & 2 & 3 & 4 & 5 & 6 & 8 \\ 
		  \hline
		 1 & \checkmark & \checkmark & \checkmark & \checkmark & $\emptyset$ & \checkmark & \checkmark \\ 
		 2 & $\emptyset$ & \checkmark & \checkmark & $\emptyset$ & \checkmark & \checkmark & $\emptyset$ \\
 		 3 & $\emptyset$ & $\emptyset$ & \checkmark & \checkmark & \checkmark & \checkmark & $\emptyset$  \\
		 4 & $\emptyset$ & $\emptyset$ & $\emptyset$ & \checkmark & \checkmark & \checkmark & $\emptyset$ \\
		 5 & $\emptyset$ & $\emptyset$ & $\emptyset$ & $\emptyset$ & \checkmark & $\emptyset$ & $\emptyset$ \\
		 6 & $\emptyset$  & $\emptyset$  & $\emptyset$  & $\emptyset$  & $\emptyset$ & \checkmark & $\emptyset$ \\
		8 & $\emptyset$ & \checkmark & \checkmark & \checkmark & \checkmark & \checkmark & \checkmark \\
	\end{tabular}
\end{center} 
\caption{The nonempty subsets $S_{ij}$ in case (3) $K_6, K_8 \neq \emptyset$, $K_7 = \emptyset$.} \label{fig:tabla_co4tent_3}
\end{figure}
For its part, the matrices $\mathbb C_1$ and $\mathbb C_6$ are as follows:

\[  \mathbb C_1 = \bordermatrix{ & K_1\cr
				S_{12}\ \textbf{L} & \cdots \cr
		         S_{11}\            & \cdots \cr
                S_{16]}\ \textbf{L} & \cdots }\
                \begin{matrix}
                \textcolor{red}{\bullet} \\ \\ \textcolor{blue}{\bullet} \\ 
                \end{matrix} \qquad
     \mathbb C_6 = \bordermatrix{ & K_6\cr
                S_{66}\            & \cdots \cr
                S_{26}\ \textbf{R} & \cdots \cr
                S_{36}\ \textbf{R} & \cdots \cr
                S_{46}\ \textbf{R} & \cdots \cr
                S_{[86}\ \textbf{R} & \cdots }\
                \begin{matrix}
                \\ \\  \textcolor{blue}{\bullet} \\  \textcolor{blue}{\bullet} \\ \textcolor{blue}{\bullet} \\ \textcolor{red}{\bullet} \\ \\ 
                \end{matrix}  \qquad \]

In case (4), the matrices $\mathbb{C}_1$, $\mathbb C_2$ and $\mathbb C_3$ are analogous as in case (3). There is no matrix $\mathbb C_7$ and thus there are no LR-vertices. Notice that the subset $S_{15}$ contains only vertices that are complete to $K_1$ and thus $S_{15} = S_{[15}$. Furthermore, this subset is equivalent to $S_{[85}$ in case (1).
Moreover, in this case, the vertices in $S_{[16}$ in $\mathbb C_6$ are analogous as those vertices in $S_{[86}$ and thus the matrix $\mathbb C_6$ results analogous as in case (3).
Also notice that those vertices in $S_{[16]}$ can be placed all having one endpoint in the arc $s_{13} s_{35}$ and the other in $k_1 k_3$. It follows that $S_{54} = S_{[54]}$, thus these vertices are complete to $K$ and hence $S_{[54]}= S_{[16]}$. Moreover, those vertices in $S_{65}$ are complete to $K_5$ and thus we can consider $S_{65} = \emptyset$ and $S_{65} = S_{[16}$.

\begin{figure}[h!]	 
\begin{center}
	\begin{tabular}{ c | c c c c c c c c} 
		 \hline
		 $i\setminus j$ & 1 & 2 & 3 & 4 & 5 & 6 \\ 
		  \hline
		 1 & \checkmark & \checkmark & \checkmark & \checkmark & \checkmark & \checkmark \\ 
		 2 & $\emptyset$ & \checkmark & \checkmark & $\emptyset$ & \checkmark & \checkmark \\
 		 3 & $\emptyset$ & $\emptyset$ & \checkmark & \checkmark & \checkmark & \checkmark \\
		 4 & $\emptyset$ & $\emptyset$ & $\emptyset$ & \checkmark & \checkmark & \checkmark \\
		 5 & $\emptyset$ & $\emptyset$ & $\emptyset$ & $\emptyset$ & \checkmark & $\emptyset$ \\
		 6 & $\emptyset$  & $\emptyset$  & $\emptyset$  & $\emptyset$  & $\emptyset$ & \checkmark \\
	\end{tabular}
\end{center} 
\caption{The nonempty subsets $S_{ij}$ in case (4) $K_6 \neq \emptyset$, $K_7, K_8 = \emptyset$.} \label{fig:tabla_co4tent_4}
\end{figure}

Finally, let us consider case (5). When considering those vertices in $S_{54}$, it follows easily that $S_{54} = S_{54]}$ and thus these vertex subset is equivalent to those vertices in $S_{75}$ (in case (1) ) that are complete to $K_7$. Hence, we consider these vertices as in $S_{75}$ and $S_{54} = \emptyset$.
The subset $S_{15}$ of vertices of $S$ is split in three distinct subsets: $S_{15]}$, $S_{[15}$ and $S_{[15]}$. The rows representing vertices in $S_{15]}$ are pre-colored with blue and labeled with L, only in $\mathbb C_1$, and are equivalent to those vertices in $S_{16]}$ in case (1). For their part, the rows that represent $S_{[15}$ are pre-colored with red and labeled with R, and they appear only in $\mathbb C_5$. These rows are equivalent to those in $S_{[85}$ in case (1). Finally, the vertices in $S_{[15]}$ are represented by uncolored empty LR-rows in $\mathbb C_7$, resulting equivalent to those vertices in $S_{[86]}$ in case (1).

\begin{figure}[h!]	 
\begin{center}
	\begin{tabular}{ c | c c c c c c c c} 
		 \hline
		 $i\setminus j$ & 1 & 2 & 3 & 4 & 5 & 7 \\ 
		  \hline
		 1 & \checkmark & \checkmark & \checkmark & \checkmark & \checkmark & \checkmark \\ 
		 2 & $\emptyset$ & \checkmark & \checkmark & $\emptyset$ & \checkmark & \checkmark \\
 		 3 & $\emptyset$ & $\emptyset$ & \checkmark & \checkmark & \checkmark & $\emptyset$ \\
		 4 & $\emptyset$ & $\emptyset$ & $\emptyset$ & \checkmark & \checkmark & $\emptyset$ \\
		 5 & $\emptyset$ & $\emptyset$ & $\emptyset$ & $\emptyset$ & \checkmark & $\emptyset$ \\
		 7 & $\emptyset$  & $\emptyset$ & $\emptyset$  & \checkmark  & \checkmark & \checkmark \\
	\end{tabular}
\end{center} 
\caption{The nonempty subsets $S_{ij}$ in case (5) $K_7 \neq \emptyset$, $K_6, K_8 = \emptyset$.} \label{fig:tabla_co4tent_5}
\end{figure}

Therefore, it suffices to see what happens if $K_6, K_7, K_8 \neq \emptyset$, since the matrices defined in the cases (2) to (5) have the same rows or less that each of the corresponding matrices $\mathbb{C}_1, \ldots, \mathbb{C}_8$ defined for case (1). In other words, the case $K_6, K_7, K_8 \neq \emptyset$ is the most general of all. 

\vspace{2mm}
Let us suppose that $K_6, K_7, K_8 \neq \emptyset$.
The Claims in Chapter \ref{chapter:partitions} and the following prime circle model allow us to assume that some subsets of $S$ are empty.

\begin{figure}[h!]
\centering
\includegraphics[scale=1]{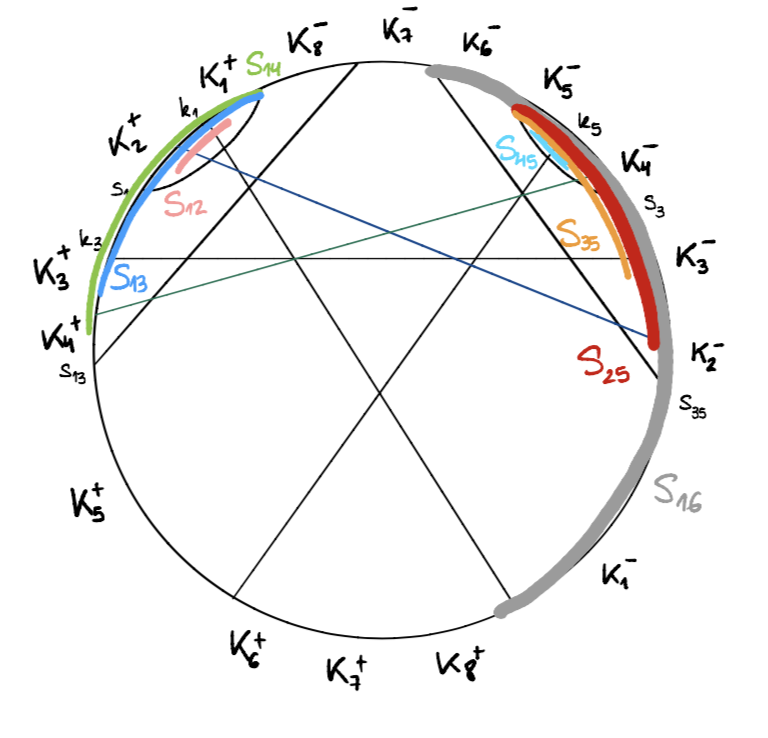}
	  \caption{A circle model for the co-$4$-tent graph.} \label{fig:co4tent_model}    
\end{figure}

We denote $S_{87}$ to the set of vertices in $S$ that are complete to $K_1, \ldots, K_6$, are adjacent to $K_7$ and $K_8$ but are not complete to $K_8$, and analogously $S_{76}$ is the set of vertices in $S$ that are complete to $K_1, \ldots, K_5, K_8$, are adjacent to $K_6$ and $K_7$ but are not complete to $K_6$. Hence, $S_{76]}$ denotes the vertices of $S$ that are complete to $K_1, \ldots, K_6, K_8$ and are adjacent to $K_7$.

\begin{remark}
	Claim \ref{claim:tent_0} remains true if $G$ contains an induced co-$4$-tent.
	The proof is analogous as in the tent case.
\end{remark}

\subsection{Split circle equivalence} \label{subsec:co4tent2}

In this subsection, we will show results analogous to Lemmas \ref{lema:equiv_circle_2nested_4tent_sinLR} and \ref{lema:B6_2nested_4tent}.

\begin{lema} \label{lema:equiv_circle_2nested_co4tent}
	If $\mathbb C_1$, $\mathbb C_2, \ldots, \mathbb C_8$ are not $2$-nested, then $G$ contains one of the forbidden subgraphs in $\mathcal{T}$ or $\mathcal{F}$. 
\end{lema}

\begin{proof}
	Using the argument of symmetry, we will prove this for the matrices $\mathbb C_1$, $\mathbb C_2$, $\mathbb C_3$, $\mathbb C_6$ and $\mathbb C_7$.
	
\begin{mycases}
Let us suppose that one of the matrices $\mathbb C_i$ is not $2$-nested. By Lemma \ref{lema:2-nested_if}, suppose that $\mathbb C_i$ is not partially $2$-nested. 
The structure of the proof is analogous as in Lemmas \ref{lema:equiv_circle_2nested_tent}, \ref{lema:equiv_circle_2nested_4tent_sinLR} and \ref{lema:B6_2nested_4tent}, and as in those lemmas we notice that, if $G$ is $\{ \mathcal{T}, \mathcal{F} \}$-free, then in particular, for each $i=1, \ldots, 8$, $\mathbb C_i$ contains no $M_0$, $M_{II}(4)$, $M_V$ or $S_0(k)$ for every even $k \geq 4$ since these matrices are the adjacency matrices of non-circle graphs.

	\case Suppose that one of the matrices $\mathbb C_i$ is not admissible, for some $i=1,2,3,6,7$.
	
	\subcase Suppose first that $\mathbb C_1$ is not admissible.
	Hence, since $\mathbb C_1$ has no uncolored labeled rows, or any rows labeled with R or LR, then $\mathbb C_1$ contains either $D_0$ or $S_2(k)$.
	Suppose that $\mathbb C_1$ contains $D_0$. Let $v_1$ and $v_2$ in $S_{12}$, $k_{11}$ and $k_{12}$ in $K_1$ such that $k_{1i}$ is adjacent to $v_i$ and nonadjacent to $v_{i+1}$ (mod 2), for $i=1,2$.

Notice that, if $v_1$ and $v_2$ have empty intersection in $K_2$, then we find a $4$-tent induced by $\{ k_{21}$, $k_{11}$, $k_{12}$, $k_{22}$, $v_1$, $v_2$, $ s_{35}\}$. The same holds for any two vertices $v_1$ and $v_2$ in $S_{12} \cup S_{16} \cup S_{17}$ (considering $s_1$ instead of $s_{35}$), hence we may assume that there is a vertex $k_i$ in $K_i$ --for $i=2,6,7$ as appropriate-- adjacent to both $v_1$ and $v_2$.

Thus, if both $v_1$ and $v_2$ lie in $S_{12}$, then we find a net${}\vee{}K_1$ induced by $\{ k_{11}$, $k_{12}$, $k_{2}$, $k_{3}$, $v_1$, $v_2$, $ s_{35}\}$. If $v_1$ and $v_2$ both lie in or $S_{16]}$ or $S_{17}$, then we find a net${}\vee{}K_1$ induced by $\{ k_{11}$, $k_{12}$, $k_{2}$, $k_{8}$, $v_1$, $v_2$, $ s_{35}\}$ (since $K_8 \neq \emptyset$, however the same holds using any vertex in $K_7$ nonadjacent to both $v_1$ and $v_2$). 
If $v_1$ in $S_{12}$ and $v_2$ in $S_{16]}$, then we find $M_{II}(4)$ induced by $\{ k_{11}$, $k_{2}$, $k_{5}$, $k_{12}$, $v_1$, $v_2$, $s_{13}$, $ s_{35}\}$. If $v_2$ in $S_{17}$ is analogous changing $k_5$ by $k_6$.

Suppose there is $S_2(j)$ as a subconfiguration of $\mathbb C_1$, and suppose $j$ is even, thus $v_1$ and $v_j$ lie both in $S_{12}$ or both in $S_{16} \cup S_{17}$.
If both lie in $S_{12}$, then we find $F_2(j+1)$ induced by $\{ k_{11}, \ldots, k_{1(j-1)}$, $k_{2}$, $k_{3}$, $v_1, \ldots, v_j$, $ s_{35}\}$. 
If instead both lie in $S_{16} \cup S_{17}$, then we find $F_1(j+1)$ induced by $\{ k_{11}, \ldots, k_{1(j-1)}$, $k_{3}$, $v_1, \ldots, v_j$, $ s_{1}\}$.
Suppose $j$ is odd, then $v_1$ in $S_{12}$ and $v_2$ in $S_{16} \cup S_{17}$, or viceversa.
In the first case, we find $F_1(j+2)$ induced by $\{ k_{11}, \ldots, k_{1(j-1)}$, $k_{2}$, $k_{3}$, $v_1, \ldots, v_j$, $s_1$, $s_{35}\}$. In the second case, we find $F_2(j)$ induced by $\{ k_{11}, \ldots, k_{1(j-1)}$, $k_{3}$, $v_1, \ldots, v_j \}$, therefore $\mathbb C_1$ is admissible.

\subcase Suppose $\mathbb C_2$ is not admissible. Since $\mathbb C_2$ has no uncolored labeled rows, or LR rows, or blue rows labeled with R, or red rows labeled with L, then $\mathbb C_2$ contains either $D_0$, $D_2$, $S_2(j)$ for some $j$ even or $S_3(j)$ for some $j$ odd.
Suppose there is $D_0$. Let $v_1$ and $v_2$ be the rows of $D_0$, and $k_{21}$ and $k_{22}$ in $K_2$ such that $v_i$ is adjacent to $k_{2i}$ and is nonadjacent to $k_{2(i+1)}$ (mod 2) for $i=1,2$. If $v_1$ and $v_2$ lie in $S_{12}$, then we know by the previous case that there is a vertex $k_1$ in $K_1$ adjacent to both. However, in this case we find a tent induced by $\{ k_{21}$, $k_{1}$, $k_{22}$, $v_1$, $v_2$, $s_{35}\}$. The same holds if $v_1$ and $v_2$ lie in $S_{23} \cup S_{25} \cup S_{26}$, changing $k_1$ by $k_3$ and $s_{35}$ by $s_1$, thus there is no $D_0$.

Suppose there is $D_2$, let $v_1$ and $v_2$ be the rows of $D_2$, one is labeled with L and the other is labeled with R. Thus, $v_1$ in $S_{12}$ and $v_2$ in $S_{23} \cup S_{25} \cup S_{26}$, or viceversa. Let $k_{21}$ and $k_{22}$ in $K_2$ such that $k_{21}$ is adjacent to both $v_1$ and $v_2$ and $k_{22}$ is nonadjacent to $v_1$ and $v_2$. Then, we find $M_{II}(4)$ induced by $\{ k_{1}$, $k_{21}$, $k_{3}$, $k_{22}$, $v_1$, $v_2$, $s_1$, $s_{35}\}$, and thus there is no $D_2$.

Suppose there is $S_2(j)$ for some even $j$. If $v_1$ and $v_j$ lie in $S_{12}$, then we find $F_2(j+1)$ induced by $\{ k_{21}, \ldots, k_{2(j-1)}$, $k_{1}$, $v_1, \ldots, v_j$, $ s_{35}\}$. If instead $v_1$ and $v_j$ lie in $S_{23} \cup S_{25} \cup S_{26}$, then we also find $F_2(j+1)$ induced by $\{ k_{21}, \ldots, k_{2(j-1)}$, $k_{3}$, $v_1, \ldots, v_j$, $ s_{1}\}$, and hence there is no $S_2(j)$.

Suppose there is $S_3(j)$ for some odd $j$. Thus, $v_1$ in $S_{12}$ and $v_2$ in $S_{23} \cup S_{25} \cup S_{26}$, or viceversa. In that case, we find $F_2(j+2)$ induced by $\{ k_{21}, \ldots, k_{2(j-1)}$, $k_{1}$, $k_{3}$, $v_1, \ldots, v_j$, $s_1$, $s_{35}\}$, and therefore $\mathbb C_2$ is admissible. 

\subcase Suppose $\mathbb C_3$ is not admissible. Since there are no LR-rows, or uncolored labeled rows, then there is either $D_0$, $D_1$, $D_2$, $S_2(j)$ or $S_3(j)$.

Suppose there is $D_0$, let $v_1$ and $v_2$ be the rows of $D_0$ and $k_{31}$ and $k_{32}$ in $K_3$ the columns of $D_0$.
The vertices $v_1$ and $v_2$ lie in $S_{13}$, $S_{34}$, $S_{35}$, $S_{36}$ or $S_{23}$.
First notice that, in either case, if the intersection is empty in $K_1$ (resp.\ $K_i$ for $i=2,3,4,5$), then we find a $4$-tent induced by $\{ k_{11}$, $k_{31}$, $k_{32}$, $k_{12}$, $v_1$, $v_2$, $s_{35} \}$ (resp.\ $s_1$, $s_{13}$, $s_5$).

If $v_1$ and $v_2$ both lie in $S_{13}$, then we find a tent induced by $\{ k_1$, $k_{31}$, $k_{32}$, $v_1$, $v_2$, $s_{35} \}$. The same holds if both lie in $S_{35}$ or $S_{36}$.
If $v_1$ and $v_2$ lie in $S_{34}$, then we find net${}\vee{}K_1$ induced by $\{ k_{31}$, $k_{4}$, $k_{32}$, $k_{5}$, $v_1$, $v_2$, $s_{5}\}$. The same holds by symmetry if both lie in $S_{23}$. 
If $v_1$ in $S_{13}$ and $v_2$ in $S_{23}$, then we find $M_{II}(4)$ induced by $\{ k_{1}$, $k_{2}$, $k_{31}$, $k_{32}$, $v_1$, $v_2$, $s_1$, $s_{35}\}$. The same holds if $v_1$ in $S_{35} \cup S_{36}$ and $v_2$ in $S_{23}$, therefore there is no $D_0$.

Suppose there is $D_1$, let $v_1$ and $v_2$ be the rows of $D_1$ and $k_{3}$ in $K_3$ be the non-tag column of $D_1$.
Suppose that $v_1$ in $S_{13}$ and $v_2$ in $S_{34}$. Then, we find $F_1(5)$ induced by $\{ k_{1}$, $k_{3}$, $k_{4}$, $k_{5}$, $v_1$, $v_2$, $s_5$, $s_{13}$, $s_{35}\}$. The same holds by symmetry if $v_1$ in $S_{35} \cup S_{36}$ and $v_2$ in $S_{23}$, thus there is no $D_1$.
	
Suppose there is $D_2$, let $v_1$ and $v_2$ be the rows of $D_2$, and $k_{31}$ and $k_{32}$ in $K_3$ be the columns of $D_2$. If $v_1$ in $S_{13}$ and $v_2$ in $S_{35} \cup S_{36}$, then we find $M_{II}(4)$ induced by $\{ k_{1}$, $k_{5}$, $k_{31}$, $k_{32}$, $v_1$, $v_2$, $s_{13}$, $s_{35}\}$. The other case is analogous, therefore there is no $D_2$.

Suppose there is $S_2(j)$ with $j$ even. If $v_1$ and $v_j$ in $S_{13}$, then we find $F_1(j+1)$ induced by $\{ k_{31}, \ldots, k_{3(j-1)}$, $k_{1}$, $v_1, \ldots, v_j$, $s_{35}\}$. If instead $v_1$ and $v_j$ in $S_{34}$, then we find $F_1(j)$ induced by $\{ k_{31}, \ldots, k_{3(j-1)}$, $k_{4}$, $k_{5}$, $v_1, \ldots, v_j \}$. It is analogouos by symmetry if $v_1$ and $v_j$ are colored with blue, thus there is no $S_2(j)$ with $j$ even, hence suppose $j$ is odd.
If $v_1$ in $S_{13}$ and $v_j$ in $S_{23}$, then we find $F_2(j)$ induced by $\{ k_{31}, \ldots, k_{3(j-1)}$, $k_{1}$, $v_1, \ldots, v_j \}$. If instead $v_1$ in $S_{23}$ and $S_{13}$, then we find $F_1(j+2)$ induced by $\{ k_{31}, \ldots, k_{3(j-1)}$, $k_{1}$, $k_{2}$, $v_1, \ldots, v_j$, $s_1$, $s_{35}\}$. It is analgous for the other cases.

Suppose there is $S_3(j)$. If $j$ is even, then $v_1$ in $S_{13}$ and $v_j$ in $S_{34}$, or the analogous blue labeled rows. However, in that case we find $F_1(j+3)$ induced by $\{ k_{31}, \ldots, k_{3(j-1)}$, $k_{1}$, $k_{4}$, $k_5$, $v_1, \ldots, v_j$, $s_1$, $s_5$, $s_{35}\}$.
If instead $j$ is odd, then $v_1$ in $S_{13}$ and $S_{35} \cup S_{36}$ or the analogous labeled rows. In that case, we find $F_1(j+2)$ induced by $\{ k_{31}, \ldots, k_{3(j-1)}$, $k_{1}$, $k_{5}$, $v_1, \ldots, v_j$, $s_{13}$, $s_{35}\}$, therefore $\mathbb C_3$ is admissible.

\subcase Suppose $\mathbb C_6$ is not admissible. Since there are no LR-rows or uncolored labeled rows, or rows labeled with L, then there is either $D_0$ or $S_2(j)$.
Suppose there is $D_0$, let $v_1$ and $v_2$ be the rows of $D_0$ and $k_{61}$ and $k_{62}$ in $K_6$ be the columns of $D_0$.
If $v_1$ and $v_2$ lie in $S_{26} \cup S_{36} \cup S_{46}$, then we find a net${}\vee{}K_1$ induced by $\{ k_1$, $k_4$, $k_{61}$, $k_{62}$, $v_1$, $v_2$, $s_{13} \}$.
Once more, if the intersection in $K_4$ is empty, then we find a $4$-tent induced by $\{ k_{41}$, $k_{61}$, $k_{62}$, $k_{42}$, $v_1$, $v_2$, $s_{13} \}$.

If $v_1$ and $v_2$ in $S_{76} \cup S_{[86}$, then we find a tent induced by $\{ k_1$, $k_{61}$, $k_{62}$, $v_1$, $v_2$, $s_{35} \}$.
If instead $v_1$ in $S_{26} \cup S_{36} \cup S_{46}$ and $v_2$ in $S_{76} \cup S_{[86}$, then we find $M_{II}(4)$ induced by $\{ k_{61}$, $k_{62}$, $k_{1}$, $k_{4}$, $v_1$, $v_2$, $s_{13}$, $s_{35}\}$, thus there is no $D_0$.

Suppose there is $S_2(j)$. If $j$ is even, then $v_1$ and $v_j$ lie in $S_{26} \cup S_{36} \cup S_{46}$. In that case, we find $F_2(j+1)$ induced by $\{ k_{61}, \ldots, k_{6(j-1)}$, $k_{1}$, $k_{4}$, $v_1, \ldots, v_j$, $s_{13}\}$.
If instead $j$ is odd, then $v_1$ in $S_{26} \cup S_{36} \cup S_{46}$ and $v_2$ in $S_{76} \cup S_{[86}$, or viceversa. In the first case, we find $F_2(j+1)$ induced by $\{ k_{61}, \ldots, k_{6(j-1)}$, $k_{1}$, $k_{4}$, $v_1, \ldots, v_j$, $s_{13}$, $s_{35}\}$. In the second case, we find $F_2(j)$ induced by $\{ k_{61}, \ldots, k_{6(j-1)}$, $k_{1}$, $v_1, \ldots, v_j \}$, and therefore $\mathbb C_6$ is admissible.

\subcase Finally, suppose $\mathbb C_7$ is not admissible. Notice that, if there is $D_8$, then we find a tent, and if there is $D_9$, then we find $F_0$. Since there are no red labeled rows, then there is either $D_0$, $D_1$, $D_6$, $D_7$, $S_1(j)$, $S_2(j)$ with even $j$, $S_3(j)$ with even $j$, $S_4(j)$ with even $j$, $S_5(j)$ with even $j$, $S_6(j)$ or $S_7(j)$. 

Suppose there is $D_0$, let $v_1$ and $v_2$ be the rows, and $k_{71}$, $k_{72}$ in $K_7$ be the columns of $D_0$.
If $v_1$ and $v_2$ lie in $S_{[74} \cup S_{75} \cup S_{76}$, then we find a net${}\vee{}K_1$ induced by $\{ k_{71}$, $k_{72}$, $k_{4}$, $k_{6}$, $v_1$, $v_2$, $s_{35}\}$. If instead $v_1$ and $v_2$ lie in $S_{17} \cup S_{[27} \cup S_{87}$, since $v_1$ and $v_2$ are not complete to $K_8$, then there is either a $4$-tent (if there is no vertex in $K_8$ adjacent to both, induced by $\{ k_{71}$, $k_{81}$, $k_{82}$, $k_{72}$, $v_1$, $v_2$, $s_{13} \}$), or a net${}\vee{}K_1$ induced by $\{ k_{71}$, $k_{72}$, $k_{8}$, $k_{2}$, $v_1$, $v_2$, $s_{13} \}$, therefore there is no $D_0$.

Suppose there is $D_1$, let $v_1$ and $v_2$ be the rows, and $k_7$ in $K_7$ be the non-tag column. Let $v_1$ in $S_{74]} \cup S_{75} \cup S_{76}$ (notice that $v_1$ is complete to $K_8$ and is not complete to $K_6$) and $v_2$ in $S_{17} \cup S_{[27} \cup S_{87}$ (is not complete to $K_8$ and is complete to $K_6$). Thus, we find $M_{II}(4)$ induced by $\{ k_{8}$, $k_{3}$, $k_{6}$, $k_{7}$, $v_1$, $v_2$, $s_{13}$, $s_{35}\}$, hence there is no $D_1$.
Suppose there is $D_6$, let $v_1$, $v_2$ and $v_3$ be the rows where $v_3$ is an LR-row, and $k_{71}$ and $k_{72}$ in $K_7$ be the columns of $D_6$. In that case, $v_1$ lies in $S_{74]} \cup S_{75} \cup S_{76}$, $v_2$ in $S_{17} \cup S_{[27} \cup S_{87}$ and $v_3$ in $S_{76]}$, hence we find a $4$-tent induced by $\{ k_{71}$, $k_{8}$, $k_{6}$, $k_{72}$, $v_1$, $v_2$, $v_{3}\}$, therefore there is no $D_6$.
Suppose there is $D_7$, let $v_1$ be any row labeled with either L or R, and $v_2$ and $v_3$ LR-rows in $S_{76]}$. In either case, there is a vertex $k_i$ in $K_i$ with $i \neq 7$ such that $v_1$ is adjacent to $k_i$, and hence we find a net${}\vee{}K_1$ induced by $\{ k_{71}$, $k_{72}$, $k_{73}$, $k_{i}$, $v_1$, $v_2$, $v_{3} \}$, thus there is also no $D_7$.

Suppose there is $S_1(j)$, and suppose that $j$ is even. Since $v_1$ and $v_j$ correspond to rows labeled with either L or R, in either case $v_1$ and $v_j$ are complete to $K_4$. Hence, we find an odd $(j-1)$-sun with center induced by $\{ k_{71}, \ldots, k_{7(j-2)}$, $k_{4}$, $v_1, \ldots, v_j \}$. Moreover, if $j$ is odd, then we find a $(j-1)$-sun induced by the same subset.

Suppose there is $S_2(j)$ where $j$ is even. If $v_1$ and $v_j$ are labeled with L, then they are both complete to $K_6$ and $K_5$. Analogously, if they are labeled with R, then they are both complete to $K_8$ and $K_2$. In the first case, we find $F_2(j+1)$ induced by $\{ k_{71}, \ldots, k_{7(j-1)}$, $k_{5}$, $k_{6}$, $v_1, \ldots, v_j$, $s_{35}\}$. It is analogous if they are labeled with R.

Suppose there is $S_3(j)$ where $j$ is even. However, we find a $j$-sun induced by $\{ k_{71}$, $\ldots$, $k_{7(j-1)}$, $k_{5}$, $v_1, \ldots, v_j \}$ and thus it is not possible.

If there is $S_4(j)$ with even $j$, then we find a $j-1$-sun with center induced by $\{ k_{71}$, $\ldots$, $k_{7(j-2)}$, $k_{5}$, $v_1, \ldots, v_j \}$. 

If instead there is $S_5(j)$ with $j$ even, then we find $F_2(j+1)$ induced by $\{ k_{71}, \ldots, k_{7(j-1)}$, $k_{6}$, $k_{4}$, $v_1, \ldots, v_j$, $s_{35}\}$ if $v_1$ and $v_j$ lie in $S_{74} \cup S_{75} \cup S_{76}$. It is analogous if $v_1$ and $v_j$ lie in $S_{27} \cup S_{17} \cup S_{87}$ using $k_8$, $k_2$ and $s_{13}$.

Finally, if there is $S_6(j)$, then we find $M_{II}(j)$, and if there is $S_7(j)$ then we find a $j$-sun if $j$ is even, and a $j$-sun with center if $j$ is odd.

Therefore $\mathbb C_i$ is admissible for every $i=1,2,3,6,7$.

\case Let $C= \mathbb C_i$ and suppose that $C$ is not LR-orderable, then $C^*_\tagg$ contains either a Tucker matrix or $M_4'$, $M_4''$, $M_5'$, $M_5''$, $M'_2(k)$, $M''_2(k)$, $M_3'(k)$, $M_3''(k)$, $M_3'''(k)$ for some $k \geq 4$ (see Figure \ref{fig:forb_LR-orderable_tags}). 

The proof of this case is analogous as in Lemma \ref{lema:B6_2nested_4tent}, since in most situations we only use the fact that $C$ is admissible. Moreover, whenever we consider two labeled rows $v$ and $w$ labeled with distinct letters, we have at least two vertices $k_6$ in $K_6$ and $k_8$ in $K_8$ such that $v$ is adjcent to $k_6$ and nonadjacent to $k_8$ and $w$ is adjacent to $k_8$ and nonadjacent to $k_6$. Moreover, there is always a vertex $k_4$ in $K_4$ that is adjacent to both. This holds whether they are labeled with the same letter or not.


	\case Therefore, we may assume that $\mathbb C_i$ is admissible and LR-orderable but is not partially $2$-nested. Since there are no uncolored labeled rows and those colored rows are labeled with either L or R and do not induce any of the matrices $\mathcal{D}$, then in particular no pair of pre-colored rows of $\mathbb C_i$ induce a monochromatic gem or a monochromatic weak gem, and there are no badly-colored gems since every LR-row is uncolored, therefore $\mathbb C_i$ is partially $2$-nested.

\case Finally, let us suppose that $C= \mathbb C_i$ is partially $2$-nested but is not $2$-nested. As in the previous cases, we consider $C$ ordered with a suitable LR-ordering.
	Let $C'$ be a matrix obtained from $C$ by extending its partial pre-coloring to a total $2$-coloring. It follows from Lemma \ref{lema:B_ext_2-nested} that, if $C'$ is not $2$-nested, then either there is an LR-row for which its L-block and R-block are colored with the same color, or $C'$ contains a monochromatic gem or a monochromatic weak gem or a badly-colored doubly weak gem. 	
	
	If $C'$ contains a monochromatic gem where the rows that induce such a gem are not LR-rows, then the proof is analogous as in the tent case. Thus, we may assume that at least one of the rows is an LR-row and hence let $i=7$.
	
	\subcase \textit{Let us first suppose there is an LR-row $w$ for which its L-block $w_L$ and R-block $w_R$ are colored with the same color. 	}
	If these two blocks are colored with the same color, then there is either one odd sequence of rows $v_1, \ldots, v_j$ that force the same color on each block, or two distinct sequences, one that forces the same color on each block. 
	
	\subsubcase If there is one odd sequence $v_1, \ldots, v_j$ that forces the color on both blocks, then the proof is analogous as in \ref{lema:B6_2nested_4tent}. 
	
	\subsubcase Suppose there are two independent sequences $v_1, \ldots, v_j$ and $x_1, \ldots, x_l$  that force the same color on $w_L$ and $w_R$, respectively. Suppose without loss of generality that $w_L$ and $w_R$ are colored with red.
	If $j=1$ and $l=1$, then we find $D_6$, which is not possible. Hence, we assume that either $j>1$ or $l>1$. 
	Suppose that $j>1$ and $l>1$, thus there is one labeled row in each sequence. We may assume that $v_j$ is labeled with L and $x_l$ is labeled with R, since LR-ordering used to color $B'$ is suitable. As in the proof of Lemma \ref{lema:B6_2nested_4tent}, we assume throughout the proof that each row in each sequence forces the coloring on both the previous and the next row in its sequence. Thus in this case, $v_2, \ldots, v_j$  is contained in $w_L$ and $x_2, \ldots, x_l$ is contained in $w_R$. 
	Moreover, $w$ represents a vertex in $S_{76]}$, $v_j$ lies in $S_{74]} \cup S_{75} \cup S_{76}$ and $x_l$ lies in $S_{[27} \cup S_{17} \cup S_{87}$, and thus both are colored with blue and $j$ and $l$ are both odd.
	If $x_l$ lies in $S_{[27} \cup S_{17} \cup S_{87}$, since there is a $k_4$ in $K_4$ adjacent to both $v_j$ and $x_l$, then we find $F_2(j+l+1)$ contained in the submatrix induced by each row and column on which the rows in $w$ and both sequences are not null and the column representing $k_i$.
	The proof is analogous if either $j=1$ or $l=1$.

	Hence, we assume there is either a monochromatic weak gem in which one of the rows is an LR-row or a badly-colored doubly-weak gem in $C'$, for the case of a monochromatic gem or a monochromatic weak gem where one of the rows is an L-row (resp.\ R-row) and the other is unlabeled is analogous to the tent case. Moreover, if an LR-row and an unlabeled row (or a row labeled with L or R) induce a monochromatic gem, then in particular these rows induce a monochromatic weak gem. 
	
	However, the proof follows analogously as in Lemma \ref{lema:B6_2nested_4tent} and therefore, If $G$ is $\{ \mathcal{T}, \mathcal{F} \}$-free, then $\mathbb{C}_i$ is $2$-nested for each $i=1, 2,\ldots, 8$.

\end{mycases}	 
\end{proof}

\begin{defn} \label{def:matrices_C_por_colores}
We define the matrices $\mathbb{C}_r$, $\mathbb{C}_b$, $\mathbb{C}_{r-b}$ and $\mathbb C_{b-r}$ as in Section \ref{subsec:4tent4}. Similarly, we have the following Lemma for these matrices.
\end{defn}

\begin{lema} \label{lema:matrices_union_son_nested_co4tent}
	Suppose that $\mathbb C_i$ is $2$-nested for each $i =1,2 \ldots, 8$. If $\mathbb C_r$, $\mathbb C_b$, $\mathbb C_{r-b}$ or $\mathbb C_{b-r}$ are not nested, then $G$ contains $F_0$ as a minimal forbidden induced subgraph for the class of circle graphs.
\end{lema}

\begin{proof}
Suppose that $\mathbb C_r$ is not nested, and let $v_1$ and $v_2$ be the vertices represented by the rows that induce a $0$-gem in $\mathbb C_r$. The rows in $\mathbb C_r$ represent vertices in the following subsets of $S$: $S_{12}$, $S_{[13}$, $S_{[14}$, $S_{34}$, $S_{74]}$, $S_{75}$, $S_{76}$, $S_{82]}$, $S_{83}$, $S_{84}$, $S_{85}$, $S_{[86}$, $S_{86]}$ or $S_{87}$. Notice that, by definition, these last two subsets are not complete to $K_8$.

Notice that the vertices in $S_{86]} \cup S_{87}$ do not induce $0$-gems in $\mathbb C_r$.

\begin{mycases}
	\case Suppose that $v_1$ in $S_{12}$. Since $S_{[13}$ and $S_{[14}$ are complete to $K_1$ and $S_{82]}$ is complete to $K_2$, the only possibility is that $v_2$ in $S_{12}$. In that case, we find $F_0$ induced by $\{ v_1$, $v_2$, $s_{35}$, $k_{11}$, $k_{12}$, $k_{21}$, $k_{22}$, $k_4 \}$, where $k_{11}$ and $k_{12}$ in $K_1$, $k_{21}$ and $k_{22}$ in $K_2$ and $k_4$ in $K_4$.
 We find the same forbidden subgraph if $v_1$ and $v_2$ lie both in $S_{34}$, with vertices $k_{31}$, $k_{32}$ in $K_3$, $k_{41}$ and $k_{42}$ in $K_4$, $k_5$ in $K_5$ and $s_5$ instead of $s_{35}$.
 
 \case Let $v_1$ in $S_{[13} \cup S_{[14}$. 

	 \subcase If $v_1$ in $S_{[14}$, then $v_2$ lies in $S_{34}$ or in $S_{82]} \cup S_{83} \cup S_{84}$ since every vertex in $S_{12}$, $S_{[13}$ is contained in every vertex of $S_{[14}$, and every vertex in $S_{[14}$ is properly contained in every vertex of $S_{74]} \cup S_{75} \cup S_{76} \cup S_{85} \cup S_{[86}$. 
	 If $v_2$ in $S_{34}$, then we find $F_0$ induced by $\{ v_1$, $v_2$, $s_{5}$, $k_{1}$, $k_{3}$, $k_{41}$, $k_{42}$, $k_5 \}$.
	 If instead $v_2$ in $S_{82]} \cup S_{83} \cup S_{84}$, then we find $F_0$ induced by $\{ v_1$, $v_2$, $s_{35}$, $k_{8}$, $k_{1}$, $k_{2}$, $k_{4}$, $k_5 \}$, since there is a vertex $k_4$ in $K_4$ adjacent to $v_1$ and nonadjacent to $v_2$ and a vertex $k_8$ in $K_8$ adjacent to $v_2$ and nonadjacent to $v_1$ (which is represented in the $0$-gem by the column $c_L$).

	\subcase If $v_1$ in $S_{[13}$, then $v_2$ lies in $S_{34}$ or in $S_{82]} \cup S_{83}$. However, the first is not possible since $\mathbb C_3$ is admissible. The proof if $v_2$ lies in $S_{82]} \cup S_{83}$ follows analogously as in the previous subcase. 
	
	\case Suppose $v_1$ in $S_{34}$. Since $\mathbb C_3$ is admissible and $S_{74]}$ is complete to $K_4$, then the only possibility is that $v_2$ in $S_{84}$. We find $F_0$ induced by $\{ v_1$, $v_2$, $s_{5}$, $k_{2}$, $k_{3}$, $k_{41}$, $k_{42}$, $k_5 \}$.

\end{mycases}

Suppose that $\mathbb C_b$ is not nested, and let $v_1$ and $v_2$ be the vertices represented by the rows that induce a $0$-gem in $\mathbb C_b$. The rows in $\mathbb C_b$ represent vertices in the following subsets of $S$: $S_{23}$, $S_{25]}$, $S_{26}$, $S_{[27}$, $S_{16]}$, $S_{17}$, $S_{35]}$, $S_{36}$, $S_{45}$, $S_{[46}$, $S_{74]}$, $S_{75}$, $S_{76}$, $S_{82]}$, $S_{83}$, $S_{84}$, $S_{[85}$, $S_{[86}$, $S_{86]}$ or $S_{87}$. 

Notice that the vertices in $S_{86]}$, $S_{87}$, $S_{82]}$, $S_{83}$, $S_{84}$ do not induce $0$-gems in $\mathbb C_r$. The same holds for those vertices in $S_{74]}$, $S_{75}$ and $S_{76}$, however in this case this follows from the fact that $\mathbb C_7$ is admissible.

\begin{mycases}
	\case Suppose $v_1$ in $S_{23}$. Since $\mathbb C_3$ is admissible, then $v_2$ lies in $S_{23} \cup S_{25]} \cup S_{26}$. If $v_2$ in $S_{23}$, then we find $F_0$ induced by $\{ v_1$, $v_2$, $s_{1}$, $k_{1}$, $k_{21}$, $k_{22}$, $k_{31}$, $k_{32} \}$. If $v_2$ in $S_{25]}$ or $S_{26}$, then we find $F_0$ induced by the same subset changing $k_{32}$ for some vertex in $K_5$ or $K_6$, respectively. 
	
	\case Let $v_1$ in $S_{25]} \cup S_{26}$, thus $v_2$ in $S_{26} \cup S_{36} \cup S_{46}$. We assume that $v_1$ in $S_{25]}$, since the proof is analogous if $v_1$ in $S_{26}$. We find $F_0$ induced by $\{ v_1$, $v_2$, $s_{1}$, $k_{1}$, $k_{21}$, $k_{22}$, $k_{5}$, $k_6 \}$ if $v_2$ in $S_{26}$. If instead $v_2$ in $S_{36}$ or $S_{46}$, then the subset is the same with the exception of $k_{22}$, which is replaced by an analogous vertex in $K_3$ or $K_4$, respectively.
	
	\case Suppose $v_1$ in $S_{[27}$. Thus, $v_2$ in $S_{16]} \cup S_{17} \cup S_{86]} \cup S_{87}$. Since $v_2$ is never complete to $K_8$ and both vertices induce a $0$-gem, we find $F_0$ induced by $\{ v_1$, $v_2$, $s_{13}$, $k_{1}$, $k_{2}$, $k_{6}$, $k_{7}$, $k_{8} \}$.
	
	\case Suppose $v_1$ in $S_{16]}$. Thus, $v_2$ in $S_{17}$. Since $K_8 \neq \emptyset$, we find $F_0$ induced by $\{ v_1$, $v_2$, $s_{13}$, $k_{11}$, $k_{12}$, $k_{6}$, $k_{7}$, $k_{8} \}$. 
	
	\case Suppose $v_1$ in $S_{35]}$. Thus, $v_2$ in $S_{36} \cup S_{46}$. We find $F_0$ induced by $\{ v_1$, $v_2$, $s_{13}$, $k_{1}$, $k_{31}$, $k_{32}$, $k_{5}$, $k_{6} \}$ if $v_2$ in $S_{36}$, and if $v_2$ in $S_{46}$ we change $k_{32}$ for an analogous vertex in $K_4$.
	
	\case Suppose $v_1$ in $S_{17}$. Thus, $v_2$ in $S_{86]} \cup S_{87}$. Since $v_2$ is not complete to $K_8$, then we find $F_0$ induced by $\{ v_1$, $v_2$, $s_{13}$, $k_{8}$, $k_{11}$, $k_{12}$, $k_{6}$, $k_{7} \}$.
	
\end{mycases}

Suppose that $\mathbb C_{r-b}$ is not nested, and let $v_1$ and $v_2$ be the vertices represented by the rows that induce a $0$-gem. The rows in $\mathbb C_{r-b}$ represent vertices in either $S_{86]}$ or $S_{87}$.

Suppose that $v_1$ in $S_{86]}$ and $v_2$ in $S_{87}$. Since none of the vertices is complete to $K_8$, $\mathbb C_8$ is admissible and these rows are R-rows in $\mathbb C_8$, then there is no $D_0$ and thhus there are three vertices $k_{81}$, $k_{82}$ and $k_{83}$ in $K_8$ such that $k_{81}$ is nonadjacent to both $v_1$ and $v_2$, $k_{83}$ is adjacent to both $v_1$ and $v_2$ and $k_{82}$ is adjacent to $v_1$ and nonadjacent to $v_2$. We find $F_0$ induced by $\{ v_1$, $v_2$, $s_{13}$, $k_{81}$, $k_{82}$, $k_{83}$, $k_{6}$, $k_{7} \}$. It follows analogously if both vertices lie in $S_{87}$, and if both lie in $S_{[86}$ only changing $k_7$ for an analogous $k_{62}$ in $K_6$.

\vspace{1mm}
Suppose that $\mathbb C_{b-r}$ is not nested, and let $v_1$ and $v_2$ be the vertices represented by the rows that induce a $0$-gem. The rows in $\mathbb C_{b-r}$ represent vertices in $S_{74]}$, $S_{75}$, $S_{76}$, $S_{82]}$, $S_{83}$, $S_{84}$, $S_{[85}$ or $S_{[86}$.

\begin{mycases}
	\case Suppose that $v_1$ and $v_2$ in $S_{74]} \cup S_{75} \cup S_{76}$. In either case, $v_1$ and $v_2$ are not complete to $K_6$ by definition. Since $\mathbb C_6$ is admissible, thus there is no $D_0$ and there are vertices $k_{61}$ and $k_{62}$ in $K_6$ such that $v_1$ is nonadjacent to $k_{61}$ and $k_{62}$ and $v_2$ is adjacent to $k_{61}$ and is nonadjacent to $k_{62}$. We find $F_0$ induced by $\{ v_1$, $v_2$, $s_{35}$, $k_{71}$, $k_{72}$, $k_{4}$, $k_{61}$, $k_{62} \}$ if $v_1$ and $v_2$ lie in $S_{76}$. It follows analogously if $v_1$ or $v_2$ lie in $S_{74]} \cup S_{75}$ changing $k_{61}$ for an analogous vertex $k_5$ in $K_5$.
	
	\case Suppose that $v_1$ and $v_2$ in $S_{82]} \cup S_{83} \cup S_{84} \cup S_{[85} \cup S_{[86}$. Since every vertex in $S_{[85}$ and $S_{[86}$ is complete to $K_8$, then none of these vertices induce a $0$-gem in $\mathbb C_{b-r}$. Thus, $v_1$ and $v_2$ lie in $S_{82]} \cup S_{83} \cup S_{84}$. Moreover, since every vertex in $S_{82]}$ is complete to $K_2$, then it is not possible that both vertices lie in $S_{82]}$. Let $k_{81}$ and $k_{82}$ in $K_8$ such that $v_1$ is adjacent to both and $v_2$ is adjacent to $k_{82}$ and is nonadjacent to $k_{81}$. Notice that in that case we are assuming that, if one of the vertices lies in $S_{82]}$, then such vertex is $v_1$. If $v_2$ in $S_{83}$, then we find $F_0$ induced by $\{ v_1$, $v_2$, $s_{35}$, $k_{81}$, $k_{82}$, $k_{2}$, $k_{3}$, $k_{5} \}$. If instead $v_2$ in $S_{84}$, we find $F_0$ with the same subset only changing $k_3$ for some analogous $k_4$ in $K_4$.
	
	\case Suppose that $v_1$ in $S_{74]} \cup S_{75} \cup S_{76}$ and $v_2$ in $S_{82]} \cup S_{83} cup S_{84} \cup S_{[85} \cup S_{86}$. Notice that, if $v_2$ in $S_{82]} \cup S_{83} \cup S_{84}$, then $v_2$ is contained in $v_1$ and thus such vertices cannot induce a $0$-gem in $\mathbb C_{b-r}$. Thus, $v_2$ in $S_{[85} \cup S_{[86}$. In this case, there is a vertex $k_6$ in $K_6$ that is nonadjacent to both $v_1$ and $v_2$ since none of these vertices is complete to $K:6$ by definition and $\mathbb C_6$ is admissible. If $v_1$ in $S_{74]} \cup S_{75}$, then we find we find $F_0$ induced by $\{ v_1$, $v_2$, $s_{35}$, $k_{7}$, $k_{8}$, $k_{4}$, $k_{5}$, $k_{6} \}$. If instead $v_1$ in $S_{76}$ and $v_1$ and $v_2$ induce a $0$-gem, then $v_2$ in $S_{[86}$. We find $F_0$ with the same subset as before, only changing $k_5$ for some analogous $k_{62}$ in $K_6$.
\end{mycases}

This finishes the proof.

\end{proof}

The main result of this section is the following theorem, which follows directly from the previous lemmas.

\begin{teo} \label{teo:finalteo_co4tent}
	Let $G=(K,S)$ be a split graph containing an induced co-$4$-tent. Then, $G$ is a circle graph if and only if $\mathbb C_1,\mathbb C_2,\ldots,\mathbb C_8$ are $2$-nested and $\mathbb C_r$, $\mathbb C_b$, $\mathbb C_{r-b}$ and $\mathbb C_{b-r}$ are nested.
\end{teo} 

\begin{proof}

Necessity is clear by the previous lemmas and the fact that the graphs in families $\mathcal{T}$ and $\mathcal{F}$ are all non-circle. Suppose now that each of the matrices $\mathbb C_1,\mathbb C_2,\ldots,\mathbb C_8$ is $2$-nested and the matrices $\mathbb C_r$, $\mathbb C_b$, $\mathbb C_{r-b}$ or $\mathbb C_{b-r}$ are nested.
Let $\Pi$ be the ordering for all the vertices in $K$ obtained by concatenating each suitable LR-ordering $\Pi_i$ for $i \in \{1, 2,\ldots, 8\}$.

Consider the circle divided into sixteen pieces as in Figure \ref{fig:co4tent_model}. For each $i\in\{1$,$2$,$\ldots$,$8\}$ and for each vertex $k_i \in K_i$ we place a chord having one endpoint in $K_i^+$ and the other endpoint in $K_i^-$, in such a way that the ordering of the endpoints of the chords in $K_i^+$ and $K_i^-$ is $\Pi_i$. Throughout the following, we will consider the circular ordering  clockwise.

Let us see how to place the chords for each subset $S_{ij}$ of $S$. 

The vertices with exactly endpoint in $K_7^{-}$ that are not LR-vertices in $\mathbb C_7$ are $S_{74]} \cup S_{75} \cup S_{76}$ and $S_{[27} \cup S_{17} \cup S_{87}$. Since $\mathbb C_7$ is admissible, the vertices in $S_{74]} \cup S_{75} \cup S_{76}$ and $S_{[27} \cup S_{17} \cup S_{87}$ do not intersect in $K_7$.
	Moreover, since there are no pre-colored red rows, then there are no vertices with exactly one endpoint in $K_7^{+}$.
	Furthermore, the vertices in $S_{76]}$ and $S_{[86]}$ are represented by LR-rows in $\mathbb C_7$. These last ones are exactly those empty LR-rows. Since $\mathbb C_7$ is $2$-nested, then all of these vertices can be drawned in the circle model. It follows that, if $S_{[86]} \neq \emptyset$, then either $S_{74]} \cup S_{75} = \emptyset$ or $S_{[27} \cup S_{17} \cup S_{87} = \emptyset$.
	On the other hand, those nonempty LR-rows in $\mathbb C_7$ correspond to vertices in $S_{76]}$. Each of these vertices with two blocks in $\mathbb C_7$ have one endpoint in $K_7^{+}$, placed according to the ordering $\Pi_7$ of the nonempty columns of its red block, and the other endpoint placed in $K_7^{-}$ according to the ordering $\Pi_7$ of the nonempty columns of its blue block. It follows analogously for those nonempty LR-vertices with exactly one block.

Notice that in $\mathbb C_1$ (resp.\ in $\mathbb K_5$ by symmetry) there are no R-rows (resp.\ L-rows). Since $\mathbb C_1$ is $2$-nested, then all the vertices that have exactly one endpoint in $K_1^{-}$ (resp.\ $K_5^{+}$) are nested and thus such endpoint can be placed without issues. The same holds for every vertex with both endpoints in $K_1^{-}$ and $K_5^{+}$.
	Moreover, the only vertices with exactly one endpoint in $K_1^{+}$ may be those in $S_{12}$, for all the vertices in $S_{[13} \cup S_{[14}$ are nested and have the endpoint corresponding to $K_1$ placed between $s_{14}^{-}$ and the first endpoint of a vertex in $S_{82]} \cup S_{83} \cup S_{84}$ (or $s_{13}^{-}$ if this set is empty). The vertices in $S_{12}$ are nested, and thus each endpoint of these vertices may be placed in the ordering given by $\Pi_2$ and $\Pi_1$, respectively, between $s_1^{+}$ and $s_1^{-}$.

The only vertices that have exactly one endpoint in $K_2^{+}$ are those in $S_{12}$. The vertices that have exactly one endpoint in $K_2^{-}$ are those in $S_{23} \cup S_{25]} \cup S_{26}$. Since $\mathbb C_2$ is $2$-nested and $\mathbb C_b$ is nested, then these vertices are all nested and thus we can place the chords according to the ordering $\Pi_2$. Those vertices in $S_{[27}$ have the endpoint corresponding to $K_2$ placed right after $s_{35}^{-}$, and before any of the chords with endpoint in $K_1^{-}$.
	The same holds by symmetry for those chords with exactly one endpoint in $K_4^{+}$ and $K_4^{-}$.

The vertices with exactly one endpoint in $K_3^{+}$ are $S_{34}$ and $S_{[13} \cup S_{83}$. Since $\mathbb C_3$ is admissible, the vertices in $S_{34}$ and $S_{[13} \cup S_{83}$ do not intersect in $K_3$. Moreover, since $\mathbb C_3$ is $2$-nested and $\mathbb C_r$ and $\mathbb C_{b-r}$ are nested, then the vertices in $S_{[13} \cup S_{83}$ are nested and thus we can place both of its endpoints following the ordering given by $\Pi_3$.
	The vertices with exactly one endpoint in $K_3^{-}$ are those in $S_{23}$ (which we have already shown where to place) and those in $S_{35]} \cup S_{36}$. These last vertices are nested since $\mathbb C_b$ is nested and thus we place both its endpoints according to $\Pi_3$. Notice that, since $\mathbb C_3$ is admissible, then the vertices in $S_{23}$ and $S_{35]}  \cup S_{36}$ do not intersect in $K_3$.
	
	Since $\mathbb C_{b-r}$ is nested, if $S_{[85} \neq \emptyset$, then $S_{74]} = \emptyset$, and viceversa. The same holds for $S_{[86}$ and $S_{74]} \cup S_{75}$. Moreover, if $S_{[85} \neq \emptyset$, then every vertex in $S_{[85}$ is nested in $S_{75}$, and if $S_{[86} \neq \emptyset$, then every vertex in $S_{[86}$ is nested in $S_{76}$. It follows analogously by symmetry for those vertices in $S_{[27} \cup S_{17} \cup S_{87} \cup S_{86] \cup S_{16]}}$.
	
	Those vertices with exactly one endpoint in $K_6^{+}$ are those in $S_{76} \cup S_{[86}$. These vertices are nested since $\mathbb C_6$ is $2$-nested and $\mathbb C_{b-r}$ is nested. Thus, if these subsets are nonempty, then $S_{74]} \cup S_{75} = \emptyset$. Therefore, we can place both its enpoints according to $\Pi_6$, one in $K_6^{+}$ and the other between $s_{13}^{-}$ and $s_{35}^{+}$.
	The vertices that have exactly one endpoint in $K_6^{-}$ are those in $S_{26} \cup S_{36} \cup S_{46}$, and since $\mathbb C_b$ is nested, then these vertices are all nested and therefore we place both its endpoints according to $\Pi_6$.
	
Finally, all the vertices represented by unlabeled rows in each $\mathbb C_i$ for $i=1, 2, \ldots, 8$ represent the vertices in $S_{ii}$. These vertices are entirely colored with either red or blue, and are either disjoint or nested with every other vertex colored with its color. Hence, we place both endpoints of the corresponding chord in $K_i^{+}$ if it is colored with red, and in $K_i^{-}$ if it is colored with blue, according to the ordering $\Pi_i$ given for $K_i$.
	
This gives the guidelines for a circle model for $G$.

\end{proof}

\section{Split circle graphs containing an induced net} \label{sec:circle5}

Let $G=(K,S)$ be a split graph. If $G$ is a minimally non-circle graph, then it contains either a tent, or a $4$-tent, or a co-$4$-tent, or a net as induced subgraphs.
In the previous sections, we have addressed the problem of having a split minimally-non-circle graph that contains an induced tent, $4$-tent and co-$4$-tent, respectively. Let us consider a split graph $G$ that contains no induced tent, $4$-tent or co-$4$-tent, and suppose there is a net subgraph in $G$. 

\begin{figure}[h!]
\centering
\includegraphics[scale=.5]{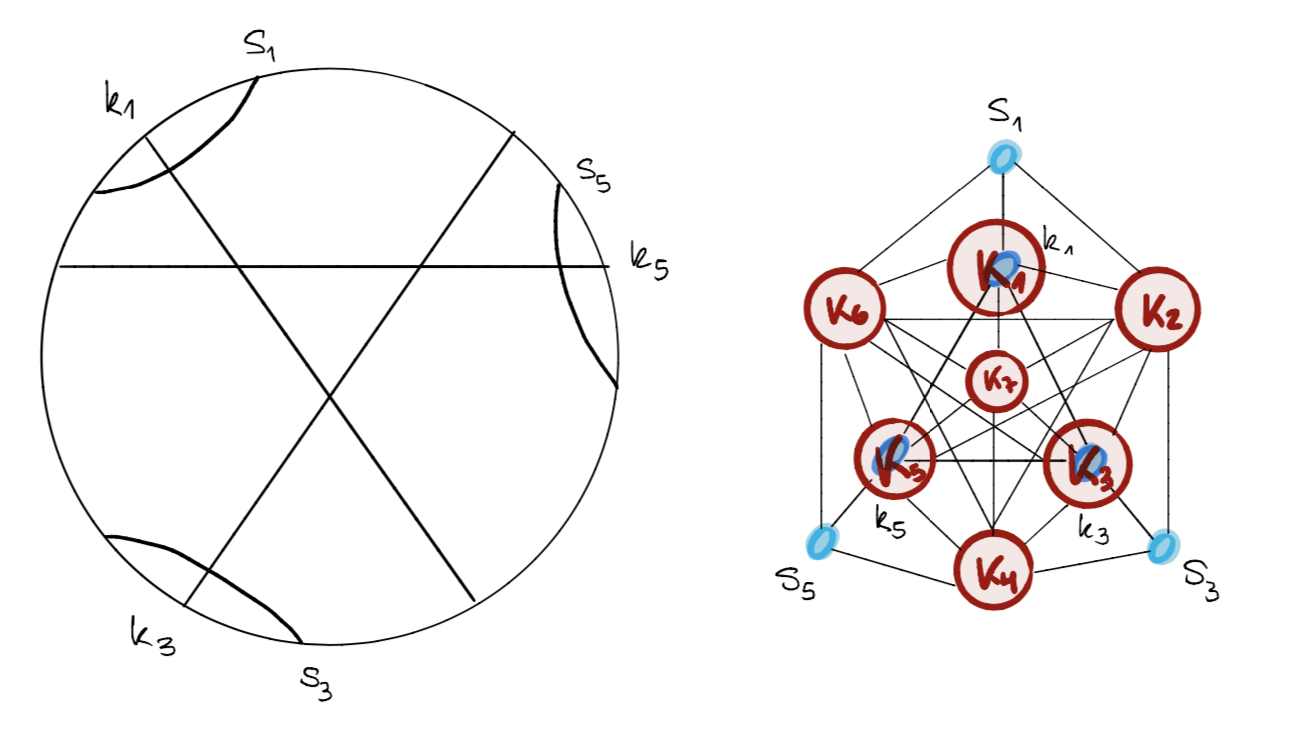}
		\label{fig:net_model}    
	  \caption{A circle model for the net graph and the partitions of $K$.}
\end{figure}

We define $K_i$ as the subset of vertices in $K$ that are adjacent only to $s_i$ if $i=1,3,5$, and if $i=2,4,6$ as those vertices in $K$ that are adjacent to $s_{i-1}$ and $s_{i+1}$. We define $K_7$ as the subset of vertices in $K$ that are nonadjacent to $s_1$, $s_3$ and $s_5$.
Let $s$ in $S$. We denote $T(s)$ to the vertices that are false twins of $s$.

\begin{remark}
The net is not a prime graph. Moreover, if $K_i = \emptyset$, $K_j = \emptyset$ for any pair $i, j \in \{2,4,6\}$, then $G$ is not prime.
For example, if $K_2 = \emptyset$ and $K_4 = \emptyset$, then a split decomposition can be found considering the subgraphs $H_1 = K_3 \cup T(s_3)$ and $H_2 = G \setminus T(s_3)$.
\end{remark}

Since in the proof we consider a minimally non-circle graph $G$, it follows from the previous remark that at least two of $K_2$, $K_4$ and $K_6$ must be nonempty so that $G$ results prime. However, in that case we find a $4$-tent as an induced subgraph. Therefore, as a consequence of this and the previous sections, we have now proven the characterization theorem given at the begining of the chapter.

\begin{teo}
Let $G=(K,S)$ be a split graph. Then, $G$ is a circle graph if and only if $G$ is $\{ \mathcal{T}, \mathcal{F}\}$-free (See Figures \ref{fig:forb_T_graphs2} and \ref{fig:forb_F_graphs2}).
\end{teo}


\part{Minimal completions}


\selectlanguage{spanish}%
\chapter*{Introducción}

Dado un grafo $G$ y una clase de grafos $\Pi$, un problema de modificación de grafos consiste en estudiar cómo agregar o borrar minimalmente vértices o aristas de $G$ de forma tal de obtener un grafo que pertenezca a la clase $\Pi$.

Dado que los grafos pueden ser utilizados para representar diversos problemas del mundo real y estructuras teóricas, no es difícil observar que los problemas de modificación permiten modelar un amplio número de aplicaciones prácticas en distintos campos. Algunos ejemplos son: redes; álgebra numérica; biología molecular; bases de datos, etc. Es entonces natural que estos problemas sean ampliamente estudiados. 

Una clase de grafos $\Pi$ es una familia de grafos que tiene la propiedad $\Pi$, por ejemplo, $\Pi$ puede ser la propiedad de ser cordal, o planar, o perfecto, etc. 

El problema de modificación que estudiamos es el problema de $\Pi$-completación.
Una $\Pi$-completación de un grafo $G=(V,E)$ es un supergrafo $H = (V, E \cup F)$ tal que $H$ pertenece a la clase $\Pi$ y $E \cap F = \emptyset$. En otras palabras, queremos hallar un conjunto de aristas $F$ tal que, al agregarlas a $G$, el grafo resultante pertenezca a la clase $\Pi$. Las aristas en $F$ se llaman \emph{aristas de relleno}.
Una $\Pi$-completación se dice \emph{mínima }si para todo conjunto de aristas $F'$ tal que $H'= (V, E \cup F')$ pertenece a $\Pi$, entonces $|F'| \geq |F|$. Una $\Pi$-completación es \emph{minimal }si para todo subconjunto $F' \subset F$, el supergrafo $H'= (V, E \cup F')$ no pertenece a $\Pi$.

El problema de calcular una completación mínima de un grafo arbitrario a una clase de grafos específica ha sido ampliamente estudiado, ya que tiene aplicaciones en áreas tan diversas como biología molecular, álgebra computacional, y más precisamente en aquellas áreas que involucran el modelado basado en grafos donde las aristas que no están se deben a la falta de data, como por ejemplo en problemas de clustering de datos \cite{GGKS95,NSS01}.
Desafortunadamente, las completaciones mínimas de grafos arbitrarios a clases de grafos específicas como los cografos, grafos bipartitos, grafos cordales, etc., son NP-hard de computar \cite{NSS01,BBD06,Y81}. 

Por esta razón, la investigación actual en este tópico se enfoca en hallar completaciones minimales de grafos arbitrarios a clases de grafos específicas de la forma más eficiente desde el punto de vista computacional. Más aún, aunque el problema de completar minimalmente es y ha sido muy estudiado, se desconocen caracterizaciones estructurales para la mayoría de los problemas para los cuales se ha dado un algoritmo polinomial para hallar tal completación. Estudiar la estructura de las completaciones minimales puede permitir hallar algoritmos de reconocimiento eficientes. 

Las completaciones minimales de un grafo arbitrario a un grafo de intervalos y de intervalos propios han sido estudiadas en \cite{CT13,RST06}.
En estos casos particulares, una completación minimal se puede hallar en $\mathcal{O}(n^2)$ y $\mathcal{O}(n+m)$, repsectivamente. Sin embargo, no hay resultados en la literatura que se refieran a la complejidad del problema de reconocimiento, en ninguno de los dos casos.

La motivación más conocida para los problemas de Modificación Mínima de Intervalos viene de la biología molecular, y es una de las razones principales por las cuales los grafos de intervalos comenzaron a estudiarse. En un artículo de 1959 \cite{B59}, Benzer mostró fuertes evidencias de que la colección de ADN que compone a un gen bacterial es lineal, de la misma forma que la estructura de los genes en los cromosomas. Esta estructura lineal puede ser representada como intervalos en la recta real que se solapan, ergo, como un grafo de intervalos. Sin embargo, el mapeo de la estructura genética se hace por observación indirecta. Es decir, esta estructura lineal no es observada de forma directa, sino que es inferida por cómo se pueden recombinar algunos de los fragmentos del genoma original.  
Para estudiar diversas propiedades de una cierta secuencia de ADN, la pieza original de ADN es fragmentada en pedazos más pequeños. Estos fragmentos luego son clonados varias veces utilizando varios métodos biológicos, y toman el nombre de clones. En este proceso, la posición de cada clon en el pedazo original de ADN se pierde, aunque como generalmente varias copias de la misma pieza de ADN se fragmentan de diferentes formas, algunos clones se solapan. El problema de reconstruir el arreglo original de los clones en la secuencia original se denomina mapeo físico de ADN. Decidir si dos clones se solapan o no es la parte crítica en la cual pueden surgir errores, dado que es un proceso basado en información parcial. 
Sabemos que, una vez que se decide una forma de ordenar estos clones consistente en términos de solapamiento, el modelo resultando debería representar a un grafo de intervalos. Sin embargo, puede haber algún falso positivo o falso negativo, debido a la interpretación errónea de alguno de los datos. Corregir el modelo para deshacerse de inconsistencias es entonces equivalente a remover o agregar aristas al grafo que representa el set de datos, de modo tal que se convierta en un grafo de intervalos. Por supuesto, queremos cambiar el grafo original lo menos posible. Más aún, cuando todos los clones tienen el mismo tamaño, es decir, que la secuencia de ADN ha sido fragmentada en partes iguales, el grafo resultando debe ser no sólo de intervalos, sino de intervalos propios. 

Se mostró en \cite{KF79, Y81, GJ79, GGKS95} que el problema de $\Pi$-completación es NP-completo si $\Pi$ es la familia de los grafos cordales, de intervalos o de intervalos propios.


\selectlanguage{english}%
\chapter{Introduction}

Given a graph $G$ and a graph class $\Pi$, a graph modification problem consists in studying how to minimally add or delete vertices or edges from $G$ such that the resulting graph belongs to the class $\Pi$.

As graphs can be used to represent various real world and theoretical structures, it is not difficult to see that these modification problems can model a large number of practical applications in several different fields. Some examples are: networks reliability; numerical algebra; molecular biology; computer vision; and relational databases. It is thus natural that such problems have been widely studied.

A graph class $\Pi$ is a family of graphs having the property $\Pi$, for example, $\Pi$ can be the property of being chordal, or planar, or perfect, etc.

The modification problem we studied is the $\Pi$-completion problem. 
A $\Pi$-completion of a graph $G=(V,E)$ is a supergraph $H = (V, E \cup F)$ such that $H$ belongs to $\Pi$ and $E \cap F = \emptyset$. In other words, we want to find a set of edges $F$ such that, when added to $G$, the resulting graph belongs to the class $\Pi$. The edges in $F$ are referred to as \emph{fill edges}.
A $\Pi$-completion is \emph{minimum }if for any set of edges $F'$ such that $H'= (V, E \cup F')$ belongs to $\Pi$, then $|F'| \geq |F|$. A $\Pi$-completion is \emph{minimal }if for any proper subset $F' \subset F$, the supergraph $H'= (V, E \cup F')$ does not belong to $\Pi$. 


The problem of calculating a minimum completion in an arbitrary graph to a specific graph class has been rather studied, since it has
applications in areas such as molecular biology, computational algebra, and more specifically in those areas that involve mod\-elling 
based in graphs where the missing edges are due to lack of data, for example in data clustering problems \cite{GGKS95,NSS01}.
Unfortunately, minimum completions of arbitrary graphs to specific graph classes, such as cographs, bipartite graphs, chordal graphs, etc., 
have been showed to be NP-hard to compute \cite{NSS01,BBD06,Y81}. 

For this reason, current research on this topic is focused in finding minimal com\-ple\-tions of arbitrary graphs to specific graph classes in the most
efficient way possible from the computational point of view.
And even though the minimal completion problem is and has been rather studied, structural characterizations are still unknown for most 
of the problems for which a polynomial algorithm to find such a completion has been given. 
Studying the structure of minimal completions may allow to find efficent recognition algorithms.

Minimal completions from an arbitrary graph to interval graphs and proper interval graphs have been studied in \cite{CT13,RST06}.
In these particular cases, a minimal completion can be found in $\mathcal{O}(n^2)$ and $\mathcal{O}(n+m)$ respectively, but there are no results in the literature that refer to the complexity of the recognition problem in both cases.

The most well known motivation for Minimum Interval Modification problems, comes from molecular biology, and it is one of the main reasons why interval graphs started being studied in the first place. In a paper from 1959 \cite{B59}, Benzer first gave strong evidences that the collection of DNA composing a bacterial gene was linear, just like the structure of the genes themselves in the chromosome. This linear structure could be represented as overlapping intervals on the real line, and therefore as an interval graph. However, mapping of the genetic structure is done by indirect observation. That is, such linear structure is not observed directly, but it is inferred by how various fragments of the original genome can be recombined.
In order to study various properties of a certain DNA sequence, the original piece of DNA is fragmented into smaller pieces. This fragments are then cloned many times using various biological methods, and take the name of clones. In this process the position of each clone on the original stretch of DNA is lost, but since usually many copies of the same piece of DNA are fragmented in different ways, some clones will overlap. The problem of reconstructing the original arrangements of the clones in the original sequence is called physical mapping of DNA. Deciding whether two clones overlap or not is the critical part where errors may arise, since it is a process based on partial information. We know that once we decide an arrangement of these clones consistent with the overlapping, the resulting model should represent an interval graph. However, there might be some false positive or false negatives, due to erroneous interpretation of some data. Correcting the model to get rid of inconsistencies is then equivalent to remove or add edges to the graph representing the dataset, so that it becomes interval. Of course we want to change it as little as possible. Moreover, when all the clones have the same size, i.e., the DNA sequence has been fragmented in equal parts, the resulting graph should be not only interval, but proper interval.

It was shown in \cite{KF79, Y81, GJ79, GGKS95} that the minimum $\Pi$-completion problem is NP-complete if $\Pi$ is the family of chordal, interval, or proper interval graphs.

In the following sections we give some basic definitions and state some of the known structural characterizations for chordal, interval and proper interval graphs, which will be useful in the next chapter.

\section{Basic definitions}

A graph $G$ is \emph{chordal} if every cycle of length greater or equal to $4$ has a chord, which is an edge that is not part of the cycle but connects two vertices of the cycle.

\vspace{0.5mm}
We say $G$ is an \emph{interval graph } if $G$ admits an intersection model consisting of intervals in the real line. It has one vertex for
each interval in the family and an edge between every pair of vertices represented by intervals that intersect.
In particular, $G$  is a \emph{unit interval graph } if there is a model in which every interval has length 1, and $G$ is a \emph{proper interval graph} if $G$ admits a model such that no interval is properly included in any other.
Interval, unit interval and proper interval graphs are all subclasses of chordal graphs.

The neighbourhood of a vertex  $x$ in $V$ is the set $N(x) = \{ v \in V \mid v \mbox{ is adjacent to } x \}$. 
If $X \subseteq V$, we define $N_{X} (w) = \{ v \in X \subseteq V \mid v \mbox{ is adjacent to } w \}$. When $X=V$ we will simply denote it $N(w)$.

\vspace{0.5mm}
Three independent vertices form an \emph{asteroidal triple (AT) } if, for each two, there is a path $P$ from one to the other such that $P$ does not pass through a neighbor of the third one.

\vspace{0.5mm}
Let $u$ and $v$ in $V$ be two nonadjacent vertices. A set $S \subseteq V$ is a \emph{$u,v$-minimal separator} if $u$ and $v$ belong to distinct connected components in $G \left[ V \setminus S \right]$, and $S$ is minimal with this property. 
We say indistinctly that $S$ is a minimal separator if such vertices $u$ and $v$ exist.

\vspace{0.5mm}
Let $G$ and $H$ be two graphs. We say that \emph{$G$ is $H$-free} if there is no subgraph isomorphic to $H$ in $G$.

\section{Known characterizations of interval and proper interval graphs}

We now give a list of properties and characterization theorems that will be strongly used in the following chapter.

\begin{lema} \label{lema_separadores_2comp} \cite{KK98}
	Let $G=(V,E)$ be a graph, and $S \subseteq V$. Then, $S$ is a minimal separator if and only if $G \left[ V \setminus S \right]$ has at least two connected components $C_1$, $C_2$ such that $N(C_1)=N(C_2)=S$.
\end{lema}

\begin{lema} \label{lema_separadores_KK} \cite{KK98}
	Let $G=(V,E)$ be a graph. If $a$ and $b$ are nonadjacent vertices in $G$, then there is a unique $a,b$-minimal separator $S$ such that $S \subseteq N(a)$.
\end{lema}

\begin{lema} \label{lema_separador_clique} \cite{D61}
	If $G = (V,E)$ is a chordal graph, then every minimal separator is a clique.
\end{lema}

\begin{teo} \label{teo_caract_intAT} \cite{LB62}
	$G$ is an interval graph if and only if $G$ is chordal and AT--free.
\end{teo}

\begin{teo} \label{teo_caract_int_prop} \cite{J92}
The following properties are equivalent:
	\begin{itemize}
		\item $G$ is a proper interval graph
		\item $G$ is chordal and contains no claw, net or tent as induced subgraphs (See Figure \ref{fig:forb_chordal})
		\item $G$ is an interval graph and contains no claws
	\end{itemize}
\end{teo}

\begin{figure}[h]
\centering
\includegraphics[scale=.4]{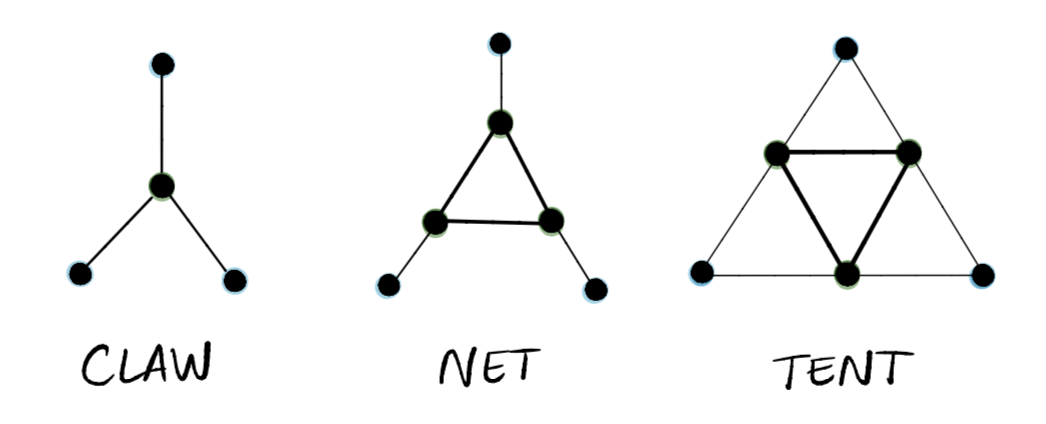}
\caption{Some of the forbidden induced subgraphs for proper interval graphs.} \label{fig:forb_chordal}
\end{figure}

\begin{teo} \label{teo_coinciden} \cite{R69}
	The class of unit interval graphs coincides with the class of proper interval graphs.
\end{teo}


\selectlanguage{spanish}%
\chapter*{Completaciones minimales de grafos de intervalos propios}

El resultado principal de este capítulo es el Teorema \ref{teo:caract_minimal_ifcase_}. En la primera sección se dan una serie de definiciones y propiedades básicas sobre separadores minimales en un grafo de intervalos.

\begin{defn} 
	Sea $G = (V, E)$ un grafo conexo, $S$ un separador minimal de $G$, y $C_i$ una componente conexa de $G \left[ V \setminus S \right]$.
	Definimos el \emph{núcleo $A_i(S)$} como el conjunto de vértices $v$ en $C_i$ tales que existe al menos un vértice $s$ en el separador $S$ de modo tal que $v$ y $s$ son adyacentes.
\end{defn}

\begin{prop} 
	Sea $G = (V, E)$ un grafo de intervalos propios conexo y sea $S$ un separador minimal de $G$. Entonces, todo núcleo $A_i$ es una clique para $i= 1,2$.
\end{prop}

\begin{defn} 
Sea $G$ un grafo de intervalos, $H$ una completación de $G$ a intervalos propios, y sea $e=(v,w)$ en $F$ una arista de relleno.
	\begin{enumerate}
		\item Decimos que \emph{$e$ es tipo I}, si existen un separador minimal $S$ de $H$ y un núcleo $A$ tales que $v$ y $w$ pertenecen ambos a $A$.
		\item Decimos que \emph{$e$ es tipo II}, si $e$ \textit{no es tipo I} y existen al menos un separador minimal $S$ de $H$ y un núcleo $A$ para los cuales $v$ está en $S$, $w$ está en $A$, de forma tal que si se borra $e$, entonces deja de haber un ordenamiento nuclear en $A$. 
		\item Decimos que \emph{$e$ es tipo III} si $e$ no es tipo I, existen al menos un separador minimal $S$ de $H$ y un núcleo $A$ para los cuales $v$ está en $S$, $w$ está en $A$, y para los cuales cada separador minimal $S$ y núcleo $A$, si $e$ se borra, entonces sigue existiendo un orden nuclear en $A$.
		\item Decimos que \emph{$e$ es tipo IV}, si $e$ no es tipo I y, para todo separador minimal $S$, o bien ambos $v,w \in S$ o bien ambos $v,w \not\in S$
	\end{enumerate}
\end{defn}

Esta definición induce una partición de las aristas en $F$. El resultado principal de este capítulo es la siguiente condición necesaria para que una completación a intervalos propios sea minimal cuando el grafo de input es de intervalos.

\begin{teo} \label{teo:caract_minimal_ifcase_}
	Sea $G = (V,E)$ un grafo conexo de intervalos y sea $H = (V, E \cup F)$ una completación de $G$ a intervalos propios. Si $H$ es minimal, entonces toda arista $e$ en $F$ o bien es de tipo I, o es de tipo II.
\end{teo}


\selectlanguage{english}%
\chapter{Minimal completion of proper interval graphs} \label{chapter:completions}

In this chapter, we study how to structurally characterize a minimal completion of an interval graph to a proper interval graph. In Section \ref{section:preliminares_completaciones}, we define and characterize some orderings for the vertices that are strongly based in the minimal separators of an interval graph. In Section \ref{section:teorema_completaciones}, we define the types of edges that can be found in any completion of an interval graph. Afterwards, we state and prove a necessary condition for a minimal completion in this particular case.

\section{Preliminaries} \label{section:preliminares_completaciones}

In this section, we will start giving some definitions and properties that will allow us to describe in the next section all the types of edges that can be found in a completion of an interval graph and state Theorem \ref{teo:caract_minimal_ifcase}. These definitions and properties include a necessary condition regarding the ordering of the vertices for any proper interval graph.

The following property allows us to assume from now on that the graph $G$ is connected.

\begin{prop} \cite{LMP10} \label{PIC_p1}
	Let $G = (V, E)$ be a graph, let $\mathcal{C} (G)= \{C_1, \ldots, C_k \}$ be the set of all connected components of $G$ and let $H = (V, E \cup F)$ be a $\Pi-$completion. 
	Then, $H$ is a minimal $\Pi-$completion of $G$ if and only if $H \left[ C_i \right]$
	is a minimal $\Pi$-completion of $G \left[ C_i \right]$ for every connected component $C_i \in \mathcal{C} (G)$.
\end{prop}

\vspace{0.5mm}
\begin{defn} \label{PIC_def1}
	Let $G = (V, E)$ a connected graph, $S$ a minimal separator of $G$, and let $C_i$ be a connected component of $G \left[ V \setminus S \right]$.
	We define the \emph{nucleus $A_i(S)$} as the set of vertices $v$ in $C_i$ for which there is at least one vertex $s$ in the separator $S$ such that $v$ and $s$ are adjacent.
		
	In this regard, $A_i(S)$ will refer as needed in each case by abuse of language both of the vertex set $A_i(S)$ and the induced subgraph $G \left[ A_i(S) \right]$.
	Moreover, we will use $A_i = A_i(S)$ whenever it is clear which is the minimal separator. 
\end{defn}

\begin{prop} \label{PIC_p2}
	Let $H = (V, E)$ be a connected proper interval graph. 
	Then, for every minimal separator $S$ of $H$, the subgraph $H \left[ V \setminus S \right]$ has exactly two connected components.
\end{prop}

\begin{proof}
		By Lemma \ref{lema_separadores_2comp}, there are at least two distinct connected components $C_1,C_2$ of $H \left[ V \setminus S \right]$ such that $N(C_1)=N(C_2)=S$.
	Toward a contradiction, let $C_3$ be a nonempty connected component of $H \left[ V \setminus S \right]$ such that $C_1 \neq C_3$ and $C_2 \neq C_3$.

		Notice that, if we consider any three vertices $x_i$ in $C_i$ for each $i=1, 2, 3$, then these vertices are nonadjacent.
		Since $C_3$ is nonempty and $H$ is connected, there are vertices $v_3$ in $C_3$ and $s$ in $S$ such that $v_3$ is adjacent to $s$.
		Similarly, let $v_1$ in $C_1$ and $v_2$ in $C_2$ such that $v_1$ and $v_2$ are both adjacent to the vertex $s$.
		Hence, the set $ \{ v_1,v_2,v_3, s \}$ induces a claw and this contradicts the hypothesis of $H$ being a proper interval graph.
\end{proof}

By proposition \ref{PIC_p2}, we will assume from now on that, if $H$ is a connected proper interval graph, then for every minimal separator $S$ of $H$, the subgraph $H \left[ V \setminus S \right]$ has \emph{exactly} two connected components.


\begin{prop} \label{PIC_p3}
	Let $H = (V, E)$ be a connected proper interval graph, $S$ a minimal separator of $H$ and let $A_i$ be a nucleus of the separator $S$, for $i=1,2$.
	For every pair of vertices $v, w$ in $A_i$ with a common neighbour $s$ in $S$, then $(v,w)$ is an edge of $E$.
\end{prop}

\begin{proof}
	Suppose to the contrary that $v$ and $w$ in $A_1$ are both adjacent to some vertex $s$ in $S$, and that the edge $(v,w)$ is not in $E$.
	 
	Since $S$ is a minimal separator and $H$ is connected, then $A_2$ is nonempty. Thus, let $z$ in $A_2$ such that $z$ is adjacent to $s$.
	Hence, the set $\{ v,w,s,z \}$ induces a claw in $H$ and this results in a contradiction.
\end{proof}

\begin{cor} \label{PIC_cor1}
	Under the previous hypothesis, if $|S| = 1$, then $A_i$ is a clique for $i=1, 2$.
\end{cor}


\begin{prop} \label{PIC_p4}
	Let $G=(V,E)$ be an interval graph, $S$ a minimal separator of $G$ such that $|S|>1$, and let $A$ be a nucleus of the separator $S$. 
	
	If $s_1$ and $s_2$ in $S$, then $N_{A} (s_1) \cap N_{A} (s_2)$ is a nonempty set.
\end{prop}

\begin{proof} 
	Let $s_1$ and $s_2$ in $S$. Suppose there are two nonadjacent vertices $v_1$ and $v_2$ in $A_1$ such that $s_1$ is adjacent to $v_1$ and nonadjacent to $v_2$, and $s_2$ is adjacent to $v_2$ and nonadjacent to $v_1$. 
	Since $C_1$ is connected, there is a simple path $\mathcal{P}$ in $C_1$ that joins $v_1$ and $v_2$.
If there is a vertex in $\mathcal{P}$ nonadjacent to either $s_1$ or $s_2$, then we find a cycle of length greater or equal than $4$. In particular, the same holds if $\mathcal{P} \cap (C_1 \setminus A_1)$ is nonempty because $s_1$ and $s_2$ are adjacent. 
	
	Hence, suppose that $\mathcal{P} \subseteq A_1$ and every vertex in $\mathcal{P}$ is adjacent to both $s_1$ and $s_2$.
Since $S$ is a minimal separator, there are vertices $x_1$ and $x_2$ in $A_2$ such that $x_i$ is adjacent to $s_i$ for $i=1,2$.
	In particular, since $x_1$ and $x_2$ are both in $C_2$ --which is a connected component of $G\left[ V\setminus S \right]$--, there is a path $\mathcal{P}^{'}$ joining $x_1$ and $x_2$ such that $\mathcal{P}^{'}$ is entirely contained in $C_2$.
	
	We claim that the set $\{x_1, v_1, v_2 \}$ induces an $AT$. It is clear that $x_1$, $v_1$ and $v_2$ are three independent vertices.
	If $x_1$ is also adjacent to $s_2$, then we have the path $\mathcal{P} \subseteq A_1$ connecting $v_1$ and $v_2$, and the following paths: 
	$$ \mathcal{P}_1 : x_1 \rightarrow s_1 \rightarrow v_1 $$ 
	$$\mathcal{P}_2 : x_1 \rightarrow s_2 \rightarrow v_2 $$ 
 
	The proof is analogous if $x_1 = x_2$. If instead $x_1$ is nonadjacent to $s_2$, then we have $\mathcal{P}$ joining $v_1$ and $v_2$, 
	$\mathcal{P}_1$ defined as above joining $x_1$ and $v_1$, and the path:
	$$ \mathcal{P}_2 : x_1 \xrightarrow[]{\mathcal{P}^{'}} x_2 \rightarrow s_2 \rightarrow v_2 $$ 	
 	
 	and thus $G$ is not an interval graph, which results in a contradiction. Hence, the vertices $v_1$ and $v_2$ are adjacent.
	However, since $v_1$ is adjacent to $s_1$, $v_2$ is adjacent to $s_2$, and $s_1$ is adjacent to $s_2$ for $S$ is a minimal separator of a chordal graph, either $v_1$ is adjacent to $s_2$, or $v_2$ is adjacent to $s_1$ and therefore $N_{A_1} (s_1) \cap N_{A_1} (s_2)$ is nonempty.
\end{proof}

\begin{cor}
	Under the hypothesis of Proposition \ref{PIC_p4}, if $N_{A_i} (s) = \{ v \}$, then $v$ is complete to $S$.
\end{cor}

\begin{prop} \label{PIC_p5}
	Let $G = (V, E)$ be a connected proper interval graph. 
	If $S$ is a minimal separator of $G$ such that $|S|>1$, then, for $i= 1,2$, every nucleus $A_i$ is a clique.
\end{prop}

\begin{proof}
	We will prove this result for $A_1$. If $|A_1| = 1$, then the proposition holds. 
	
	If $|A_1| = 2$, then by Propositions \ref{PIC_p3} and \ref{PIC_p4}, both vertices are adjacent.
	
	Suppose that $|A_1| \geq 3$, and let $v_1$ and $v_2$ in $A_1$ be two nonadjacent vertices. 
	
	By definition of nucleus, there are vertices $s_1$ and $s_2$ in $S$ such that $s_i$ is adjacent to $v_i$, for each $i=1,2$. 
	Since $v_1$ and $v_2$ are nonadjacent, by Proposition \ref{PIC_p3}, $s_1 \neq s_2$, $v_1$ is nonadjacent to $s_2$ and $v_2$ is nonadjacent to $s_1$.
	
	By Proposition \ref{PIC_p4}, there are vertices $w_1$ in $A_1$ and $w_2$ in $A_2$ such that $w_1$ and $w_2$ are adjacent to both $s_1$ and $s_2$. It is clear that $w_1 \neq v_1$ and $w_1 \neq v_2$.
	
	Since $v_1$ and $w_1$ are adjacent to $s_1$, by Proposition \ref{PIC_p3}, $v_1$ is adjacent to $w_1$, and the same holds for $v_2$ and $w_1$. 
	Therefore, the set $\{ v_1, z_1, v_2, s_1, s_2, z_2 \}$ induces a tent and this results in a contradiction, for the tent is a forbidden subgraph for proper interval graphs.
\end{proof}



\begin{defn} \label{PIC_def2}
	Let $G$ be a graph, $S$ a minimal separator of $G$ and $A$ a nucleus of $S$. 
	\emph{A nuclear ordering for $A$} is an ordering $v_1, \ldots, v_k$ of the vertices of $A$ such that for every pair of vertices $v_i$ and $v_j$,
	if $i<j$, then $N_{S}(v_i) \subseteq N_{S}(v_j)$.
	
	\textit{Notation:} If $\sigma$ is a nuclear ordering for the nucleus $A = \{ v_1, \ldots, v_k \}$, 
	we denote $v_1 <_{\sigma} v_2 <_{\sigma} \ldots <_{\sigma} v_k$.
\end{defn}

\begin{prop} \label{PIC_p6} 
	Let $H$ be a connected proper interval graph, $S$ a minimal separator of $H$ and $A$ a nucleus of $S$. 
	If $v_1$ and $v_2$ in $A$, then $N_{S}(v_1) \cap N_{S}(v_2)$ is nonempty. Moreover, there is a nuclear ordering $\sigma$ for $A$.
\end{prop}

\begin{proof}
	Let $v_1$ and $v_2$ in $A$.  Let us see that either $N_{S}(v_1) \subseteq N_{S}(v_2)$ or $N_{S}(v_2) \subseteq N_{S}(v_1)$.
	
	Toward a contradiction, suppose there is a vertex $s_1$ in $N_{S}(v_1)$ such that $s_1 \notin N_{S}(v_2)$, and a vertex
	$s_2$ in $N_{S}(v_2)$ such that $s_2 \notin N_{S}(v_1)$. 
			
	Since $S$ and $A$ are cliques --by Propositions \ref{PIC_p3} and \ref{PIC_p5}-- and $H$ is chordal,
	$v_1$ and $v_2$ are adjacent and also $s_1$ is adjacent to $s_2$. Thus, the set $\{ v_1,v_2,s_1,s_2 \}$ induces a $C_4$ and this results in a contradiction.	
	
	Therefore, either $N_{S}(v_1) \subseteq N_{S}(v_2)$ or $N_{S}(v_2) \subseteq N_{S}(v_1)$, and since any two vertices in $A$ are comparable, this induces a nuclear ordering in $A$.
\end{proof}

\begin{cor}
	For each nucleus $A$, there is a vertex $v \in A$ such that $v$ is complete to $S$.
\end{cor}

\begin{prop} \label{PIC_p7}
	Let $H$ be a proper interval graph and $S$ a minimal separator of $H$. 
	Then, there is a vertex ordering $s_1, s_2, \ldots, s_m$ for $S$ such that
	$$ N_{A_1}(s_1) \supseteq \ldots \supseteq N_{A_1}(s_m) \mbox{, and }$$
	$$ N_{A_2}(s_1) \subseteq \ldots \subseteq  N_{A_2}(s_m) $$
	
	We call this a \emph{bi-ordering for $S$}, and we denote it regarding the nucleus corresponding each direction. 
	For example, the previous would be denoted as
	$s_1 \geq_{A_1} \ldots \geq_{A_1} s_m$ and $s_1 \leq_{A_2} \ldots \leq_{A_2} s_m$.
\end{prop}

\begin{proof}
	Suppose to the contrary that there is a minimal separator $S$ of $H$ such that every decreasing ordering of its vertices regarding $A_1$ is not an increasing ordering regarding $A_2$.
	
	Let $s_1 \geq_{A_1} \ldots \geq_{A_1} s_m$ be a decreasing ordering of $S$ regarding $A_1$.
	Suppose without loss of generality $s_1 \not\leq_{A_2} s_2$, and $s_2 <_{A_2} s_1 \leq_{A_2} s_3 \leq_{A_2} \ldots \leq_{A_2} s_m$.

	Notice that, if $s_2 =_{A_2} s_1$, then the given ordering regarding $A_1$ holds for $A_2$, thus
	since the ordering is total between vertices in $S$, we may assume a strict ordering for $A_2$.
	
	Moreover, if $s_1 =_{A_1} s_2$, then we can swap $s_1$ and $s_2$ in the ordering regarding $A_1$ and thus this new ordering results in a bi-ordering for $S$.
	
	\vspace{0.5mm}
	Suppose $s_1 >_{A_1} s_2$. Hence, there is a vertex $x_1$ in $A_1$ such that $s_1$ is adjacent to $x_1$ and $s_2$ is nonadjacent to $x_1$.
	Let $x_2$ in $A_2$ such that $s_1$ is adjacent to $x_2$ and $s_2$ is nonadjacent to $x_2$. 
	We can find such a vertex for we are assuming $s_2 <_{A_2} s_1$.
	These four vertices induce a claw, and therefore this results in a contradiction since $H$ is proper interval.
	
	\vspace{0.5mm}
	This argument holds for every pair of vertices in $S$ for which the position given by the order in the other nucleus cannot be inverted.	
\end{proof}

\section{A necessary condition} \label{section:teorema_completaciones}

In this section, we will use the properties and definitions given in the previous section to define all the types of edges that may arise in a completion of an interval graph, and we will state and prove a necessary condition for any minimal completion to proper interval graphs when the input graphs is an interval graph, which is the main result of this chapter.

\begin{defn} 
Let $G$ be an interval graph, $H$ a completion of $G$ to proper interval, and let $e=(v,w)$ in $F$ be a fill edge.
	\begin{enumerate}
		\item We say \emph{$e$ is type I}, if there is a minimal separator $S$ of $H$ and a nucleus $A$ such that $v$ and $w$ are both vertices in $A$. \label{edge_typeI}
		\item We say \emph{$e$ is type II}, if $e$ \textit{is not type I} and there is at least one minimal separator $S$ of $H$ and a nucleus $A$ for which $v$ in $S$, $w$ in $A$, such that if $e$ is deleted, then there is no nuclear ordering in $A$. \label{edge_typeII}
		\item We say \emph{$e$ is type III} if $e$ is not type I, there is at least one minimal separator $S$ of $H$ and nucleus $A$ for which $v$ in $S$ and $w$ in $A$, and for each such minimal separator $S$ and nucleus $A$, if $e$ is deleted, then there is still a nuclear ordering in $A$. \label{edge_typeIII} 
		\item We say \emph{$e$ is type IV}, if $e$ is not type I and, for every minimal separator $S$, either both $v,w \in S$ or both $v,w \not\in S$ \label{edge_typeIV}
	\end{enumerate}
\end{defn}

Notice that this definition induces a partition of the edges in $F$.
Moreover, the definition of type IV edge can be restated as follows: \textit{$e$ is type IV if for every minimal separator $S$ such that $e$ and $S$ intersect, then $v$ and $w$ are both vertices in $S$.}

\begin{teo} \label{teo:caract_minimal_ifcase}
	Let $G = (V,E)$ be a connected interval graph and let $H = (V, E \cup F)$ be a completion of $G$ to proper interval. If $H$ is minimal, then every edge $e$ in $F$ is either type I or type II.
\end{teo}

\begin{proof}
	
	\vspace{1mm}
	Suppose $H$ is minimal. We will see that every edge is either type I or type II.
	Toward a contradiction, suppose there is an edge $e$ in $F$ such that $e$ is either a type III or type IV edge. If $e$ is removed, then we will find a subset $F'$ of $F$ for which $H' = (V, E \cup F')$ is a completion of $G$ to proper interval.
	
	\vspace{2mm}
	\begin{mycases}
	\case Suppose the edge $e$ is type III.

	\vspace{1mm}
	 Since $e$ is type III, there is a minimal separator $S$ and a nucleus $A_1$ such that $e = (s,v)$, with $s$ in $S$, $v$ in $A_1$. We denote $F' = F \setminus \{e\}$.
	
	If $H$ is minimal and $e$ in $F$ is deleted, then the resulting graph $H' = H \setminus \{ e \}$ is either not an interval graph, or $H'$ contains an induced claw. 
	Hence, by Theorems \ref{teo_caract_intAT} and \ref{teo_caract_int_prop}, we have three possible subcases:

	\begin{enumerate}[1)]
		\item The resulting subgraph $H'$ contains an induced cycle $C_n$, with $n \geq 4$ (thus, $H'$ is not a chordal graph), or \label{teo_casoa}
		\item $H'$ contains an AT (in this case, $H'$ is chordal but $H'$ is not an interval graph), or \label{teo_casob}
		\item $H'$ is an interval graph but contains an induced claw (thus, $H'$ is an interval graph and $H'$ is not a proper interval graph). \label{teo_casoc}
	\end{enumerate}

		
	\vspace{1mm}
	Let $W \subset V$ a vertex subset, and $F \subset E$ an edge subset. We denote by $N_{W,F}(v)$ to those neighbours of the vertex $v$ in $W$ that are connected to $v$ by edges in $F$. 
		
	\vspace{1mm}
	\begin{remark} \label{obs:nuclearordering_t3}
	Let $\sigma_1$ be a nuclear ordering for $A_1$ in $H$ given by \\*
	$ v_1 \leq_{\sigma_1} v_2 \leq_{\sigma_1} \ldots \leq_{\sigma_1} v_t $, such that $v_j = v$ for some $j$ in $\{ 1, \ldots, t \}$.
	
	Let $\sigma_2$ be the -partial- ordering induced by $\sigma_1$ in the nucleus $A_1$ once the edge $e$ is deleted, which we will refer to simply as the induced ordering and which we denote by $\leq_{\sigma_2}$.
	
	Since $e$ is type III, if $e$ is deleted, then we can find a nuclear ordering for $A_1$. However, we cannot assert that the induced ordering is indeed a nuclear ordering.
	
	A few observations:
	  
	\begin{itemize}\itemsep1pt
		\item The inclusion $N_{S,F'}(v_j) \subseteq N_{S,F'}(v_{j+i})$ holds for every $i$ in $\{ 1, \ldots, t-j \}$, thus, con\-sid\-er\-ing the edge set $E \cup F'$ we see that $v=v_j \leq_{\sigma_2} v_{j+1} \leq_{\sigma_2} \ldots \leq_{\sigma_2} v_t$ holds as for $\sigma_1$.
		\item Suppose $s \in N_{S, F}(v_j)$ and $s \not\in N_{S,F}(v_i)$ for every $v_i \leq_{\sigma_1} v_j$. Then, the induced ordering $\sigma_2$ does not change for $v_1, \ldots, v_j$.
		\item Suppose instead that $s \in N_{S, F}(v_j) \cap N_{S,F}(v_i)$ for some $v_i <_{\sigma_1} v_j$, then we set $k$ to be $\min \{ i : s \in N_{S,F}(v_i) \}$. Notice that $k < j$. 
		  If $e$ is deleted, then $s \notin N_{S, F'} (v_j)$. However, since $N_{S,F}(v_k) \subseteq N_{S,F}(v_j)$ and $s \in N_{S, F'} (v_k)$, then $N_{S,F'}(v_j) \subset N_{S,F'}(v_k)$ and hence we have that $N_{S,F}(v_k) = \ldots = N_{S,F}(v_j)$, since $s$ is the only element removed from the neighbourhood of $v_j$. 
			
			Therefore, the induced ordering $\sigma_2$ must necessarily be $$ v_1 \leq_{\sigma_2} \ldots \leq_{\sigma_2} v_{k-1} \leq_{\sigma_2} v_j \leq_{\sigma_2} v_k \leq_{\sigma_2} \ldots \leq_{\sigma_2} v_t $$
	\end{itemize}
	\end{remark}	
	
	\vspace{1mm}	
	\subcase Suppose that if $e = (s, v)$ is deleted, then we find a cycle. 
	Since $S$ and $A_1$ are cliques and $H$ is chordal, this cycle must have length 4 at the most. Moreover, it is induced by a set $\{v, s, w_1, s_1\}$ for some vertices $w$ in $A_1$ and $s_1$ in $S$ such that $v$ is adjacent to $w$ and $s_1$, and $w$ is adjacent to $s$. 
	
	Since $s_1 \in N_{S,F}(v)$ and $s_1 \notin N_{S,F}(w)$, thus $w <_{\sigma_1} v$ and the inequality is \emph{strict}.
	By Remark \ref{obs:nuclearordering_t3}, if $e$ is deleted, then the induced ordering $\sigma_2$ satisfies 
	$N_{S, F'}(w) = N_{S,F'}(v)$. However, $s \in N_{S,F'}(w)$ which results in a contradiction.

\begin{remark} \label{obs:5partitiones_sep}
	For each minimal separator $S$, we can partition the vertices of the graph into 5 disjoint sets: $C_1\setminus A_1, A_1, S, A_2$ and $C_2\setminus A_2$ (see Figure \ref{fig_separadores}).
	
	Since $S$, $A_1$ and $A_2$ are cliques, the only way two independent vertices may belong to the same set is if they both lie in either $C_1\setminus A_1$ or $C_2\setminus A_2$.
	
	
	\begin{center}
	  \includegraphics[scale=.3]{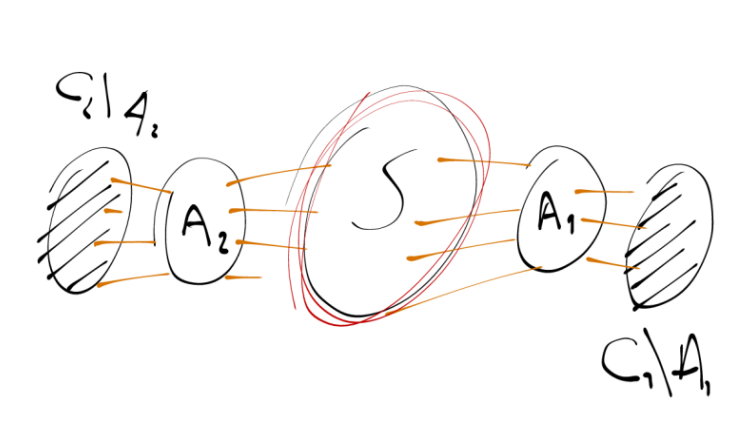}
	  \captionof{figure}{Scheme of the partition of the graph $H$}
		    \label{fig_separadores}
	\end{center}

\end{remark}

	\vspace{2mm}
	
	\subcase Suppose now that if $e =(s,v)$ is deleted, then there is an $AT$ in the subgraph $H' = (V, E \cup F') = H \setminus \{e\}$ induced by some independent vertices $w_1, w_2$ and $w_3$.
	 
	Since there are no $AT$'s in $H$ (for $H$ is an interval graph), there is a path $P_{1,2}$ in $H^{'}$ joining $w_1$ and $w_2$, such that 
	there is a vertex $w$ in $P_{1,2}$ adjacent to $w_3$ through the edge $e$. 
	Hence, $w$ is nonadjacent to $w_3$ in $H^{'}$. Thus, either $w=v$ and $w_3=s$, or $w=s$ and $w_3=v$.
	
	\vspace{1mm}
	Let us suppose first that $w=v$ and $w_3 = s$.
	
	\begin{claim} \label{claim:PIC_1}
		Under the previous hypothesis, $w_1$ and $w_2$ are both in $C_1 \setminus A_1$.
	\end{claim}
	
	To prove this, we divide in cases according to the 5 partitions described in Remark \ref{obs:5partitiones_sep}.	
	
	First of all, since $S$ is a clique and $w_3=s$ lies in $S$, then $w_1 \notin S$ and $w_2 \notin S$. Furthermore, since $A_1$ and $A_2$ are cliques, the vertices $w_1$ and $w_2$ cannot belong to the same nucleus.
	
	On one hand, we may assert that $w_1 \notin C_2$, for if this is the case, since $w$ lies in $C_1$ and $w$ is a vertex in $P_{1,2}$, then the path $P_{1,2}$ goes through the set $S$ and thus, the path contains at least one neighbour of $w_3$ in $S$, which results in a contradiction for $w$ is, by hypothesis, the only vertex adjacent to $w_3$ in the path $P_{1,2}$.

	In an analogous way, we may assert that it is not possible to have $w_1$ in $A_1$ and $w_2$ in $C_1 \setminus A_1$, for we cannot find a path joining $s$ and $w_2$ without going through neighbours of $w_1$ in $A_1$.
	
	Therefore, the only remaining possibility is $w_1$ and $w_2$ in $C_1 \setminus A_1$. \QED
	
	\vspace{3mm}
	Let us study now the relationship between $w$ and $w_1$, $w_2$. A couple of observations:
	\begin{enumerate}[(1)]
		\item There is no path joining $w_1$ and $w_2$ entirely contained in $C_1\setminus A_1$, for if this was the case, then we can find an $AT$ in $H$, which results in a contradiction since $H$ is an interval graph. \label{item:enumerate_case1_2_1}
	
	\item Since the set $\{ w_1, w_2, w_3 \}$ induces an $AT$ in $H^{'}$ and $w_3$ is adjacent to $w$ through $e$, the vertex $w$ is nonadjacent to either $w_1$ or $w_2$ for if not, then we find a claw in $H$ induced by $\{w_1, w_2, w, w_3 \}$.
	Notice that this implies that the set $N_{A_1}(w_1) \cap N_{A_1}(w_2)$ is empty, since by definition every vertex in a nucleus is adjacent to at least one vertex in the separator, and thus the same argument holds. \label{item:enumerate_case1_2_2}
	\end{enumerate}
	
	Summing up the results in (\ref{item:enumerate_case1_2_1}), (\ref{item:enumerate_case1_2_2}) and Claim \ref{claim:PIC_1}, $w$ is nonadjacent to either $w_1$ or $w_2$, and thus there are vertices $v_1$ and $v_2$ in $A_1$ such that $v_1$ is adjacent to $w_1$ and is nonadjacent to $w_2$, and analogously $v_2$ is adjacent to $w_2$ and is nonadjacent to $w_1$. 
	Notice that $v_1$ is adjacent to $v_2$ since they both lie in the same nucleus.
	
	Suppose first that $w \neq v_1$ and $w \neq v_2$. Hence, the path $w_1 \rightarrow v_1 \rightarrow v_2 \rightarrow w_2$ joins $w_1$ and $w_2$ in $H$ and contains no neighbour of $w_3$, therefore $\{w_1, w_2, w_3 \}$ is an $AT$ in $H$, which results in a contradiction.

	\vspace{0.5mm}
	Suppose now that $w \neq v_1$ and $w = v_2$.
	First of all, if $N_{A_1}(w_1) \cap N_{A_1}(s)$ is nonempty, then we can find a $w_1,w_2$-minimal separator such that $e$ belongs to one of the nucleus as follows: 
	Let $S' = N_{A_1}(w_1)$. Since there is no path connecting $w_1$ and $w_2$ entirely included in $C_1\setminus A_1$, $S'$ results in a minimal separator such that $e$ lies in one of the nucleus, which is not possible since $e$ is type III.
	
	Hence, $N_{A_1}(w_1) \cap N_{A_1}(s)$ is empty. Let $x$ be a vertex in $N_{A_1} (w_1)$ such that $x$ is nonadjacent to $s$. 
	Since $x$ in $A_1$ and using the definition of nucleus, there is a vertex $s_1$ in $S$ such that $s_1$ is adjacent to $x$ and $s_1 \neq w_3$. 
	Since $w = v$ is adjacent to $w_3 = s$ in $H$ and $x$ is nonadjacent to $w_3$, thus $w >_{\sigma_1} x$ and therefore $w$ is adjacent to $z$ for every $z$ in $N_S(x)$. 
	In particular, $w$ is a neighbour of $s_1$ (see Figure \ref{fig:C1-A1}).

\begin{center}
	  \includegraphics[scale=.6]{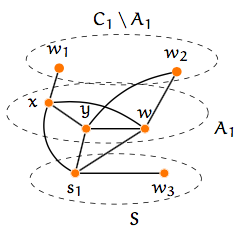}
	  \captionof{figure}{$w_1, w_2 \in C_1 \setminus A_1$ nonadjacent; $w_3 = s$ and $w = v$.}
		    \label{fig:C1-A1}
\end{center}

	Let $w'$ in $A_2$ adjacent to $s_1$. We have the following paths:
	$$ P_1: w' \rightarrow s_1 \rightarrow x \rightarrow w_1 $$
	$$ P_2: w' \rightarrow s_1 \rightarrow w \rightarrow w_2 $$
	$$ P_3: w_1 \rightarrow x \rightarrow w \rightarrow w_2 $$
	
	None of these paths goes through neighbours of the excluded vertex in each case, and $e \not\in P_i$ for each $i=1,2,3$.
	Therefore, $\{ w', w_1, w_2 \}$ induces an AT in $H$ and this contradicts the hypothesis of completion.

	\vspace{4mm}
	Conversely, suppose that $w=s$ and $w_3 = v$.			
	It is straightforward that $w_1$ and $w_2$ do not belong to $A_1$, for $w_3 \in A_1$ and $A_1$ is a clique. 
	Moreover, if $w_1$ lies in $C_1 \setminus A_1$, then every path joining $w_1$ and $w_2$ goes through neighbours of $w_3$ in $A_1$, unless such a path is entirely contained in $C_1 \setminus A_1$, including both vertices $w_1$ and $w_2$.
	Moreover, notice that if there is a path joining $w_1$ and $w_2$ entirely contained in $C_1 \setminus A_1$, then we find an $AT$ in $H$ given by $\{w_1, w_2, w_3\}$, for we have a path joining $w_1$ and $w_2$ that does not contain the edge $e$ and the paths in $H'$ joining every other pair of vertices in the $AT$, which results in a contradiction.	
		
	Hence, if there is a path joining $w_1$ and $w_2$ that goes through $s$ to avoid every other neighbour of $w_3$, then $w_1$ must lie in $S$ and $w_2$ in $C_2$, for they do not belong to the clique $A_1$ and also they do not lie in $C_1 \setminus A_1$. 
	Furthermore, $w_2 \notin C_2$ since any path joining $w_2$ and $w_3$ goes through neighbours of $w_1$ in $S$, therefore this case is not possible either.

	\vspace{2mm}
	
	\subcase Suppose that we delete $e$ and find an induced claw. Such a claw is induced by $v$, $s$ and two more vertices $w_1$ and $w_2$.
	
	Since $v$ and $s$ are nonadjacent in $H'$, $w_1$ is nonadjacent to $s$ and $v$, and $w_2$ is adjacent to $v$, $s$ and $w_1$. 
	If $w_1$ in $C_1\setminus A_1$, then we can find a subset $T$ of $N_{A_1}(w_1)$ such that $T$ is a $w_1,v$-minimal separator. Since $w_2$ is adjacent to $v$, $w_1$ and $s$, then $e$ is contained in one of the nucleus of $T$, which results in a contradiction since $e$ is not type I.
	
	\vspace{0.5mm}
	The other possibility, is having a vertex $w_2$ in $S$ adjacent to $w_1$, $v$ and $s$, and $w_1$ in $A_2$ nonadjacent to $s$. 
	
	By Lemma \ref{lema_separadores_KK}, there is exactly one $w_1,s-$minimal separator $T$ such that $T \subset N(w_1)$. 
	Applying the definition of $w_1,s-$minimal separator and since $N_S(w_1) \subseteq N_S(s)$, then $w_1$ lies in one of the nucleus $A_1(T)$ and $s \in A_2(T)$. 
	Furthermore, $v \notin T$ and $w_2$ in $T$, thus $e$ is contained in the nucleus $A_1(T)$, for $w_2$ is adjacent to both $v$ and $s$, and this contradicts the hypothesis of $e$ not being a type I edge.
	
	Therefore, since for every subcase \ref{teo_casoa}, \ref{teo_casob} and \ref{teo_casoc} the hypothesis of minimality does not hold, then the edge $e$ is not type III.

	\vspace{4mm}
	\case Suppose that the edge $e$ is type IV.
	\vspace{1mm}
			
	Let $S$ be a minimal separator such that $e = (s_1, s_2)$ for $s_1$ and $s_2$ in $S$. Suppose first that $s_1$ is not universal in $H$, thus there is a vertex $v$ in $V$ nonadjacent to $s_1$. 
	By Lemma \ref{lema_separadores_KK}, there is exactly one $v,s_1$-minimal separator $S'$ contained in $N(s_1)$. Suppose without loss of generality that $v$ in $A_1(S')$ and $s_1$ in $A_2(S')$. Since $s_2$ in $N(s_1)$, hence $s_2$ in $A_2(S')$ or $s_2$ in $S'$, which results in a contradiction since $e$ is type IV. Therefore, $s_1$ is a universal vertex and the proof is analogous by symmetry for $s_2$.
	
	Notice that, since $s_1$ and $s_2$ are universal vertices in $H$, for each minimal separator $S$, the sets $C_i(S) \setminus A_i(S)$ are empty for $i=1,2$.
	
	If the edge $e$ is deleted, then the resulting graph $H^{'}$ is not chordal and has two kinds of cycles: the ones induced by the vertices $s_1$, $s_2$, any vertex $v_1$ in $A_1$ and any vertex $v_2$ in $A_2$, and, if $|S|>2$, the cycles induced by the vertices $s_1$, $s_2$, any vertex $v$ in a nucleus $A_i$ and some other vertex $s_3$ in $S$. 
	
	In the sequel, we will find a subset $J$ of fill edges such that the proper subset $F \setminus (J \cup \{e\})$ of $F$ results a completion of the original graph $G$ to a proper interval graph, and thus contradicting the minimality of $H$.	
	
	\subcase  We will suppose first that $S$ has exactly three elements $s_1$, $s_2$ and $s_3$, and once this is proved we will see the case $|S|=2$.

	Let $B_i, B_j$ be a partition of the nucleus $A_1$. Thus, $|B_i|=i, |B_j|=j$ for some $i,j=0,\ldots,|A_1|$ and $i+j=|A_1|$.	
	
	For each partition $B_i$, $B_j$ of the vertices in the nucleus $A_1$, we denote $F_{i,j}$ to the edge subset $\{ (s_1, b ) : b \in B_j \} \cup \{ (s_2, b ) : b \in B_i \}$. Analogously, we define $F'_{i,j}$ for every partition $D_i$, $D_j$ of the vertices in the nucleus $A_2$. 
	
	Let $a_1$ in $A_1$ and $a_2$ in $A_2$. Both vertices are adjacent to $s_1$ and $s_2$. When $e$ is deleted, there is a $C_4$ in $H'$ induced by the set $\{ a_1, s_1, a_2, s_2 \}$.	 
	Thus, there is either a partition $B_i$, $B_j$ of $A_1$ for some $i, j = 0, \ldots, |A_1|$, $i + j = |A_1|$, such that $F_{i,j}$ is a subset of $F$, or there is a partition $D_i$, $D_j$ of $A_2$ for some $i, j = 0, \ldots, |A_2|$, $i + j = |A_2|$, such that $F'_{i,j}$ is a subset of $F$. This follows, for if not, $G$ would not be not chordal since $s_1$ and $s_2$ are universal vertices and thus, in particular, $s_1$ and $s_2$ are adjacent to every vertex in $A_1$ and $A_2$.
	
	Suppose without loss of generality that there is a partition $B_i$, $B_j$ of $A_1$ such that $F_{i,j}$ is a subset of $F$ and $F_{i,j} \neq \emptyset$.
	
	Furthermore, let $a_2$ in $A_2$, $b_1$ in $B_i$ and $b_2$ in $B_j$. 
	Since $A_1$ is a clique, the subset $\{ b_1, b_2, s_2, a_2, s_1 \}$ induces a cycle in $H \setminus (F_{i,j} \cup \{e\})$.
	Hence, the edge subset $F_1 = \{ (b_1, b_2) : b_1 \in B_i, b_2 \in B_j \}$ is a subset of $F$.		
	
	Let $B_i, B_j$ be a partition of the nucleus $A_1$ as stated above.
	For each partition $B_i$,$B_j$, we denote $X_{i,j}(A_1)$ to the subgraph of $H$ resulting of deleting the edge $e$, every edge in $F_1$, and every edge in $F_{i,j}$. We denote $X_{i,j}(A_2)$ to the subgraph of $H$ defined analogously by a partition $D_i$, $D_j$ of the nucleus $A_2$.
	
	For a graphic idea of this definition see Figure \ref{fig:Xij}.
	
	\begin{center}
	    \includegraphics[scale=.4]{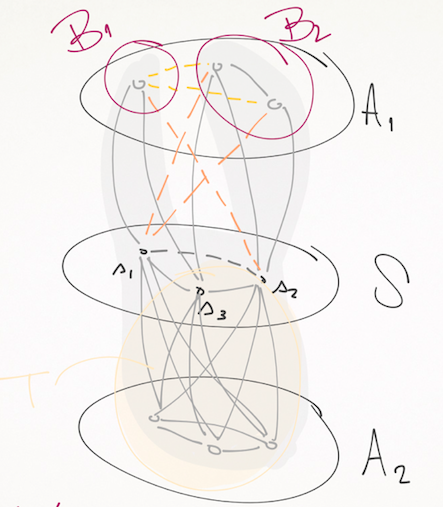}
	    \captionof{figure}{An example of a subgraph $X_{i,j}(A_1)$.}
	    \label{fig:Xij}
	\end{center}	
	
	As a consequence of the previous paragraphs, we have the following claim.
	
	\begin{claim}
		Under the previous hypothesis, there is either a partition $B_i$, $B_j$ of the nucleus $A_1$ or a partition $D_l$, $D_k$ of the nucleus $A_2$ such that $G$ is a subgraph of $X_{i,j}(A_1)$ or $X_{l,k}(A_2)$.
	\end{claim}
		
	Suppose without loss of generality that $B_i, B_j$ is a partition of $A_1$ such that $G$ is a subgraph of $X_{i,j}(A_1)$, and let $J = F_1 \cup F_{i,j} \cup \{e\}$ be the subset of every fill edge in $H$ that was deleted to obtain $X_{i,j}(A_1)$. 
	
	\begin{remark} \label{obs:indep_Xij}
		There is no independent set of size 3 or more in $X_{i,j}(A_1)$.

		Toward a contradiction, suppose there are independent vertices. Hence, the only possibility is $v$ in $A_1$, $w$ in $A_2$ and $s$ in $S$.
		Remember that $s_1$ and $s_2$ are universal vertices, $s$ is nonadjacent to both $v$ and $w$. Thus, since the vertices $s_1$ and $s_2$ are complete to $A_2$ in the subgraph $X_{i,j}(A_1)$, then $s \neq s_2$ and $s \neq s_1$. 
		On the other hand, let $s$ in $S$ such that $s \neq s_1$ and $s \neq s_2$. If there are vertices $v_1$ in $A_1$ and $v_2$ in $A_2$ such that $s$ is nonadjacent to both $v_1$ and $v_2$, then we find a claw in $H$ induced by $\{s_1, s, v_1, v_2 \}$. 
		Hence, $s$ is complete in $H$ to either $A_1$ or $A_2$. Since $J$ does not contain any edges for which $s$ is an endpoint, then $s$ is complete in $X_{i,j}(A_1)$ to either $A_1$ or $A_2$. Therefore, it is not possible to find three independent vertices in $X_{i,j}(A_1)$. 
		Moreover, this also proves that there are no $AT$'s in $X_{i,j}(A_1)$. 
	\end{remark}
	
	If $i=0$, then $j = |A_1|$ and it is easy to see by the previous remark that $X_{i,j}(A_1)$ is chordal, $AT$-free and claw-free.  
	Since $\varnothing \neq J \subseteq F$, then $X_{i,j}(A_1)$ is a completion of $G$ to proper interval graphs and this contradicts the hypothesis of $H$ being minimal.
	
	\vspace{1mm}
	Suppose that $i >0$ and $j>0$. By hypothesis, there are three vertices $s_1$, $s_2$ and $s_3$ in $S$.	
	If $N_{A_1}(s_3) \subseteq B_i$, since $\varnothing \neq B_i \neq A_1$, then we define the subset of fill edges 
	$$J_1 = F \setminus \{ (s_1, v) \in F : v \in B_j \}$$
	
	Notice that $e$ in $J_1$. Let $H_1 = (V, E \cup J_1)$.
	By Remark \ref{obs:indep_Xij}, it is clear that the subgraph $H_1$ is $AT$-free and claw-free. 
	Moreover, $H_1$ is chordal, for it is easy to see that either $N_{A_1}(s_1) \subseteq N_{A_1}(s_3) \subseteq N_{A_1}(s_2)$, or $N_{A_1}(s_3) \subseteq N_{A_1}(s_1) \subseteq N_{A_1}(s_2)$. 
	Since $J_1$ is a proper subset of $F$, $H$ is not a minimal completion of $G$ and this results in a contradiction.
	
	Analogously, if neither $B_i \not\subseteq N_{A_1}(s_3)$ and $B_j \not\subseteq N_{A_1}(s_3)$, then we define the subset of edges
	 
	$$J_2 = F \setminus \{ (s_1, v) \in F : v \in B_j \setminus N_{A_1}(s_3) \}$$
	
	We define the subgraph $H_2 = (V, E \cup J_2)$, and thus the same argument used for $H_1$ holds for $H_2$.
	
	\vspace{1mm}
	\subcase If $|S|=2$, then we claim that any graph $X_{i, j} (A_1)$ is a proper interval graph since it suffices to see that it is chordal and $AT$-free, thus we contradict the minimality.

	\vspace{0.5mm}	
	\subcase Finally, suppose that $|S|>3$. If $i=0$ and $j = |A_1|$, then we use the same argument as if $|S| = 3$. 
	Suppose that $i>0$ and $j>0$.
	
	Let $X$ be the subset of $S$ defined as 
	$$\{x \in S : x \neq s_1, x \neq s_2 \mbox{ and } x \mbox{ is not complete to } A_1 \}$$ 
	
	If $X = \varnothing$, then we define the subset of edges $J$ as in the previous case.
	
	\vspace{2mm}
	Suppose that $X$ is nonempty. 
	Let $s_3$ in $X$ be a vertex such that $N_{A_1}(s_3) \supseteq N_{A_1} (x)$, for every vertex $x$ in $X$.
	
	If $B_i \not\subseteq N_A(s_3)$ and $B_j \not\subseteq N_A(s_3)$, then we define the subgraph $H_2 = (V, E \cup J_2)$ as in the previous case with the subset of edges $J_2$. 
	
	If instead either $B_i \subseteq N_A(s_3)$ or $B_j \subseteq N_A(s_3)$, then we define the subgraph $H_1 = (V, E \cup J_1)$ as in the previous case with the subset of edges $J_1$.
	
	In both cases, we find a proper subgraph $H_i$ of $H$ such that $H_i$ is a completion of $G$ to proper interval, and this results in a contradiction of the minimality.
	
	\vspace{1mm}
	Therefore, if the completion is minimal, then there are no type III or type IV edges.
	
\end{mycases}	
\end{proof}	


\selectlanguage{spanish}%
\chapter*{Conclusiones y trabajo futuro}

Los resultados más importantes en esta tesis son el Teorema \ref{teo:circle_split_caract} en el Capítulo \ref{chapter:split_circle_graphs} de la Parte I, y el Teorema \ref{teo:caract_minimal_ifcase} en el Capítulo \ref{chapter:completions} de la Parte II. En el Teorema \ref{teo:circle_split_caract}, damos una caracterización por subgrafos miniamles prohibidos para aquellos grafos split que son circle, y en el Teorema \ref{teo:caract_minimal_ifcase} enunciamos y demostramos una condición necesaria para que la completación de un grafo de intervalos a un grafo de intervalos propios sea minimal. 

\section*{Parte I}

Los capítulos \ref{chapter:partitions} y \ref{chapter:2nested_matrices} están dedicados a construir las bases y herramientas necesarias para demostrar el Teorema \ref{teo:circle_split_caract}.
Más precisamente, definimos las matrices $2$-nested, y enunciamos y demostramos una caracterización de estas matrices por subconfiguraciones prohibidas que nos permite representar y caracterizar las matrices de adyacencia de aquellos grafos split estudiados en el Capítulo \ref{chapter:split_circle_graphs}.
Algunos de los resultados dados en el Capítulo 3 han sido publicados en \cite{PDGS19}, y el resto de los resultados del capítulo están siendo preparados en un manuscrito para ser enviados para su publicación. 
En el Capítulo \ref{chapter:split_circle_graphs} nos referimos al problema de caracterizar los grafos círculo al restringirlos a la clase de grafos split. A su vez, este capítulo se divide en 5 secciones: una introducción a los resultados estructurales conocidos para la clase de grafos círculo, y una sección para cada caso de la prueba del Teorema \ref{teo:circle_split_caract}. Este resultado devino en una caracterización por subgrafos inducidos prohibidos para aquellos grafos split que además son círculo. Por su parte, este resultado será enviado pronto para su publicación. 

Dejamos algunas posibles continuaciones de este trabajo.

\begin{itemize}
\item Hallamos una lista de subgrafos inducidos prohibidos para aquellos grafos split que además son círculo. ¿Es cierto que esta lista de subgrafos es además minimal?

\item Recordemos que los grafos split son aquellos grafos cordales para los cuales su complemento también es cordal, y que el grafo $A''_n$ con $n = 3$ que aparece re\-pre\-sen\-ta\-do en la Figura \ref{fig:local_complement_A''3} es un ejemplo de grafo cordal que no es ni círculo ni split. Se desprende de este ejemplo que el Teorema \ref{teo:circle_split_caract} no vale si consideramos un grafo cordal en vez de un grafo split, dado que hay más subgrafos prohibidos que no son considerados en la lista de subgrafos prohibidos dada. Sin embargo, el Teorema \ref{teo:circle_split_caract} es de hecho un buen primer paso para caracterizar a los grafos círculo por subgrafos inducidos prohibidos dentro de la clase de los grafos cordales, el cual aún permanece como problema abierto. 

\item Los grafos split pueden ser reconocidos en tiempo lineal. ¿Cuál es la complejidad de reconocimiento de los grafos círculo dentri de la clase de los grafos split?

\item Otra continuación posible de este trabajo podría ser estudiar la caracterización de aquellos grafos círculo cuyo complemento también es un grafo círculo. 

\item Caracterizar a los grafos Helly circle por subgrafos inducidos prohibidos. La clase de los grafos Helly circle fue caracterizada por subgrafos inducidos prohibidos dentro de la clase de los grafos círculo \cite{DGR10}.
Más aún, sería interesante hallar una descomposición análoga a la descomposición split para los grafos círculo, es decir, de forma tal que los grafos Helly circle sean cerrados bajo esta descomposición.
\end{itemize}

\section*{Parte II}

En el Capítulo 2, damos algunas propiedades referidas al ordenamiento de los vértices de un grafo de intervalos basado en los separadores minimales del grafo que vale tanto para grafos de intervalos como para grafos de intervalos propios, y definimos una par\-ti\-ción de las aristas de relleno de acuerdo a su relación con los separadores minimales del grafo. En la última parte de este capítulo, dada una completación $H$ a intervalos propios de un grafo de intervalos $G$, enunciamos y demostramos una condición necesaria para que $H$ sea minimal.

Con respecto al problema de la completación minimal estudiado en el Capítulo \ref{chapter:completions} de la Parte II, tenemos las siguientes conjeturas:

\begin{conj}
	Conjeturamos que la condición dada en el Teorema \ref{teo:caract_minimal_ifcase} también es una condición suficiente. Además, en ese caso la complejidad de completar de forma minimal a intervalos propios cuando el grafo de input es de intervalos resulta polinomial. 
\end{conj}

\begin{conj}
La completación mínima a intervalos propios cuando el grafo de input es de intervalos es un problema NP-completo.
\end{conj}

Queremos continuar trabajando en estas conjeturas para obtener resultados más fuertes para publicar en un artículo.  

Siguiendo una línea similar que la que nos llevó a estudiar el problema del Capítulo \ref{chapter:completions}, permanece abierto el problema de caracterizar y determinar la complejidad de la completación mínima y minimal de grafos arco-circulares a grafos arco-circulares propios.


\selectlanguage{english}%
\chapter*{Final remarks and future work}
\addcontentsline{toc}{chapter}{Final remarks and future work}

The main results in this thesis are Theorem \ref{teo:circle_split_caract} in Chapter \ref{chapter:split_circle_graphs} of Part I, and Theorem \ref{teo:caract_minimal_ifcase} in Chapter \ref{chapter:completions} of Part II. In Theorem \ref{teo:circle_split_caract}, we give a characterization by minimal forbidden subgraphs for those split graphs that are circle, and in Theorem \ref{teo:caract_minimal_ifcase} we state and prove a necessary condition for a completion to proper interval graphs to be minimal when the input graph is an interval graph.

\section*{Part I}

Chapters \ref{chapter:partitions} and \ref{chapter:2nested_matrices}, were devoted to build the foundations and necessary tools to prove Theorem \ref{teo:circle_split_caract}.
More precisely, we define $2$-nested matrices and then state and prove a characterization of these matrices by forbidden subconfigurations that allows us to represent and characterize the adjacency matrices of those split graphs studied in Chapter \ref{chapter:split_circle_graphs}.
Some of the results given in Chapter $3$ have been published in \cite{PDGS19}, and the remaining results are being prepared in a manuscript to be submitted for publication.
In Chapter \ref{chapter:split_circle_graphs} we address the problem of characterizing circle graphs when restricted to split graphs. In turn, this chapter is divided into 5 sections: an introduction to the known structural characterizations of circle graphs, and one section for each case of Theorem \ref{teo:circle_split_caract}. This work resulted in a characterization by forbidden induced subgraphs for those split graphs that are also circle. For its part, this result will be shortly submitted for publication.

We leave some possible continuations of this work.

\begin{itemize}
\item We have found a characterization by forbidden induced subgraphs for those split graphs that are also circle. Are all the subgraphs given in Theorem \ref{teo:circle_split_caract} also minimally non-circle?

\item Recall that split graphs are those chordal graphs for which its complement is also a chordal graph, and that the graph $A''_n$ with $n = 3$ depicted in Figure \ref{fig:local_complement_A''3} is a chordal graph that is neither circle nor a split graph. It follows from this example that Theorem \ref{teo:circle_split_caract} does not hold if we consider chordal graphs instead of split graphs, for there are more forbidden subgraphs that are not considered in the given list. However, Theorem \ref{teo:circle_split_caract} is indeed a good first step to characterize circle graphs by forbidden induced subgraphs within the class of chordal graphs, which remains as an open problem.

\item Given that split graphs can be recognized in linear-time: is it possible to recognize a split circle graph in linear-time?

\item Another possible continuation of this work would be studying the characterization of those circle graphs whose complement is also a circle graph.  

\item Characterize Helly circle graphs by forbidden induced subgraphs. The class of Helly circle graphs was characterized by forbidden induced subgraphs within circle graphs in \cite{DGR10}.
Moreover, it would be interesting to find a decomposition analogous as the split decomposition is for circle graphs, this is, such that Helly circle graphs are closed under this decomposition.
\end{itemize}

\section*{Part II}

In Chapter $2$, we give some properties regarding the ordering of the vertices of an interval graph using minimal separators which hold both for interval and proper interval graphs, and we define a partition of the fill edges according to their relationship with the minimal separators of the graph. In the last part of this chapter, given a completion $H$ to proper interval graphs of an interval graph $G$, we state and prove a necessary condition for $H$ to be minimal.

With regard to the minimal completion problem studied in Chapter \ref{chapter:completions} of Part II, we have the following conjectures:

\begin{conj}
	We conjecture that the only if case of Theorem \ref{teo:caract_minimal_ifcase} holds. Furthermore, in that case the complexity of completing minimally to proper interval graphs when the input is an interval graph is polynomial.
\end{conj}

\begin{conj}
The minimum completion to proper interval graphs when the input graph is interval is NP-complete.
\end{conj}

We would like to continue working on these conjectures in order to obtain a stronger result for an article. 

Following a similar line as the one that led to the problem studied in Chapter \ref{chapter:completions}, it remains as an open problem the characterization and complexity of minimum and minimal completions to proper circular-arc graphs, when the input graph is circular-arc.

\chapter*{Appendix}
\addcontentsline{toc}{chapter}{Appendix}

\begin{lema}\label{lema:no_son_circle}
The graphs depicted in Figures~\ref{fig:forb_F_graphs} and~\ref{fig:forb_T_graphs} are non-circle graphs.
\end{lema}

\begin{proof}

Let us first consider an odd $k$-sun with center, where $v_1, \ldots, v_k$ are the vertices of the clique of size $k$, $w_1, \ldots, w_k$ are the \emph{petals} (the vertices of degree $2$) and $x$ is the \emph{center} (the vertex adjacent to every vertex of the clique of size $k$). If $k=3$, we consider the local complement with respect to the center and we obtain $BW_3$. If $k=5$ or $k=7$, we apply local complementation with respect to $x$, $w_1$, $\ldots$ and $w_k$, obtaining $W_5$ and $W_7$, respectively.
If instead $k=9$, we first apply local complementation with respect to $x$, $w_1$, $\ldots$ and $w_k$ and we obtain $W_9$. Once we have this wheel, we apply local complement with respect to $v_1$, $v_2$ and $v_9$ and obtain $W_6$ induced by $\{v_3,\ldots,v_8\}$. If we now consider the local complement with respect to $v_3$, $v_6$ and $x$ (in that order), we find $\overline{C_6}$, which is not a circle graph because it is locally equivalent to $W_5$.
More in general, for every $k \geq 11$ we can obtain a $W_k$ considering the sequence described at the beginning of the paragraph. Once we have a wheel, if we apply local complementation by $v_1$, $v_2$ and $v_k$, then we obtain $W_{k-3}$. We can repeat this until $k-3 < 8$, in which case either $k=5$, $k=6$ or $k=7$ and we reduce this to one of the previous cases.



Let us consider now an even $k$-sun, where $v_1, \ldots, v_k$ are the vertices of the clique of size $k$ and $w_1, \ldots, w_k$ are the \emph{petals} (the vertices of degree $2$), where $w_i$ is adjacent to $v_i$ and $v_{i+1}$. If $k=4$, then we apply local complementation with respect to the sequence $w_1$, $w_2$, $w_3$, $w_4$, $v_1$, $w_4$, $w_1$, $w_3$, $v_3$, $w_2$ and $v_1$, and we obtain $\overline{C_6}$, which is locally equivalent to $W_5$.
If $k=6$ and we apply local complementation to the sequence $w_1$, $w2, \ldots, w_k$, then we obtain $\overline{C_6}$.
Let us consider an even $k\geq 8$, thus $k=2j$ for some $j\geq 4$ and let $l=\frac{k-8}{2}$. We apply local complementation with respect to the sequence $v_1$, $v_{j+1}$, $v_2$, $v_k$, $\ldots$, $v_{2+l}$, $v_{k-l}$ and we find $W_5$ or $W_7$ induced by $\{v_1$, $v_{j+1}$, $v_{j+1-2}$, $v_{j+1-1}$, $v_{j+1+1}$, $v_{j+1+2} \}$ or by $\{v_1$, $v_{j+1}$, $v_{j+1-3}$, $v_{j+1-2}$, $v_{j+1-1}$, $v_{j+1+1}$, $v_{j+1+2}$, $v_{j+1+3} \}$, depending on whether $k \equiv 2 \mod 4$ or $k \equiv 0 \mod 4$, respectively.



Let us consider $M_{II}(k)$ for some even $k \geq 4$, where $v_1, \ldots, v_k$ are the vertices of the complete induced subgraph of size $k$, $w_1, \ldots, w_{k-2}$ are the petals, where $w_i$ is adjacent to $v_i$ and $v_{i+1}$ for each $i=1, \ldots, k-2$, $x_1$ is adjacent to $v_2, \ldots, v_k$ and $x_2$ is adjacent to $v_1, \ldots, v_{k-2}, v_k$.
If $k=4$, then we apply local complementation by the sequence $x_1$, $x_2$, $w_1$ and $w_2$ and we obtain $W_5$ induced by $\{ v_1$, $v_2$, $v_3$, $v_4$, $x_1$, $x_2 \}$.
If instead $k \geq 6$, then we apply local complementation by the sequence $v_k$, $w_1$, $\ldots, w_{k-2}$ and we obtain $W_{k+1}$ induced by $\{ v_1, \ldots, v_{k-1}$, $x_1$, $x_2$, $v_k\}$. The proof follows analogously from here as the one given for the the odd $k$-sun with center case.


Let us consider $F_{1}(k)$ for some odd $k \geq 5$, where $v_1, \ldots, v_{k-1}$ are the vertices of the complete induced subgraph of size $k-1$, $w_1, \ldots, w_{k-2}$ are the petals, where $w_i$ is adjacent to $v_i$ and $v_{i+1}$ for each $i=1, \ldots, k-1$, $x_1$ is adjacent to $v_2, \ldots, v_{k-1}$ and $x_2$ is adjacent to $v_1, \ldots, v_{k-2}$.
If $k=5$, then we apply local complementation by the sequence $v_4$, $x_1$, $w_2$, $w_3$ and $w_1$ and we obtain $W_5$ induced by $\{ v_1$, $v_2$, $v_3$, $v_4$, $x_1$, $x_2 \}$.
If instead $k \geq 7$, then we apply local complementation by the sequence $w_1$, $\ldots, w_{k-2}$ and we obtain $\overline{C_{k+1}}$ induced by $\{ v_1, \ldots, v_{k-1}$, $x_1$, $x_2 \}$. The proof follows analogously as the one given for the the even $k$-sun case.


Let us consider a tent${}\vee{}K_1$ given by $\{ v_1$, $v_2$, $v_3$, $s_{12}$, $s_{23}$, $s_{31}$, $x\}$, where $v_1$, $v_2$ and $v_3$ are all adjacent, $x$ is the universal vertex and $s_{i(i+1)}$ is adjacent to $v_i$ and $v_{i+1}$ (subindexes are modulo $3$).
If we apply local complementation by the sequence $s_{12}$, $s_{23}$, $v_2$ and $x$, then we obtain $\overline{C_6}$, which is locally equivalent to $W_5$.


Let us consider $F_0$ given by $\{ v_1$, $v_2$, $v_3$, $v_4$, $v_5$, $s_{13}$, $s_{24}$, $s_{35} \}$, where $v_1$, $\ldots$, $v_5$ are the vertices in the complete induced subgraph of size $5$ and $s_{i(i+2)}$ is adjacent to $v_i$, $v_{i+1}$ and $v_{i+2}$ for each $i=1,2,3$.
If we apply local complementation by $v_3$, then we obtain an induced subgraph isomorphic to $M_{II}(4)$, which contains an induced subgraph locally equivalent to $W_5$.


Let us consider $M_{II}(k)$ for $k \geq 3$, where $v_1, \ldots, v_{k+1}$ are the vertices of the complete induced subgrpah of size $k+1$, $w_1, \ldots, w_{k-2}$ are the petals, where $w_i$ is adjacent to $v_i$ and $v_{i+1}$ for each $i=1, \ldots, k-2$ and $x$ is adjacent to $v_2, \ldots, v_{k-2}, v_k$.
If $k=3$, then we apply local complementation by $v_2$ and obtain a tent${}\vee{}K_1$.
If instead we consider any even $k \geq 4$ and we apply local complementation by the sequence $v_{k+1}$, $w_1$, $\ldots, w_{k-1}$ and we obtain $W_{k+1}$, thus the proof from now on is analogous as the one given for the the odd $k$-sun with center case.


Let us consider $F_{2}(k)$ for some odd $k \geq 5$, where $v_1, \ldots, v_{k}$ are the vertices of the complete induced subgraph of size $k$, $w_1, \ldots, w_{k-1}$ are the petals, where $w_i$ is adjacent to $v_i$ and $v_{i+1}$ for each $i=1, \ldots, k-1$ and $x$ is adjacent to $v_2, \ldots, v_{k-1}$.
If we apply local complementation by the sequence $w_1$, $\ldots, w_{k-1}$ and we obtain $\overline{C_{k+1}}$ and thus the proof follows analogously from here as the one given for the the even $k$-sun case.


Let us consider $M_V$ given by $\{ v_1$, $v_2$, $v_3$, $v_4$, $v_5$, $w_{1}$, $w_{2}$, $w_{3}$, $w_4 \}$, where $v_1$, $\ldots$, $v_5$ are the vertices in the complete induced subgraph of size $5$, $w_1$ is adjacent to $v_1$ and $v_2$, $w_2$ is adjacent to $v_3$ and $v_4$, $w_3$ is adjacent to $v_1, \ldots, v_4$ and $w_4$ is adjacent to $v_1$, $v_4$ and $v_5$.
If we apply local complementation by the sequence $w_3$, $v_3$, $v_2$, $w_2$ and $w_1$, then we obtain an induced subgraph isomorphic to $BW_3$ induced by $\{ w_1$, $w_3$, $w_2$, $v_4$, $w_4$, $v_1$, $v_5 \}$.


Finally, let us consider $M_{IV}$ given by $\{ v_1$, $v_2$, $v_3$, $v_4$, $v_5$, $v_6$, $w_{1}$, $w_{2}$, $w_{3}$, $w_4 \}$, where $v_1$, $\ldots$, $v_6$ are the vertices in the complete induced subgraph of size $6$, $w_1$ is adjacent to $v_1$ and $v_2$, $w_2$ is adjacent to $v_3$ and $v_4$, $w_3$ is adjacent to $v_5$ and $v_6$ and $w_4$ is adjacent to $v_2$, $v_4$ and $v_6$.
If we apply local complementation by the sequence $v_6$, $w_4$, $v_3$, $v_4$, $v_6$, $w_3$ and $w_1$, then we obtain an induced subgraph isomorphic to $W_5$ induced by $\{ v_6$, $v_5$, $v_3$, $v_1$, $v_2$, $w_2 \}$.


\end{proof}

\end{document}